\documentclass[11pt, a4paper]{article}

\usepackage{authblk}
\usepackage[margin=1in]{geometry}
\usepackage{xcolor}

\usepackage{amssymb}
\usepackage{amsfonts}
\usepackage{amsmath}
\usepackage{amsthm}
\allowdisplaybreaks[1]

\usepackage{enumerate}
\usepackage{enumitem}
\usepackage{float}
\usepackage{indentfirst}
\usepackage{booktabs}
\usepackage{array}
\usepackage{colortbl}
\usepackage{subfigure}
\usepackage{makecell}

\usepackage{listings}
\usepackage{xurl} % Allow linkbreaks inside hyperlinks
\usepackage[pagebackref,plainpages=false,colorlinks,linkcolor=teal,anchorcolor=blue,citecolor=teal]{hyperref}

\hypersetup{breaklinks=true}

\usepackage[english]{babel}
\usepackage[T1]{fontenc}
\AtBeginDocument{%
  \DeclareFontShape{T1}{cmr}{m}{scit}{<->ssub*cmr/m/sc}{}%
}

\usepackage{cleveref}
\crefname{ineq}{inequality}{inequalities}
\creflabelformat{ineq}{#2{\upshape(#1)}#3}
\crefname{fact}{fact}{facts}
\crefname{equation}{equation}{equations}
\crefname{algorithm}{algorithm}{algorithms}
\Crefname{algorithm}{Algorithm}{Algorithms}
\crefname{algocf}{algorithm}{algorithms}
\Crefname{algocf}{Algorithm}{Algorithms}
\crefname{remark}{remark}{remarks}
\crefname{conjecture}{conjecture}{conjectures}
\crefname{problem}{problem}{problems}
\numberwithin{equation}{section}

\usepackage{thmtools}
\declaretheorem[style=plain,numberwithin=section]{theorem}
\declaretheorem[style=plain,numberlike=theorem]{lemma,corollary}
\declaretheorem[style=remark,numberlike=theorem]{remark}
\declaretheorem[style=definition,numberlike=theorem]{definition}
\declaretheorem[style=plain,numberlike=theorem]{proposition}
\declaretheorem[style=plain,numberlike=theorem]{fact}
\declaretheorem[style=definition,numberlike=theorem]{problem,conjecture}

\usepackage{tikz}
\usetikzlibrary{quantikz2}

\usepackage[normalem]{ulem}

\usepackage[algoruled]{algorithm2e}

\usepackage{enumerate}
\usepackage{framed}
\usepackage[title]{appendix}
\usepackage{colortbl}
\usepackage{comment}
\usepackage[modulo]{lineno}
\usepackage{cite}

\usepackage[refpage]{nomencl}

\usepackage{mdframed}
\usepackage[nodisplayskipstretch]{setspace}
\usepackage{tablefootnote}

\allowdisplaybreaks[4]

\newcommand{\parheading}[1]{%
  \par\addvspace{1em}%
  \noindent\emph{#1}\enspace\ignorespaces%
}

\newcommand{\BPP}{\textnormal{\textsf{BPP}}\xspace}
\newcommand{\BQP}{\textnormal{\textsf{BQP}}\xspace}

\newcommand{\QMA}{\textnormal{\textsf{QMA}}\xspace}

\newcommand{\QIPtwo}{\textnormal{\textsf{QIP(2)}}\xspace}

\newcommand{\AM}{\textnormal{\textsf{AM}}\xspace}
\newcommand{\coAM}{\textnormal{\textsf{coAM}}\xspace}

\newcommand{\SZK}{\textnormal{\textsf{SZK}}\xspace}
\newcommand{\QSZK}{\textnormal{\textsf{QSZK}}\xspace}

\newcommand{\BQUSPACE}{\texorpdfstring{\textnormal{\textsf{BQ}\textsubscript{U}\textsf{SPACE}}}\xspace}

\newcommand{\BQSPACE}{\textnormal{\textsf{BQSPACE}}\xspace}

\newcommand{\RQUSPACE}{\texorpdfstring{\textnormal{\textsf{RQ}\textsubscript{U}\textsf{SPACE}}}\xspace}
\newcommand{\coRQUSPACE}{\texorpdfstring{\textnormal{\textsf{coRQ}\textsubscript{U}\textsf{SPACE}}}\xspace}
\newcommand{\Lspace}{\textnormal{\textsf{L}}\xspace}

\newcommand{\BPL}{\textnormal{\textsf{BPL}}\xspace}
\newcommand{\BQL}{\textnormal{\textsf{BQL}}\xspace}
\newcommand{\BQUL}{\texorpdfstring{\textnormal{\textsf{BQ}\textsubscript{U}\textsf{L}}}\xspace}
\newcommand{\RQL}{\textnormal{\textsf{RQL}}\xspace}
\newcommand{\RQUL}{\texorpdfstring{\textnormal{\textsf{RQ}\textsubscript{U}\textsf{L}}}\xspace}
\newcommand{\coRQL}{\textnormal{\textsf{coRQL}}\xspace}
\newcommand{\coRQUL}{\texorpdfstring{\textnormal{\textsf{coRQ}\textsubscript{U}\textsf{L}}}\xspace}

\newcommand{\NC}{\textnormal{\textsf{NC}}\xspace}
\newcommand{\DETs}{\mathsf{DET}^*}

\newcommand{\textoverline}[1]{$\overline{\mbox{#1}}$}

\newcommand{\CertQSDlog}{\textnormal{\textsc{CertQSD}\textsubscript{log}}\xspace}

\newcommand{\coCertQSD}{\texorpdfstring{\textnormal{\textoverline{\textsc{CertQSD}}}}\xspace}
\newcommand{\coCertQHS}{\texorpdfstring{\textnormal{\textoverline{\textsc{CertQHS}}}}\xspace}
\newcommand{\coCertQSDlog}{\texorpdfstring{\textnormal{\textoverline{\textsc{CertQSD}}\textsubscript{log}}}\xspace}
\newcommand{\coCertQHSlog}{\texorpdfstring{\textnormal{\textoverline{\textsc{CertQHS}}\textsubscript{log}}}\xspace}

\newcommand{\GapQSD}{\textnormal{\textsc{GapQSD}}\xspace}
\newcommand{\GapQED}{\textnormal{\textsc{GapQED}}\xspace}
\newcommand{\GapQJS}{\textnormal{\textsc{GapQJS}}\xspace}
\newcommand{\GapQHS}{\textnormal{\textsc{GapQHS}}\xspace}
\newcommand{\GapQSDlog}{\texorpdfstring{\textnormal{\textsc{GapQSD}\textsubscript{log}}}\xspace}
\newcommand{\GapQEDlog}{\texorpdfstring{\textnormal{\textsc{GapQED}\textsubscript{log}}}\xspace}
\newcommand{\GapQJSlog}{\texorpdfstring{\textnormal{\textsc{GapQJS}\textsubscript{log}}}\xspace}
\newcommand{\GapQHSlog}{\texorpdfstring{\textnormal{\textsc{GapQHS}\textsubscript{log}}}\xspace}

\newcommand{\QSD}{\textnormal{\textsc{QSD}}\xspace}
\newcommand{\QED}{\textnormal{\textsc{QED}}\xspace}
\newcommand{\QJSP}{\textnormal{\textsc{QJSP}}\xspace}

\newcommand{\SD}{\textnormal{\textsc{SD}}\xspace}
\newcommand{\ED}{\textnormal{\textsc{ED}}\xspace}
\newcommand{\JS}{\textnormal{\textsc{JS}}\xspace}
\newcommand{\JSP}{\textnormal{\textsc{JSP}}\xspace}

\renewcommand{\bra}[1]{\langle{#1}\rvert}
\renewcommand{\ket}[1]{\lvert{#1}\rangle}

\newcommand{\ketbra}[2]{\ensuremath{\ket{#1}\!\bra{#2}}}
\newcommand{\innerprod}[2]{\langle #1 | #2 \rangle}
\newcommand{\innerprodF}[2]{\langle #1 , #2 \rangle}

\newcommand{\Tr}{\mathrm{Tr}}
\newcommand{\rank}{\mathrm{rank}}
\newcommand{\Out}{\mathrm{Out}}
\newcommand{\dx}{\mathrm{d}x}
\newcommand{\dd}{\mathrm{d}}
\newcommand{\dtheta}{\mathrm{d}\theta}
\DeclareMathOperator\erf{erf}

\newcommand{\td}{\mathrm{T}}

\newcommand{\F}{\mathrm{F}}
\newcommand{\QJS}{\mathrm{QJS}}
\newcommand{\binH}{\mathrm{H_2}}
\newcommand{\HS}{\mathrm{HS}^2}
\renewcommand{\S}{\mathrm{S}}
\renewcommand{\H}{\mathrm{H}}

\newcommand{\Img}{\mathrm{Img}}
\newcommand{\spanset}{\mathrm{span}}
\newcommand{\sign}{\mathrm{sgn}}
\newcommand{\even}{\mathrm{even}}
\newcommand{\odd}{\mathrm{odd}}
\newcommand{\inter}{\mathrm{int}}
\newcommand{\exter}{\mathrm{ext}}
\newcommand{\Eval}{\mathrm{Eval}}
\newcommand{\yes}{{\rm yes}}
\newcommand{\no}{{\rm no}}
\newcommand{\Sqrt}{{\rm Sqrt}}

\newcommand{\Had}{\textnormal{\textsc{H}}\xspace}
\newcommand{\CNOT}{\textnormal{\textsc{CNOT}}\xspace}
\newcommand{\T}{\textnormal{\textsc{T}}\xspace}
\newcommand{\SWAP}{\textnormal{\textsc{SWAP}}\xspace}

\renewcommand{\Pr}[1]{\mathrm{Pr}\!\left[#1\right]}

\newcommand{\binset}{\{0,1\}}

\newcommand{\In}{\mathrm{in}}
\renewcommand{\Out}{\mathrm{out}}
\newcommand{\SV}{\mathrm{(SV)}}
\newcommand{\tpsi}{\tilde{\psi}}

\DeclareMathOperator\poly{poly}
\DeclareMathOperator\polylog{polylog}

\DeclarePairedDelimiter\rbra{\lparen}{\rparen}
\DeclarePairedDelimiter\sbra{\lbrack}{\rbrack}
\DeclarePairedDelimiter\cbra{\{}{\}}
\DeclarePairedDelimiter\abs{\lvert}{\rvert}
\DeclarePairedDelimiter\norm{\lVert}{\rVert}
\DeclarePairedDelimiter\ceil{\lceil}{\rceil}

\newcommand{\rmi}{\mathrm{i}}

\newcommand{\bbC}{\mathbb{C}}
\newcommand{\bbE}{\mathbb{E}}

\newcommand{\bbR}{\mathbb{R}}
\newcommand{\bbN}{\mathbb{N}}
\newcommand{\bfa}{\mathbf{a}}
\newcommand{\bfb}{\mathbf{b}}
\newcommand{\bfc}{\mathbf{c}}
\newcommand{\bfu}{\mathbf{u}}
\newcommand{\bfv}{\mathbf{v}}
\newcommand{\bfy}{\mathbf{y}}

\newcommand{\calA}{\mathcal{A}}
\newcommand{\calH}{\mathcal{H}}
\newcommand{\calI}{\mathcal{I}}
\newcommand{\calL}{\mathcal{L}}

\newcommand{\calT}{\mathcal{T}}

\newcommand{\sfB}{\mathsf{B}}
\newcommand{\sfF}{\mathsf{F}}

\newcommand{\sfO}{\mathsf{O}}

\newcommand{\sfR}{\mathsf{R}}
\newcommand{\sfS}{\mathsf{S}}

\newcommand{\sfY}{\mathsf{Y}}
\newcommand{\sfZ}{\mathsf{Z}}

\begin{document}
% Reduce the line spacing between equation and text
\setlength{\abovedisplayskip}{6pt}
\setlength{\belowdisplayskip}{6pt}

\title{Space-bounded quantum state testing\\
\Large{via space-efficient quantum singular value transformation}}
\author[1]{Fran\c{c}ois Le Gall\thanks{Email: legall@math.nagoya-u.ac.jp}}
\author[2,1]{Yupan Liu\thanks{Email: yupan.liu@epfl.ch}} 
\author[3,1]{Qisheng Wang\thanks{Email: QishengWang1994@gmail.com}}
\affil[1]{Graduate School of Mathematics, Nagoya University}
\affil[2]{School of Computer and Communication Sciences, \'Ecole Polytechnique F\'ed\'erale de Lausanne}
\affil[3]{School of Computer Science, Shanghai Jiao Tong University}

\date{}
\maketitle
\pagenumbering{roman}
\thispagestyle{empty}

\begin{abstract}
Driven by exploring the power of quantum computation with a limited number of qubits, we present a novel complete characterization for space-bounded quantum computation, which encompasses settings with one-sided error (unitary \coRQL{}) and two-sided error (\BQL{}), approached from a quantum (mixed) state testing perspective: 

\begin{itemize}
    \item The \textit{first} family of natural complete problems for unitary $\coRQL$, namely \textit{space-bounded quantum state certification} for trace distance and Hilbert--Schmidt distance; 
    \item A new family of (arguably simpler) natural complete problems for \BQL{}, namely \textit{space-bounded quantum state testing} for trace distance, Hilbert--Schmidt distance, and (von Neumann) entropy difference. 
\end{itemize}

In the space-bounded quantum state testing problem, we consider two logarithmic-qubit quantum circuits (devices) denoted as $Q_0$ and $Q_1$, which prepare quantum states $\rho_0$ and $\rho_1$, respectively, with access to their ``source code''. Our goal is to decide whether $\rho_0$ is $\epsilon_1$-close to or $\epsilon_2$-far from $\rho_1$ with respect to a specified distance-like measure. Interestingly, unlike time-bounded state testing problems, which exhibit computational hardness depending on the chosen distance-like measure (either \QSZK{}-complete or \BQP{}-complete),  our results reveal that the space-bounded state testing problems, considering all three measures, are computationally as easy as preparing quantum states. 

Our results primarily build upon \textit{a space-efficient variant} of the quantum singular value transformation (QSVT) introduced by \hyperlink{cite.GSLW19}{Gilyén, Su, Low, and Wiebe (STOC 2019)}, which is of independent interest. Our technique provides a unified approach for designing space-bounded quantum algorithms. Specifically, we show that implementing QSVT for any bounded polynomial that approximates a piecewise-smooth function incurs only a constant overhead in terms of the space required for (special forms of) the projected unitary encoding. 
\end{abstract}

\newpage
\tableofcontents
\thispagestyle{empty}
\newpage
\pagenumbering{arabic}
%%%%%%%%%%%%%%%%%%%%%%%%%%%%%%%%%%%%%%%%%%%%%%%%%%%%%

\section{Introduction}

In recent years, exciting experimental advancements in quantum computing have been achieved, but concerns about their scalability persist. It thus becomes essential to characterize the computational power of feasible models of quantum computation that operate under restricted resources, such as \textit{time} (i.e., the number of gates in the circuit) and \textit{space} (i.e., the number of qubits on which the circuit acts).
This paper specifically focuses on the latter aspect: What is the computational power of quantum computation with a limited number of qubits? 

Previous studies on complete problems of space-bounded quantum computation~\cite{Wat99,Wat03,vMW12} have primarily focused on well-conditioned versions of standard linear algebraic problems~\cite{TS13,FL18,FR21} and have been limited to the two-sided error scenario. In contrast, we propose a novel family of complete problems that not only characterize the \textit{one-sided error scenario (and extend to the two-sided scenario)} but also arise from a quantum property testing perspective. Our new complete problems are arguably more natural and simpler, driven by recent intriguing challenges of verifying the intended functionality of quantum devices. 

Consider the situation where a quantum device is designed to prepare a quantum (mixed) state $\rho_0$, but a possibly malicious party could provide another quantum device that outputs a different $n$-qubit (mixed) state $\rho_1$, claiming that $\rho_0 \!\approx_{\epsilon}\! \rho_1$. The problem of testing whether $\rho_0$ is $\epsilon_1$-close to or $\epsilon_2$-far from $\rho_1$ with respect to a specified distance-like measure, given the ability to produce copies of $\rho_0$ and $\rho_1$, is known as \textit{quantum state testing}~\cite[Section 4]{MdW16}. 
Quantum state testing (resp., distribution testing) typically involves utilizing sample accesses to quantum states $\rho_0$ and $\rho_1$ (resp., distributions $D_0$ and $D_1$) and determining the number of samples required to test the closeness between quantum states (resp., distributions). 
This problem is a quantum (non-commutative) generalization of classical property testing, which is a fundamental problem in theoretical computer science (see~\cite{Goldreich17}), specifically (tolerant) distribution testing (see~\cite{Canonne20}). Moreover, this problem is an instance of the emerging field of quantum property testing (see~\cite{MdW16}), which aims at designing quantum testers for the properties of quantum objects. 

In this paper, we investigate quantum state testing problems where quantum states $\rho_0$ and $\rho_1$ are preparable by \textit{computationally constrained resources}, specifically state-preparation circuits (viewed as the ``source code'' of devices) that are (\textit{log})\textit{space-bounded}. 
Our main result conveys a conceptual message that testing quantum states prepared in bounded space is (computationally) as \textit{easy} as preparing these states in a space-bounded manner. Consequently, we can introduce the first family of natural $\coRQUL$-complete promise problems since Watrous~\cite{Wat01} introduced unitary \RQL{} and \coRQL{} (known as $\RQUL$ and $\coRQUL$, respectively) in 2001, as well as a new family of natural \BQL{}-complete promise problems. 

Our main technique is a \textit{space-efficient variant} of the quantum singular value transformation (QSVT)~\cite{GSLW19},  distinguishing itself from prior works primarily focused on time-efficient QSVT. As time-efficient QSVT provides a unified framework for designing time-efficient quantum algorithms~\cite{GSLW19,MRTC21}, we believe our work indicates a unified approach to designing space-bounded quantum algorithms, potentially facilitating the discovery of new complete problems for \BQL{} and its one-sided error variants.
Subsequently, we will first state our main results and then provide justifications for the significance of our results from various perspectives.

\subsection{Main results}
We will commence by providing definitions for time- and space-bounded quantum circuits. We say that a quantum circuit $Q$ is \textit{\emph{(}poly\emph{)}time-bounded} if $Q$ is polynomial-size and acts on $\poly(n)$ qubits. Likewise, we say that a quantum circuit $Q$ is \textit{\emph{(}log\emph{)}space-bounded} if $Q$ is polynomial-size and acts on $O(\log{n})$ qubits. It is worthwhile to note that primary complexity classes, e.g., \BQL{}, \coRQUL{}, and \BPL{}, mentioned in this paper correspond to \textit{promise problems}. 

\paragraph{Complete characterizations of quantum logspace from state testing.}
While prior works~\cite{TS13,FL18,FR21} on \BQL{}-complete problems have mainly focused on well-conditioned versions of standard linear algebraic problems (in $\DETs$), our work takes a different perspective by exploring quantum property testing. Specifically, we investigate the problem of \textit{space-bounded quantum state testing}, which aims to test the closeness between two quantum states that are preparable by (log)space-bounded quantum circuits (devices), with access to the corresponding ``source code'' of these devices.

We begin by considering a computational problem that serves as a ``white-box'' space-bounded counterpart of \textit{quantum state certification}~\cite{BOW19}, equivalent to quantum state testing with one-sided error.
Our first main theorem (\Cref{thm:space-bounded-quantum-state-certification-RQL-complete-informal}) demonstrates the \textit{first} family of natural $\coRQUL$-complete problems in the context of space-bounded quantum state certification with respect to the trace distance ($\td$) and the squared Hilbert--Schmidt distance ($\HS$). 
\begin{theorem}[Informal version of \Cref{thm:space-bounded-quantum-state-certification-RQL-complete}]
    \label{thm:space-bounded-quantum-state-certification-RQL-complete-informal}
   The following space-bounded quantum state certification problems are $\coRQUL$-complete: for any $\alpha(n) \geq 1/\poly(n)$, decide whether 
    \begin{enumerate}[label={\upshape(\roman*)}, itemsep=0.33em, topsep=0.33em, parsep=0.33em]
        \item $\coCertQSD_{\log}$\emph{:} $\rho_0=\rho_1$ or $\td(\rho_0,\rho_1) \geq \alpha(n)$;
        \item $\coCertQHS_{\log}$\emph{:} $\rho_0=\rho_1$ or $\HS(\rho_0,\rho_1) \geq \alpha(n)$.
    \end{enumerate}
\end{theorem}

By relaxing the error requirement from one-sided to two-sided, we extend space-bounded quantum state testing to include two further distance-like measures: the quantum entropy difference, denoted by $\S(\rho_0)-\S(\rho_1)$, and the quantum Jensen-Shannon divergence ($\QJS_2$). As a result, we identify a new family of natural \BQL{}-complete problems, presented in our second main theorem:
\begin{theorem}[Informal version of \Cref{thm:space-bounded-quantum-state-testing-BQL-complete}]
    \label{thm:space-bounded-quantum-state-testing-BQL-complete-informal}
    The following space-bounded quantum state testing problems are $\BQL$-complete: for any $\alpha(n)$ and $\beta(n)$ such that $\alpha(n)-\beta(n) \geq 1/\poly(n)$, or for any $g(n) \geq 1/\poly(n)$, decide whether
    \begin{enumerate}[label={\upshape(\roman*)}, itemsep=0.33em, topsep=0.33em, parsep=0.33em]
        \item $\GapQSD_{\log}$\emph{:} $\td(\rho_0,\rho_1) \geq \alpha(n)$ or $\td(\rho_0,\rho_1) \leq \beta(n)$; \label{thmitem:GapQSDlog-BQL-complete}
        \item $\GapQED_{\log}$\emph{:} $\S(\rho_0)-\S(\rho_1) \geq g(n)$ or $\S(\rho_1)-\S(\rho_0) \geq g(n)$; \label{thmitem:GapQEDlog-BQL-complete}
        \item $\GapQJS_{\log}$\emph{:} $\QJS_2(\rho_0,\rho_1) \geq \alpha(n)$ or $\QJS_2(\rho_0,\rho_1) \leq \beta(n)$; \label{thmitem:GapQJSlog-BQL-complete}
        \item $\GapQHS_{\log}$\emph{:} $\HS(\rho_0,\rho_1) \geq \alpha(n)$ or $\HS(\rho_0,\rho_1) \leq \beta(n)$. \label{thmitem:GapQHSlog-BQL-complete}
    \end{enumerate}
\end{theorem}

It is noteworthy that our algorithm for \GapQSDlog{} in \Cref{thm:space-bounded-quantum-state-testing-BQL-complete-informal}\ref{thmitem:GapQSDlog-BQL-complete} exhibits a \textit{polynomial advantage} in space over the best known classical algorithms~\cite{Wat02}. Watrous implicitly showed in~\cite[Proposition 21]{Wat02} that \GapQSDlog{} is contained in the class \NC{}, which corresponds to (classical) poly-logarithmic space.

\paragraph{Space-efficient quantum singular value transformation.} Proving our main theorems mentioned above poses a significant challenge: establishing the containment in the relevant class (\BQL{} or \coRQUL{}), which is also the difficult direction for showing the known family of \BQL{}-complete problems~\cite{TS13,FL18,FR21}.

Proving the containment for the one-sided error scenario is not an effortless task: such a task is not only already relatively complicated for \coCertQHSlog{}, but also additionally requires novel techniques for \coCertQSDlog{}. On the other hand, for two-sided error scenarios, while showing the containment is straightforward for \GapQHSlog{}, it still demands sophisticated techniques for all other problems, such as \GapQSDlog{}, \GapQEDlog{}, and \GapQJSlog{}.

As explained in \Cref{subsec:proof-technique-space-efficient-QSVT}, our primary technical contribution and proof technique involve developing a space-efficient variant of the quantum singular value transformation (QSVT), which constitutes our last main theorem (\Cref{thm:space-efficient-QSVT-informal}). 

\subsection{Background on space-bounded quantum computation}

Watrous~\cite{Wat99,Wat03} initiated research on space-bounded quantum computation and showed that fundamental properties, including closure under complement, hold for $\BQSPACE[s(n)]$ with $s(n)\geq \Omega(\log{n})$. 
Watrous also investigated classical simulations of space-bounded quantum computation (with unbounded error), presenting deterministic simulations in $O(s^2(n))$ space and unbounded-error randomized simulations in $O(s(n))$ space.   
A decade later, van Melkebeek and Watson~\cite{vMW12} provided a simultaneous $\tilde{O}(t(n))$ time and $O(s(n)+\log{t(n)})$ space unbounded-error randomized simulation for a bounded-error quantum algorithm in $t(n)$ time and $s(n)$ space. 
The complexity class corresponding to space-bounded quantum computation with $s(n)=\Theta(\log(n))$ is known as \BQL{}, or $\BQUL{}$ if only \textit{unitary} gates are permitted. 

Significantly, several developments over the past two decades have shown that $\BQL{}$ is well-defined, independent of the following factors in chronological order:
\begin{itemize}
    \item \textbf{The choice of gateset}. The Solovay--Kitaev theorem~\cite{Kitaev97} establishes that most quantum classes are gateset-independent, given that the gateset is closed under adjoint and all entries in gates have reasonable precision. The work of~\cite{vMW12} presented a space-efficient counterpart of the Solovay--Kitaev theorem, implying that \BQL{} is also \textit{gateset-independent}. 
    \item \textbf{Error reduction}. Repeating $\BQUL$ sequentially necessitates reusing the workspace, making it unclear how to reduce errors for $\BQUL{}$ as intermediate measurements are not allowed. To address this issue, the work of~\cite{FKLMN16} adapted the witness-preserving error reduction for \QMA{}~\cite{MW05} with several other ideas to the space-efficient setting.
    \item \textbf{Intermediate measurements}. In the space-bounded scenario, the principle of deferred measurement is not applicable since this approach leads to an exponential increase in space complexity. Initially, \BQL{} appeared to be seemingly more powerful than $\BQUL{}$ since we cannot directly demonstrate that $\BPL \subseteq \BQUL$. Recently, Fefferman and Remscrim~\cite{FR21} (as well as~\cite{GRZ21,GR22}) proved the equivalence between \BQL{} and $\BQUL{}$, indicating a space-efficient approach to eliminating intermediate measurements. 
\end{itemize}

\paragraph{\BQL{}-complete problems.}
Identifying natural complete problems for the class \BQL{} (or $\BQUL$) is a crucial and intriguing question. 
Ta-Shma~\cite{TS13} proposed the first candidate \BQL{}-complete problem, building upon the work of Harrow, Hassidim, and Lloyd~\cite{HHL09} which established a \BQP{}-complete problem for inverting a (polynomial-size) well-conditioned matrix. Specifically, Ta-Shma showed that inverting a well-conditioned matrix with polynomial precision is in \BQL{}. Similarly, computing eigenvalues of an Hermitian matrix is also in \BQL{}. These algorithms offer a quadratic space advantage over the best-known classical algorithms that saturate the classical simulation bound~\cite{Wat99,Wat03,vMW12}.
Fefferman and Lin~\cite{FL18} later improved upon this result to obtain the first natural $\BQUL$-complete problem by ingeniously utilizing amplitude estimation to avoid intermediate measurements.

More recently, Fefferman and Remscrim~\cite{FR21} further extended this natural $\BQUL$-complete problem (or \BQL{}-complete, equivalently) to a \textit{family} of natural \BQL{}-complete problems. They showed that a well-conditioned version of standard $\DETs$-complete problems is \BQL{}-complete, where $\DETs$ denotes the class of problems that are $\NC^1$ (Turing) reducible to \textsc{intDET}, including well-conditioned integer determinant (\textsc{DET}), well-conditioned matrix powering (\textsc{MATPOW}), and well-conditioned iterative matrix product (\textsc{ITMATPROD}), among others.

\paragraph{$\RQUL{}$- and $\coRQUL{}$-complete problems.}
Watrous~\cite{Wat01} introduced the one-sided error counterpart of $\BQUL$, namely $\RQUL$ and $\coRQUL$, and developed error reduction techniques. Moreover, Watrous proved that the undirected graph connectivity problem (\textsc{USTCON}) is in $\RQUL \cap \coRQUL$ whereas Reingold~\cite{Reingold08} demonstrated that \textsc{USTCON} is in \Lspace{} several years later. 
It is noteworthy that the question of whether intermediate measurements offer computational advantages in one-sided error scenarios, specifically \RQUL{} vs.~\RQL{} and \coRQUL{} vs.~\coRQL{}, remains open.
Recently, Fefferman and Remscrim~\cite{FR21} proposed a ``verification'' version of the well-conditioned iterative matrix product problem (\textsc{vITMATPROD}) as a \textit{candidate} \coRQL{}-complete problem. However, although this problem is known to be \coRQL{}-hard, its containment remains \textit{unresolved}. 
Specifically, \textsc{vITMATPROD} requires to decide whether a single entry in the product of polynomially many well-conditioned matrices is equal to zero. 

\subsection{Time-bounded and space-bounded distribution and state testing}

We summarize prior works and our main results for time-bounded\footnote{The problem of \textit{time-bounded distribution (resp., state) testing} aims to test the closeness between two distributions (resp., states) that are preparable by (poly)time-bounded circuits (devices), with access to the corresponding ``source code'' of these devices.} and space-bounded distribution and state testing with respect to $\ell_1$ norm, entropy difference, and $\ell_2$ norm in \Cref{table:space-bounded-quantum-state-testing}. 

Interestingly, the sample complexity of testing the closeness of quantum states (resp., distributions) depends on the choice of distance-like measures,\footnote{It is noteworthy that the quantum entropy difference is not a distance.} including the one-sided error counterpart known as \textit{quantum state certification}~\cite{BOW19}. In particular, for distance-like measures such as the $\ell_1$ norm, called total variation distance in the case of distributions~\cite{CDVV14} and trace distance in the case of states~\cite{BOW19}, as well as classical entropy difference~\cite{JVHW15,WY16} and its quantum analog~\cite{AISW20,OW21}, the sample complexity of distribution and state testing is polynomial in the dimension $N$.  However, for distance-like measures such as the $\ell_2$ norm, called Euclidean distance in the case of distributions~\cite{CDVV14} and Hilbert--Schmidt distance in the case of states~\cite{BOW19}, the sample complexity is \textit{independent} of dimension $N$.

\begin{table}[!htp]
\centering
\begin{tabular}{cccc}
    \toprule
    & $\ell_1$ norm & $\ell_2$ norm & Entropy\\
    \midrule
    \makecell{Classical\\Time-bounded} 
    & \makecell{\SZK{}-complete\footref{footnote:GapSD-upper-bound}\\ \footnotesize{\cite{SV97,GSV98}}} 
    & \makecell{\BPP{}-complete\\ \footnotesize{Folklore}} 
    & \makecell{\SZK{}-complete\\ \footnotesize{\cite{GV99,GSV98}}}\\
    \midrule
    \makecell{Quantum\\Time-bounded} & 
    \makecell{\QSZK{}-complete\footref{footnote:GapQSD-upper-bound} \\ \footnotesize{\cite{Wat02,Wat09}}} &
    \makecell{\BQP{}-complete \\ \footnotesize{\cite{BCWdW01,ARS+21}}} &
    \makecell{\QSZK{}-complete \\ \footnotesize{\cite{BASTS10,Wat09}}}\\
    \midrule
    \makecell{Quantum\\Space-bounded} &
    \makecell{\BQL{}-complete\\ \footnotesize{\Cref{thm:space-bounded-quantum-state-testing-BQL-complete-informal}\ref{thmitem:GapQSDlog-BQL-complete}}} & 
    \makecell{\BQL{}-complete\\ \footnotesize{\cite{BCWdW01} and \Cref{thm:space-bounded-quantum-state-testing-BQL-complete-informal}\ref{thmitem:GapQHSlog-BQL-complete}}} & 
    \makecell{\BQL{}-complete\\ \footnotesize{\Cref{thm:space-bounded-quantum-state-testing-BQL-complete-informal}\ref{thmitem:GapQEDlog-BQL-complete}}} \\
    \bottomrule
\end{tabular}
\caption{Time- and space-bounded distribution or state testing.}
\label{table:space-bounded-quantum-state-testing}
\end{table}

As depicted in \Cref{table:space-bounded-quantum-state-testing}, this phenomenon that the required sample complexity for distribution and state testing, with polynomial precision and exponential dimension, depends on the choice of distance-like measure has reflections on time-bounded quantum state testing:
\begin{itemize}
    \item For $\ell_1$ norm and entropy difference, the time-bounded scenario is \textit{seemingly much harder than} preparing states or distributions since $\QSZK \subseteq \BQP$ and $\SZK \subseteq \BPP$ are unlikely.     
    \item For $\ell_2$ norm, the time-bounded scenario is \textit{as easy as} preparing states or distributions. 
\end{itemize}

However, interestingly, a similar phenomenon \textit{does not appear} for space-bounded quantum state testing. Although no direct classical counterpart has been investigated before in a complexity-theoretic fashion, namely space-bounded distribution testing, there is another closely related model (a version of streaming distribution testing) that does not demonstrate an analogous phenomenon either, as we will discuss in \Cref{subsec:space-bounded-testing}. 

\subsubsection{Time-bounded distribution and state testing}
\label{subsec:time-bounded-testing}

We review prior works on time-bounded state (resp., distribution) testing, with a particular focus on testing the closeness between states (resp., distributions) that are preparable by (poly)time-bounded quantum (resp., classical) circuits (device), with access to the ``source code'' of corresponding devices. For time-bounded distribution testing, we also recommend a brief survey~\cite{GV11} by Goldreich and Vadhan. 

\paragraph{$\ell_1$ norm scenarios.} Sahai and Vadhan~\cite{SV97} initiated the study of the time-bounded distribution testing problem, where distributions $D_0$ and $D_1$ are \textit{efficiently samplable}, and the distance-like measure is the total variation distance. Their work named this problem \textsc{Statistical Difference} (\SD{}). In particular, the promise problem $(\alpha,\beta)$-\SD{} asks whether $D_0$ is $\alpha$-far from or $\beta$-close to $D_1$ with respect to $\|D_0-D_1\|_{\rm TV}$. 
Although sampling from the distribution is in \BPP{},\footnote{Rigorously speaking, as an instance in \SD{}, sample-generating circuits are not necessarily (poly)time-uniform.} testing the closeness between these distributions is \SZK{}-complete~\cite{SV97,GSV98}, where \SZK{} is the class of promise problems possessing statistical zero-knowledge proofs. 
It is noteworthy that the \SZK{} containment of $(\alpha,\beta)$-\SD{} for any $\alpha(n)-\beta(n) \geq 1/\poly(n)$ is currently unknown.\footnote{The works of~\cite{SV97,GSV98} demonstrated that $(\alpha,\beta)$-\SD{} is in \SZK{} for any constant $\alpha^2-\beta > 0$. The same technique works for the parameter regime $\alpha^2(n)-\beta(n) \geq 1/O(\log{n})$. However, further improvement of the parameter regime requires new ideas, as clarified in~\cite{Goldreich19}.
Recently, the work of~\cite{BDRV19} improved the parameter regime to $\alpha^2(n)-\beta(n) \geq 1/\poly(n)$ by utilizing a series of tailor-made reductions. Currently, we only know that $(\alpha,\beta)$-\SD{} for $\alpha(n)-\beta(n) \geq 1/\poly(n)$ is also in $\AM\cap\coAM$~\cite{BL13}. \label{footnote:GapSD-upper-bound}}
In addition, we note that \SZK{} is contained in $\AM \cap \coAM$~\cite{Fortnow87,AH91}.

Following the pioneering work~\cite{SV97}, Watrous~\cite{Wat02} introduced the time-bounded quantum state testing problem, where two quantum states $\rho_0$ and $\rho_1$ that are preparable by time-bounded quantum circuits $Q_0$ and $Q_1$, respectively, as well as the distance-like measure is the trace distance. This problem is known as the \textsc{Quantum State Distinguishability} (\QSD{}), specifically, $(\alpha,\beta)$-\QSD{} asks whether $\rho_0$ is $\alpha$-far from or $\beta$-close to $\rho_1$ with respect to  $\td(\rho_0,\rho_1)$. Analogous to its classical counterpart, \QSD{} is \QSZK{}-complete~\cite{Wat02,Wat09}, whereas the \QSZK{} containment for any $\alpha(n) - \beta(n) \geq 1/\poly(n)$ remains an open question.\footnote{Like \SD{} and \SZK{}, the techniques in~\cite{Wat02,Wat09} show that $(\alpha,\beta)$-\QSD{} is in \QSZK{} for $\alpha^2(n)-\beta(n)\geq 1/O(\log{n})$, and the same limitation also applies to the quantum settings. A recent result~\cite{Liu23} following the line of work of~\cite{BDRV19} improved the parameter regime to $\alpha^2(n)-\sqrt{2\ln{2}}\beta(n) \geq 1/\poly(n)$, but the differences between classical and quantum distances make it challenging to push the bound further.  \label{footnote:GapQSD-upper-bound}} 

\begin{sloppypar}
\paragraph{Entropy difference scenarios.} Beyond the $\ell_1$ norm, another distance-like measure commonly considered in time-bounded quantum state (or distribution) testing is the (quantum) entropy difference, which also corresponds to the (quantum) Jensen-Shannon divergence. The promise problem \textsc{Entropy Difference} (\ED{}), first introduced by Goldreich and Vadhan~\cite{GV99} following the work of~\cite{SV97}, asks whether efficiently samplable distributions $D_0$ and $D_1$ satisfy $\H(D_0)-\H(D_1) \geq g$ or $\H(D_1)-\H(D_0) \geq g$ for $g=1$. They proved that \ED{} is \SZK{}-complete. 
Ben-Aroya, Schwartz, and Ta-Shma~\cite{BASTS10} further investigated the promise problem \textsc{Quantum Entropy Difference} (\QED{}), which asks whether $\S(\rho_0)-\S(\rho_1) \geq g$ or $\S(\rho_1)-\S(\rho_0) \geq g$, for efficiently preparable quantum states $\rho_0$ and $\rho_1$ and $g=1/2$. They showed that \QED{} is \QSZK{}-complete. 
Moreover, the \SZK{} (resp., \QSZK{}) containment for \ED{} (resp., \QED{}) automatically holds for any $g(n) \geq 1/\poly(n)$. 
\end{sloppypar}

Furthermore, Berman, Degwekar, Rothblum, and Vasudevan~\cite{BDRV19} demonstrated that the Jensen-Shannon divergence problem (\JSP{}), asking whether $\JS(D_0,D_1) \geq \alpha$ or $\JS(D_0,D_1) \leq \beta$ for efficiently samplable distributions $D_0$ and $D_1$, is \SZK{}-complete. Their work accomplished this result by reducing the problem to \ED{}, and this containment applies to  $\alpha(n)-\beta(n) \geq 1/\poly(n)$. Recently,  Liu~\cite{Liu23} showed a quantum counterpart, referred to as the \textsc{Quantum Jensen-Shannon Divergence Problem} (\QJSP{}), is \QSZK{}-complete. 
Notably, the quantum Jensen-Shannon divergence is a special instance of the Holevo $\chi$ quantity~\cite{Holevo73JS}.\footnote{The quantum Jensen-Shannon divergence coincides with the Holevo $\chi$ quantity on size-$2$ ensembles with a uniform distribution, which arises in the Holevo bound~\cite{Holevo73JS}. See \cite[Theorem 12.1]{NC10}.} 

\paragraph{$\ell_2$ norm scenarios.} 
For the quantum setting, it is straightforward that applying the SWAP test~\cite{BCWdW01} to efficiently preparable quantum states $\rho_0$ and $\rho_1$ can lead to a \BQP{} containment, in particular, additive-error estimations of $\Tr(\rho_0^2)$, $\Tr(\rho^2_1)$, and $\Tr(\rho_0\rho_1)$ with polynomial precision. 
Recently, the work of~\cite{ARS+21} observed that time-bounded quantum state testing with respect to the squared Hilbert--Schmidt distance is \BQP{}-complete. 
For the classical setting, namely the squared Euclidean distance, the \BPP{}-completeness is relatively effortless.\footnote{Specifically, we achieve \BPP{} containment by following the approach in \cite[Theorem 7.1]{BCHTV19}. On the other hand, the \BPP{} hardness owes to the fact that the squared Euclidean distance between the distribution $(p_{\rm acc}, 1-p_{\rm acc})$ from the output bit of any \BPP{} algorithm and the distribution $(1,0)$ is $(1-p_{\rm acc})^2$.}

\subsubsection{Space-bounded distribution and state testing} 
\label{subsec:space-bounded-testing}

To the best of our knowledge, no prior work has specifically focused on space-bounded distribution testing from a complexity-theoretic perspective. Instead, we will review prior works that are (closely) related to this computational problem. Afterward, we will delve into space-bounded quantum state testing, which constitutes the main contribution of our work. 

\paragraph{Space-bounded distribution testing and related works.}
We focus on a computational problem involving two $\poly(n)$-size classical circuits $C_0$ and $C_1$, which generate samples from the distributions $D_0$ and $D_1$ respectively. Each circuit contains a read-once polynomial-length random-coins tape.\footnote{It is noteworthy that random coins are provided as \textit{input} to classical circuits $C_0$ and $C_1$ for generating samples from the corresponding distributions in the time-bounded scenario, such as \SD{} and \ED{}.} The input length and output length of the circuits are $O(\log{n})$. The task is to decide whether $D_0$ is $\alpha$-far from or $\beta$-close to $D_1$ with respect to some distance-like measure. Additionally, we can easily observe that space-bounded distribution testing with respect to the squared Euclidean distance ($\ell_2$ norm) is \BPL{}-complete, much like its time-bounded counterpart.

Several models related to space-bounded distribution testing have been investigated previously. Earlier streaming-algorithmic works~\cite{FKSV02,GMV06} utilize \textit{entries} of the distribution as the data stream, with entries given in different orders for different models. On the other hand, a later work~\cite{CLM10} considered a data stream consisting of a sequence of i.i.d.~samples drawn from distributions and studied low-space streaming algorithms for distribution testing. 

Regarding (Shannon) entropy estimation, previous streaming algorithms considered worst-case ordered samples drawn from $N$-dimensional distributions and required $\polylog(N/\epsilon)$ space, where $\epsilon$ is the additive error. Recently, Acharya, Bhadane, Indyk, and Sun~\cite{ABIS19} addressed the entropy estimation problem with i.i.d.~samples drawn from distributions as the data stream and demonstrated the first $O(\log(N/\epsilon))$ space streaming algorithm. The sample complexity, viewed as the time complexity, was subsequently improved in~\cite{AMNW22}. 

However, for the total variation distance ($\ell_1$ norm), previous works focused on the trade-off between the sample complexity and the space complexity (memory constraints), achieving only a nearly-log-squared space streaming algorithm~\cite{DGKR19}. 

Notably, the main differences between the computational and streaming settings lie in how we access the sampling devices.\footnote{Of course, not all distributions can be described as a polynomial-size circuit (i.e., a succinct description).} In the computational problem, we have access to the ``source code'' of the devices and can potentially use them for purposes like ``reverse engineering''. Conversely, the streaming setting utilizes the sampling devices in a ``black-box'' manner, obtaining i.i.d.~samples. As a result, a logspace streaming algorithm will result in a \BPL{} containment.\footnote{In particular, the sample-generating circuits $C_0$ and $C_1$ in space-bounded distribution testing can produce the i.i.d.~samples in the data stream.}

\paragraph{Space-bounded quantum state testing.}
Among the prior works on streaming distribution testing, particularly entropy estimation, the key takeaway is that the space complexity of the corresponding computational problem is $O(\log(N/\epsilon))$. This observation leads to a conjecture that the computational hardness of space-bounded distribution and state testing is \textit{independent} of the choice of commonplace distance-like measures. Our work, in turn, provides a positive answer for space-bounded quantum state testing. 

Space-bounded state testing with respect to the squared Hilbert--Schmidt distance ($\ell_2$ norm) is \BQL{}-complete, as established in \Cref{thm:space-bounded-quantum-state-testing-BQL-complete-informal}\ref{thmitem:GapQHSlog-BQL-complete}. In particular, the \BQL{} containment follows from the SWAP test~\cite{BCWdW01}, analogous to the time-bounded scenario. Furthermore, establishing \BQL{} hardness, as well as \coRQUL{}-hardness for state certification, is not challenging (see \Cref{lemma:GapQHSlog-BQLhard}).

Regarding space-bounded state testing with respect to the trace distance ($\ell_1$ norm), we note that \cite[Proposition 21]{Wat02} implicitly established an \NC{} containment. The \BQL{}-hardness, as well as \coRQUL{}-hardness for state certification, is adapted from~\cite{ARS+21}. Similarly, we derive the \BQL{}-hardness for space-bounded state testing with respect to the quantum Jensen-Shannon divergence and the quantum entropy difference, building on the previous work~\cite{Liu23}. 

Finally, we devote the remainder of this section to our main technique (\Cref{thm:space-efficient-QSVT-informal}), and consequently, we present \BQL{} (resp., \coRQUL{}) containment for state testing (resp., certification) problems for other distance-like measures beyond the squared Hilbert--Schmidt distance.

\subsection{Proof technique: Space-efficient quantum singular value transformation}
\label{subsec:proof-technique-space-efficient-QSVT}

The quantum singular value transformation (QSVT)~\cite{GSLW19} is a powerful and efficient framework for manipulating the singular values $\{\sigma_i\}_i$ of a linear operator $A$, using a corresponding projected unitary encoding $U$ of $A=\tilde{\Pi} U \Pi$ for projections $\tilde{\Pi}$ and $\Pi$.\footnote{Regardless of QSVT, it is noteworthy that the concept of block-encoding, specifically a unitary dilation $U$ of a contraction $A$ (see \Cref{footnote:unitary-dilation}), is already used in quantum logspace for powering contraction matrices~\cite{GRZ21}.} The singular value decomposition is $A=\sum_i \sigma_i \ket{\tpsi_i}\bra{\psi_i}$ where $\ket{\tpsi_i}$ and $\ket{\psi_i}$ are left and right singular vectors, respectively. QSVT has numerous applications in  quantum algorithm design, and is even considered a grand unification of quantum algorithms~\cite{MRTC21}. 
To implement the transformation $f^{\SV}(A)=f^{\SV}(\tilde{\Pi} U \Pi)$, we require a degree-$d$ polynomial $P_d$ that satisfies two conditions. Firstly, $P_d$ well-approximates $f$ on the interval of interest $\calI$, with $\max_{x \in \calI\setminus \calI_{\delta}}|P_d(x)-f(x)| \leq \epsilon$, where $\calI_{\delta} \subseteq \calI\subseteq [-1,1]$ and typically $\calI_{\delta}\coloneqq(-\delta,\delta)$. Secondly, $P_d$ is bounded with $\max_{x \in [-1,1]}|P_d(x)| \leq 1$. The degree of $P_d$ depends on the precision parameters $\delta$ and $\epsilon$, with $d=O(\delta^{-1} \log{\epsilon^{-1}})$, and all coefficients of $P_d$ can be computed efficiently.

According to~\cite{GSLW19}, we can use alternating phase modulation to implement $P_d^{\SV}(\tilde{\Pi} U \Pi)$,\footnote{This procedure is a generalization of quantum signal processing, as explained in \cite[Section II.A]{MRTC21}.} which requires a sequence of rotation angles $\Phi \in \bbR^d$. 
For instance, consider $P_d(x)=T_d(x)$ where $T_d(x)$ is the $d$-th Chebyshev polynomial (of the first kind), then we know that $\phi_1 = (1-d)\pi/2$ and $\phi_j=\pi/2$ for all $j\in\{2,3,\cdots,d\}$. 
QSVT techniques, including the pre-processing and quantum circuit implementation, are generally \textit{time-efficient}. Additionally, the quantum circuit implementation of QSVT is already \textit{space-efficient} because implementing QSVT with a degree-$d$ bounded polynomial for any $s(n)$-qubit projected unitary encoding requires $O(s(n))$ qubits, where $s(n)\geq \Omega(\log{n})$.
However,  the pre-processing in the QSVT techniques is typically not space-efficient. 
Indeed, prior works on the pre-processing for QSVT, specifically angle-finding algorithms in~\cite{Haah19,CDG+20,DMWL21}, which have time complexity polynomially dependent on the degree $d$, do not consider the space-efficiency. 
Therefore, the use of previous angle-finding algorithms may lead to an \textit{exponential} increase in space complexity.
This raises a fundamental question on making the pre-processing space-efficient as well: 
\begin{problem}[Space-efficient QSVT]
    \label{prob:space-efficient-QSVT}
    Can we implement a degree-$d$ QSVT for any $s(n)$-qubit projected unitary encoding with $d \leq 2^{O(s(n))}$, using only $O(s(n))$ space in both the pre-processing and quantum circuit implementation?
\end{problem}

\paragraph{QSVT via averaged Chebyshev truncation.} 
A space-efficient QSVT associated with Chebyshev polynomials is implicitly shown in~\cite{GSLW19}, as the angles for any Chebyshev polynomial $T_k(x)$ are explicitly known. This insight sheds light on \Cref{prob:space-efficient-QSVT} and suggests an alternative pre-processing approach for QSVT: Instead of finding rotation angles, it seems suffice to find projection coefficients of Chebyshev polynomials. 

Recently, Metger and Yuen~\cite{MY23} realized this approach and constructed bounded polynomial approximations of the sign and \textit{shifted} square-root functions with exponential precision in polynomial space by utilizing Chebyshev truncation, which offers a partial solution to \Cref{prob:space-efficient-QSVT}.\footnote{To clarify, we can see from~\cite{MY23} that directly adapting their construction shows that implementing QSVT for any $s(n)$-qubit block-encoding with $O(s(n))$-bit precision requires $\poly(s(n))$ classical and quantum space for any $s(n) \geq \Omega(\log{n})$. 
However, \Cref{prob:space-efficient-QSVT} (space-efficient QSVT) seeks to reduce the dependence of $s(n)$ in the space complexity from \textit{polynomial} to \textit{linear}. \label{footnote:QSVT-log-vs-polylog}} 
The key ingredient behind their approach is the degree-$d$ Chebyshev truncation $\tilde{P}_d(x)=\frac{c_0}{2} + \sum_{k=1}^d c_k T_k$ where $T_k$ is the $k$-th Chebyshev polynomial (of the first kind) and $c_k\coloneqq\frac{2}{\pi}\int_{-1}^1 \frac{f(x)T_k(x)}{\sqrt{1-x^2}} \dx$. This provides a \textit{nearly best} uniform approximation compared to the best degree-$d$ polynomial approximation with error $\varepsilon_d(f)$ for the function $f \colon [-1,1] \rightarrow \bbR$. In particular, $\tilde{P}_d$ satisfies $\max_{x\in[-1,1]} |\tilde{P}_d(x) - f(x)| \leq O(\log{d}) \varepsilon_d(f)$. 

Our construction achieves an error bound \textit{independent} of $d$ via a carefully chosen \textit{average} of the Chebyshev truncation, known as the \textit{de La Vall\'ee Poussin partial sum}, $\hat{P}_{d'}(x) = \frac{1}{d} \sum_{l=d}^{d'} \tilde{P}_l(x) = \frac{\hat{c}_0}{2} + \sum_{k=1}^{d'} \hat{c}_k T_k(x)$, with a slightly larger degree $d'=2d-1$. The degree-$d$ averaged Chebyshev truncation $\hat{P}_{d'}$ satisfies $\max_{x\in[-1,1]} |\hat{P}_{d'}(x)-f(x)| \leq 4 \varepsilon_d(f)$. 

Once we have a space-efficient polynomial approximation for the function $f$ (pre-processing), we can establish a space-efficient QSVT associated with $f$ for \textit{bitstring indexed encodings} that additionally require projections $\tilde{\Pi}$ and $\Pi$ spanning the corresponding subset of $\{\ket{0},\ket{1}\}^{\otimes s}$,\footnote{To ensure that $\tilde{\Pi} U \Pi$ admits a matrix representation, we require the basis of projections $\tilde{\Pi}$ and $\Pi$ to have a well-defined order, leading us to focus exclusively on bitstring indexed encoding. Additionally, for simplicity, we assume no ancillary qubits are used here, and refer to \Cref{def:bitstring-indexed-encoding} for a formal definition.} as stated in \Cref{thm:space-efficient-QSVT-informal}: With the space-efficient QSVT associated with Chebyshev polynomials $T_k(x)$, it suffices to implement the averaged Chebyshev truncation polynomial by LCU techniques~\cite{BCC+15} and to renormalize the bitstring indexed encoding by robust oblivious amplitude amplification (if necessary and applicable). 

A refined analysis indicates that applying an averaged Chebyshev truncation to a bitstring indexed encoding for any $d' \leq 2^{O(s(n))}$ and $\epsilon \geq 2^{-O(s(n))}$ requires $O(s(n))$ qubits and deterministic $O(s(n))$ space, provided that an evaluation oracle $\Eval_{P_d}$ estimates coefficients $\{\hat{c}_k\}_{k=0}^{d'}$ of the averaged Chebyshev truncation with $O(\log(\epsilon^2/d))$ precision. However, our approach causes a \textit{quadratic} dependence of the degree $d$ in the query complexity to $U$. 

\begin{theorem}[Space-efficient QSVT, informal version of \Cref{thm:space-efficient-QSVT}]
    \label{thm:space-efficient-QSVT-informal}
    Let $f\colon\bbR \rightarrow \bbR$ be a continuous function bounded on $\calI \subseteq [-1,1]$. If there exists a degree-$d$ polynomial $P^*_d$ that approximates $h\colon[-1,1] \rightarrow \bbR$, where $h$ approximates $f$ only on $\calI$ with additive error at most $\epsilon$, such that $\max_{x \in [-1,1]} |h(x)-P^*_d(x)| \leq \epsilon$, then the degree-$d$ averaged Chebyshev truncation yields another degree-$d'$ polynomial $P_{d'}$, with $d'=2d-1$, satisfying the following conditions:
    \[\max_{x \in \calI} |f(x)-P_{d'}(x)| \leq O(\epsilon) \quad \text{and} \quad \max_{x \in [-1,1]} |P_{d'}(x)| \leq 1.\]
    Furthermore, we have an algorithm $\calA_f$ that computes any coefficient $\{\hat{c}_k\}_{k=0}^{d'}$ of the averaged Chebyshev truncation polynomial $P_{d'}$ space-efficiently. The algorithm is deterministic for continuously bounded $f$, and bounded-error randomized for piecewise-smooth $f$. Additionally, for any $s(n)$-qubit bitstring indexed encoding $U$ of $A=\tilde{\Pi} U \Pi$ with $d'\leq 2^{O(s(n))}$, we can implement the quantum singular value transformation $P_{d'}^{\SV}(A)$ using $O(d^2\|\hat{\bfc}\|_1)$ queries\footnote{The dependence of $\|\hat{\bfc}\|_1$ arises from renormalizing the bitstring indexed encoding via amplitude amplification.} to $U$ with $O(s(n))$ qubits. 
    It is noteworthy that $\|\hat{\bfc}\|_1$ is bounded by $O(\log{d})$ in general, and we can further improve to a constant norm bound for twice continuously differentiable functions.    
\end{theorem}

Our techniques in \Cref{thm:space-efficient-QSVT-informal} offer three advantages over the techniques proposed by~\cite{MY23}. Firstly, our techniques can handle any \textit{piecewise-smooth function}, such as the (normalized) logarithmic function $\ln(1/x)$, the multiplicative inverse function $1/x$, and the square-root function $\sqrt{x}$;\footnote{Our technique can imply a better norm bound $\|\hat{\bfc}\| \leq O(1)$. See \Cref{remark:square-function} for the details.} whereas the techniques from~\cite{MY23} are restricted to continuously bounded functions whose second derivative of the integrand in $\{\hat{c}_k\}_{k=1}^{d'}$ is at most $\poly(d)$ on the interval $\calI=[-1,1]$, such as the sign function and the \textit{shifted} square-root function $\sqrt{(x+1)/2}$.\footnote{Specifically, the second derivative $|f''(x)|$ of the shifted square-root function $f(x)\coloneqq\sqrt{(x+1)/2}$ is unbounded at $x=-1$. Nevertheless, we can circumvent this point by instead considering $g_{\delta}(x)=\sqrt{(1-\delta)(x+1)/2+\delta}$ with the second derivative $|g''_{\delta}(-1)|=O(\delta^{-3/2})$, as shown in \cite[Lemma 2.11]{MY23}.} Secondly, our techniques are \textit{constant overhead} in terms of the space complexity of the bitstring indexed encoding $U$, while the techniques from~\cite{MY23} are only \textit{poly-logarithmic overhead}. Thirdly, our techniques have an error bound independent of $d$, unlike the $\log{d}$ factor in~\cite{MY23}, simplifying parameter trade-offs for applying the space-efficient QSVT to concrete problems. 

In addition, it is noteworthy that applying the space-efficient QSVT with the sign function will imply a unified approach to error reduction for the classes \BQUL{}, \coRQUL{}, and \RQUL{}. 

\paragraph{Computing the coefficients.}
We will implement the evaluation oracle $\Eval_{P_d}$ to prove \Cref{thm:space-efficient-QSVT-informal}. To estimate the coefficients $\{\hat{c}_k\}_{k=0}^{d'}$ in the averaged Chebyshev truncation for any function $f$ that is bounded on the interval $\calI=[-1,1]$, we can use standard numerical integral techniques,\footnote{We remark that using a more efficient numerical integral technique, such as the exponentially convergent trapezoidal rule, may improve the required space complexity for computing coefficients by a constant factor.} given that the integrand's second derivative in $\{\hat{c}_k\}_{k=0}^{d'}$ is bounded by $\poly(d)$.

However, implementing the evaluation oracle for piecewise-smooth functions $f$ on an interval $\calI \subsetneq [-1,1]$ is relatively complicated. We cannot simply apply averaged Chebyshev truncation to $f$. Instead, we consider a low-degree Fourier approximation $g$ resulting from implementing smooth functions to Hamiltonians~\cite[Appendix B]{vAGGdW17}. We then make the error vanish outside $\calI$ by multiplying with a Gaussian error function, resulting in $h$ which approximates $f$ \textit{only} on $\calI$. Therefore, we can apply averaged Chebyshev truncation and our algorithm for bounded functions to $h$ through a somewhat complicated calculation. 

Finally, we need to compute the coefficients of the low-degree Fourier approximation $g$. Interestingly, this step involves the \textit{stochastic matrix powering problem}, which lies at the heart of space-bounded derandomization, e.g.,~\cite{SZ99,CDSTS23,PP23}. We utilize space-efficient random walks on a directed graph to estimate the power of a stochastic matrix. Consequently, we can only develop a bounded-error randomized algorithm $\calA_f$ for piecewise-smooth functions.\footnote{The (classical) pre-processing in space-efficient QSVT is \textit{not} part of the deterministic Turing machine producing the quantum circuit description in the \BQL{} model (\Cref{def:BQUSPACE}). Instead, we treat it as a component of quantum computation, allowing the use of randomized algorithms since $\BPL \subseteq \BQL$~\cite{FR21}.}

\subsection{Proof overview: A general framework for quantum state testing}

Our framework enables space-bounded quantum state testing, specifically for proving \Cref{thm:space-bounded-quantum-state-certification-RQL-complete-informal,thm:space-bounded-quantum-state-testing-BQL-complete-informal}, and is based on the one-bit precision phase estimation~\cite{Kitaev95}, also known as the \textit{Hadamard test}~\cite{AJL09}. Prior works~\cite{TS13, FL18} have employed (one-bit precision) phase estimation in space-bounded quantum computation.

To address quantum state testing problems, we reduce them to estimating $\Tr(P_{d'}(A)\rho)$, where $\rho$ is a (mixed) quantum state prepared by a quantum circuit $Q_{\rho}$, $A$ is an Hermitian operator block-encoded in a unitary operator $U_A$, and $P_{d'}$ is a space-efficiently computable degree-$d'$ polynomial obtained from some degree-$d$ averaged Chebyshev truncation with $d'=2d-1$. Similar approaches have been applied in \textit{time-bounded} quantum state testing, including fidelity estimation~\cite{GP22} and subsequently trace distance estimation~\cite{WZ23}.

\begin{figure}[!htp]
    \centering
    \begin{quantikz}[wire types={q,b,b,b}, classical gap=0.07cm, row sep=0.75em]
        \lstick{$\ket{0}$} & \gate{H} & \ctrl{1} & \gate{H} & \meter{} \\
        \lstick{$\ket{\bar 0}$} &  & \gate[2]{U_{P_{d'}(A)}}  &  &  \\
        \lstick{$\ket{\bar 0}$} & \gate[2]{\quad Q_{\rho} \quad}  &  &  &  \\
        \lstick{$\ket{\bar 0}$} & & & & 
    \end{quantikz}
    \caption{General framework for quantum state testing $\calT(Q_{\rho},U_A,P_{d'})$.}
    \label{fig:framework-state-testing-intro}
\end{figure}
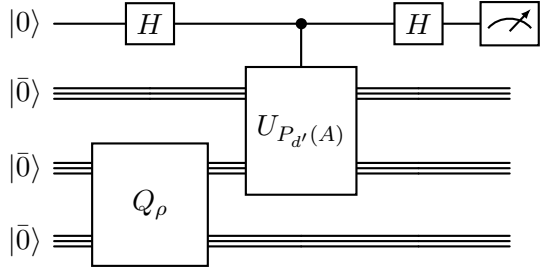

To implement a unitary operator $U_{P_{d'}(A)}$ that (approximately) block-encodes $P_{d'}(A)$ in a space-efficient manner, we require $P_{d'}$ to meet the conditions specified in \Cref{thm:space-efficient-QSVT-informal}. As illustrated in \Cref{fig:framework-state-testing-intro}, we denote the quantum circuit as $\calT(Q_{\rho}, U_A, P_{d'})$, where we exclude the precision for simplicity. The measurement outcome of $\calT(Q_{\rho}, U_A, P_{d'})$ will be $0$ with a probability close to $\frac{1}{2} \big(1+\Tr(P_{d'}(A)\rho)\big)$. This property allows us to estimate $\Tr(P_{d'}(A)\rho)$ within an additive error $\epsilon$ using $O(1/\epsilon^2)$ sequential repetitions, resulting in a \BQL{} containment. 

As an example of the application, $\mathcal{T}(Q_{i}, U_{\frac{\rho_0-\rho_1}{2}}, P_{d'}^{\sign})$ is utilized in \GapQSD{}, where $U_{\frac{\rho_0-\rho_1}{2}}$ is a block-encoding of $\frac{\rho_0-\rho_1}{2}$, and $P_{d'}^{\sign}$ is a space-efficient polynomial approximation of the sign function. 
Notably, this algorithm can be viewed as a two-outcome measurement $\{\hat{\Pi}_0,\hat{\Pi}_1\}$ where $\hat{\Pi}_0 = \frac{1}{2}I+\frac{1}{2} P_{d'}^{\sign}\big(\frac{\rho_0-\rho_1}{2}\big)$.
Similarly, $\mathcal{T}(Q_{i}, U_{\rho_{i}}, P_{d'}^{\ln})$ is utilized in \GapQED{}, where $U_{\rho_{i}}$ is a block-encoding of $\rho_{i}$ for $i\in\binset$, and $P_{d'}^{\ln}$ is a space-efficient polynomial approximation of the normalized logarithmic function. Both $P_{d'}^{\sign}$ and $P_{d'}^{\ln}$ can be obtained by employing \Cref{thm:space-efficient-QSVT-informal}.

\paragraph{Making the error one-sided.} The main challenge is constructing a unitary $U$ of interest, such as $\calT(Q_{\rho},U_A,P_{d'})$, that accepts with a \textit{certain fixed} probability $p$ for \textit{yes} instances ($\rho_0=\rho_1$), while having a probability that polynomially deviates from $p$ for \textit{no} instances. 
As an example, we consider \coCertQHSlog{} and express $\HS(\rho_0,\rho_1)$ as a linear combination of $\Tr(\rho_0^2)$, $\Tr(\rho_1^2)$, and $\Tr(\rho_0\rho_1)$.
We can then design a unitary quantum algorithm that satisfies the requirement for \textit{yes} instances based on the SWAP test~\cite{BCWdW01}, and consequently, we can achieve perfect completeness by applying the exact amplitude amplification~\cite{BBHT98,BHMT02}. The analysis demonstrates that the acceptance probability polynomially deviates from $1$ for \textit{no} instances. By applying error reduction for \coRQUL{}, the resulting algorithm is indeed in \coRQUL{}. 

Moving on to \coCertQSDlog{}, we consider the quantum circuit $U_{i} = \mathcal{T}(Q_{i}, U_{\frac{\rho_0-\rho_1}{2}}, P_{d'}^{\sign})$ for $i \in \binset$. Since our space-efficient QSVT \textit{preserves parity}, specifically, the approximation polynomial $P_{d'}^{\sign}$ satisfies $P_{d'}^{\sign}(0)=0$,\footnote{Let $f$ be any odd function such that space-efficient QSVT associated with $f$ can be implemented by \Cref{thm:space-efficient-QSVT-informal}. It follows that the corresponding approximation polynomial $P_{d'}^{(f)}$ is also odd. See \Cref{remark:QSVT-parity-preserving}.} the requirement for \textit{yes} instances is satisfied. Then, we can similarly achieve the \coRQUL{} containment for \coCertQSDlog{}. 

\subsection{Discussion and open problems}

Since space-efficient quantum singular value transformation (QSVT) offers a unified framework for designing quantum logspace algorithms, it suggests a new direction to find applications of space-bounded quantum computation. 
An intriguing candidate is solving positive semi-definite programming (SDP) programs with constant precision~\cite{JY11, AZLO16}. 
A major challenge in achieving a \BQL{} containment for this problem is that iteratively applying the space-efficient QSVT super-constantly many times may lead to a bitstring indexed encoding requiring $\omega(\log{n})$ ancillary qubits, raising the question:

\begin{enumerate}[label={\upshape(\alph*)}]
\item Is it possible to have an approximation scheme (possibly under certain conditions) that introduces merely $O(1)$ additional ancillary qubits in the bitstring indexed encoding per iteration, such that applying space-efficient QSVT $\log{n}$ times results in a bitstring indexed encoding with at most $O(\log{n})$ ancillary qubits? \label{open-problems:iterative-space-efficient-QSVT}
\end{enumerate}

Furthermore, as quantum distances investigated in this work are all instances of a quantum analog of symmetric $f$-divergence, there is a natural question on other instances:
\begin{enumerate}[label={\upshape(\alph*)}]
    \setcounter{enumi}{1}
    \item Can we demonstrate that space-bounded quantum state testing problems with respect to other quantum distance-like measures are also \BQL{}-complete? 
\end{enumerate}

In addition, there is a question on improving the efficiency of the space-efficient QSVT:\footnote{Another technical limitation of our approach is that the pre-processing time in our space-efficient QSVT is $\poly(d,\epsilon^{-1})$, whereas in the time-efficient QSVT it is $\poly(d)$. For polynomial approximations of degree $d=O(\delta^{-1} \log(\epsilon^{-1}))$, such as for the sign function, this difference yields an \textit{exponentially} worse dependence on $\epsilon^{-1}$, leading to weaker results in \Cref{footnote:space-efficient-error-reduction-subtly}. Nevertheless, this limitation seems \textit{fundamental} and does not appear improvable in general, as uniform polynomial approximations of the (signed) positive power function in~\cite{LW25entropy,LW25Lalpha} has degree $d=\poly(\epsilon^{-1})$, matching the $\epsilon^{-1}$ dependence of our pre-processing.} 
\begin{enumerate}[label={\upshape(\alph*)}]
    \setcounter{enumi}{2}
    \item Recently, a query complexity lower bound $\Omega(d)$ for matrix functions~\cite{MS23} implies that time-efficient QSVT~\cite{GSLW19} is time-optimal. Can we improve the query complexity of $U$ and $U^{\dagger}$ in space-efficient QSVT for smooth functions from $O(d^2)$ to $O(d)$? This improvement would make QSVT optimal for both (quantum) time and space. \label{open-problems:space-efficient-QSVT-query-complexity}
\end{enumerate}

Notably, the pre-processing in QSVT techniques, which is not necessarily classical in general, usually involves finding the sequence of $z$-axis rotation angles. Our approach, however, uses averaged Chebyshev truncation and the LCU technique. A general solution thus seems to involve developing a space-efficient (quantum) angle-finding algorithm. 

\subsection{Related works}

\paragraph{More on quantum state testing problems.}
Testing the spectrum of quantum states was studied in~\cite{OW21}: for example, whether a quantum state is maximally mixed or $\epsilon$-far away in trace distance from mixed states can be tested using $\Theta(N/\epsilon^2)$ samples.
Later, it was generalized in~\cite{BOW19} to quantum state certification with respect to fidelity and trace distance. 
Estimating distinguishability measures of quantum states~\cite{ARS+21} is another topic, including the estimation of fidelity~\cite{FL11,WZC+22,GP22} and trace distance~\cite{WGL+22,WZ23}.

Entropy estimation of quantum states has been widely studied in the literature. 
Given quantum purified access, it was shown in~\cite{GL20} that the von Neumann entropy $\S(\rho)$ can be estimated within additive error $\epsilon$ with query complexity $\tilde O(N/\epsilon^{1.5})$. 
If we know the reciprocal $\kappa$ of the minimum non-zero eigenvalue of $\rho$, then $\S(\rho)$ can be estimated with query complexity $\tilde O(\kappa^2/\epsilon)$~\cite{CLW20}.
$\S(\rho)$ can be estimated within multiplicative error $\epsilon$ with query complexity $\tilde O(n^{\frac{1}{2}+\frac{1+\eta}{2\epsilon^2}})$~\cite{GHS21}, provided that $\S(\rho) = \Omega(\epsilon + 1/\eta)$.
If $\rho$ is of rank $r$, then $\S(\rho)$ can be estimated with query complexity $\tilde O(r/\epsilon^2)$~\cite{WGL+22}.
Estimating the R\'enyi entropy $S_\alpha(\rho)$ given quantum purified access was first studied in~\cite{SH21}, and then was improved in~\cite{WGL+22,LWZ22}. 
In addition, the work of~\cite{GH20} investigates the (conditional) hardness of \GapQED{} with logarithmic depth or constant depth. 

\paragraph{Further implications of our work.} By analyzing the optimal prover strategies underlying the $\mathsf{QIP(2)}$ and $\mathsf{co\text{-}QIP(2)}$ proof systems for \QSZK{}, as presented in~\cite{Wat02}, one can establish that these prover strategies can be approximately implemented in quantum linear space, yielding a slight improvement of the (quantum) upper bound for \QSZK{}, namely $\mathsf{QIP(2)} \cap \mathsf{co\text{-}QIP(2)}$ with a quantum linear-space (and thus single-exponential-time) honest prover~\cite{LLW25}.\footnote{A weaker version of this result, specifically that \QSZK{} is in \QIPtwo{} with a quantum linear-space honest prover, appeared in the second arXiv version of this work (as well as in the second-named author's PhD thesis~\cite[Section 6.3]{Liu25}), but was removed from the current version. Additionally, the improvement in~\cite{LLW25} also applies to the non-interactive variant of \QSZK{}: in particular, \textsf{NIQSZK} is in $\mathrm{qq}\text{-}\mathsf{QAM}$ with a quantum linear-space honest prover.} 
In addition, our new \BQL{}-complete problem \GapQSDlog{} implies that a space-bounded variant of quantum statistical zero-knowledge~\cite{Wat02,Wat09}, where the verifier's actions are restricted to be \textit{unitary} (denoted by $\mathsf{QSZK_\mathrm{U}L}$), is contained in \BQL{} (and indeed satisfies $\mathsf{QSZK_\mathrm{U}L}=\BQL$)~\cite{LLNW24}. This stands in contrast to the time-bounded setting, in which \QSZK{} is unlikely to collapse to \BQP{}, and computational advantages are typically gained from interaction.

\section{Preliminaries}
\label{sec:preliminaries}

We assume that the reader is familiar with quantum computation and the theory of quantum information. 
For an introduction, the textbooks by~\cite{NC10} and~\cite{deWolf19} provide a good starting point, while for a more comprehensive survey on quantum complexity theory, refer to~\cite{Watrous08}. 

In addition, we adopt the convention that the logarithmic function $\log$ has a base of $2$, denoted by $\log(x)\coloneqq\log_2(x)$ for any $x \in \bbR^{+}$. 
Moreover, we define $\tilde{O}(f)\coloneqq O(f \polylog(f))$. 
Lastly, for the sake of simplicity, we utilize the notation $\ket{\bar{0}}$ to represent $\ket{0}^{\otimes a}$ with $a>1$.

\subsection{Singular value decomposition and transformation}
\label{subsec:singular-value-decomp-and-trans}
We recommend~\cite{Bhatia96,HJ12,Halmos87} for comprehensive textbooks on matrix analysis and linear algebra. 
For any $\tilde{d} \times d$ (complex) matrix $A$, there is a \textit{singular value decomposition} of $A$ such that $A= \sum_{i=1}^{\min\{d,\tilde{d}\}} \sigma_i \ket{\tilde{\psi}_i}\bra{\psi_i}$, where: 
\begin{itemize}[itemsep=0.33em,topsep=0.33em,parsep=0.33em]
    \item The \textit{singular values} $\sigma_1 \geq \sigma_2 \geq \cdots \geq \sigma_{\min\{d,\tilde{d}\}} \geq 0$, where non-zero singular values $\sigma_i$ are the square roots of non-zero eigenvalues of $A^{\dagger}A$ or $AA^{\dagger}$.
    \item $\ket{\tilde{\psi}_1}, \cdots, \ket{\tilde{\psi}_{\tilde{d}}}$ form an orthonormal basis and are eigenvectors of $AA^{\dagger}$. 
    \item $\ket{\psi_{1}}, \cdots, \ket{\psi_{d}}$ form an orthonormal basis and are eigenvectors of $A^{\dagger}A$. 
\end{itemize}

Notably, the largest singular value of $A$ coincides with the operator norm of $A$, specifically $\|A\|\coloneqq\|A\|_{2\rightarrow 2} = \sigma_1(A)$. 
Let $\| \ket{\psi} \|_2 \coloneqq \sqrt{\innerprod{\psi}{\psi}}$ be the Euclidean norm of a vector $\ket{\psi}$. 
Next, we list the families of matrices that are commonly used in this work. It is noteworthy that they all admit the singular value decomposition: 
\begin{itemize}
    \item \textbf{Hermitian matrices}. $H^{\dagger} = H$, if and only if $\bra{\psi} H \ket{\psi} \in \bbR$ for all $\ket{\psi}$ such that $\|\ket{\psi}\|_2=1$, if and only if the absolute values of the eigenvalues of $H$ coincide with its singular values. 
    \item \textbf{Unitary matrices}. $UU^{\dagger}= U^{\dagger}U=I$, if and only if $\|U\ket{\psi}\|_2 = \|U^{\dagger}\ket{\psi}\|_2 = 1$ for all $\ket{\psi}$ such that $\|\ket{\psi}\|_2=1$. Equivalently, all eigenvalues $\lambda_i$ of $U$ have modulus $|\lambda_i|=1$, implying that all singular values of $U$ are $1$. 
    \item \textbf{Positive semi-definite matrices}. $P=CC^{\dagger}$ for some matrix $C$, if and only if $\bra{\psi}P\ket{\psi} \geq 0$ for all $\ket{\psi}$ such that $\|\ket{\psi}\|_2=1$, if and only if all eigenvalues of $P$ are non-negative. 
    \item \textbf{Orthogonal projection matrices}. $\Pi^2=\Pi=\Pi^{\dagger}$, if and only if $\Pi^2=\Pi$ and $\|\Pi \ket{\psi}\|_2 \leq 1$ for all $\ket{\psi}$ such that $\|\ket{\psi}\|_2=1$. Equivalently, all eigenvalues of $\Pi$ are either $0$ or $1$. For the last characterization, see~\cite[III.75]{Halmos87} and~\cite[Corollary 3.4.3.3]{HJ12}.  
    \item \textbf{Partial isometries}. $G G^{\dagger} G = G$, if and only if $G^{\dagger} G G^{\dagger} = G^{\dagger}$, if and only if $\|G \ket{\psi}\|_2=1$ for all $\ket{\psi} \in \ker(G)^{\perp}$ such that $\|\ket{\psi}\|_2=1$, if and only if $G^{\dagger} G$ is an orthogonal projection onto $\ker(G)^{\perp}$ (see~\cite[Exercise III.76.5]{Halmos87}). Consequently, the non-zero singular values of a partial isometry $G$ are all $1$. Moreover, an injective partial isometry is an isometry, and an invertible partial isometry is unitary. 
\end{itemize}

For any matrix $A$ satisfying $\|A\| \leq 1$, there is a unitary $U$ with orthogonal projections $\tilde{\Pi}$ and $\Pi$ such that $A = \tilde{\Pi} U \Pi$.\footnote{As indicated in~\cite[2.7.P2]{HJ12}, such a matrix $U$ is called a \textit{unitary dilation} of $A$.  This unitary dilation $U$ exists if and only if $A$ is a contraction, namely $\|A\| \leq 1$. \label{footnote:unitary-dilation}} With these definitions in place, we can view the singular value decomposition as the \textit{projected unitary encoding} (see \Cref{def:bitstring-indexed-encoding}): 

%Definition 2.3.1 in~\cite{Gilyen19}
%Definition 11 in~\cite{GSLW18}
\begin{definition}[Singular value decomposition of a projected unitary, adapted from Definition 7 in~\cite{GSLW19}]
    \label{def:SVD-projected-unitary}
    Given a projected unitary encoding of $A$, denoted by $U$, associated with orthogonal projections $\Pi$ and $\tilde{\Pi}$ on a finite-dimensional Hilbert space $\calH_U$: $A=\tilde{\Pi} U \Pi$.
    Then the singular value decomposition of $A$ ensures that there exist orthonormal bases of $\Img(\Pi)$ and $\Img\rbra[\big]{\tilde{\Pi}}$ such that:
    \begin{itemize}[itemsep=0.33em,topsep=0.33em,parsep=0.33em]
    \item $\Pi$: $\left\{ \ket{\psi_{i}}: i \in [d] \right\}$, where $d\coloneqq\rank(\Pi)$, of a subspace $\Img(\Pi)=\spanset\left\{\ket{\psi_i}\right\}$;
    \item $\tilde{\Pi}$: $\big\{ \ket{\tilde{\psi}_i}: i\in[\tilde{d}] \big\}$, where $\tilde{d}\coloneqq\rank(\tilde{\Pi})$, of a subspace $\Img(\tilde{\Pi})=\spanset\big\{\ket{\tilde{\psi}_i}\big\}$. 
    \end{itemize}
\end{definition}

We say that a function $f \colon \bbR \rightarrow \bbC$ is \textit{even} if $f(-x)=f(x)$ for all $x\in \bbR$, and that it is \textit{odd} if $f(-x)=-f(x)$ for all $x\in\bbR$. Next, we define the \textit{singular value transformation} of matrices:  

\begin{definition}[Singular value transformation by even or odd functions, adapted from Definition 9 in~\cite{GSLW19}]
    \label{def:matrix-SV-function}
    Let $f\colon \bbR \rightarrow \bbC$ be an even or odd function. We consider a linear operator $A\in \bbC^{\tilde{d}\times d}$ satisfying the singular value decomposition $A = \sum_{i=1}^{\min\{d,\tilde{d}\}} \sigma_i \ket{\tpsi_i}\bra{\psi_i}$. We define the \textit{singular value transformation} corresponding to $f$ as follows: 
    \[f^{\SV}(A) \coloneqq \begin{cases}
        \sum_{i=1}^{\min\{d,\tilde{d}\}} f(\sigma_i) \ket{\tpsi_i}\bra{\psi_i},&\text{for odd }f,\\
        \sum_{i=1}^d f(\sigma_i) \ket{\psi_i}\bra{\psi_i},&\text{for even }f.
    \end{cases}\]
    Here, $\sigma_i\coloneqq0$ for $i \in \{\min\{d,\tilde{d}\}\!+\!1,\cdots,d\!-\!1,d\}$. 
    For any Hermitian matrix $A$, $f^{\SV}(A)=f(A)$.     
\end{definition}

Finally, for any $d\times d$ Hermitian matrix $A$, there is a \textit{spectral decomposition} of $A$ such that $A = \sum_{i=1}^d \lambda_i \ket{\psi_i}\bra{\psi_i}$ where all eigenvalues $\{\lambda_i\}_{i=1}^d$ are real and $\{ \ket{\psi_i} \}_{i=1}^{d}$ is an orthonormal basis. As a consequence, if $f$ is an even or odd function, $f(A) = \sum_{i=1}^d f(\lambda_i) \ket{\psi_i}\bra{\psi_i} = f^{\SV}(A)$ can be achieved by singular value transformation defined in \Cref{def:matrix-SV-function}.

\subsection{Distances and divergences for quantum states}
We will provide an overview of relevant quantum distances and divergences, along with useful inequalities among different quantum distance-like measures. 
Additionally, we recommend~\cite[Section 3.1]{BOW19} for a nice survey on quantum distances and divergences. 
We say that a square matrix $\rho$ is a quantum state if $\rho$ is positive semi-definite and $\Tr(\rho)=1$.

\begin{definition}[Quantum distances and divergences]
    \label{def:quantum-distances}
    For any quantum states $\rho_0$ and $\rho_1$, we define several distance-like measures and relevant quantities: 
    \begin{itemize}
        \item \textnormal{\textbf{Trace distance}.} $\td(\rho_0,\rho_1)\coloneqq\frac{1}{2}\Tr|\rho_0-\rho_1|=\frac{1}{2}\Tr(((\rho_0-\rho_1)^\dagger(\rho_0-\rho_1))^{1/2}).$
        \item \textnormal{\textbf{(Uhlmann) Fidelity}.} $\F(\rho_0,\rho_1)\coloneqq\Tr|\sqrt{\rho_0}\sqrt{\rho_1}|$.
        \item \textnormal{\textbf{Squared Hilbert--Schmidt distance}.} $\HS(\rho_0,\rho_1)\coloneqq\frac{1}{2}\Tr(\rho_0-\rho_1)^2.$
        \item \textnormal{\textbf{von Neumann entropy}.} $\S(\rho)\coloneqq-\Tr(\rho \ln \rho)$ for any quantum state $\rho$. 
        \item \textnormal{\textbf{Quantum Jensen-Shannon divergence}.} $\QJS(\rho_0,\rho_1) \coloneqq \S\big(\tfrac{\rho_0+\rho_1}{2}\big)-\tfrac{\S(\rho_0)+\S(\rho_1)}{2}$. 
    \end{itemize}
\end{definition}

The trace distance and the squared Hilbert--Schmidt distance reach the minimum of $0$ when $\rho_0$ equals $\rho_1$, while the fidelity attains a maximum value of $1$. Additionally, there are two equalities when at least one of the two states is a pure state --- A quantum state $\rho$ is a pure state if and only if $\Tr(\rho^2)=1$, equivalently $\rho=\ket{\psi}\bra{\psi}$ for some $\ket{\psi}$ satisfying $\|\ket{\psi}\|_2=1$: 
\begin{itemize}
\item For a pure state $\rho_0$ and a mixed state $\rho_1$, $\F^2(\rho_0,\rho_1) = \Tr(\rho_0 \rho_1)$.
\item For two pure states $\rho_0$ and $\rho_1$, $\Tr(\rho_0 \rho_1) = 1 - \HS(\rho_0,\rho_1)$.
\end{itemize}

Moreover, we have $\HS(\rho_0,\rho_1) = \frac{1}{2}(\Tr(\rho_0^2) + \Tr(\rho_1^2)) - \Tr(\rho_0\rho_1)$. Additionally, Fuchs and van de Graaf~\cite{FvdG99} showed a well-known inequality between the trace distance and the fidelity:
\begin{lemma}[Trace distance vs.~fidelity, adapted from~\cite{FvdG99}]
\label{lemma:traceDist-vs-fidelity}
For any states $\rho_0$ and $\rho_1$, 
\[1-\F(\rho_0,\rho_1) \leq \td(\rho_0,\rho_1) \leq \sqrt{1-\F^2(\rho_0,\rho_1)}.\]
\end{lemma}

The joint entropy theorem (\Cref{lemma:joint-entropy-theorem}) enhances our understanding of entropy in classical-quantum states and is necessary for our usages of the von Neumann entropy. 

\begin{lemma}[Joint entropy theorem, adapted from Theorem 11.8(5) in~\cite{NC10}]
    \label{lemma:joint-entropy-theorem}
    Suppose $p_i$ are probabilities corresponding to a distribution $D$, $\ket{i}$ are orthogonal states of a system $A$, and $\{\rho_i\}_i$ is any set of density operators for another system $B$. Then 
    \[\S\rbra*{ \sum_i p_i\ket{i}\bra{i} \otimes \rho_i } = \H(D) + \sum_i p_i \S(\rho_i).\]
\end{lemma}

Let us now turn our attention to the quantum Jensen-Shannon divergence, which is defined in~\cite{MLP05}. 
For simplicity, we define $\QJS_2(\rho_0,\rho_1)\coloneqq\QJS(\rho_0,\rho_1)/\ln 2$ using the base-$2$ (matrix) logarithmic function.  Notably, when considering size-$2$ ensembles with a uniform distribution, the renowned Holevo bound~\cite{Holevo73JS} (see Theorem 12.1 in~\cite{NC10}) indicates that the \textit{quantum Shannon distinguishability} studied in~\cite{FvdG99} is at most the quantum Jensen-Shannon divergence. Consequently, this observation yields inequalities between the trace distance and the quantum Jensen-Shannon divergence.\footnote{For a detailed proof of these inequalities, please refer to \cite[Section 2.2]{Liu23}.}

\begin{lemma}[Trace distance vs.\ quantum Jensen-Shannon divergence, adapted from~\cite{FvdG99,Holevo73JS,BH09}]
    \label{lemma:QJS-vs-traceDist}
    For any quantum states $\rho_0$ and $\rho_1$, we have  
    \[1-\binH\rbra*{\frac{1-\td(\rho_0,\rho_1)}{2}} \leq \QJS_2(\rho_0,\rho_1) \leq \td(\rho_0,\rho_1).\]
    \noindent Here, the binary entropy $\binH(p)\coloneqq-p\log(p)-(1-p)\log(1-p)$.
\end{lemma}

\subsection{Space-bounded quantum computation}

We say that a function $s(n)$ is \textit{space-constructible} if there exists a deterministic space $s(n)$ Turing machine that takes $1^n$ as input and outputs $s(n)$ in the unary encoding. 
Moreover, we say that a function $f(n)$ is $s(n)$-\textit{space computable} if there exists a deterministic space $s(n)$ Turing machine that takes $1^n$ as an input and output $f(n)$. 
Our definitions of space-bounded quantum computation are formulated in terms of \textit{quantum circuits}, whereas many prior works focused on \textit{quantum Turing machines}~\cite{Wat09,Wat03,vMW12}. For a discussion on the equivalence between space-bounded quantum computation using \textit{quantum circuits} and \textit{quantum Turing machines}, we refer readers to \cite[Appendix A]{FL18} and \cite[Section 2.2]{FR21}. 

We begin by defining time-bounded and space-bounded quantum circuit families and then proceed to the corresponding complexity class $\BQUSPACE[s(n)]$. We use the abbreviated notation $C_x$ to indicate that the circuit $C_{|x|}$ takes input $x$.

\begin{definition}[Time- and space-bounded quantum circuit families]
    \label{def:unitary-quantum-circuit}
    A (unitary) quantum circuit is a sequence of quantum gates, each of which belongs to some fixed gateset that is universal for quantum computation, such as $\{\Had, \CNOT, \T\}$. 
    For a promise problem $\calL = (\calL_{\yes},\calL_{\no})$, we say that a family of quantum circuits $\{C_x: x\in \calL\}$ is $t(n)$-time-bounded if there is a deterministic Turing machine that, on any input $x \in \calL$, runs in time $O(t(|x|))$, and outputs a description of $C_x$ such that $C_x$ accepts (resp., rejects) if $x \in \calL_{\yes}$ (resp., $x \in \calL_{\no}$). 
    
    \noindent Similarly, we say that a family of quantum circuits $\{C_x: x\in \calL\}$ is $s(n)$-space-bounded if there is a deterministic Turing machine that, on any input $x \in \calL$, runs in space $O(s(|x|))$ (and hence time $2^{O(s(|x|))}$), and outputs a description of $C_x$ such that $C_x$ accepts (resp., rejects) if $x \in \calL_{\yes}$ (resp., $x \in \calL_{\no}$); furthermore, $C_x$ acts on $O(s(|x|))$ qubits and has $2^{O(s(|x|))}$ gates.
\end{definition}

\begin{definition}[{$\BQUSPACE[s(n),a(n),b(n)]$}, adapted from~{\cite[Definition 5]{FR21}}]
    \label{def:BQUSPACE}
    Let $s\colon\bbN \rightarrow \bbN$ be a space-constructible function such that $s(n) \geq \Omega(\log{n})$. Let $a(n)$ and $b(n)$ be functions that are computable in deterministic space $s(n)$. 
    A promise problem $(\calL_{\yes},\calL_{\no})$ is in $\BQUSPACE[s(n),a(n),b(n)]$ if there exists a family of $s(n)$-space-bounded  (unitary) quantum circuits $\{C_x\}_{x\in\calL}$, where $n=|x|$, satisfying the following:
    \begin{itemize}[itemsep=0.33em,topsep=0.33em,parsep=0.33em]
        \item The output qubit is measured in the computational basis after applying $C_x$. We say that $C_x$ \textit{accepts} $x$ if the measurement outcome is $1$, whereas $C_x$ \textit{rejects} $x$ if the outcome is $0$. 
        \item If $x \in \calL_{\yes}$, $\Pr{C_x \text{ accepts } x} \geq a(|x|)$. 
        \item If $x \in \calL_{\no}$, $\Pr{C_x \text{ accepts } x} \leq b(|x|)$. 
    \end{itemize}
\end{definition}
We remark that \Cref{def:BQUSPACE} is \textit{gateset-independent}, given that the gateset is closed under adjoint and all entries in chosen gates have reasonable precision. This property is due to the space-efficient Solovay--Kitaev theorem presented in~\cite{vMW12}. 
Moreover, we can achieve error reduction for $\BQUSPACE[s(n),a(n),b(n)]$ as long as $a(n)-b(n) \geq 2^{-O(s(n))}$, which follows from~\cite{FKLMN16} or our space-efficient QSVT-based construction in \Cref{subsec:BQUL-error-reduction}. 
We thereby define $\BQUSPACE[s(n)]\coloneqq\BQUSPACE[s(n),2/3,1/3]$ to represent (two-sided) bounded-error unitary quantum space, and $\BQUL\coloneqq\BQUSPACE[O(\log{n})]$ to denote unitary quantum logspace. 

We next consider general space-bounded quantum computation, which allows \textit{intermediate quantum measurements}. As indicated in \cite[Section 4.1]{AKN98}, for any quantum channel $\Phi$ mapping from density matrices on $k_1$ qubits to density matrices on $k_2$ qubits, we can exactly simulate this quantum channel $\Phi$ by a unitary quantum circuit acting on $2k_1+k_2$ qubits. Therefore, we extend \Cref{def:unitary-quantum-circuit} to \textit{general quantum circuits}, which allows local operations, such as intermediate measurements in the computational basis, resetting qubits to their initial states, and tracing out qubits. 
Now we proceed with a definition on $\BQSPACE[s(n)]$.

\begin{definition}[{$\BQSPACE[s(n),a(n),b(n)]$}, adapted from~{\cite[Definition 7]{FR21}}]
    \label{def:BQSPACE}
    Let $s\colon\bbN \rightarrow \bbN$ be a space-constructible function such that $s(n) \geq \Omega(\log{n})$. Let $a(n)$ and $b(n)$ be functions that are computable in deterministic space $s(n)$. 
    A promise problem $(\calL_{\yes},\calL_{\no})$ is in $\BQSPACE[s(n),a(n),b(n)]$ if there exists a family of $s(n)$-space-bounded general quantum circuits $\{\Phi_x\}_{x\in\calL}$, where $n=|x|$, such that the following holds:
    \begin{itemize}[itemsep=0.33em,topsep=0.33em,parsep=0.33em]
        \item The output qubit is measured in the computational basis after applying $\Phi_x$. We say that $\Phi_x$ \textit{accepts} $x$ if the measurement outcome is $1$, whereas $\Phi_x$ \textit{rejects} $x$ if the outcome is $0$.
        \item If $x \in \calL_{\yes}$, $\Pr{\Phi_x \text{ accepts } x} \geq a(|x|)$.
        \item If $x \in \calL_{\no}$, $\Pr{\Phi_x \text{ accepts } x} \leq b(|x|)$. 
    \end{itemize}
\end{definition}

It is noteworthy that unitary quantum circuits, which correspond to unitary channels, are a specific instance of general quantum circuits that correspond to quantum channels. We thus infer that $\BQUSPACE[s(n)] \subseteq \BQSPACE[s(n)]$ for any $s(n) \geq \Omega(\log{n})$. However, the opposite direction was a long-standing open problem. 
Recently, Fefferman and Remscrim~\cite{FR21} demonstrated a remarkable result that $\BQSPACE[s(n)] \subseteq \BQUSPACE[O(s(n))]$. 
In addition, it is evident that $\BQSPACE[s(n)]$ can achieve error reduction since it admits sequential repetition simply by resetting working qubits. Therefore, we define $\BQSPACE[s(n)]\coloneqq\BQSPACE[s(n),2/3,1/3]$ to represent (two-sided) bounded-error general quantum space, and denote general quantum logspace by $\BQL\coloneqq\BQSPACE[O(\log{n})]$. 

\vspace{1em}
We now turn our attention to \textit{one-sided} bounded-error unitary quantum space $\RQUSPACE[s(n)]$ and $\coRQUSPACE[s(n)]$ for $s(n) \geq \Omega(\log{n})$. These complexity classes were first introduced by Watrous~\cite{Wat01} and have been further discussed in~\cite{FR21}. We proceed with the definitions:
\begin{itemize}
    \item $\RQUSPACE[s(n),a(n)]\coloneqq\BQUSPACE[s(n),a(n),0]$; 
    \item $\coRQUSPACE[s(n),b(n)]\coloneqq\BQUSPACE[s(n),1,b(n)]$. 
\end{itemize}
Note that $\RQUSPACE[s(n),a(n)]$ and $\coRQUSPACE[s(n),b(n)]$ can achieve error reduction, as shown in~\cite{Wat01} or our space-efficient QSVT-based construction in~\Cref{subsec:BQUL-error-reduction}. We define 
\begin{align*}
    \RQUSPACE[s(n)] &\coloneqq \BQUSPACE\sbra*{ s(n),1/2,0 },\\
    \coRQUSPACE[s(n)] &\coloneqq \BQUSPACE\sbra*{ s(n),1,1/2 }
\end{align*}
to represent one-sided bounded-error unitary quantum space, as well as logspace counterparts 
\begin{align*}
    \RQUL &\coloneqq \RQUSPACE[O(\log{n})],\\
    \coRQUL &\coloneqq \coRQUSPACE[O(\log{n})].
\end{align*}

\begin{remark}[\RQUL{} and \coRQUL{} are gateset-dependent]
    \label{remark-RQUL-gateset-dependent}
    We observe that changing the gateset using the space-efficient Solovay--Kitaev theorem~\cite{vMW12} can cause errors, revealing the \textit{gateset-dependence} of unitary quantum space classes with one-sided bounded-error. To address this issue, we adopt a larger gateset $\mathcal{G}$ for $\RQUSPACE[s(n)]$ and $\coRQUSPACE[s(n)]$, which includes any single-qubit gates whose amplitudes can be computed in deterministic $O(s(n))$ space. 
\end{remark}

\subsection{Polynomial approximation via averaged Chebyshev truncation}
\label{subsec:Chebyshev-polys-and-truncated-expansion}

We begin by defining Chebyshev polynomials and then introduce Chebyshev truncation and averaged Chebyshev truncation, with the latter commonly known as the \textit{de La Vall\'ee Poussin partial sum}. These concepts are essential to our space-efficient quantum singular value transformation techniques (space-efficient QSVT, see \Cref{sec:space-efficient-QSVT}). For a comprehensive review of Chebyshev series and Chebyshev expansion, we refer the reader to~\cite[Chapter 3]{Rivlin90}. 

\begin{definition}[Chebyshev polynomials]
The Chebyshev polynomials (of the first kind) $T_k(x)$ are defined via the following recurrence relation: 
\[T_0(x)\coloneqq 1, T_1(x)\coloneqq x, \text{ and } T_{k+1}(x)\coloneqq 2x T_k(x)-T_{k-1}(x).\] 
For $x \in [-1,1]$, an equivalent definition is $T_k(\cos \theta) = \cos(k \theta)$.
\end{definition}

To use Chebyshev polynomials (of the first kind) for Chebyshev expansion, we first need to define an inner product between two functions, $f$ and $g$, as long as the following integral exists:

\begin{equation}
    \label{eq:polynomial-inner-product}
    \innerprodF{f}{g} \coloneqq \frac{2}{\pi} \int_{-1}^1 \frac{f(x)g(x)}{\sqrt{1-x^2}} \dx = \frac{2}{\pi} \int_{-\pi}^0 f(\cos\theta)g(\cos\theta) \dtheta.
\end{equation}

The Chebyshev polynomials form an orthonormal basis in the inner product space induced by $\innerprodF{\cdot}{\cdot}$ defined in \Cref{eq:polynomial-inner-product}.
As a result, any continuous and integrable function $f: [-1,1] \rightarrow \bbR$ whose Chebyshev coefficients satisfy $\lim_{k \rightarrow \infty} c_k=0$, where $c_k$ is defined in \Cref{eq:Chebyshev-expansion}, has a Chebyshev expansion given by:
\begin{equation}
    \label{eq:Chebyshev-expansion}
    f(x)=\frac{1}{2} c_0 T_0(x) + \sum_{k=1}^{\infty} c_k T_k(x), \text{ where }  c_k\coloneqq\innerprodF{T_k}{f}.
\end{equation} 

A natural approach to approximating functions with a Chebyshev expansion is to consider the truncated version of the Chebyshev expansion $\tilde{P}_d = c_0/2 + \sum_{k=1}^d c_k T_k$, denoted as \textit{Chebyshev truncation}. Remarkably, $\tilde{P}_d$ provides a \textit{nearly best} uniform approximation to $f$: 
\begin{lemma}[Nearly best uniform approximation by Chebyshev truncation, adapted from Theorem 3.3 in~\cite{Rivlin90}]
\label{lemma:truncated-Chebyshev-expansion}
For any continuous and integrable function $f\colon[-1,1]\rightarrow \bbR$, let $\varepsilon_d(f)$ be the truncation error that corresponds to the degree-$d$ best uniform approximation on $[-1,1]$ to $f$, then the degree-$d$ Chebyshev truncation polynomial $\tilde{P}_d$ satisfies
\[\varepsilon_d(f) \leq \max_{x\in[-1,1]} |f(x)-\tilde{P}_d(x)| \leq \Big( 4+\frac{4}{\pi^2} \log{d} \Big) \varepsilon_d(f)\]

\noindent Consequently, if there is a degree-$d$ polynomial $P^*_d\in\bbR[x]$ such that $\max_{x\in[-1,1]} |f(x)-P^*_d(x)| \leq \epsilon$, then the degree-$d$ Chebyshev truncation polynomial $\tilde{P}_d$ satisfies 
\[\max_{x\in[-1,1]} |f(x)-\tilde{P}_d(x)| \leq O(\epsilon \log d).\]
\end{lemma}

It is noteworthy that the proof of \Cref{lemma:truncated-Chebyshev-expansion} in~\cite{Rivlin90} relies only on the linear decay of Chebyshev coefficients $c_k$ for any Chebyshev expansion. However, for functions with a Chebyshev expansion whose Chebyshev coefficients decay almost exponentially, Chebyshev truncation is ``asymptotically'' as good as the best uniform approximation: 

\begin{lemma}[A sufficient condition that Chebyshev truncation is ``asymptotically'' best, adapted from~Equation (3.44) in~\cite{Rivlin90}]
    \label{lemma:truncated-Chebyshev-expansion-verygood}
    For any function $f$ that admits a Chebyshev expansion, consider a degree-$d$ Chebyshev truncation polynomial $\tilde{P}_d$, and let $\varepsilon_d(f)$ be the truncation error corresponding to the degree-$d$ best uniform approximation on $[-1,1]$ to $f$.  If the Chebyshev coefficients of $f$ satisfy $\sum_{j=2}^{\infty} |c_{d+j}| \leq \eta |c_{d+1}|$, then 
    \[\varepsilon_d(f) \leq \max_{x\in[-1,1]} |f(x)-\tilde{P}_d(x)| \leq \frac{4}{\pi} (1+\eta) \varepsilon_d(f).\]
\end{lemma}

Although \Cref{lemma:truncated-Chebyshev-expansion-verygood} improves the truncation error in \Cref{lemma:truncated-Chebyshev-expansion} from $O(\epsilon \log{d})$ to $O(\epsilon)$, it only applies to a fairly narrow range of functions, such as sine and cosine functions. 
Using an average of Chebyshev truncations, known as the de La Vall\'ee Poussin partial sum, we obtain the degree-$d$ \textit{averaged Chebyshev truncation} $\hat{P}_{d'}$, which is a polynomial of degree $d'=2d-1$:
\begin{equation}
    \label{eq:averaged-Chebyshev-truncation}
    \hat{P}_{d'}(x) \coloneqq \frac{1}{d} \sum_{l=d}^{d'} \tilde{P}_l(x) 
    = \frac{\hat{c}_0}{2} + \sum_{k=1}^{d'} \hat{c}_k T_k(x) \text{ where } \hat{c}_k = \begin{cases}
        c_k ,& 0 \leq k \leq d\\
        \frac{2d-k}{d} c_k,& k > d
    \end{cases},
\end{equation}
we can achieve the truncation error $4 \epsilon$ for any function that admits Chebyshev expansion. 
\begin{lemma}[Asymptotically best approximation by averaged Chebyshev truncation, adapted from Exercise 3.4.7 in~\cite{Rivlin90}]
    \label{lemma:averaged-Chebyshev-truncation}
    For any function $f$ that has a Chebyshev expansion, consider the degree-$d$ averaged Chebyshev truncation $\hat{P}_{d'}$ defined in \Cref{eq:averaged-Chebyshev-truncation}. 
    Let $\varepsilon_d(f)$ be the truncation error corresponding to the degree-$d$ best uniform approximation on $[-1,1]$ to $f$. If there exists a degree-$d$ polynomial $P^*_d\in\bbR[x]$ such that $\max_{x\in[-1,1]} |f(x)-P^*_d(x)| \leq \epsilon$, then
    \[ \max_{x\in[-1,1]} \big| f(x) - \hat{P}_{d'}(x) \big| \leq 4 \varepsilon_d(f) \leq 4 \max_{x\in[-1,1]} |f(x)-P^*_d(x)| \leq 4\epsilon. \]
\end{lemma}

\vspace{1em}
Lastly, since the $\ell_1$ norm of the coefficient vector corresponding to the polynomial approximation plays a key role in our space-efficient QSVT (\Cref{sec:space-efficient-QSVT}), we provide upper bounds for the coefficient vector $\hat{\bfc} \coloneqq (\hat{c}_0,\cdots,\hat{c}_{d'})$ in \Cref{lemma:coefficient-vector-norm-bound}. Interestingly, Chebyshev coefficients $c_k$ (and so do $\hat{c}_k$) decay a bit faster if the function $f$ becomes a bit smoother.

\begin{lemma}[$\ell_1$-norm bounds on the averaged truncated Chebyshev coefficient vector]
\label{lemma:coefficient-vector-norm-bound}
For any function $f$ that admits a Chebyshev expansion and is bounded with $\max_{x\in[-1,1]} |f(x)| \leq B_0$ for some constant $B_0>0$, we have the following $\ell_1$-norm bounds for the coefficient vector $\hat\bfc$ that corresponds to the degree-$d$ averaged Chebyshev truncation $\widehat P_{d'}$ with $d' = 2d-1$:
\begin{enumerate}[label={\upshape(\roman*)},itemsep=0.33em,topsep=0.33em,parsep=0.33em]
    \item \label{thmitem:l1-norm-bound-general}For any such function $f$ satisfying our conditions, we have
    $\|\hat\bfc\|_1 \le O(B_0\sqrt{d})$;
    \item \label{thmitem:l1-norm-bound-first-conti}If $f(\cos\theta)$ is absolutely continuous on $[-\pi,0]$ and     its angular first derivative satisfies $\int_{-\pi}^{0} \abs*{ \frac{\dd}{\dd\theta}f(\cos\theta) }\dd\theta \leq B_1$, then $\|\hat\bfc\|_1 \leq O(B_0+B_1\log d)$;
    \item \label{thmitem:l1-norm-bound-twice-conti}If $f$ is additionally (at least) twice continuously differentiable and its angular second derivative satisfies $\int_{-\pi}^{0} \abs*{ \frac{\dd^2}{\dd\theta^2}f(\cos\theta) } \dd\theta \leq B_2$, then $\|\hat\bfc\|_1 \leq O(B_0+B_2)$.
\end{enumerate}
\end{lemma}

\begin{proof}
Let $g(\theta) \coloneqq f(\cos\theta)$, and recall that $c_k=\frac{2}{\pi}\int_{-\pi}^{0} g(\theta)\cos(k\theta)\,d\theta$. The $\ell_1$ norm of the coefficient vector $\hat{\bfc}$ can be expressed as
\begin{equation}
    \label{eq:coeff-vec-l1-norm}
    \norm{\hat{\bfc}}_1 = \abs{c_0} + \sum_{k=1}^d \abs{c_k} + \sum_{k=d+1}^{2d-1} \frac{2d-k}{d} \abs{c_k} \coloneqq \abs{c_0} + \sum_{k=1}^{2d-1} w_k \abs{c_k}.
\end{equation}

We begin by bounding the zeroth coefficient:
\begin{equation}
    \label{eq:zeroth-coeff-upper-bound}
    \abs{c_0} = \abs*{ \frac{2}{\pi}\int_{-\pi}^{0} g(\theta)\,\dd\theta }
    \leq \frac{2}{\pi}\int_{-\pi}^{0}|g(\theta)|\,\dd\theta 
    \leq 2B_0.
\end{equation}

\parheading{\Cref{thmitem:l1-norm-bound-general}: the general case.} 
By Parseval's identity (e.g.,~\cite[Theorem 1.3(ii)]{SS03}) for the cosine Fourier coefficients of $g$, we obtain 
\begin{equation}
    \label{eq:nonzeroth-coeff-general}
    \sum_{k=1}^{\infty}\abs{c_k}^2 \leq \frac{\abs{c_0}^2}{2}+\sum_{k=1}^{\infty}\abs{c_k}^2 = \frac{2}{\pi}\int_{-\pi}^{0}\abs{g(\theta)}^2\,\dd\theta \leq 2B_0^2.
\end{equation}
Here, the last inequality uses $\abs{g(\theta)}\leq B_0$. 

Applying the Cauchy--Schwarz inequality to \Cref{eq:coeff-vec-l1-norm}, it follows that
\begin{subequations}
\label{eq:coeff-general-bound}
\begin{align}
    \norm{\hat{\bfc}}_1 &= \abs{c_0} + \sum_{k=1}^{2d-1} w_k \abs{c_k} \\
    &\leq 2B_0 + \rbra*{ \sum_{k=1}^{2d-1} w_k^2 }^{1/2} \rbra*{ \sum_{k=1}^{2d-1}|c_k|^2 }^{1/2}\\
    &\leq 2B_0 + \rbra*{ \frac{4d}{3} }^{1/2} \sqrt{2} B_0 \leq O(B_0\sqrt{d}).
\end{align}
\end{subequations}
Here, the second line follows from \Cref{eq:zeroth-coeff-upper-bound}, and the third line follows from \Cref{eq:nonzeroth-coeff-general} and the fact that
\[ \sum_{k=1}^{2d-1} w_k^2 = \sum_{k=1}^{d}1+\sum_{k=d+1}^{2d-1}\rbra*{ \frac{2d-k}{d} }^2 = d+\frac{1}{d^2}\sum_{j=1}^{d-1}j^2 \leq \frac{4d}{3}. \]

\parheading{\Cref{thmitem:l1-norm-bound-first-conti}: absolutely continuous with bounded angular first derivative.} 
For every $k\geq 1$, integration by parts gives
\begin{subequations}
\label{eq:coeff-first-conti}
\begin{align}
    \abs{c_k} &= \abs*{\frac{2}{\pi}\int_{-\pi}^{0} g(\theta)\cos(k\theta)\,\dd\theta}\\
    &= \abs*{\frac{2}{\pi} \cdot \frac{g(\theta)\sin(k\theta)}{k} \bigg|_{-\pi}^{0} - \frac{2}{\pi k}\int_{-\pi}^{0} g'(\theta)\sin(k\theta)\,\dd\theta}\\
    &\leq \frac{2}{\pi k}\int_{-\pi}^{0}|g'(\theta)|\,\dd\theta\\
    &\leq \frac{2B_1}{\pi k}.
\end{align}
\end{subequations}
Here, the third line follows from $ \frac{g(\theta)\sin(k\theta)}{k} \big|_{-\pi}^{0} = 0$ as $\sin(0)=\sin(-k\pi)=0$, and the last line uses the condition $\int_{-\pi}^{0}|g'(\theta)|\,\dd\theta\leq B_1$. 

Plugging \Cref{eq:zeroth-coeff-upper-bound,eq:coeff-first-conti} into \Cref{eq:coeff-vec-l1-norm}, we conclude that 
\[ \norm{\hat{\bfc}}_1 = \abs{c_0} + \sum_{k=1}^{2d-1} w_k \abs{c_k} 
\leq \abs{c_0} + \sum_{k=1}^{2d-1}  \abs{c_k} \leq 2B_0 + \frac{2B_1}{\pi}\sum_{k=1}^{2d-1}\frac{1}{k} \leq O(B_0 + B_1 \log{d}). \]
Here, the first inequality uses the fact that $0 \leq w_k \leq 1$, and the last inequality follows from the Euler--Maclaurin formula. 

\parheading{\Cref{thmitem:l1-norm-bound-twice-conti}: twice continuously differentiable with bounded angular second derivative.} 
Analogously to \Cref{eq:coeff-first-conti}, integrating by
parts once more, for any $k\ge 1$, it follows that
\begin{subequations}
\label{eq:coeff-twice-conti}
\begin{align}
    \abs{c_k} &= \abs*{- \frac{2}{\pi k}\int_{-\pi}^{0} g'(\theta)\sin(k\theta)\,\dd\theta}\\
    &= \frac{2}{\pi k} \abs*{ -\frac{g'(\theta)\cos(k\theta)}{k} \bigg|_{-\pi}^{0} +\frac{1}{k}\int_{-\pi}^{0} g''(\theta)\cos(k\theta)\,\dd\theta }\\
    &\leq \frac{2}{\pi k^2}\int_{-\pi}^{0}|g''(\theta)|\,\dd\theta\\
    &\leq \frac{2B_2}{\pi k^2}.
\end{align}
\end{subequations}
Here, the boundary term vanishes in the third line because $g'(-\pi)=g'(0)=0$, as $g'(\theta)=-f'(\cos\theta)\sin\theta$, and the last line uses the condition $\int_{-\pi}^{0}\left|g''(\theta)\right|\,\dd\theta \leq B_2$.

Plugging \Cref{eq:zeroth-coeff-upper-bound,eq:coeff-twice-conti} into \Cref{eq:coeff-vec-l1-norm}, we conclude that 
\[ \norm{\hat{\bfc}}_1 \leq \abs{c_0} + \sum_{k=1}^{2d-1}  \abs{c_k} \leq 2B_0 + \frac{2B_2}{\pi}\sum_{k=1}^{2d-1}\frac{1}{k^2} \leq O(B_0 + B_2). \]
Here, the last inequality also follows from the Euler--Maclaurin formula. 
\end{proof}

\subsection{Tools for space-bounded randomized and quantum algorithms}

Our convention assumes that for any algorithm $\calA$ in bounded-error randomized time $t(n)$ and space $s(n)$, $\calA$ outputs the correct value with probability at least $2/3$  (viewed as ``success probability''). 
We first proceed with space-efficient success probability estimation.

\begin{lemma}[Space-efficient success probability estimation by sequential repetitions]
\label{lemma:success-probability-estimation}
Let $\calA$ be a randomized (resp., quantum) algorithm that outputs the correct value with probability $p$, has time complexity $t(n)$, and space complexity $s(n)$. We can obtain an additive-error estimation $\hat{p}$ such that $|p - \hat{p}| \leq \epsilon$, where $\epsilon \geq 2^{-O(s(n))}$. Moreover, this estimation can be computed in bounded-error randomized (resp., quantum) time $O(\epsilon^{-2} t(n))$ and space $O(s(n))$. 
\end{lemma}

\begin{proof}
Consider a $m$-time sequential repetition of the algorithm $\calA$, and let $X_i$ be a random variable indicating whether the $i$-th repetition succeeds, then we obtain a random variable $X = \frac{1}{m} \sum_{i=1}^m X_i$ such that $\bbE[X]=p$.
Now let $\hat{X}=\frac{1}{m} \sum_{i=1}^m \hat{X}_i$ be the additive-error estimation, where $\hat{X}_i$ is the outcome of $\calA$ in the $i$-th repetition. By the Chernoff--Hoeffding bound (e.g.,~\cite[Theorem 4.12]{MU17}), we know that 
\[\Pr{|\hat{X}-p|\geq \epsilon} \leq 2\exp(-2m\epsilon^2).\] 
By choosing $m=\ceil*{2/\epsilon^{2}}$, this choice of $m$ ensures that this procedure based on $\calA$ succeeds with probability at least $2/3$.

Furthermore, the space complexity of our algorithm is $O(\log{m}) = O(\log{1/\epsilon}) = O(s(n))$ since we can simply reuse the workspace. Also, the time complexity is $m\cdot t(n) = O(\epsilon^{-2} t(n))$ as desired.
\end{proof}

Notably, when applying \Cref{lemma:success-probability-estimation} to a quantum algorithm, we introduce intermediate measurements to retain space complexity through reusing working qubits. While space-efficient success probability estimation without intermediate measurements is possible,\footnote{Fefferman and Lin~\cite{FL18} noticed that one can achieve space-efficient success probability estimation for quantum algorithms without intermediate measurements via quantum amplitude estimation~\cite{BHMT02}. } we will use \Cref{lemma:success-probability-estimation} for convenience, given that $\BQL=\BQUL$~\cite{FR21}.

\vspace{1em}
The SWAP test was originally proposed for pure states in~\cite{BCWdW01}. Subsequently, in~\cite{KMY09}, it was demonstrated that the SWAP test can also be applied to mixed states. 

\begin{lemma}[SWAP test for mixed states, adapted from~{\cite[Proposition 9]{KMY09}}]
\label{lemma:swap-test}
    Let $\rho_0$ and $\rho_1$ be two quantum states, which may be mixed. 
    There exists a $(2n+1)$-qubit quantum circuit that outputs $0$ with probability $\frac{1+\Tr(\rho_0\rho_1)}{2}$, using a single query to each corresponding state-preparation circuit $Q_0$ and $Q_1$, and $O(n)$ one- and two-qubit quantum gates. 
\end{lemma}

A matrix $B$ is said to be \textit{sub-stochastic} if all its entries are non-negative and the sum of entries in each row (respectively, column) is strictly less than $1$. Moreover, a matrix $B$ is \textit{row-stochastic} if all its entries are non-negative and the sum of entries in each row is equal to $1$.

\begin{lemma}[Sub-stochastic matrix powering in bounded space]
\label{lemma:substochastic-matrix-powering}
Let $B$ be an $l \times l$ upper-triangular sub-stochastic matrix, where each entry of $B$ requires at most $\ell$-bit precision. Then, there exists an explicit randomized algorithm that computes the matrix power $B^k[s,t]$ in $O\rbra{\log{l} + \log{k}}$ space and $O(\ell k)$ time. Specifically, the algorithm accepts with probability $B^k[s,t]$. 
\end{lemma}

\begin{proof}
Our randomized algorithm leverages the equivalence between space-bounded randomized computation and Markov chains, see~\cite[Section 2.4]{Saks96} for a detailed introduction. 

First, we construct a row-stochastic matrix $\hat{B}$ from $B$ by adding an additional column and row. Let $\hat{B}[i,j]$ denote the entry at the $i$-th column and the $j$-th row of $\hat{B}$. Specifically, 
\[\hat{B}[i,j] \coloneqq \begin{cases}
B[i,j], & \text{ if } 1 \leq i,j \leq l;\\
1-\sum_{s=j}^{l} B[s,j],& \text{ if } i=l+1 \text{ and } 1\leq j \leq l+1;\\
0, & \text{ if } 1 \leq i \leq l \text{ and } j=l+1.\\
\end{cases}\]

Next, we view $\hat{B}$ as a transition matrix of a Markov chain since $\hat{B}$ is row-stochastic. We consequently have a random walk on the directed graph $G=(V,E)$ where $V=\{1,2,\cdots,l\}\cup \{\perp\}$ and $(u,v)\in E$ iff $\hat{B}(u,v)>0$. 
In particular, the probability that a $k$-step random walk starting at node $s$ and ending at node $t$ is exactly $\hat{B}^k[s,t]=B^k[s,t]$. This is because the walker who visits the dummy node $\perp$ will not reach other nodes. 

Finally, note that $\hat{B}$ is a $(l+1)\times (l+1)$ matrix, the matrix powering of $\hat{B}^k$ can be computed in $O(\log{l}+\log{k})$ space. 
In addition, the overall time complexity is $O(\ell k)$ since we simulate the dyadic rationals (with $\ell$-bit precision) of a single transition exactly by $\ell$ coin flips. 
\end{proof}

\section{Space-efficient quantum singular value transformations}
\label{sec:space-efficient-QSVT}

We begin by defining the \textit{projected unitary encoding} and its special forms, viz. the bitstring indexed encoding and the block-encoding. 
\begin{definition}[Projected unitary encoding and its special forms, adapted from~\cite{GSLW19}] 
    \label{def:bitstring-indexed-encoding}
    Let $U$ be an $(\alpha, a, \epsilon)$-projected unitary encoding of a linear operator $A$ if $\|A- \alpha \tilde{\Pi} U \Pi\| \leq \epsilon$, where $U$ and orthogonal projections $\tilde{\Pi}$ and $\Pi$ act on $s+a$ qubits, and both $\rank(\tilde{\Pi})$ and $\rank(\Pi)$ are at least $2^a$ ($a$ is viewed as the number of ancillary qubits).
Furthermore, we are interested in two special forms of the projected unitary encoding: 
\begin{itemize}[itemsep=0.33em,topsep=0.33em,parsep=0.33em]
    \item \textbf{Bitstring indexed encoding.} We say that a projected unitary encoding is a \textit{bitstring indexed encoding} if both orthogonal projections $\tilde{\Pi}$ and $\Pi$ span on $\tilde{S},S \subseteq \{\ket{0},\ket{1}\}^{\otimes(a+s)}$, respectively.\footnote{Typically, to ensure these orthogonal projections coincide with space-bounded quantum computation, we additionally require the corresponding subsets $\tilde{S}$ and $S$ admit space-efficient set membership, namely deciding the membership of these subsets is in deterministic $O(s+a)$ space.} 
    In particular, for any $\ket{\tilde{s_i}} \in \tilde{S}$ and $\ket{s_j} \in S$, we have a matrix representation $A_{\tilde{S},S}(i,j) \coloneqq \bra{\tilde{s}_i} U \ket{s_j}$ of $A$. 
    \item \textbf{Block encoding.} We say that a projected unitary encoding is a block-encoding if both orthogonal projections are of the form $\Pi=\tilde{\Pi}=\ket{0}\bra{0}^{\otimes a}\otimes I_s$. We use the shorthand $A=(\bra{\bar{0}}\otimes I_s) U (\ket{\bar{0}}\otimes I_s)$ for convenience. 
\end{itemize} 
\end{definition}
See \Cref{subsec:singular-value-decomp-and-trans} for definitions of singular value decomposition and transformation. 
With these definitions in place, we present the main (informal) theorem in this section:
\begin{theorem}[Space-efficient QSVT]
    \label{thm:space-efficient-QSVT}
    Let $f\colon\bbR \rightarrow \bbR$ be a continuous function bounded on the closed interval of interest $\calI \subseteq [-1,1]$. If there exists a degree-$d$ polynomial $P^*_d$ that approximates $h\colon[-1,1] \rightarrow \bbR$, where $h$ approximates $f$ only on $\calI$ with additive error at most $\epsilon$, such that $\max_{x \in [-1,1]} |h(x)-P^*_d(x)| \leq \epsilon$, then degree-$d$ averaged Chebyshev truncation yields another degree-$d'$ polynomial $P_{d'}$, with $d'=2d-1$, satisfying the following conditions: 
   \[ \max_{x \in \calI} |f(x)-P_{d'}(x)| \leq O(\epsilon) \quad \text{and} \quad \max_{x \in [-1,1]} |P_{d'}(x)| \leq 1. \]
    \noindent Moreover, there is a space-efficient classical algorithm for computing any entry in the coefficient vector $\hat{\bfc}$ of the averaged Chebyshev truncation polynomial $P_{d'}$\emph{:} 
    \begin{itemize}[itemsep=0.33em,topsep=0.33em,parsep=0.33em]
        \item If $f$ is a continuously bounded function with $\max_{x\in[-1,1]} |f''(x)| \leq \poly(d)$,\footnote{This conclusion also applies to a linear combination of bounded functions, provided that the coefficients are bounded and can be computed deterministically and space-efficiently.} then any entry in the coefficient vector $\hat{\bfc}$ can be computed in deterministic $O(\log{d})$ space; 
        \item If $f$ is a piecewise-smooth function, then any entry in the coefficient vector $\hat{\bfc}$ can be computed in bounded-error randomized $O(\log{d})$ space.
    \end{itemize}
    \noindent Furthermore, for any $(1,a,0)$-bitstring indexed encoding $U$ of $A=\tilde{\Pi} U \Pi$, acting on $s+a$ qubits where $a(n) \leq s(n)$, and any $P_{d'}$ with $d'\leq 2^{O(s(n))}$, we can implement an 
    $(\alpha, a+\log d+O(1), \epsilon_{\alpha})$-bitstring indexed encoding of the quantum singular value transformation $P_{d'}^{\SV}(A)$ that acts on $O(s(n))$ qubits using $O(d^2\eta_{\alpha})$ queries to $U$, where $\epsilon_{\alpha}$ is specified in \Cref{thm:LCU-averaged-chebyshev-truncation}. 
    Here, $\alpha = \|\hat\bfc\|_1$ with $\eta_\alpha = 1$ in general, and particularly $\alpha = 1$ with $\eta_{\alpha} = \|\hat\bfc\|_1$ if $P^{\SV}_{d'}(A)$ is a partial isometry.
    It is noteworthy that $\|\hat{\bfc}\|_1$ is bounded by $O(\log{d})$ in general, and can be improved to a constant bound for twice continuously differentiable functions.
\end{theorem}

We remark that we can apply \Cref{thm:space-efficient-QSVT} to general forms of the projected unitary encoding $U$ with orthogonal projections $\Pi$ and $\tilde{\Pi}$, as long as such an encoding meets the conditions: (1) The basis of $\Pi$ and $\tilde{\Pi}$ admits a well-defined order; (2) Both controlled-$\Pi$ and controlled-$\tilde{\Pi}$ admit computationally efficient implementation. We note that bitstring indexed encoding defined in \Cref{def:bitstring-indexed-encoding} trivially meets the first condition, and a sufficient condition for the second condition is that the corresponding subsets $S$ and $\tilde{S}$ have space-efficient set membership. 

\vspace{1.5em}
Next, we highlight the main technical contributions leading to our space-efficient quantum singular value transformations (\Cref{thm:space-efficient-QSVT}). To approximately implement a space-efficient QSVT $f^{\SV}(A)$, we require \textit{the pre-processing} to find a space-efficient polynomial approximation $P^{(f)}_{d'} \approx f$ on $\calI$. These polynomial approximations are detailed in \Cref{subsec:coefficients-approx}:
\begin{itemize}[itemsep=0.33em,topsep=0.33em,parsep=0.33em,leftmargin=2em]
    \item We provide deterministic space-efficient polynomial approximations for \textit{continuously bounded} functions (\Cref{lemma:space-efficient-bounded-funcs}) using averaged Chebyshev truncation (see \Cref{subsec:Chebyshev-polys-and-truncated-expansion}), including the sign function (\Cref{corr:space-efficient-sign}). 
    \item We present bounded-error randomized space-efficient polynomial approximations for \textit{piecewise-smooth} functions (\Cref{thm:space-efficient-smooth-funcs}), such as the normalized logarithmic function (\Cref{corr:space-efficient-log}). To achieve this, we adapt the time-efficient technique in~\cite[Lemma 37]{vAGGdW17} to the space-efficient scenario by leveraging space-efficient random walks (\Cref{lemma:space-efficient-low-weight-approx}).
\end{itemize}

With an appropriate polynomial approximation $P_{d'}^{(f)}$, we can implement the space-efficient QSVT $P_{f,d'}^{\SV}(A)$, as established in \Cref{subsec:poly-applied-to-unitary-encodings} (specifically \Cref{thm:LCU-averaged-chebyshev-truncation}). It is worth noting that a space-efficient QSVT for Chebyshev polynomials is implicitly shown in~\cite{GSLW19}, as stated in \Cref{lemma:Chebyshev-poly-implementation}. We establish \Cref{thm:LCU-averaged-chebyshev-truncation} by combining this result with the LCU technique (\Cref{lemma:space-efficient-LCU}) and the renormalization procedure (\Cref{lemma:renormalizing-encodings}, if necessary and applicable). 

In addition to these general techniques, we provide explicit space-efficient QSVT examples in \Cref{subsec:space-efficient-QSVT-examples}, including those for the sign function (\Cref{corr:sign-polynomial-implementation}) and the normalized logarithmic function (\Cref{corr:log-polynomial-implementation}). 
Notably, the former leads to a simple proof of space-efficient error reduction for unitary quantum computations (\Cref{subsec:BQUL-error-reduction}). 

\subsection{Space-efficient bounded polynomial approximations}
\label{subsec:coefficients-approx}

% Corollary 66 in~\cite{GSLW18}
We provide a systematic approach for constructing \textit{space-efficient} polynomial approximations of real-valued piecewise-smooth functions, which is a space-efficient counterpart of~\cite[Corollary 23]{GSLW19}. Notably, our algorithm (\Cref{lemma:space-efficient-bounded-funcs}) is \textit{deterministic} for continuous functions that are bounded on the interval $[-1,1]$. However, for general piecewise-smooth functions, we only introduce a \textit{randomized} algorithm (\Cref{thm:space-efficient-smooth-funcs}). In addition, please refer to \Cref{subsec:Chebyshev-polys-and-truncated-expansion} as a brief introduction to Chebyshev polynomial and (averaged) Chebyshev truncation. 

\subsubsection{Continuously bounded functions}

We design a space-efficient algorithm for computing the coefficients of a polynomial approximation with high accuracy for continuously bounded functions. Our approach leverages the averaged Chebyshev truncation, specifically \textit{the de La Vall\'ee Poussin partial sum}, in conjunction with numerical integration, namely \textit{the composite trapezium rule}. 

\begin{lemma}[Space-efficient polynomial approximations for bounded functions]
\label{lemma:space-efficient-bounded-funcs}
For any continuous function $f$ that is bounded with $\max_{x\in[-1,1]} |f(x)| \leq B$ for some known constant $B>0$. Let $P^*_{f,d}$ be a degree-$d$ polynomial with the same parity as $f$ satisfying $\max_{x\in[-1,1]} \!|f(x)-P^*_{f,d}  \!(x)| \leq \epsilon$. By employing the degree-$d$ averaged Chebyshev truncation, we can obtain a degree-$d'$ polynomial $P^{(f)}_{d'}$ that has the same parity as $P^*_{f,d} $ and satisfies $\max_{x\in[-1,1]} |f(x)-P^{(f)}_{d'}(x)| \leq 4\epsilon$.\footnote{It is noteworthy that for any even function $f$, the degree of $P^{(f)}_{d'}$ is $2d-2$ rather than $2d-1$. Nevertheless, for the sake of convenience, we continue to choose $d'=2d-1$.} This polynomial $P^{(f)}_{d'}$ is defined as a linear combination of Chebyshev polynomials $T_k(\cos \theta) = \cos(k \theta)$ with $d'=2d-1$ and the integrand $F_k(\theta)\coloneqq\cos(k\theta)f(\cos{\theta})$: 
\begin{equation}
    \label{eq:averaged-truncated-Chebyshe-expansion-theta}
    P^{(f)}_{d'} = \frac{\hat{c}_0}{2} + \sum_{k=1}^{d'} \hat{c}_k T_k, \text{ where } c_k = \frac{2}{\pi} \int_{-\pi}^0 F_k(\theta) \dd\theta \text{ and } \hat{c}_k=\begin{cases}
        c_k ,& 0 \leq k \leq d\\
        \frac{2d-k}{d} c_k,& k>d
    \end{cases}.
\end{equation}
If the integrand $F_k(\theta)$ satisfies $\max_{\xi\in [-\pi,0]} |F_k''(\xi)| \leq O(d^\gamma)$ for every $0\leq k\leq d'$ and some constant $\gamma$, then any entry of the coefficient vector $\hat{\bfc}=(\hat{c}_0,\cdots,\hat{c}_{d'})$, up to additive error $\epsilon$ for $\|\hat{\bfc}\|_1$, can be computed in deterministic time $O(d^{(\gamma+1)/2}\epsilon^{-1/2} t(\ell))$ and space $O(\log(d^{(\gamma+3)/2} \epsilon^{-3/2}B))$, where $\ell=O(\log(d^{(\gamma+3)/2}\epsilon^{-3/2}))$ and evaluating $F_k(\theta)$ in $\ell$-bit precision is in deterministic time $t(\ell)$ and space $O(\ell)$. Furthermore, the coefficient vector $\hat{\bfc}$ has the following $\ell_1$ norm bounds:
 \begin{enumerate}[label={\upshape(\roman*)},itemsep=0.33em,topsep=0.33em,parsep=0.33em]
    \item \label{thmitem:polyApprox-l1-norm-bound-general}For any function $f$ satisfying our conditions, we have $\|\hat{\bfc}\|_1 \leq O(B\sqrt{d})$; 
    \item \label{thmitem:polyApprox-l1-norm-bound-first-conti}If $f(\cos\theta)$ is absolutely continuous on $[-\pi,0]$ and satisfies $\int_{-\pi}^0 \abs*{ \frac{\dd}{\dd\theta} f(\cos\theta) } \dd\theta \leq O(B)$, then $\|\hat{\bfc}\|_1 \leq O(B\log d)$;
    \item \label{thmitem:polyApprox-l1-norm-bound-twice-conti}If $f$ is additionally twice continuously differentiable and satisfies $\int_{-\pi}^0 \abs*{ \frac{\dd^2}{\dd\theta^2} f(\cos\theta) } \dd\theta \leq O(B)$, then $\|\hat{\bfc}\|_1 \leq O(B)$.
\end{enumerate}
\end{lemma}

\begin{proof}
We begin with the polynomial approximation $P^{(f)}_{d'}$ obtained from the degree-$d$ averaged Chebyshev truncation expressed in \Cref{eq:averaged-truncated-Chebyshe-expansion-theta}. The degree $d'$ is $2d-1$ if $f$ is odd, and $2d-2$ if $f$ is even. To bound the truncation error of $P^{(f)}_{d'}$, we require a degree-$d$ polynomial $P^*_{f,d} $ such that $\max_{x\in[-1,1]} |f(x)-P^*_{f,d} (x)| \leq \epsilon$. By utilizing \Cref{lemma:averaged-Chebyshev-truncation}, we obtain the desired error bound $\max_{x\in[-1,1]} |f(x)-P^{(f)}_{d'}(x)| \leq 4\epsilon$.

\paragraph*{Computing the coefficients.}
To compute the coefficients $\hat{c}_k$ for $0 \leq k \leq d'$, it suffices to compute the Chebyshev coefficients $c_k$ for $0 \leq k \leq 2d-1$. Noting that $c_k = \frac{2}{\pi}\int_{-\pi}^{0} F_k(\theta) \dtheta$ where $F_k(\theta)\coloneqq\cos(k \theta) f(\cos{\theta})$, we can estimate the numerical integration using the composite trapezium rule, e.g.,~\cite[Section 7.5]{SM03}. 
The application of this method yields the following:
\begin{subequations}
    \label{eq:integration-summations}
    \begin{align}
        \int_{-\pi}^{0} \!F_k(\theta) \dd\theta &\approx \frac{\pi}{m} \rbra*{ \frac{F_k(\theta_0)}{2} + \sum_{l=1}^{m-1} F_k(\theta_l) + \frac{F_k(\theta_m)}{2} },\\
        \text{ where } \theta_l &\coloneqq \frac{\pi l}{m}-\pi \text{ for } l=0,1,\cdots,m.
    \end{align}
\end{subequations}
Moreover, we know the upper bound on the numerical errors for computing the coefficient $c_k$:
\begin{equation}
    \label{eq:integration-errors}
    \varepsilon_{d',k}^{(f)} \coloneqq \sum_{l=1}^m \abs*{  \int_{\theta_{l-1}}^{\theta_l}  F_k(\theta)\dd\theta - \frac{\pi}{2m} \cdot \left( F_k(\theta_{l-1}) + F_k(\theta_l) \right) } \leq \frac{\pi^3}{12m^2} \max_{\xi\in[-\pi,0]} \abs*{ F_k''(\xi) }.
\end{equation}
To obtain an upper bound on the number of intervals $m$, we need to ensure that the error of the numerical integration is within 
\[\varepsilon_{d'}^{(f)} = \sum_{k=0}^d \varepsilon^{(f)}_{d',k} + \sum_{k=d+1}^{d'} \frac{2d-k}{d} \varepsilon^{(f)}_{d',k} \leq \sum_{k=0}^{d'} \varepsilon^{(f)}_{d',k} \leq \epsilon.\]
Plugging the assumption $|F_k''(\xi)| \leq O(d^{\gamma})$ into \Cref{eq:integration-errors}, by choosing an appropriate value of $m=\Theta(\epsilon^{-1/2}d^{(\gamma+1)/2})$, we establish that $\varepsilon^{(f)}_{d'} \leq O(d^{\gamma+1})/m^2 \leq \epsilon$. Moreover, to guarantee that the accumulated error is $O(\epsilon/d)$ in \Cref{eq:integration-summations}, we need to evaluate the integrand $F_k(\theta)$ with $\ell$-bit precision, where $\ell=O(\log{(dm/\epsilon)})=O(\log(\epsilon^{-3/2}d^{(\gamma+3)/2}))$. 

Lastly, the desired $\ell_1$-norm bounds for the coefficient vector $\hat{\bfc}$ follow directly from \Cref{lemma:coefficient-vector-norm-bound} by setting that $B_0\coloneqq B$ and, in \Cref{thmitem:polyApprox-l1-norm-bound-first-conti,thmitem:polyApprox-l1-norm-bound-twice-conti}, using the assumed angular derivative bounds as $B_1=O(B)$ and $B_2=O(B)$.

\paragraph*{Analyzing time and space complexity.}
The presented numerical integration algorithm is deterministic, and therefore, the time complexity for computing the integral is $O(m t(\ell))$, where $t(\ell)$ is the time complexity for evaluating the integrand $F_k(\theta)$ within $2^{-\ell}$ accuracy (i.e., $\ell$-bit precision) in $O(\ell)$ space. 
The space complexity required for computing the numerical integration is the number of bits required to index the integral intervals and represent the resulting coefficients. To be specific, the space complexity is 
\begin{align*}
    \max\big\{O(\log{m}),O(\ell),\log\|\hat{\bfc}\|_{\infty}\big\} 
    &\leq O\big(\max\big\{ \log\big(\epsilon^{-\frac{3}{2}}d^{\frac{\gamma+3}{2}}\big),\log B \big\} \big) \\
    &\leq O\big(\log\big(\epsilon^{-\frac{3}{2}}d^{\frac{\gamma+3}{2}}B\big)\big).
\end{align*}
Here, $\|\hat{\bfc}\|_{\infty}$ satisfies the upper bound
\[\|\hat{\bfc}\|_{\infty} = \max_{0 \leq k \leq d'} \frac{2}{\pi} \abs*{\int_{-\pi}^0 \cos(k\theta) f(\cos{\theta}) \dtheta} \leq  \max_{-\pi \leq \theta \leq 0} O(|f(\cos\theta)|) \leq O(B),\] 
and the last inequality is due to the fact that 
\[\forall A, B > 0, ~\Theta(\max\{\log A, \log B\}) = \Theta(\log(AB)). \qedhere\] 
\end{proof}

It is worth noting that evaluating a large family of functions, called holonomic functions, with $\ell$-bit precision requires only \textit{deterministic} $O(\ell)$ space:

\begin{remark}[Space-efficient evaluation of holonomic functions]
    \label{remark:evaluating-funcs}
    Holonomic functions encompass several commonly used functions,\footnote{For a more detailed introduction, please refer to~\cite[Section 4.9.2]{BZ10}.} such as polynomials, rational functions, sine and cosine functions (but not other trigonometric functions such as tangent or secant), exponential functions, logarithms (to any base), the Gaussian error function, and the normalized binomial coefficients. In~\cite{CGKZ05, Mezzarobba12}, these works have demonstrated that evaluating a holonomic function with $\ell$-bit precision is achievable in deterministic time $\tilde{O}(\ell)$ and space $O(\ell)$. Prior works achieved the same time complexity, but with a space complexity of $O(\ell\log{\ell})$. 
\end{remark}

In addition, we provide an example in \Cref{remark:square-function} that achieves only a logarithmically weaker bound on $\|\hat{\bfc}\|_1$ using \Cref{lemma:space-efficient-bounded-funcs}, whereas a constant norm bound can be achieved by leveraging \Cref{thm:space-efficient-smooth-funcs} for piecewise-smooth functions. 

\begin{remark}[On the norm bound of the square-root function's polynomial approximation]
    \label{remark:square-function}
    We consider a function $\Sqrt_\delta(x)$ that coincides with $\sqrt{x}$ on the interval $[\delta, 1]$.\footnote{Since the second derivative of the square-root function $\sqrt{x}$ is unbounded at $x=0$, we cannot directly apply \Cref{lemma:space-efficient-bounded-funcs} to $\sqrt{x}$.} Specifically, $\Sqrt_\delta(x)$ is defined as $\sqrt{x}$ for $x\geq \delta$, $-\sqrt{-x}$ for $x\leq -\delta$, and $x/\sqrt{\delta}$ for $x\in (-\delta,\delta)$. 
    Since $\Sqrt_\delta(\cos\theta)$ is absolutely continuous and has angular first derivative with bounded $L_1$ norm, \Cref{lemma:space-efficient-bounded-funcs} gives $\norm{\hat{\bfc}}_1 \leq O(\log d)$; however, the angular second-derivative condition with an $O(1)$ bound need not hold for this direct construction.
\end{remark}

\vspace{1em}
We now present an example of bounded functions, specifically the sign function. 
\begin{corollary}[Space-efficient approximation to the sign function]
    \label{corr:space-efficient-sign}
    For any $\delta > 0$ and $\epsilon > 0$, there exists an explicit odd polynomial $P_{d'}^{\sign}(x)=\hat{c}_0/2+\sum_{k=1}^{d'} \hat{c}_k T_k(x) \in \bbR[x]$ of degree $d'\leq\tilde C_{\sign
    }\delta^{-1}\log{\epsilon^{-1}}$, where $d'=2d-1$ and $\tilde C_{\sign}$ is a universal constant. 
    Any entry of the coefficient vector $\hat{\bfc}\coloneqq(\hat{c}_0,\cdots,\hat{c}_{d'})$ can be computed in deterministic time $\tilde{O}\big(\epsilon^{-1/2} d^{2}\big)$ and space $O(\log(\epsilon^{-3/2} d^3))$. Furthermore, the polynomial $P_{d'}^{\sign}$ satisfies the following conditions: 
    \begin{align*}
    \forall x \in [-1,1] \setminus [-\delta,\delta],& \left|\sign(x)-P^{\sign}_{d'}(x)\right| \leq C_{\sign} \epsilon, \text{ where } C_{\sign}=5;\\
    \forall x \in [-1,1],& \left|P_{d'}^{\sign}(x) \right| \leq 1.
    \end{align*} 
    Additionally, the coefficient vector $\hat{\bfc}$ has a norm bounded by $\|\hat{\bfc}\|_1 \leq \hat{C}_{\sign}\log d'$, where $\hat{C}_{\sign}$ is another universal constant.
    Without loss of generality, we assume that $d'\geq 2$ and that $\hat{C}_{\sign}$ and $\tilde{C}_{\sign}$ are at least $1$.
\end{corollary}

\begin{proof}
We start from a degree-$d$ polynomial $\tilde{P}_d^\sign$ that well-approximates $\sign(x)$:
\begin{proposition}[Polynomial approximation of the sign function, adapted from~{\cite[Lemma 10 and Corollary 4]{LC17}}]
\label{prop:poly-approx-sgn}
For any $\delta > 0$, $x\in \bbR$, $\epsilon\in (0,\sqrt{2e\pi})$. Let $\kappa = \frac{2}{\delta}\log^{1/2}\left(\frac{\sqrt{2}}{\sqrt{\pi}\epsilon}\right)$, then 
\[g_{\delta,\epsilon}(x)\coloneqq\erf(\kappa x) \text{ satisfies that } |g_{\delta,\epsilon}(x)| \leq 1 \text{ and } \max_{|x|\geq \delta/2} \left| g_{\delta,\epsilon}(x) - \sign(x) \right| \leq \epsilon.\] 
Moreover, there is an explicit odd polynomial $\tilde{P}_d^{\sign}\in\bbR[x]$ of degree $d=O(\sqrt{(\kappa^2+\log{\epsilon^{-1}})\log{\epsilon^{-1}}})$ such that $\max_{x\in[-1,1]} \left| \tilde{P}_d^{\sign}(x)-\erf(\kappa x) \right| \leq \epsilon$.
\end{proposition}

By \Cref{prop:poly-approx-sgn}, we obtain a degree-$d$ polynomial $\tilde{P}^\sign_d$ that well approximates the function $\erf(\kappa x)$ where $\kappa=O(\delta^{-1}\sqrt{\log{\epsilon^{-1}}})$.

To utilize \Cref{lemma:space-efficient-bounded-funcs}, it suffices to upper bound the second derivative $\max_{\xi \in [-\pi,0]}|F_k''(\xi)|$ for any $0 \leq k \leq d'$, as specified in \Cref{fact:sign-func-derivative-bound}. The proof is deferred to the end of this part.

\begin{fact}
\label{fact:sign-func-derivative-bound}
Let $F_k(\theta)=\erf(\kappa \cos{\theta}) \cos(k\theta)$, it holds that
\[ \forall k\in\cbra{0,1,\dots, d'}, \quad\max_{\xi \in [-\pi,0]}|F_k''(\xi)| \leq \frac{2}{\sqrt{\pi}} \kappa + k^2 + \frac{4}{\sqrt{\pi}} \kappa^3 + \frac{4}{\sqrt{\pi}} k \kappa.\] 
\end{fact}

Note that both $\kappa$ and $k$ are at most $O(d)$. By \Cref{fact:sign-func-derivative-bound}, we have $\max_{\xi \in [-\pi,0]}|F_k''(\xi)| \leq O(d^3)$ for any $0 \leq k \leq d'$. Utilizing \Cref{lemma:space-efficient-bounded-funcs}, we obtain a polynomial approximation $P_{d'}^{\sign}(x) = \hat{c}_0/2+\sum_{k=1}^{d'} \hat{c}_k T_k(x)$ with a degree of $d'=2d-1 \leq \tilde{C}_{\sign} \delta^{-1}\log{\epsilon^{-1}}$, where $\tilde{C}_{\sign}$ is a universal constant. This polynomial satisfies $\max_{x\in [-1,1]} |\erf(\kappa x)-P_{d'}^{\sign}(x)| \leq 4\epsilon$. Then, we can derive:
\[\max_{x\in[-1,1]\setminus[-\delta,\delta]} |\sign(x)-P^{\sign}_{d'}(x)| \leq \epsilon + \max_{x\in[-1,1]} |\erf(\kappa x)-P^{\sign}_{d'}(x)| \leq C_{\sign}\epsilon, \quad \text{where } C_{\sign}=5.\]

Moreover, to bound the norm $\|\hat{\bfc}\|_1$, it suffices to consider the function $\erf(\kappa x)$ due to \Cref{prop:poly-approx-sgn}. Let $g(\theta)\coloneqq \erf(\kappa\cos\theta)$. We observe that $|g(\theta)|\leq 1$ and $g'(\theta)=-\frac{2\kappa}{\sqrt{\pi}}\sin\theta \cdot e^{-\kappa^2\cos^2\theta}$. Since $\sin\theta\leq 0$ on $[-\pi,0]$, we have $g'(\theta)\geq 0$ on this interval. Hence, it follows that $\int_{-\pi}^0 |g'(\theta)|\,\dd\theta =g(0)-g(-\pi) =\erf(\kappa)-\erf(-\kappa) =2\erf(\kappa)\leq 2$. Therefore, by \Cref{lemma:space-efficient-bounded-funcs}\ref{thmitem:polyApprox-l1-norm-bound-first-conti} and the assumption $d'\geq 2$, we obtain $\|\hat{\bfc}\|_1 \leq \hat{C}_{\sign}\log d'$ for some universal constant $\hat{C}_{\sign}$.

For the complexity of computing coefficients $\{\hat{c}_k\}_{k=1}^{d'}$, note that the evaluation of the integrand $F(\theta)$ requires $\ell$-bit precision, where $\ell=O(\log(\epsilon^{-3/2}d^{3}))$. Following $t(\ell) = \tilde{O}(\ell)$ specified in \Cref{remark:evaluating-funcs}, any entry of the coefficient vector $\hat{\bfc}$ can be computed in deterministic time $O(\epsilon^{-1/2}d^{2} t(\ell)) = \tilde{O}(\epsilon^{-1/2}d^{2})$ and space $O(\log(\epsilon^{-3/2} d^3))$. 

Finally, we note that $\max_{x\in[-1,1]}|P_{d'}^{\sign}(x)| \leq 1+\epsilon$ due to numerical errors in computing the coefficients $\{\hat{c}_k\}_{k=1}^{d'}$. We finish the proof by considering the normalized polynomial $\hat{P}_{d'}^{\sign}/(1+\epsilon)$ rather than $\hat{P}_{d'}$ and adjusting the coefficient vector $\hat{\bfc}$ of $P_{d'}^{\sign}$ accordingly.
\end{proof}

We now give the proof of \Cref{fact:sign-func-derivative-bound} stated above:
\begin{proof}[Proof of \Cref{fact:sign-func-derivative-bound}]
Through a straightforward calculation, we have derived that 
\begin{subequations}
\label{eq:sign-func-derivative-bound}
\begin{align}
    |F_k''(\theta)| =& \frac{2}{\sqrt{\pi}} \left|\kappa \exp(-\kappa^2 \cos^2{\theta}) \cos{\theta} \cos(k \theta) \right|
        + \left|k^2\cos(k \theta) \erf(\kappa \cos(\theta))\right|\\
        &+\frac{4}{\sqrt{\pi}} \left|\kappa^3 \exp(-\kappa^2\cos^2{\theta}) \cos{\theta} \cos(k \theta) \sin^2{\theta}\right|\\
        &+\frac{4}{\sqrt{\pi}} \left|k \kappa \exp(-\kappa^2 \cos^2{\theta})\sin{\theta}\sin(k \theta)\right|\\
    \leq& \frac{2}{\sqrt{\pi}} \kappa + k^2+\frac{4}{\sqrt{\pi}}\kappa^3+\frac{4}{\sqrt{\pi}}k\kappa.
\end{align}
\end{subequations}
The last line owes to the facts that $|\erf(x)|\leq 1$, $\exp(-x^2) \leq 1$, $|\sin{x}|\leq 1$, and $|\cos{x}|\leq 1$ for any $x$. We complete the proof by noting that \Cref{eq:sign-func-derivative-bound} holds for any $0 \leq k \leq d'$. 
\end{proof}

\subsubsection{Piecewise-smooth functions}
\label{subsubsec:piecewise-smooth}

% Corollary 66 in~\cite{GSLW18}
We present a randomized algorithm for constructing bounded polynomial approximations of piecewise-smooth functions, offering a \textit{space-efficient} alternative to~\cite[Corollary 23]{GSLW19}, as described in \Cref{thm:space-efficient-smooth-funcs}. Our algorithm leverages \Cref{lemma:space-efficient-bounded-funcs,lemma:space-efficient-low-weight-approx}. 

Since this subsection mostly focuses on polynomial approximations, we introduce some notation for convenience. For a function $f\colon \calI \rightarrow \bbR$ and an interval $\calI' \subseteq \calI$, we define $\|f\|_{\calI'} \coloneqq \sup\{|f(x)| \colon x \in \calI'\}$ to denote the supremum of the function $f$ on the interval $\calI'$.

\begin{theorem}[Taylor series-based space-efficient bounded polynomial approximations]
    \label{thm:space-efficient-smooth-funcs}
    Consider a real-valued function $f\colon[x_0-r-\delta, x_0+r+\delta] \rightarrow \bbR$ such that $f(x_0+x)=\sum_{l=0}^{\infty} a_l x^l$ for all $x \in [-r-\delta, r+\delta]$, where $x_0 \in [-1,1]$, $r\in(0,2]$, $\delta \in (0,r]$. Assume that $\sum_{l=0}^{\infty} (r+\delta)^l |a_l| \leq B$ where $B > 0$. Let $\epsilon \in (0,\frac{1}{2B}]$ such that $B > \epsilon$, then there is a polynomial $P_{d'}\in\bbR[x]$ of degree $d'=2d-1 \leq O(\delta^{-1} \log(\epsilon^{-1} B))$, corresponding to some degree-$d$ averaged Chebyshev truncation, such that any entry of the coefficient vector $\hat{\bfc}$ can be computed in bounded-error randomized time $\tilde{O}(\max\{ (\delta')^{-5} \epsilon^{-2} B^2, d^2\epsilon^{-1/2} \})$ and space $O(\log(d^3(\delta')^{-4}\epsilon^{-3/2}B))$ where $\delta'\coloneqq\frac{\delta}{2(r+\delta)}$, such that
    \begin{align*}
        \|f(x) - P(x)\|_{[x_0-r,x_0+r]} & \leq O(\epsilon),\\
        \|P(x)\|_{[-1,1]} & \leq O(\epsilon) + \|f(x)\|_{[x_0-r-\delta/2,x_0+r+\delta/2]} \leq O(\epsilon) + B,\\
        \|P(x)\|_{[-1,1] \setminus [x_0-r-\delta/2, x_0+r+\delta/2]} & \leq O(\epsilon).
    \end{align*}
    Furthermore, the coefficient vector $\hat{\bfc}$ of $P_{d'}$ has a norm bounded by $\|\hat{\bfc}\|_1 \leq O(B\sqrt d)$.
\end{theorem}

The main ingredient, and the primary challenge, for demonstrating \Cref{thm:space-efficient-smooth-funcs} is to construct a low-weight approximation using Fourier series, as shown in~\cite[Lemma 37]{vAGGdW17}, which requires computing the powers of sub-stochastic matrices in bounded space (\Cref{lemma:substochastic-matrix-powering}). 

\begin{lemma}[Space-efficient low-weight approximation by Fourier series]
\label{lemma:space-efficient-low-weight-approx}
Let $0 < \delta,\epsilon < 1$ and $f\colon\bbR \rightarrow \bbR$ be a real-valued function such that $|f(x)-\sum_{k=0}^K a_k x^k| \leq \epsilon/4$ for all $x \in \calI_{\delta}$, the interval $\calI_{\delta}\coloneqq[-1+\delta,1-\delta]$ and $\|\bfa\|_1 \leq O(\max\{\epsilon^{-1},\delta^{-1}\})$. Then there is a coefficient vector $\bfc\in \bbC^{2M+1}$ such that
\begin{itemize}
    \item For even functions, $\left| f(x) - \sum_{m=-M}^M c^{(\even)}_m \!\cos(\pi x m) \right| \leq \epsilon$ for any $x \in \calI_{\delta}$; 
    \item For odd functions, $\left| f(x) - \sum_{m=-M}^M c^{(\odd)}_m \!\sin\!\left(\pi x \!\left(m\!+\!\frac{1}{2}\right) \right) \right| \leq \epsilon$ for any $x \in \calI_{\delta}$;
    \item Otherwise, $\Big| f(x) - \sum_{m=-M}^M \big(c^{(\even)}_m \!\cos(\pi x m) \!+\! c^{(\odd)}_m \!\sin\!\left(\pi x \!\left(m\!+\!\frac{1}{2}\right) \right)\big) \Big| \leq \epsilon$ for any $x \in \calI_{\delta}$. 
\end{itemize}
Here $M\coloneqq\max\left(2\lceil \delta^{-1} \ln(4\|a\|_1\epsilon^{-1}) \rceil,0\right)$ and $\|\bfc\|_1 \leq \|\bfa\|_1$. Moreover, the coefficient vector $\bfc$ can be computed in bounded-error randomized time $\tilde{O}(\delta^{-5}\epsilon^{-2})$ and space $ O(\log(\delta^{-4}\epsilon^{-1}))$. 
\end{lemma}

\begin{proof}
We begin by noticing that the truncation error of $\sum_{k=0}^K a_k x^k$, as shown in~\cite[Theorem A.4]{SM03}, is $(1-\delta)^{k+1} \leq e^{-\delta(k+1)} \leq \epsilon$, implying that $K \geq \Omega(\delta^{-1} \ln{\epsilon^{-1}})$. 
Without loss of generality, we can assume that $\|\bfa\|_1 \geq \epsilon/2$.\footnote{This is because if $\|\bfa\|_1 < \epsilon/2$, then $\|f\|_{\calI_{\delta}} \leq \|f(x)-\sum_{k=0}^K a_k x^k\|_{\calI_{\delta}} + \|\sum_{k=0}^K a_k x^k \|_{\calI_{\delta}} \leq \epsilon/4+\|\bfa\|_1 < \epsilon$, implying that $M=0$ and $\bfc=0$.} 

\paragraph*{Construction of polynomial approximations.}
Our construction involves three approximations, as described in Lemma 37 of~\cite{vAGGdW17}. We defer the detailed proofs of all three approximations to the end of this subsection. 

The first approximation combines the assumed $\sum_{k=0}^K a_k x^k$ with $\arcsin(x)$'s Taylor series.
\begin{proposition}[First approximation]
    \label{prop:first-low-weight-approx}
    Let $\hat{f}_1(x)\! \coloneqq \!\sum_{k=0}^K a_k x^k$ such that $\|f-\hat{f}_1\|_{\calI_{\delta}} \leq \epsilon/4$. Then we know that $\hat{f}_1(x) = \sum_{k=0}^K a_k \sum_{l=0}^{\infty} b_l^{(k)} \sin^{l}\left(\frac{x \pi}{2}\right)$ where the coefficients $b_l^{(k)}$ satisfy that
    \begin{equation}
    \label{eq:logQSVT-recursive-formula}
    b_l^{(k+1)}=\sum_{l'=0}^l b_{l'}^{(k)} b_{l-l'}^{(1)}, \text{ where } b^{(1)}_l=\begin{cases}
    0 & \text{ if } l \text{ is even,}\\
    \binom{l-1}{ \frac{l-1}{2} } \frac{2^{-l+1}}{l}\cdot\frac{2}{\pi} & \text{ if } l \text{ is odd.}
    \end{cases}
    \end{equation}
    Furthermore, the coefficients $\{b^{(k)}_l\}$ satisfies the following: 
    \begin{enumerate}[label={\upshape(\arabic*)}]
        \item $\|\bfb^{(k)}\|_1=1$ for all $k\geq 1$; 
        \item $\bfb^{(k)}$ is entry-wise non-negative for all $k \geq 1$; 
        \item $b^{(k)}_l=0$ if $l$ and $k$ have different parities. 
    \end{enumerate}
\end{proposition}

The second approximation truncates the series at $l=L$, and bounds the truncation error. 
\begin{proposition}[Second approximation]
    \label{prop:second-low-weight-approx}
    Let $\hat{f}_2(x)\coloneqq\! \sum_{k=0}^K a_k \sum_{l=0}^L b_l^{(k)} \sin^l\!\left(\frac{x\pi}{2}\right)$, where $L\coloneqq\lceil \delta^{-2} \ln(4\|\bfa\|_1 \epsilon^{-1})\rceil$, then we have $\|\hat{f}_1-\hat{f}_2\|_{\calI_{\delta}} \leq \epsilon/4$.  
\end{proposition}

The third approximation approximates the functions $\sin^l(x)$ in $\hat{f}_2(x)$ using a tail bound of the binomial distribution. Notably, this construction not only quadratically improves the dependence on $\delta$, but also ensures that the integrand's second derivative is \textit{bounded} when combined with \Cref{lemma:space-efficient-bounded-funcs}.
\begin{proposition}[Third approximation]
    \label{prop:third-low-weight-approximation}
    Let $\hat{f}_3(x)$ be polynomial approximations of $f$ that depends on the parity of $f$ such that $\|\hat{f}_2\!-\!\hat{f}_3\|_{\calI_{\delta}} \!\leq\! \epsilon/2$ and $M\!=\!\lfloor \delta^{-1}\ln(4\|\bfa\|_1\epsilon^{-1}) \!\rfloor$, then we have 
    \begin{align*}
    \hat{f}_3^{(\even)}(x) &\coloneqq \textstyle\sum\limits_{k=0}^K a_k\! \sum\limits_{\hat{l}=0}^{L/2} (-1)^{\hat{l}} 2^{-2\hat{l}} b_{2\hat{l}}^{(k)} \sum\limits_{m'=\hat{l}-M}^{\hat{l}+M}  (-1)^{m'} \binom{2\hat{l}}{m'} \cos(\pi x (m'-\hat{l})),\\
    \hat{f}_3^{(\odd)}(x) &\coloneqq \textstyle\sum\limits_{k=0}^{K} a_k \sum\limits_{\hat{l}=0}^{(L-1)/2} (-1)^{\hat{l}+1} 2^{-2\hat{l}-1} b_{2\hat{l}+1}^{(k)} \sum\limits_{m'=\hat{l}+1-M}^{\hat{l}+1+M}  (-1)^{m'} \binom{2\hat{l}+1}{m'} \sin\!\big(\pi x \big(m'-\hat{l}-\tfrac{1}{2}\big)\big).
    \end{align*}
    Therefore, we have that $\hat{f}_3(x) \coloneqq \hat{f}_3^{(\even)}(x)$ if $f$ is even, whereas $\hat{f}_3(x) \coloneqq \hat{f}_3^{(\odd)}(x)$ if $f$ is odd. In addition, if $f$ is neither even or odd, then $\hat{f}_3(x)\coloneqq\hat{f}^{(\even)}_3(x) + \hat{f}^{(\odd)}_3(x)$. 
\end{proposition}

We adopt the third approximation as our construction by rearranging the summations and introducing a new parameter $m$. The value of $m$ is defined as $m\coloneqq m'-\hat{l}$ if $f$ is even and $m\coloneqq m'-\hat{l}-1$ if $f$ is odd. Moreover, the definition of $m$ depends on the parity of $l=2\hat{l}+1$\footnote{In particular, the summand in $\hat{f}_3(x)$ is $c_m^{(\even)} \cos(\pi x m) + c_m^{(\odd)} \sin\!\big(\pi x\big( m+\frac{1}{2}\big)\big)$ if $f$ is neither even nor odd. } if $f$ is neither even nor odd. By applying this approach, we can derive the following:
\begin{subequations}
    \label{eq:low-weight-approx-rearranging-terms}
    \begin{align}
        \hat{f}^{(\even)}_3(x) &= \textstyle \sum\limits_{m=-M}^M c^{(\even)}_m \cos(\pi x m),\\
        \text{ where } c^{(\even)}_m &\coloneqq (-1)^m \sum\limits_{k=0}^K a_k \sum\limits_{\hat{l}=0}^{L/2} b^{(k)}_{2\hat{l}} \tbinom{2\hat{l}}{m+\hat{l}} 2^{-2\hat{l}}; \\
        \hat{f}^{(\odd)}_3(x) &= \textstyle \sum\limits_{m=-M}^M c^{(\odd)}_m \sin\big(\pi x \big(m+\tfrac{1}{2}\big)\big),\\ 
        \text{ where } c^{(\odd)}_m &\coloneqq(-1)^m \sum\limits_{k=0}^K a_k \sum\limits_{\hat{l}=0}^{(L-1)/2} b^{(k)}_{2\hat{l}+1} \tbinom{2\hat{l}+1}{m+\hat{l}+1} 2^{-2\hat{l}-1}.
    \end{align}
\end{subequations}

We then notice that the rearrangement of terms in \Cref{eq:low-weight-approx-rearranging-terms} can be directly applied to the definition of $\hat{f}_3(x)$ in \Cref{prop:third-low-weight-approximation}. As a consequence, we obtain the following bound on the accumulative error: 
\[\|f - \hat{f}_3\|_{\calI_{\delta}} \leq \|f - \hat{f}_1\|_{\calI_{\delta}} + \|\hat{f}_1 - \hat{f}_2\|_{\calI_{\delta}} + \|\hat{f}_2 - \hat{f}_3\|_{\calI_{\delta}} \leq \epsilon.\] 
Additionally, we remark that $\|\bfc\|_1 \leq \|\bfa\|_1$, which follows from $\|\bfb^{(k)}\|_1=1$ (see \Cref{prop:first-low-weight-approx}) and $\sum_{m=0}^l \binom{l}{m}=2^l$. 

\paragraph*{Analyzing time and space complexity.}
To evaluate the bounded polynomial approximation $\hat{f}_3(x)$ with $\epsilon$ accuracy, it is necessary to approximate the summand with $\ell$-bit precision, where $\ell=O(\log(KLM\epsilon^{-1}))=O(\log(\delta^{-4}\epsilon^{-1}))$. Since the summand is a product of a constant number of holonomic functions, approximating $b^{(k)}_l$ with $\ell$-bit precision is sufficient. Other quantities in the summand can be evaluated with the desired accuracy in deterministic time $\tilde{O}(\ell)$ and space $O(\ell)$ as stated in \Cref{remark:evaluating-funcs}.

We now present a bounded-error randomized algorithm for estimating $b^{(k)}_l$. As $\bfb^{(1)}$ is entry-wise non-negative and $\sum_{i=1}^l b_i^{(1)} < \|\bfb^{(1)}\|_1=1$ following \Cref{prop:first-low-weight-approx}, we can express the recursive formula in \Cref{eq:logQSVT-recursive-formula} as the matrix powering of a sub-stochastic matrix $B_1$:
\[
B_1^k \coloneqq \begin{pmatrix}
b_1^{(1)} & b_2^{(1)} & \cdots & b_{l-1}^{(1)} & b_l^{(1)}\\
0 & b_1^{(1)} & \cdots & b_{l-2}^{(1)} & b_{l-1}^{(1)}\\
\vdots & \vdots & \ddots & \vdots & \vdots\\
0 & 0 & \cdots & b_{1}^{(1)} & b_{2}^{(1)}\\
0 & 0 & \cdots & 0 & b_{1}^{(1)}\\
\end{pmatrix}^k
=
\begin{pmatrix}
b_1^{(k)} & b_2^{(k)} & \cdots & b_{l-1}^{(k)} & b_l^{(k)}\\
0 & b_1^{(k)} & \cdots & b_{l-2}^{(k)} & b_{l-1}^{(k)}\\
\vdots & \vdots & \ddots & \vdots & \vdots\\
0 & 0 & \cdots & b_{1}^{(k)} & b_{2}^{(k)}\\
0 & 0 & \cdots & 0 & b_{1}^{(k)}\\
\end{pmatrix}
\coloneqq B_k.
\]

In addition, we approximate the sub-stochastic matrix $B_1$ by dyadic rationals with $\ell$-bit precision, denoted as $\hat{B}_1$. Utilizing \Cref{lemma:substochastic-matrix-powering}, we can compute any entry $\hat{B}_1^k[s,t]$ with a randomized algorithm that runs in $O(\ell k)$ time and $\log(l+1)$ space with acceptance probability $\hat{B}_1^k[s,t]$.
To evaluate $\hat{B}_1^k[s,t]$ with an additive error of $\epsilon$, we use the sequential repetitions outlined in \Cref{lemma:success-probability-estimation}. Specifically, we repeat the algorithm $m=2\epsilon^{-2} \ln(KLM) = O(\epsilon^{-2}\log(\delta^{-4}))$ times, and each turn succeeds with probability at least $1-1/(3KLM)$. Noting that the number of the evaluation of ${b^{(k)}_l}$ for computing $\hat{f}_3(x)$ is $O(KLM)$, and by the union bound, we can conclude that the success probability of evaluating all coefficients in $\bfc$ is at least $2/3$. 

Finally, we complete the proof by analyzing the overall computational complexity. It is evident that our algorithm utilizes $O(\ell+\log{m}) = O(\log(\delta^{-4}\epsilon^{-3}))$ space because indexing $m$ repetitions requires additional $O(\log{m})$ bits.
Moreover, since there are $O(KLM)$ summands in $\hat{f}_3(x)$, and evaluating $b^{(k)}_l$ takes $m$ repetitions with time complexity $O(\ell K)$ for a single turn, the overall time complexity is $O(KLM\cdot\ell K\cdot\epsilon^{-2}\log(KLM))=\tilde{O}(\delta^{-5}\epsilon^{-2})$. 
\end{proof}

\vspace{1em}
Now we present the proof of \Cref{thm:space-efficient-smooth-funcs}, which is a space-efficient and randomized algorithm for constructing bounded polynomial approximations for piecewise-smooth functions. 

\begin{proof}[Proof of \Cref{thm:space-efficient-smooth-funcs}]
% Corollary 66 in~\cite{GSLW18}
Our approach is based on Theorem 40 in~\cite{vAGGdW17} and Corollary 23 in~\cite{GSLW19}. Firstly, we obtain a Fourier approximation $\hat{f}(x)$ of the given function $f(x)$ by truncating it using \Cref{lemma:space-efficient-low-weight-approx}. Next, we ensure that $\hat{f}(x)$ is negligible outside the interval $[-x_0-r,x_0+r]$ by multiplying it with a suitable rectangle function, denoted as $h(x)$. Finally, we derive a space-efficient polynomial approximation $\hat{h}(x)$ of $h(x)$ by applying \Cref{lemma:space-efficient-bounded-funcs}.

\paragraph*{Construction of a bounded function.} 
Let us begin by defining a linear transformation $L(x)\coloneqq\frac{x-x_0}{r+\delta}$ that maps $[x_0-r-\delta,x_0+r+\delta]$ to $[-1,1]$. For convenience, we denote $g(y)\coloneqq f(L^{-1}(y))$ and $b_l\coloneqq a_l(r+\delta)^l$, then it is evident that $g(y)\coloneqq\sum_{l=0}^{\infty} b_ly^l$ for any $y\in[-1,1]$. 

To construct a Fourier approximation by \Cref{lemma:space-efficient-low-weight-approx}, we need to bound the truncation error $\varepsilon_J^{(g)}$. We define $\delta'\coloneqq\frac{\delta}{2(r+\delta)}$ and $J\coloneqq\lceil (\delta')^{-1} \log(12B\epsilon^{-1}) \rceil$. This ensures that the truncation error $\varepsilon_J^{(g)}\coloneqq\big|g(y)-\sum_{j=0}^{J-1} b_jy^j\big|$ for any $y\in[-1+\delta',1-\delta']$ satisfies the following: 
\[\varepsilon_J^{(g)} = \abs*{ \sum_{j=J}^{\infty} b_jy^j } \leq \sum_{j=J}^{\infty} \big| b_j(1-\delta')^j \big| \leq (1-\delta')^J \sum_{j=J}^{\infty} |b_j| \leq (1-\delta')^J B \leq e^{-\delta' J} B \leq \frac{\epsilon}{12} \coloneqq \frac{\epsilon'}{4}. \]
Afterward, let $\hat{\bfb}\coloneqq(b_0,b_1,\cdots,b_{J-1})$, then we know that $\|\hat{\bfb}\|_1 \leq \|\bfb\|_1 \leq B$ by the assumption. Now we utilize \Cref{lemma:space-efficient-low-weight-approx} and obtain the Fourier approximation $\hat{g}(y)$:

\begin{equation}
    \label{eq:smooth-funcs-Fourier-approxY}
    \hat{g}(y) \!\coloneqq\! \begin{cases}
     \sum_{m=-M}^M c_m^{(\even)} \cos(\pi ym), & \text{if } f \text{ is even}\\
     \sum_{m=-M}^M c_m^{(\odd)} \sin\!\big(\pi y \big(m\!+\!\tfrac{1}{2}\big)\big), & \text{if } f \text{ is odd}\\
     \sum_{m=-M}^M \Big(c_m^{(\even)} \!\cos(\pi ym) + c_m^{(\odd)} \!\sin\big(\pi y \big(m\!+\!\tfrac{1}{2}\big)\big)\Big), & \text{otherwise}
    \end{cases}.
\end{equation}
By appropriately choosing $M=O\big( (\delta')^{-1}\!\log\big(\|\hat{\bfb}\|_1/\epsilon'\big)\big) = O\big( r \delta^{-1}\!\log\big(B/\epsilon\big) \big)$, we obtain that the vectors of coefficients $\bfc^{(\even)}$ and $\bfc^{(\odd)}$ satisfy $\|\bfc^{(\even)}\|_1 \leq \|\hat{\bfb}\|_1 \leq \|\bfb\|_1 \leq B$ and similarly $\|\bfc^{(\odd)}\|_1 \leq B$. Plugging $f(x)=g(L(x))$ into \Cref{eq:smooth-funcs-Fourier-approxY}, we conclude that $\hat{f}(x)=\hat{g}(L(x))$ is a Fourier approximation of $f$ with an additive error of $\epsilon/3$ on the interval $[x_0-r-\delta/2, x_0+r+\delta/2]$: 
\begin{equation*}
    \hat{f}(x) = \hat{g}\Big( \frac{x\!-\!x_0}{r\!+\!\delta} \Big) = \begin{cases}
     \sum\limits_{m=-M}^M c_m^{(\even)} \!\!\cos\!\big(\pi m \big(\frac{x-x_0}{r+\delta}\big) \big), & \text{if } f \text{ is even}\\
     \sum\limits_{m=-M}^M c_m^{(\odd)} \!\!\sin\!\big(\pi \big(m+\tfrac{1}{2}\big) \big(\frac{x-x_0}{r+\delta}\big) \big), & \text{if } f \text{ is odd}\\
     \sum\limits_{m=-M}^M c_m^{(\even)} \!\!\cos\!\big(\pi m \big(\frac{x-x_0}{r+\delta}\big) \big) + c_m^{(\odd)} \sin\!\big(\pi \big(m+\tfrac{1}{2}\big) \big(\frac{x-x_0}{r+\delta}\big) \big), & \text{otherwise}
    \end{cases}.
\end{equation*}

\paragraph*{Making the error negligible outside the interval. }
Subsequently, we define the function $h(x) = \hat{f}(x) \cdot R(x)$ such that it becomes negligible outside the interval of interest, i.e., $[x_0-r-\delta/2,x_0+r+\delta/2]$. 
Here, the approximate rectangle function $R(x)$ is $\tilde{\epsilon}$-close to $1$ on the interval $[x_0-r,x_0+r]$, and is $\tilde{\epsilon}$-close to $0$ on the interval $[-1,1] \setminus [x_0-r-2\tilde{\delta}, x_0+r+2\tilde{\delta}]$, where $\tilde{\epsilon}\coloneqq\epsilon/(3B)$ and $\tilde{\delta}\coloneqq\delta/4$. Moreover, $|R(x)|\leq 1$ for any $x \in [-1,1]$. 
Similar to~\cite[Lemma 29]{GSLW19}, $R(x)$ can be expressed as a linear combination of Gaussian error functions: 
\begin{align*}
    R(x) &\coloneqq \frac{1}{2} \rbra*{ \erf\!\big( \kappa (x - x_0 + r + \tilde{\delta}) \big) - \erf \big( \kappa (x - x_0 - r - \tilde{\delta}) \big) },\\
    \text{where } \kappa &\coloneqq \frac{2}{\tilde{\delta}} \log^{1/2} \rbra*{\frac{\sqrt{2}}{\sqrt{\pi} \tilde{\epsilon}}} = \frac{8}{\delta}  \log^{1/2} \rbra*{\frac{\sqrt{18}B}{\sqrt{\pi}\epsilon}}.
\end{align*}

\paragraph*{Bounded polynomial approximation via averaged Chebyshev truncation. }
We present an algorithmic, space-efficient, randomized polynomial approximation method using averaged Chebyshev truncation to approximate the function $h(x)\coloneqq\hat{f}(x) \cdot R(x)$. As suggested in \Cref{prop:time-efficient-smooth-funcs}, we use an explicit polynomial approximation $P^*_d(x)$ of the bounded function $h(x)$ of degree $d=O(\delta^{-1} \log(B\epsilon^{-1}))$ that satisfies the conditions specified in \Cref{eq:smooth-funcs-opt}.

% Corollary 66 in~\cite{GSLW18}
\begin{proposition}[Bounded polynomial approximations based on a local Taylor series, adapted from~{\cite[Corollary 23]{GSLW19}}]
\label{prop:time-efficient-smooth-funcs}
Let $x_0\in [-1,1]$, $r\in(0,2]$, $\delta\in(0,r]$ and let $f\colon[x_0-r-\delta, x_0+r+\delta] \rightarrow \bbR$ and be such that $f(x_0+x)\!\coloneqq\!\sum_{l=0}^{\infty} a_l x^l$ for all $x\in[-r\!-\!\delta,r\!+\!\delta]$. Suppose $B>0$ is such that $\sum_{l=0}^{\infty} (r+\delta)^l |a_l| \leq B$. Let $\epsilon \in \big(0,\!\frac{1}{2B}\big]$, there is a $\epsilon/3$-precise Fourier approximation $\hat{f}(x)$ of $f(x)$ on the interval $[x_0\!-r-\delta/2, x_0\!+\!r\!+\!\delta/2]$, where $\hat{f}(x)\!\coloneqq\!\sum_{m=-M}^M\! \operatorname{Re}\!\Big[\tilde{c}_m e^{-\frac{\rmi \pi m}{2(r+\delta)} x_0} e^{\frac{\rmi \pi m}{2(r+\delta)} x}\Big]$ and $\|\tilde{\bfc}\|_1 \leq B$. We have an explicit polynomial $P^*_d\in\bbR[x]$ of degree $d=O(\delta^{-1}\log(B\epsilon^{-1}))$ s.t. 
    \begin{subequations}
    \label{eq:smooth-funcs-opt}
    \begin{align}
        \|\hat{f}(x)R(x) - P^*_d(x)\|_{[x_0-r,x_0+r]} & \leq \epsilon,\\
        \|P^*_d(x)\|_{[-1,1]} & \leq \epsilon + \|\hat{f}(x)R(x)\|_{[x_0-r-\delta/2,x_0+r+\delta/2]} \leq \epsilon + B,\\
        \|P^*_d(x)\|_{[-1,1] \setminus [x_0-r-\delta/2, x_0+r+\delta/2]} & \leq \epsilon.
    \end{align}
    \end{subequations}
\end{proposition}

To utilize \Cref{lemma:space-efficient-bounded-funcs}, we need to bound the second derivative $\max_{\xi\in[-\pi,0]} |F_k''(\xi)|$, where the integrand $F_k(\theta)\coloneqq\cos(k\theta) h(\cos\theta)$ for any $0 \leq k \leq {d'}$ with $d'=2d-1$. We will calculate this upper bound directly in \Cref{fact:smooth-funcs-derivative}, and the proof is deferred to the end of this section. 

\begin{fact}
    \label{fact:smooth-funcs-derivative}
    Consider the integrand $F_k(\theta) \!=\! \sum_{m=-M}^M \!\frac{c_m}{2} \big( H_{k,m}^{(+)} \!-\! H_{k,m}^{(-)} \big)$ for any function $f$ which is either even or odd. If $f$ is even, we have that $c_m=c_m^{(\even)}$ defined in \Cref{lemma:space-efficient-low-weight-approx}, and 
    \begin{equation}
        \label{eq:smooth-funcs-derivative-even}
        H_{k,m}^{(\pm)}(\theta) \coloneqq \cos\!\rbra*{ \pi m \rbra*{ \frac{\cos{\theta} - x_0}{r+\delta} } } \cdot \cos(k\theta) \cdot \erf\!\rbra*{\kappa \rbra*{ \cos{\theta} - x_0 \pm r \pm \frac{\delta}{4} }}.
    \end{equation}
    Likewise, if $f$ is odd, we know that $c_m=c_m^{(\odd)}$ defined in \Cref{lemma:space-efficient-low-weight-approx}, and 
    \begin{equation}
        \label{eq:smooth-funcs-derivative-odd}
        H_{k,m}^{(\pm)}(\theta) \coloneqq \sin\!\rbra*{ \pi \rbra*{ m+\frac{1}{2} } \rbra*{ \frac{\cos{\theta} - x_0}{r+\delta} } } \cdot \cos(k\theta) \cdot \erf\!\rbra*{ \kappa \rbra*{ \cos{\theta} - x_0 \pm r \pm \frac{\delta}{4} } }.
    \end{equation}
    Moreover, the integrand is $F_k(\theta) \!=\! \sum_{m=-M}^M \!\Big(\frac{c^{(\even)}_m}{2} \big( \hat{H}_{k,m}^{(+)} \!-\! \hat{H}_{k,m}^{(-)} \big) + \frac{c^{(\odd)}_m}{2} \big( \tilde{H}_{k,m}^{(+)} \!-\! \tilde{H}_{k,m}^{(-)} \big)\Big)$ when $f$ is neither even nor odd, where $\hat{H}_{k,m}^{(\pm)}$ and $\tilde{H}_{k,m}^{(\pm)}$ follow from \Cref{eq:smooth-funcs-derivative-even,eq:smooth-funcs-derivative-odd}, respectively.     
    Regardless of the parity of $f$, we have that the second derivative $\abs*{F''_k(\theta)} \leq O(Bd^3)$. 
\end{fact}

Together with \Cref{fact:smooth-funcs-derivative}, we are ready to apply \Cref{lemma:space-efficient-bounded-funcs} to $h(x)=\hat{f}(x)R(x)$, resulting in a degree-$d'$ polynomial $P_{d'} = \hat{c}_0/2 + \sum_{k=1}^{d'} \hat{c}_k T_k$ where $d'=2d-1$ and $\hat{c}_k$ is defined as in \Cref{eq:averaged-truncated-Chebyshe-expansion-theta}. Since $P_{d'}$ is the degree-$d$ averaged Chebyshev truncation of the function $h$ and satisfies \Cref{eq:smooth-funcs-opt}, we define intervals $\calI_{\inter}\coloneqq[x_0-r,x_0+r]$ and $\calI_{\exter}\coloneqq [-1,1] \setminus [x_0-r-\delta/2, x_0+r+\delta/2]$ to obtain: 
\begin{subequations}
\label{eq:smooth-funcs-approxpoly-errors}
\begin{align}
\|f(x)-P_{d'}(x)\|_{\calI_{\inter}} &\leq \|f(x)-h(x)\|_{\calI_{\inter}} + \|h(x)-P_{d'}(x)\|_{\calI_{\inter}} \leq \epsilon + 4\epsilon = O(\epsilon), \\
\|P_{d'}(x)-0\|_{\calI_{\exter}} & \leq \|P_{d'}(x)-h(x)\|_{\calI_{\exter}} + \|h(x)-0\|_{\calI_{\exter}} \leq 4\epsilon + 2B \cdot \frac{\epsilon}{3B} \leq O(\epsilon).  
\end{align}
\end{subequations}

We can achieve the desired error bound by observing \Cref{eq:smooth-funcs-approxpoly-errors} implies: 
\[ \|P_{d'}(x)\|_{[-1,1]} \leq  \|P_{d'}(x)\|_{\calI_{\exter}} + \|P_{d'}(x)\|_{[-1,1] \setminus \calI_{\exter}} \leq O(\epsilon) + B. \]
Moreover, since $|R(x)|\leq 1$ and the low-weight Fourier approximation satisfies $|\hat f(x)|\leq O(B)$ on $[-1,1]$, we have $\|h(x)\|_{[-1,1]}\leq O(B)$. By \Cref{lemma:space-efficient-bounded-funcs}\ref{thmitem:polyApprox-l1-norm-bound-general}, we deduce that the norm of the coefficient vector $\hat{\bfc}$ of the polynomial $P_{d'}$ is bounded by $\|\hat{\bfc}\|_1 \leq O(B\sqrt d)$.

\paragraph*{Analyzing time and space complexity.}
The construction of $\hat{f}(x)$ can be implemented in bounded-error randomized time $\tilde{O}((\delta')^{-5}\epsilon^{-2}B^2)$ and space
$O(\log((\delta')^{-4}\epsilon^{-1}B))$, given that this construction uses \Cref{lemma:space-efficient-low-weight-approx} with $\delta'=\frac{\delta}{2(r+\delta)} \in (0,\frac{1}{2}]$ and $\epsilon'=\frac{\epsilon}{3B}$.
Having $\hat{f}(x)$, we can construct a bounded polynomial approximation $\hat{h}(x)$ deterministically using \Cref{lemma:space-efficient-bounded-funcs}. This construction can be implemented in deterministic time $O(d^{(\gamma+1)/2}\epsilon^{-1/2} t(\ell)) \leq \tilde{O}(d^2\epsilon^{-1/2})$ and space $O(\log(d^{(\gamma+3)/2}\epsilon^{-3/2}B)) \leq O(\log(d^3\epsilon^{-3/2}B))$ since the integrand $F_k(\theta)$ is a product of a constant number of (compositions of) holonomic functions (\Cref{remark:evaluating-funcs}). Therefore, our construction can be implemented in bounded-error randomized time $\tilde{O}(\max\big\{ (\delta')^{-5} \epsilon^{-2} B^2, d^2\epsilon^{-1/2} \big\})$ and space $O(\max\{\log((\delta')^{-4} \epsilon^{-1} B), \log(d^3 \epsilon^{-3/2} B) \}) \leq O(\log(d^3(\delta')^{-4}\epsilon^{-3/2}B))$.
\end{proof}

\vspace{1em}
With the aid of \Cref{thm:space-efficient-smooth-funcs}, we can provide a space-efficient polynomial approximation to the normalized logarithmic function utilized in~\cite[Lemma 11]{GL20}.

\begin{corollary}[Space-efficient polynomial approximation to the normalized logarithmic function]
    \label{corr:space-efficient-log}
    Let $\beta \in (0,1]$ and $\epsilon\in(0,1/2)$, there is an even polynomial $P^{\ln}_{d'}$ of degree $d'=2d-1 \leq \tilde{C}_{\ln} \beta^{-1} \log{\epsilon^{-1}}$, where $P^{\ln}_{d'}$ corresponds to some degree-$d$ averaged Chebyshev truncation and $\tilde{C}_{\ln}$ is a universal constant, such that 
    \begin{align*}
    \forall x \in [\beta,1],& \left| P^{\ln}_{d'}(x) - \tfrac{\ln(1/x)}{2\ln(2/\beta)} \right| \leq C_{\ln} \epsilon , \text{ where } C_{\ln} \text{ is a universal constant}, \\
    \forall x \in [-1,1],& |P^{\ln}_{d'}(x)| \leq 1.
    \end{align*}
    Moreover, the coefficient vector $\bfc^{\ln}$ of $P^{\ln}_{d'}$ has a norm bounded by $\|\bfc^{\ln}\|_1 \leq \hat{C}_{\ln} \sqrt{d'}$, where $\hat{C}_{\ln}$ is another universal constant. In addition, any entry of the coefficient vector $\bfc^{\ln}$ can be computed in bounded-error randomized time $\tilde{O}(\max\{\beta^{-5}\epsilon^{-2}, d^2\epsilon^{-1/2}\})$ and space $O(\log(d^3 \beta^{-4} \epsilon^{-3/2}))$. 
    Without loss of generality, we assume that all constants $C_{\ln}$, $\hat{C}_{\ln}$, and $\tilde{C}_{\ln}$ are at least $1$. 
\end{corollary}

\begin{proof}
Consider the function $f(x)\coloneqq\frac{\ln(1/x)}{2\ln(2/\beta)}$. We apply \Cref{thm:space-efficient-smooth-funcs} to $f$ by choosing the same parameters as in Lemma 11 of~\cite{GL20}, specifically $\epsilon'=\epsilon/2$, $x_0=1$, $r=1-\beta$, $\delta=\beta/2$, and $B=1/2$.\footnote{As indicated in Lemma 11 of~\cite{GL20}, since the Taylor series of $f(x)$ at $x=1$ is $\frac{1}{2\ln(2/\beta)} \sum_{l=1}^{\infty} \frac{(-1)^lx^l}{l}$, we obtain that $B=f\big(\frac{\beta}{2}-1\big)=\frac{1}{2\ln(2/\beta)} \sum_{l=1}^\infty \frac{(1-\beta/2)^l}{l} = -\frac{1}{2\ln(2/\beta)} \sum_{l=1}^{\infty} \frac{(-1)^{l-1}}{l} (\beta/2-1)^l = -\frac{1}{2\ln(2/\beta)}\ln\frac{\beta}{2} = \frac{1}{2}$.} This results in a space-efficient randomized polynomial approximation $\tilde{P}_{d'}\in\bbR[x]$ of degree $d'=2d-1 = O(\delta^{-1} \log(\epsilon^{-1} B)) \leq \tilde{C}_{\ln} \beta^{-1} \log{\epsilon^{-1}}$, where $\tilde{P}_{d'}$ corresponds to some degree-$d$ averaged Chebyshev truncation and $\tilde{C}_{\ln}$ is a universal constant. By appropriately choosing $\eta \leq 1/2$ such that $C'_{\ln} \epsilon = \eta/4$ for a universal constant $C'_{\ln}$, this polynomial approximation $\tilde{P}_{d'}$ satisfies the following inequalities: 
\begin{subequations}
\label{eq:conditions-log}
\begin{align}
\|f(x)-\tilde{P}_{d'}(x)\|_{[\beta,2-\beta]} & \leq C'_{\ln} \epsilon = \tfrac{\eta}{4}\\
\|\tilde{P}_{d'}(x)\|_{[-1,1]} &\leq B+ C'_{\ln} \epsilon \leq \tfrac{1}{2} + C'_{\ln} \epsilon = \tfrac{1}{2}+\tfrac{\eta}{4}\\
\|\tilde{P}_{d'}(x)\|_{[-1,\beta/2]} &\leq C'_{\ln} \epsilon = \tfrac{\eta}{4}.
\end{align}
\end{subequations}
Additionally, using \Cref{thm:space-efficient-smooth-funcs}, the coefficient vector $\bfc^{(\tilde{P})}$ of $\tilde{P}_{d'}$ satisfies 
\[ \norm[\big]{ \bfc^{(\tilde{P})} }_1 \leq O\rbra[\big]{B\sqrt{d'}} \leq \hat{C}'_{\ln}\sqrt{d'},\] where $\hat{C}'_{\ln}$ is a universal constant.
Noticing that $\delta'=\frac{\delta}{2(r+\delta)}=\frac{\beta/2}{2(1-\beta+\beta/2)}=\frac{\beta}{4(1-\beta/2)}=\Theta(\beta)$, our utilization of \Cref{thm:space-efficient-smooth-funcs} yields a bounded-error randomized algorithm that requires space 
\[O(\log(d^3 (\delta')^{-4} \epsilon^{-3/2} B)) = O(\log(d^3 \beta^{-4} \epsilon^{-3/2}))\] and time 
\[\tilde{O}(\max\{(\delta')^{-5} \epsilon^{-2} B^2, d^2 \epsilon^{-1/2} \}) = \tilde{O}(\max\{\beta^{-5}\epsilon^{-2}, d^2 \epsilon^{-1/2}\}).\] 

Furthermore, noting that the real-valued function $f(x)$ is defined only when $x>0$, $\tilde{P}_{d'}(x)$ is not an even polynomial in general. Instead, we consider $P^{\ln}_{d'}(x)\coloneqq(1+\eta)^{-1}(\tilde{P}_{d'}(x)+\tilde{P}_{d'}(-x))$ for all $x\in[-1,1]$. 
Together with \Cref{eq:conditions-log}, we have derived that: 
\begin{subequations}
\label{eq:log-polynomial-error}
\begin{align}
& \|f(x) - P^{\ln}_{d'}(x)\|_{[\beta,1]} \\
 \leq~& \big\|f(x) - \tfrac{1}{1+\eta} \tilde{P}_{d'}(x) \big\|_{[\beta,1]} + \big\|\tfrac{1}{1+\eta}\tilde{P}_{d'}(-x)\big\|_{[\beta,1]}\\
\leq~& \big\|f(x) -  \tilde{P}_{d'}(x)\big\|_{[\beta,1]}  +  \big\|\tilde{P}_{d'}(x) - \tfrac{1}{1+\eta} \tilde{P}_{d'}(x)\big\|_{[\beta,1]}  +  \big\|\tfrac{1}{1+\eta} \tilde{P}_{d'}(-x)\big\|_{[\beta,1]}\\
\leq~& \tfrac{\eta}{4} + \tfrac{\eta}{1+\eta}\cdot \big(\tfrac{1}{2} + \tfrac{\eta}{4}\big) + \tfrac{1}{1+\eta}\cdot \tfrac{\eta}{4}\\
\leq~& \eta.
\end{align}
\end{subequations}

Here, the last line owes to the fact that $\eta > 0$.
Consequently, \Cref{eq:log-polynomial-error} implies that $\|f(x)-P^{\ln}_{d'}(x)\|_{[\beta,1]} \leq 4C'_{\ln} \epsilon \coloneqq C_{\ln} \epsilon$ for another universal constant $C_{\ln}$. 
Notice $P^{\ln}_{d'}$ is an even polynomial with $\deg(P^{\ln}_{d'}) \leq \tilde{C}_{\ln} \beta^{-1} \log \epsilon^{-1}$, \Cref{eq:conditions-log} yields that: 
\begin{align*}
\|P^{\ln}_{d'}(x)\|_{[-1,1]} = \|P^{\ln}_{d'}(x)\|_{[0,1]} \leq \|\tfrac{1}{1+\eta}\tilde{P}_{d'}(x)\|_{[0,1]} + \|\tfrac{1}{1+\eta}\tilde{P}_{d'}(x)\|_{[-1,0]}
\leq \tfrac{1}{1+\eta} \cdot \tfrac{1+\eta}{2} + \tfrac{1}{1+\eta} \cdot \tfrac{\eta}{2} \leq 1. 
\end{align*}
Here, the last inequality is due to $\eta \leq 1/2$. 

It remains to verify how the coefficient vector changes when passing from $\tilde P_{d'}$ to $P^{\ln}_{d'}$. Write $\tilde P_{d'}(x)=\frac{\tilde c_0}{2}+\sum_{k=1}^{d'}\tilde c_kT_k(x)$. 
Since $T_k(-x)=(-1)^kT_k(x)$, the polynomial $P^{\ln}_{d'}(x)=(1+\eta)^{-1} \rbra[\big]{ \tilde P_{d'}(x)+\tilde P_{d'}(-x)}$ has coefficients $c^{\ln}_0=\frac{2}{1+\eta}\tilde c_0$ and $c^{\ln}_k=\frac{1+(-1)^k}{1+\eta}\tilde c_k$ for each $1 \leq k \leq d'$.
Therefore, we obtain: for a universal constant $\hat C_{\ln}$,
\[ \|\bfc^{\ln}\|_1 \leq \frac{2}{1+\eta}\abs*{\tilde c_0}+\sum_{k=1}^{d'}\frac{\abs*{1+(-1)^k}}{1+\eta}\abs*{\tilde c_k} \leq 2\norm{\bfc^{(\tilde P)}}_1 \leq 2 \hat{C}'_{\ln} \sqrt{d'} \coloneqq \hat{C}_{\ln} \sqrt{d'}.  \]
Here, the second inequality uses the fact that $2/(1+\eta) \leq 2$. The same coefficient-computation time and space bounds hold for $\bfc^{\ln}$, since each coefficient of $P^{\ln}_{d'}$ is obtained from the corresponding coefficient of $\tilde P_{d'}$ by multiplication by the explicitly computable factor $(1+(-1)^k)/(1+\eta)$.
\end{proof}

\vspace{1em}
We now provide detailed constructions in the proofs of \Cref{prop:first-low-weight-approx,prop:second-low-weight-approx,prop:third-low-weight-approximation}:

\begin{proof}[Proof of \Cref{prop:first-low-weight-approx}]
    We construct a Fourier series by a linear combination of the power of sines. We first note that $x = \frac{2}{\pi}\cdot\arcsin\!\left( \sin\left( \frac{x\pi}{2} \right) \right)$ for all $x \in [-1,1]$, and plug it into $\hat{f}_1(x)\coloneqq\sum_{k=0}^K a_kx^k$, which deduces that $\|f-\hat{f}_1\|_{\calI_{\delta}} \leq \epsilon/4$ by the assumption. Let $\bfb^{(k)}$ be the coefficients of $\left( \frac{\arcsin{y}}{\pi/2} \right)^k=\sum_{l=0}^{\infty} b_l^{(k)} y^l$ for all $y \in [-1,1]$, then we result in our first approximation. 
    Moreover, we observe that $\frac{\pi}{2}\cdot\bfb^{(1)}$ is exactly the Taylor series of arcsin, whereas we know that $\left( \frac{\arcsin{y}}{\pi/2} \right)^{k+1} = \left( \frac{\arcsin{y}}{\pi/2} \right)^{k} \cdot \left( \sum_{l=0}^{\infty} b_l^{(1)} y^l \right)$ for $k > 1$, which derives \Cref{eq:logQSVT-recursive-formula} by comparing the coefficients. 
    In addition, notice that $\|\bfb^{(k)}\|_1 = \sum_{l=0}^{\infty} b^{(k)}_l 1^l = \left(\frac{\arcsin{1}}{\pi/2}\right)^k=1$, together with straightforward reasoning follows from \Cref{eq:logQSVT-recursive-formula}, we deduce the desired property for $\{b^{(k)}_l\}$. 
\end{proof}

\begin{proof}[Proof of \Cref{prop:second-low-weight-approx}]
    We truncate the summation over $l$ in $f_1(x)$ at $l=L$, and it suffices to bound the truncation error. For all $k \in \bbN$ and $x \in [-1+\delta,1-\delta]$, we obtain the error bound:
    \[\Bigg| \sum_{l=\lfloor L \rfloor}^{\infty} b^{(k)}_l \sin^l\!\left(\tfrac{x\pi}{2}\right)\Bigg| \!
      \leq \! \sum_{l=\lfloor L \rfloor}^{\infty} b^{(k)}_l \Big| \sin^l\!\left( \tfrac{x\pi}{2} \right) \Big| \!
      \leq \! \sum_{l=\lfloor L \rfloor}^{\infty} b^{(k)}_l |1-\delta^2|^l
      \leq (1-\delta^2)^L \sum_{l=\lfloor L \rfloor}^{\infty} b^{(k)}_l \leq (1-\delta^2)^L.\]
    Here, the second inequality owing to $\forall \delta \in [0,1]$, $\sin\!\left( (1-\delta)\frac{\pi}{2} \right) \leq 1-\delta^2$, and the last inequality is due to $\|\bfb^{(k)}\|_1=1$ in \Cref{prop:first-low-weight-approx}. 
    By appropriately choosing $L\coloneqq\delta^{-2} \ln(4\|\bfa\|_1 \epsilon^{-1})$, we obtain that $\|\hat{f}_1-\hat{f}_2\|_{\calI_{\delta}} \leq \sum_{k=0}^K \abs{a_k} (1-\delta^2)^L \leq \|\bfa\|_1\cdot \exp(-\delta^2 L) \leq \epsilon/4$. 
\end{proof}

\begin{proof}[Proof of \Cref{prop:third-low-weight-approximation}]
    We upper-bound $\sin^l(x)$ in $\hat{f}_2(x)$ defined in \Cref{prop:second-low-weight-approx} using a tail bound of binomial coefficients. We obtain that 
    \[ \sin^l(z)=\rbra*{ \frac{e^{-\rmi z}-e^{\rmi z}}{-2\rmi} }^l=\rbra*{ \frac{\rmi}{2} }^l\sum_{m=0}^l (-1)^m \binom{l}{m} \exp(\rmi z(2m-l))\] 
    by a direct calculation, which implies the counterpart for real-valued functions:
    \begin{equation}
        \label{eq:power-of-sine}
        \sin^l(z) = \begin{cases}
            2^{-l} (-1)^{(l+1)/2} \sum\limits_{m'=0}^l (-1)^{m'} \binom{l}{m'} \sin(z(2m'-l)), & \text{if } l \text{ is odd;}\\
            2^{-l} (-1)^{l/2} \sum\limits_{m'=0}^l (-1)^{m'} \binom{l}{m'} \cos(z(2m'-l)), & \text{if } l \text{ is even.}
        \end{cases}
    \end{equation}
    Recall that the Chernoff bound (e.g., Corollary A.1.7 in~\cite{AS16}) which corresponds a tail bound of binomial coefficients, and assume that $l \leq L$, we have derived that:  
    \begin{equation}
        \label{eq:binomial-tail-bounds}
        \sum_{m'=0}^{\lfloor l/2 \rfloor-M} 2^{-l} \binom{l}{m'} 
        = \sum_{m'=\lceil l/2 \rceil+M}^l  2^{-l} \binom{l}{m'}
        \leq e^{-\frac{2M^2}{l}}
        \leq e^{-\frac{2M^2}{L}}
        \leq \rbra*{ \frac{\epsilon}{4\|a\|_1} }^2
        \leq  \frac{\epsilon}{4\|a\|_1}. 
    \end{equation}
    Here, we choose $M=\lceil \delta^{-1} \ln(4\|\bfa\|_1\epsilon^{-1}) \rceil$, and the last inequality is because of the assumption $\epsilon \leq 2\|\bfa\|_1$. As stated in \Cref{prop:first-low-weight-approx}, $b^{(k)}_l=0$ if $k$ and $l$ have different parities. 
    Consequently, we only need to consider all odd (resp., even) $l \leq L$ for odd (resp., even) functions. If the function $f$ is neither even nor odd, we must consider all $l \leq L$. Plugging \Cref{eq:binomial-tail-bounds} into \Cref{eq:power-of-sine}, we can derive that: 
    \begin{itemize}
        \item If $l$ is odd, 
        \begin{equation}
            \label{eq:truncating-power-of-sine-odd}
            \norm[\bigg]{ \sin^l(z) - 2^{-l} (-1)^{(l+1)/2} \sum_{m'=(l+1)/2-M}^{(l+1)/2+M} (-1)^{m'} \binom{l}{m'} \sin(z(2m'-l)) }_{\calI_{\delta}} \leq \frac{\epsilon}{2\norm*{\bfa}_1};
        \end{equation}
        \item If $l$ is even, 
        \begin{equation}
            \label{eq:truncating-power-of-sine-even}
            \norm[\bigg]{ \sin^l(z) - 2^{-l} (-1)^{l/2} \sum_{m'=l/2-M}^{l/2+M} (-1)^{m'} \binom{l}{m'} \cos(z(2m'-l)) }_{\calI_{\delta}} \leq \frac{\epsilon}{2\norm*{\bfa}_1}.
        \end{equation}
    \end{itemize}

    Plugging \Cref{eq:truncating-power-of-sine-odd,eq:truncating-power-of-sine-even} into $\hat{f}_2(x)$, and substituting $z=x\pi/2$, this equation leads to $\hat{f}_3(x)$ as desired. In addition, combining $\sum_{k=0}^K |a_k| \sum_{l=0}^{\lfloor L \rfloor} |b^{(k)}_l| \leq \sum_{k=0}^{K} |a_k| = \|\bfa\|_1$ with \Cref{eq:truncating-power-of-sine-odd,eq:truncating-power-of-sine-even}, we achieve that $\|\hat{f}_2-\hat{f}_3\|_{\calI_{\delta}} \leq \epsilon/2$. 
\end{proof}

Finally, we present the proof of \Cref{fact:smooth-funcs-derivative}:
\begin{proof}[Proof of \Cref{fact:smooth-funcs-derivative}] 
We begin by deriving an upper bound of the second derivative of the integrand $F_k(\theta)$: 
\begin{subequations}
\label{eq:smooth-funcs-derivative}
\begin{align}
    |F''_k(\theta)| &\leq \sum_{m=-M}^M  \frac{\abs{c_m}}{2} \abs*{\frac{\dd^2}{\dtheta^2} H_{k,m}^{(+)}(\theta) - \frac{\dd^2}{\dtheta^2} H_{k,m}^{(-)}(\theta)} \\
    &\leq \frac{\norm*{\bfc}_1}{2} \max_{-\pi \leq \theta \leq 0} \rbra*{ \abs*{\frac{\dd^2}{\dtheta^2} H_{k,m}^{(+)}(\theta)} + \abs*{\frac{\dd^2}{\dtheta^2} H_{k,m}^{(-)}(\theta)} }. 
\end{align}
\end{subequations}

By a straightforward calculation, we have the second derivatives of $H_{k,m}^{\pm}(\theta)$ if $f$ is even: 
\begin{align*}
    \tfrac{\dd^2}{\dtheta^2} H_{k,m}^{(\pm)}(\theta)\!=\!
    &-k^2 \cos(k \theta) \cos\!\big(\tfrac{\pi  m (\cos{\theta}-x_0)}{\delta +r}\big) \erf\!\big(\kappa  \big(\cos{\theta}-x_0 \mp r \mp \tfrac{\delta}{4}\big)\big)\\
    &-\tfrac{\pi ^2 m^2}{(\delta +r)^2} \sin^2(\theta) \cos(k \theta) \cos\!\big(\tfrac{\pi  m (\cos{\theta}-x_0)}{\delta +r}\big) \erf\!\big(\kappa  \big(\cos{\theta}-x_0 \mp r \mp \tfrac{\delta}{4}\big)\big)\\
    &+\tfrac{\pi  m}{\delta +r} \cos{\theta} \cos(k \theta) \sin\!\big(\tfrac{\pi  m (\cos{\theta}-x_0)}{\delta +r}\big) \erf\!\big(\kappa  \big(\cos{\theta}-x_0 \mp r \mp \tfrac{\delta}{4}\big)\big)\\
    &-\tfrac{2 \pi  k m}{\delta +r} \sin(\theta) \sin(k \theta) \sin\!\big(\tfrac{\pi  m (\cos{\theta}-x_0)}{\delta +r}\big) \erf\!\big(\kappa  \big(\cos{\theta}-x_0 \mp r \mp \tfrac{\delta}{4}\big)\big)\\
    &-\tfrac{2 \kappa}{\sqrt{\pi }}  \cos{\theta} \cos(k \theta) \cos\!\big(\tfrac{\pi  m (\cos{\theta}-x_0)}{\delta +r}\big) e^{-\kappa ^2 \big(\cos{\theta}-x_0 \mp r \mp \tfrac{\delta}{4}\big)^2}\\
    &-\tfrac{4 \sqrt{\pi } \kappa m}{\delta +r} \sin^2(\theta) \cos(k \theta) \sin\!\big(\tfrac{\pi  m (\cos{\theta}-x_0)}{\delta +r}\big) e^{-\kappa ^2 \big(\cos{\theta}-x_0 \mp r \mp \frac{\delta}{4}\big)^2}\\
    &+\tfrac{4 \kappa  k}{\sqrt{\pi }} \sin(\theta) \sin(k \theta) \cos\!\big(\tfrac{\pi  m (\cos{\theta}-x_0)}{\delta +r}\big) e^{-\kappa ^2 \big(\cos{\theta}-x_0 \mp r \mp \tfrac{\delta}{4}\big)^2}\\
    &-\tfrac{4 \kappa^3}{\sqrt{\pi }} \sin^2(\theta) \cos(k \theta) \cos\!\big(\tfrac{\pi  m (\cos{\theta}-x_0)}{\delta +r}\big) \big(\cos{\theta}-x_0 \mp r \mp \frac{\delta}{4}\big) e^{-\kappa ^2 \big(\cos{\theta}-x_0 \mp r \mp \tfrac{\delta}{4}\big)^2}.
\end{align*}

Noting that all functions appear in $\frac{\dd^2}{\dtheta^2} H_{k,m}^{(\pm)}(\theta)$, viz.~$\sin{x}$, $\cos{x}$, $\exp(-x^2)$, and $\erf(x)$,  are at most $1$, as well as $|x_0\pm r \pm \delta/4| \leq 7/2$, then we obtain that 
\begin{subequations}
\label{eq:summand-derivative}
\begin{align}
&\big|\tfrac{\dd^2}{\dtheta^2} H_{k,m}^{(\pm)}(\theta)\big| \\
\leq~& k^2 + \tfrac{2\kappa}{\sqrt{\pi}} + \tfrac{4\kappa k}{\sqrt{\pi}} + \tfrac{18\kappa^3}{\sqrt{\pi}} + m \cdot \big( \tfrac{\pi}{\delta+r} + \tfrac{2\pi k}{\delta+r} + \tfrac{4\sqrt{\pi}\kappa}{\delta+r} \big) + m^2 \cdot \tfrac{\pi^2}{(\delta+r)^2}\\
\leq~& (d')^2 + O(d) + O(d^2) + O(d^3) + \tfrac{M}{\delta+r} \cdot (O(1) \!+\! O(d) \!+\! O(d)) + M^2 \cdot \tfrac{O(1)}{(\delta+r)^2}\\
=~& O(d^3).
\end{align}
\end{subequations}

Here, the second line according to $k \leq d'=2d-1$ and $\kappa \leq O(d)$, also the last line is due to facts that $M \leq O(rd)$ and $1/2 \leq r/(\delta+r) \leq 1$ if $0 < \delta \leq r$ and $0 < r \leq 2$. 
Additionally, a similar argument shows that the upper bound in \Cref{eq:summand-derivative} applies to odd functions and functions that are neither even nor odd as well. This is because a direct computation yields the second derivatives of $H^{(\pm)}_{k,m}(\theta)$ when $f$ is odd: 
\begin{align*}
    \tfrac{\dd^2}{\dtheta^2} H_{k,m}^{(\pm)}(\theta)\!=\!
    &-k^2 \cos (k \theta) \sin\!\big(\tfrac{\pi  \big(m+\tfrac{1}{2}\big) (\cos \theta-x_0)}{\delta +r}\big) \erf\!\big(\kappa  \big(\cos \theta-x_0 \mp r \mp \tfrac{\delta}{4}\big)\big)\\
    &-\tfrac{\pi  \big(m+\tfrac{1}{2}\big)}{\delta +r} \cos \theta \cos (k \theta) \cos\!\big(\tfrac{\pi  \big(m+\tfrac{1}{2}\big) (\cos \theta-x_0)}{\delta +r}\big) \erf\!\big(\kappa  \big(\cos \theta-x_0 \mp r \mp \tfrac{\delta}{4}\big)\big)\\
    &-\tfrac{\pi ^2 \big(m+\tfrac{1}{2}\big)^2}{(\delta +r)^2} \sin ^2\theta \cos (k \theta) \sin\!\big(\tfrac{\pi  \big(m+\tfrac{1}{2}\big) (\cos \theta-x_0)}{\delta +r}\big) \erf\!\big(\kappa  \big(\cos \theta-x_0 \mp r \mp \tfrac{\delta}{4}\big)\big)\\
    &+\tfrac{2 \pi  k \big(m+\tfrac{1}{2}\big)}{\delta +r} \sin \theta \sin (k \theta) \cos\!\big(\tfrac{\pi  \big(m+\tfrac{1}{2}\big) (\cos \theta-x_0)}{\delta +r}\big) \erf\!\big(\kappa  \big(\cos \theta-x_0 \mp r \mp \tfrac{\delta}{4}\big)\big)\\
    &+\tfrac{4 \sqrt{\pi } \kappa  \big(m+\tfrac{1}{2}\big)}{\delta +r} \sin ^2\theta \cos (k \theta) \cos\!\big(\tfrac{\pi  \big(m+\tfrac{1}{2}\big) (\cos \theta-x_0)}{\delta +r}\big) e^{-\kappa ^2 \big(\cos \theta-x_0 \mp r \mp \tfrac{\delta}{4}\big)^2}\\
    &-\tfrac{2 \kappa}{\sqrt{\pi }}  \cos \theta \cos (k \theta) \sin\!\big(\tfrac{\pi  \big(m+\tfrac{1}{2}\big) (\cos \theta-x_0)}{\delta +r}\big) e^{-\kappa ^2 \big(\cos \theta-x_0 \mp r \mp \tfrac{\delta}{4}\big)^2}\\
    &+\tfrac{4 \kappa  k}{\sqrt{\pi }} \sin \theta \sin (k \theta) \sin\!\big(\tfrac{\pi  \big(m+\tfrac{1}{2}\big) (\cos \theta-x_0)}{\delta +r}\big) e^{-\kappa ^2 \big(\cos \theta-x_0 \mp r \mp \tfrac{\delta}{4}\big)^2}\\
    &-\tfrac{4 \kappa ^3}{\sqrt{\pi }} \sin ^2\theta \cos (k \theta) \big(\!\cos \theta\!-\!x_0 \!\mp\! r \!\mp\! \tfrac{\delta}{4}\big) \sin\!\big(\tfrac{\pi  \big(m+\tfrac{1}{2}\big) (\cos \theta-x_0)}{\delta +r}\big) e^{-\kappa ^2 \big(\cos \theta-x_0 \mp r \mp \tfrac{\delta}{4}\big)^2}.
\end{align*}

Substituting \Cref{eq:summand-derivative} into \Cref{eq:smooth-funcs-derivative}, and noticing that the coefficient vector $\|\bfc^{(\even)} + \bfc^{(\odd)}\|_1 \leq B$ regardless of the parity of $f$, we conclude that $|F''_k(\theta)| \leq O(Bd^3)$.  
\end{proof}

\subsection{Applying averaged Chebyshev truncation to bitstring indexed encodings}
\label{subsec:poly-applied-to-unitary-encodings}

% Lemma 19 of~\cite{GSLW18}
With space-efficient bounded polynomial approximations of piecewise-smooth functions, it suffices to implement averaged Chebyshev truncation on bitstring indexed encodings, as specified in \Cref{thm:LCU-averaged-chebyshev-truncation}. The proof combines \Cref{lemma:Chebyshev-poly-implementation,lemma:space-efficient-LCU,lemma:renormalizing-encodings}.
\begin{theorem}[Averaged Chebyshev truncation applied to bitstring indexed encodings]
    \label{thm:LCU-averaged-chebyshev-truncation}
    Let $A$ be an Hermitian matrix acting on $s$ qubits, and let $U$ be a $(1,a,\epsilon_1)$-bitstring indexed encoding of $A$ that acts on $s+a$ qubits. For any degree-$d$ averaged Chebyshev truncation $P_{d'}(x)=\hat{c}_0/2+\sum_{k=1}^{d'} \hat{c}_k T_k(x)$ where $d'=2d-1 \leq 2^{O(s(n))}$ and $T_k$ is the $k$-th Chebyshev polynomial \emph{(}of the first kind\emph{)}, equipped with an evaluation oracle $\Eval$ that returns $\tilde{c}_k$ with precision $\varepsilon\coloneqq O(\epsilon_2^2/d')$, assume further that either $U$ is a block-encoding or $P_{d'}$ is either even or odd. Then we have the following bitstring indexed encoding of $P_{d'}(A)$ depending on whether $P_{d'}(A)$ is a partial isometry \emph{(}up to a normalization factor\emph{)}:\footnote{This condition differs from the one that $A$ is a partial isometry. Specifically, $P_{d'}(A)$ is a partial isometry (up to a normalization factor) if $A$ is a partial isometry, whereas $\sign^{\SV}(A)$ is a partial isometry for any $A$.}
    \begin{itemize}[itemsep=0.33em,topsep=0.33em,parsep=0.33em]
        \item \textbf{Partial isometry $P_{d'}(A)$\emph{:}} We obtain a $(1,a',(144d'\sqrt{\epsilon_1}+36\epsilon_2)\|\hat{\bfc}\|_1)$-bitstring indexed encoding $V_{\rm normed}$ of $P_{d'}(A)$ that acts on $s+a'$ qubits where $a'\coloneqq a+\lceil \log{d'} \rceil+3$. 
        \item \textbf{General $P_{d'}(A)$\emph{:}} We obtain a $(\|\hat{\bfc}\|_1,\hat{a},(4d'\sqrt{\epsilon_1}+\epsilon_2)\|\hat{\bfc}\|_1)$-bitstring indexed encoding $V_{\rm unnorm}$ of $P_{d'}(A)$ that acts on $s+\hat{a}$ qubits where $\hat{a}\coloneqq a+\lceil \log{d'} \rceil+1$. 
    \end{itemize}

    \noindent Let $V$ be the bitstring indexed encoding of $P_{d'}(A)$. The implementation of $V$ requires $O(d^2\eta_{V})$ uses of $U$, $U^{\dagger}$, $\textsc{C}_{\Pi}\textsc{NOT}$, $\textsc{C}_{\tilde{\Pi}}\textsc{NOT}$, and multi-controlled single-qubit gates.\footnote{As indicated in~Figure 3(c) of~\cite{GSLW19} (see also Lemma 19 in~\cite{GSLW18}), we replace the single-qubit gates used in \Cref{lemma:Chebyshev-poly-implementation} with multi-controlled (or ``multiply controlled'') single-qubit gates.\label{footnote:multi-controlled-gates}} The description of the resulting quantum circuit of $V$ can be computed in deterministic time $\tilde{O}(d^2 \eta_{V} \log(d/\epsilon_2))$, space $O(\max\{s(n),\log(d/\epsilon_2^2)\})$, and $O(d^2 \eta_V)$ oracle calls to $\Eval$ with precision $\varepsilon$. Here, $\eta_V = \|\hat{\bfc}\|_1$ if $V=V_{\rm normed}$ whereas $\eta_V = 1$ if $V=V_{\rm unnorm}$. 

    \noindent Furthermore, our construction straightforwardly extends to any linear \emph{(}possibly non-Hermitian\emph{)} operator $A$ by simply replacing $P_{d'}(A)$ with $P^{\SV}_{d'}(A)$ defined in \Cref{def:matrix-SV-function}. 
\end{theorem}

\begin{remark}[QSVT implementations of averaged Chebyshev truncation preserve the parity]
    \label{remark:QSVT-parity-preserving}
    As shown in \Cref{prop:angles-for-Chebyshev-polys}, we can implement the quantum singular value transformation $T_k^{\SV}(A)$ \textit{exactly} for any linear operator $A$ that admits a bitstring indexed encoding, because the rotation angles corresponding to the $k$-th Chebyshev polynomials are either $\pi/2$ or $(1-k)\pi/2$, indicating that $T_k(0)=0$ for any odd $k$.
    We then implement the QSVT corresponding to the averaged Chebyshev truncation polynomial $P_{d'}(x) = \sum_{l=0}^{({d'}-1)/2} \hat{c}_{2l+1} T_{2l+1}(x)$, as described in \Cref{corr:sign-polynomial-implementation}, although the actual implementation results in a slightly different polynomial, $\tilde{P}_{d'}(x) = \sum_{l=0}^{({d'}-1)/2} \tilde{c}_{2l+1} T_{2l+1}(x)$. However, we still have $\tilde{P}_{d'}(0)=0=P_{d'}(0)$, indicating that the implementations in \Cref{thm:LCU-averaged-chebyshev-truncation} preserve the parity.
\end{remark}

We first demonstrate an approach, based on~\cite[Lemma 3.12]{MY23}, that constructs Chebyshev polynomials of bitstring indexed encodings in a space-efficient manner. 
\begin{lemma}[Chebyshev polynomials applied to bitstring indexed encodings]
    \label{lemma:Chebyshev-poly-implementation}
    Let $A$ be a linear operator acting on $s$ qubits, and let $U$ be a $(1,a,\epsilon)$-bitstring indexed encoding of $A$ that acts on $s+a$ qubits. Then, for the $k$-th Chebyshev polynomial \emph{(}of the first kind\emph{)} $T_k(x)$ of degree $k\leq 2^{O(s)}$, there exists a new $(1,a+1,4k\sqrt{\epsilon})$-bitstring indexed encoding $V$ of $T_k^{\SV}(A)$ that acts on $s+a+1$ qubits. This implementation requires $k$ uses of $U$, $U^{\dagger}$, $\textsc{C}_{\Pi}\textsc{NOT}$, $\textsc{C}_{\tilde{\Pi}}\textsc{NOT}$, and $k$ single-qubit gates. 
    Moreover, we can compute the description of the resulting quantum circuit in deterministic time $k$ and space $O(s)$. 

    \noindent Furthermore, consider $A'\coloneqq \tilde{\Pi} U \Pi$, where $\tilde{\Pi}$ and $\Pi$ are the corresponding orthogonal projections of the bitstring indexed encoding $U$. If $A$ and $A'$ satisfy the conditions $\|A-A'\| + \big\|\frac{A+A'}{2}\big\|^2 \leq 1$ and $\big\|\frac{A+A'}{2}\big\|^2 \leq \zeta$, then $V$ is a $\big(1,a+1,\frac{\sqrt{2}}{\sqrt{1-\zeta}}k\epsilon\big)$-bitstring indexed encoding of $T^{\SV}_k(A)$. 
\end{lemma}

\begin{proof}
As specified in~\Cref{prop:angles-for-Chebyshev-polys}, we first notice that we can derive the sequence of rotation angles corresponding to Chebyshev polynomials $T_k(x)$ by directly factorizing them. 

% Lemma 9 in~\cite{GSLW18}
\begin{proposition}[Chebyshev polynomials in quantum signal processing, adapted from~Lemma 6 in~\cite{GSLW19}]
    \label{prop:angles-for-Chebyshev-polys}
    Let $T_k \in \bbR[x]$ be the $k$-th Chebyshev polynomial (of the first kind). Consider the corresponding sequence of rotation angles  $\Phi\in\bbR^k$ such that $\phi_1\coloneqq(1-k)\pi/2$, and $\phi_j\coloneqq\pi/2$ for all $j\in[k] \setminus \{1\}$, then we know that $\prod_{j=1}^k \left[\begin{psmallmatrix} \exp(i\phi_j) & 0\\ 0 & \exp(-i\phi_j)\\ \end{psmallmatrix} \begin{psmallmatrix} x & \sqrt{1-x^2}\\ \sqrt{1-x^2} & -x\\ \end{psmallmatrix}\right] = \begin{psmallmatrix} T_k & \cdot\\ \cdot & \cdot \end{psmallmatrix}$.
\end{proposition}

Then, we implement the quantum singular value transformation $T_k^{\SV}(A)$, utilizing an alternating phase modulation (\Cref{prop:phase-modulation}) with the aforementioned sequence of rotation angles, denoted by $V$. 

% Theorem 17 and Lemma 19 in~\cite{GSLW18}
% \cite[Figure 1]{GSLW18}
\begin{proposition}[QSVT by alternating phase modulation, adapted from Theorem 10 and Figure 3 in~\cite{GSLW19}]
    \label{prop:phase-modulation}
    Suppose $P\in\bbC[x]$ is a polynomial, and let $\Phi\in\bbR^n$ be the corresponding sequence of rotation angles. We can construct $P^{\rm (SV)}(\tilde{\Pi} U \Pi) = \begin{cases} 
    \tilde{\Pi} U_{\Phi} \Pi, & \text{if } n \text{ is odd}\\
    \Pi U_{\Phi} \Pi, & \text{if } n \text{ is even}
    \end{cases}$ with a single ancillary qubit. 
    Moreover, this implementation in \emph{\cite[Figure 3]{GSLW19}} makes $k$ uses of $U$, $U^{\dagger}$, $\textsc{C}_{\Pi}\textsc{NOT}$, $\textsc{C}_{\tilde{\Pi}}\textsc{NOT}$, and single-qubit gates. 
\end{proposition}

Owing to the robustness of QSVT (Lemma 22 in~\cite{GSLW18}, full version of~\cite{GSLW19}), we have that 
$\big\| T_k^{\rm (SV)}(U) - T_k^{\rm (SV)}(U')\big\| \leq 4k\sqrt{\|A-A'\|} = 4k\sqrt{\epsilon}$, where $U'$ is a $(1,a,0)$-bitstring indexed encoding of $A$. 
Moreover, with a tighter bound for $A$ and $A'$, namely $\|A-A'\| + \big\|\frac{A+A'}{2}\big\|^2 \leq 1$, we can deduce that $\|T_k^{\SV}(U) - T_k^{\SV}(U')\| \leq k \frac{\sqrt{2}}{\sqrt{1-\|(A+A')/2\|^2}} \|A-A'\| \leq \frac{\sqrt{2}}{\sqrt{1-\zeta}} k\epsilon$ following~\cite[Lemma 23]{GSLW18}, indicating an improved dependence of $\epsilon$.
Finally, we can compute the description of the resulting quantum circuits in $O(\log{k}) = O(s(n))$ space and $O(k)$ times because of the implementation specified in \Cref{prop:phase-modulation}. 
\end{proof}

\vspace{1em}
We then proceed by presenting a linear combination of bitstring indexed encodings, which adapts the LCU technique proposed by Berry, Childs, Cleve, Kothari, and Somma in~\cite{BCC+15}, and incorporates a space-efficient state preparation operator. 
For a non-zero real vector $\bfy=(y_0,\ldots,y_{m-1})$, we say that
$P_{\abs{\bfy}}$ is an $\epsilon$-state preparation operator for $\bfy$ if $P_{\abs{\bfy}}\ket{\bar{0}} \coloneqq \sum_{i=0}^{m-1}\sqrt{\hat{y}'_i}\ket{i}$ for some $\hat{\bfy}'$ satisfying $ \|\abs{\bfy}/\|\bfy\|_1-\hat{\bfy}'\|_1 \leq \epsilon$. 

% Lemma 52 in~\cite{GSLW18}
\begin{lemma}[Linear combinations of bitstring indexed encodings, adapted from Lemma 29 in~\cite{GSLW19}] 
    \label{lemma:space-efficient-LCU}
    Given a matrix $A=\sum_{i=0}^{m-1} y_i A_i$ such that each linear operator $A_i$~$(0 \leq i < m)$ acts on $s$ qubits with the corresponding $(1,a,\epsilon_1)$-bitstring indexed encoding $U_i$ acting on $s+a$ qubits associated with the same projections $\tilde{\Pi}$ and $\Pi$. Also each $y_i$~$(0 \leq i < m)$ can be expressed in $O(s(n))$ bits with an evaluation oracle $\Eval$ that returns $\hat{y}_i$ with precision $\varepsilon\coloneqq O(\epsilon_2^2/m)$. Then utilizing an $\epsilon_2$-state preparation operator $P_{\abs{\bfy}}$ for the nonnegative vector $(\abs*{y_0},\ldots,\abs*{y_{m-1}})$ acting on $O(\log m)$ qubits, the diagonal unitary $D_{\bfy}=\sum_{i=0}^{m-1}\sigma_i \ket{i}\bra{i}+\rbra*{ I-\sum_{i=0}^{m-1}\ket{i}\bra{i}}$, 
    where $\sigma_i\in\{\pm 1\}$ satisfies $y_i=\sigma_i\abs*{y_i}$,
    and a $(s+a+\lceil\log{m}\rceil)$-qubit unitary $W=\sum_{i=0}^{m-1} \ket{i}\bra{i} \otimes U_i + \big(I-\sum_{i=0}^{m-1} \ket{i}\bra{i}\big)\otimes I$, we can implement a $(\|\bfy\|_1,a+\lceil\log{m}\rceil, \epsilon_1 \|\bfy\|_1 + \epsilon_2  \|\bfy\|_1)$-bitstring indexed encoding of $A$ acting on $s + a + \lceil\log m\rceil$ qubits with a single use of $W$, $P_{\abs{\bfy}}$, $P_{\abs{\bfy}}^{\dagger}$, and $D_\bfy$. 
    In addition, the (classical) pre-processing can be implemented in deterministic time $\tilde{O}(m^2 \log(m/\epsilon_2))$ and space $O(\log(m/\epsilon_2^2))$, as well as $m^2$ oracle calls to $\Eval$ with precision $\varepsilon$. 
\end{lemma}

\begin{proof}
    For the $\epsilon_2$-state preparation operator $P_{\abs{\bfy}}$ such that 
    \[P_{\abs{\bfy}}\ket{\bar{0}}=\sum_{i=0}^{m-1} \sqrt{\hat{y}'_i} \ket{i},\quad\text{where } y'_i\coloneqq \frac{\abs*{y_i}}{\norm{\bfy}_1} \quad\text{and}\quad \sum_{i=0}^{m-1} \abs*{y'_i-\hat y'_i}\leq \epsilon_2,\] 
    we utilize a scheme introduced by Zalka~\cite{Zalka98} (also independently rediscovered in~\cite{GR02} and~\cite{KM01}). 
    We make an additional analysis of the required classical computational complexity, and the proof of \Cref{prop:real-state-preparation} follows immediately afterward: 
    \begin{proposition}[Space-efficient state preparation, adapted from~{\cite{Zalka98,KM01,GR02}}]
        \label{prop:real-state-preparation}
        Given an $l$-qubit quantum state $\ket{\psi}\coloneqq\sum_{i=0}^{m-1} \sqrt{p_i}\ket{i}$, where $l=\ceil*{\log{m}}$, $p_i\geq 0$, $\sum_{i=0}^{m-1}p_i=1$, and $p_i$ are associated with an evaluation oracle $\Eval(i,\varepsilon)$ that returns $p_i$ up to accuracy $\varepsilon$, we can prepare $\ket{\psi}$ up to accuracy $\epsilon$ in deterministic time $\tilde{O}(m^2\log(m/\epsilon))$ and space $O(\log(m/\epsilon^2))$, together with $m^2$ evaluation oracle calls with precision $\varepsilon\coloneqq O(\epsilon^2/m)$. 
    \end{proposition}

    Applying \Cref{prop:real-state-preparation} to the normalized nonnegative vector $(y'_0,\ldots,y'_{m-1})$ with $\epsilon=\epsilon_2$ gives the desired $P_{|\bfy|}$ with the same deterministic time, space, and oracle-call bounds as in the statement. The signs of the coefficients are implemented separately by the diagonal unitary $D_{\bfy} = \sum_{i=0}^{m-1}\sigma_i\ket{i}\bra{i} + \rbra*{ I-\sum_{i=0}^{m-1}\ket{i}\bra{i} }$, where $\sigma_i\in\{\pm 1\}$ satisfies $y_i=\sigma_i|y_i|$.

    Now consider the bitstring indexed encoding $\rbra[\big]{ P_{\abs{\bfy}}^{\dagger} D_\bfy \otimes I_{s+a} } W \rbra[\big]{ P_{\abs{\bfy}}\otimes I_{s+a} }$ of $A$ acting on $s+a+\lceil\log{m}\rceil$ qubits. Let $y'_i\coloneqq \abs{y_i}/\norm{\bfy}_1$ so that $y_i = \sigma_i y'_i \norm{\bfy}_1$, then we obtain the implementation error:
    \begin{align*}
        & \big\| A - \|\bfy\|_1 \big(\ket{\bar{0}}\bra{\bar{0}}\otimes \tilde{\Pi}\big)  \rbra[\big]{ P_{\abs{\bfy}}^{\dagger}D_\bfy\otimes I_{s+a} } W \rbra[\big]{ P_{\abs{\bfy}} \otimes I_{s+a} }  \left(\ket{\bar{0}}\bra{\bar{0}}\otimes \Pi\right) \big\| \\
        = &  \norm[\big]{ \textstyle \norm{\bfy}_1 \sum_{i=0}^{m-1} \sigma_i y'_i A_i - \norm{\bfy}_1 \sum_{i=0}^{m-1} \sigma_i \hat y'_i \, \tilde{\Pi} U_i \Pi} \\
        \leq &  \norm[\big]{ \textstyle \norm{\bfy}_1 \sum_{i=0}^{m-1} \sigma_i y'_i \big(A_i-\tilde{\Pi} U_i \Pi\big) } + \norm[\big]{ \norm{\bfy}_1 \sum_{i=0}^{m-1} \sigma_i \big(y'_i-\hat y'_i\big) \tilde{\Pi} U_i \Pi }  \\
        \leq &  \norm{\bfy}_1 \textstyle\sum_{i=0}^{m-1} y'_i \norm{ A_i-\tilde{\Pi} U_i \Pi} + \norm{\bfy}_1 \sum_{i=0}^{m-1} \abs*{y'_i-\hat y'_i} \norm{\tilde{\Pi} U_i \Pi}  \\
        \leq &  \norm{\bfy}_1 \rbra{ \textstyle \sum_{i=0}^{m-1} y'_i } \epsilon_1 + \epsilon_2 \, \norm{\bfy}_1  \\
        = & (\epsilon_1 + \epsilon_2) \|\bfy\|_1.
    \end{align*}
    Here, the second line uses $y_i=\sigma_i y'_i\norm{\bfy}_1$, $\bra{\bar{0}}P^\dagger_{\abs{\bfy}}D_\bfy\ket{i}=\sigma_i\sqrt{\hat{y}'_i}$, and $\bra{i}P_{\abs{\bfy}}\ket{\bar{0}}=\sqrt{\hat{y}'_i}$ for each $0\leq i<m$; the third and fourth lines follow from the triangle inequality and $\abs{\sigma_i}=1$; and the fifth line follows because each $U_i$ is a $(1,a,\epsilon_1)$-bitstring indexed encoding of $A_i$, $\norm{\tilde{\Pi}U_i\Pi}\leq 1$, and the error bound guaranteed by \Cref{prop:real-state-preparation}.

    Finally, the diagonal unitary $D_{\bfy}$ is computed from the signs of the coefficients and does not change the asymptotic deterministic time or space bounds of the classical preprocessing, which completes the proof.
\end{proof}

\begin{proof}[Proof of \Cref{prop:real-state-preparation}]
    We follow the analysis presented in \cite[Section III.A]{MP16}, with a particular focus on the classical computational complexity required for this state preparation procedure. The algorithm for preparing the state $\ket{\psi}$ expresses the weight $W_x$ as a telescoping product, given by
    \begin{subequations}
        \label{eq:state-preparation-telescoping-product}
        \begin{align}
            \forall x \in \binset^l,~W_x &= W_{x_1} \cdot \frac{W_{x_1 x_2}}{W_{x_1}} \cdot \frac{W_{x_1 x_2 x_3}}{W_{x_1 x_2}} \cdots \frac{W_x}{W_{x_1\cdots x_{l-1}}},\\
            \text{ where } W_x &\coloneqq \sum_{y\in\binset^{l-|x|}} |\innerprod{xy}{\psi}|^2.
        \end{align}
    \end{subequations}

    To estimate $\ket{\psi}$ up to accuracy $\epsilon$ in the $\ell_2$ norm, it suffices to approximate each weight $W_x$ up to additive error $\varepsilon \coloneqq O(\epsilon^2/m)$, as indicated in \cite[Section III.A]{MP16}. To compute $W_{x'}$, we need $2^{l-|x'|}$ oracle calls to $\Eval(\cdot, \varepsilon)$. Evaluating all terms in \Cref{eq:state-preparation-telescoping-product} requires computing $W_{x_1}, W_{x_1x_2}, \cdots, W_x$ for any $x\in\binset^l$, which can be achieved by $2^{l-1}+2^{l-2}+\cdots+1 = 2^l$ oracle calls to $\Eval(\cdot, \varepsilon)$. As we need to compute \Cref{eq:state-preparation-telescoping-product} for all $x\in\{0,1\}^l$, the overall number of oracle calls to $\Eval(\cdot, \varepsilon)$ is $2^{2l}=m^2$. The remaining computation can be achieved in deterministic time $\tilde{O}(m^2\log(m/\epsilon))$ and space $O(\log(m/\epsilon^2))$ where the time complexity is because of the iterated integer multiplication.  
\end{proof}

\vspace{1em}
To make the resulting bitstring indexed encoding from \Cref{lemma:space-efficient-LCU} with $\alpha = 1$, we need to perform \textit{a renormalization procedure} to construct a new encoding with the desired $\alpha$. 
We achieve this by extending the proof strategy outlined by Gilyen \cite[Page 52]{Gilyen19} for block-encodings to bitstring indexed encodings. This approach works specifically for \textit{partial isometries} (up to a normalization factor $\alpha$) -- since the singular values of a partial isometry are either $0$ or $1$, it suffices to consider a space-efficient QSVT associated with some Chebyshev polynomial $T_k$, with an appropriately chosen odd $k$, such that $T_k(1/\alpha)=1$ and $T_k(0/\alpha)=0$.\footnote{Renormalizing bitstring indexed encodings of \textit{non-partial isometries} for space-efficient QSVT seems achievable by mimicking~\cite[Theorem 17]{GSLW19}. This approach cleverly uses space-efficient QSVT with the sign function (\Cref{corr:sign-polynomial-implementation}), where the corresponding encoding can be re-normalized by carefully using \Cref{lemma:renormalizing-encodings}. Nevertheless, since this renormalization procedure is not required in this paper, we leave it for future work.}

The renormalization procedure is provided in \Cref{lemma:renormalizing-encodings}.
Additionally, a similar result has been established in {\cite[Lemma 7.10]{MY23}}.

\begin{lemma}[Renormalizing bitstring indexed encoding]
    \label{lemma:renormalizing-encodings}
    Let $U$ be an $(\alpha, a, \epsilon)$-bitstring indexed encoding of $A$, where $\alpha > 1$ and $0 < \epsilon < 1$, and $A$ is a partial isometry acting on $s(n)$ qubits. We can implement a quantum circuit $V$, serving as a normalization of $U$, such that $V$ is a $(1, a+2, 36\epsilon)$-bitstring indexed encoding of $A$. This implementation requires $O(\alpha)$ uses of $U$, $U^{\dagger}$, $\textsc{C}_{\Pi}\textsc{NOT}$, $\textsc{C}_{\tilde{\Pi}}\textsc{NOT}$, and $O(\alpha)$ single-qubit gates. Moreover, the description of the resulting quantum circuit can be computed in deterministic time $O(\alpha)$ and space $O(s)$.
\end{lemma}

\begin{proof}
    Following \Cref{def:bitstring-indexed-encoding}, we have $\| A - \alpha \tilde{\Pi} U \Pi \| \leq \epsilon$, where $\tilde{\Pi}$ and $\Pi$ are the corresponding orthogonal projections. 
    Because $U$ is a $(1, a, \epsilon/\alpha)$-bitstring indexed encoding $A/\alpha$, we obtain that $\|A / \alpha\| \leq \| U \| + \epsilon/\alpha = 1 + \epsilon/\alpha$, equivalently $\|A\| \leq \alpha + \epsilon$. 
    
    \paragraph*{Adjusting the encoding through a single-qubit rotation.} 
    Consider an integer $k \coloneqq 4\ceil*{\pi(\alpha+1)/4}+1 \leq 9\alpha = O(\alpha)$ so that $k \equiv 1 \pmod 4$ and $\gamma \coloneqq (\alpha+\epsilon)\sin(\pi/2k) \leq 1$. 
    We define new orthogonal projections $\tilde \Pi' \coloneqq \tilde \Pi \otimes \ket{0}\bra{0}$ and $\Pi' \coloneqq \Pi \otimes \ket{0}\bra{0}$, and combine them with $U' = U \otimes R_\gamma$, where 
    $R_\gamma = \begin{psmallmatrix}
        \gamma & -\sqrt{1-\gamma^2} \\
        \sqrt{1-\gamma^2} & \gamma
    \end{psmallmatrix}$. 
    By noting that $\tilde{\Pi'} U' \Pi' = \gamma \tilde{\Pi} U \Pi \otimes \ket{0}\bra{0}$, we deduce that $U'$ is a $(1, a + 1, \gamma \epsilon/\alpha)$-bitstring indexed encoding of $\gamma A/\alpha \otimes \ket{0}\bra{0}$, which is consequently a $(1, a+1, 2\gamma\epsilon/\alpha)$-bitstring indexed encoding of $\sin\rbra[\big]{\frac{\pi}{2k}} A \otimes \ket{0}\bra{0}$. An error bound follows:
    \[
    \left\| \frac{\gamma}{\alpha} A - \sin\left(\frac{\pi}{2k}\right) A \right\| = \left\| \frac{\epsilon}{\alpha} \sin\left(\frac{\pi}{2k}\right) A \right\| \leq \frac{\epsilon}{\alpha} \sin\left(\frac{\pi}{2k}\right) (\alpha+\epsilon) = \frac{\gamma \epsilon}{\alpha}.
    \]

    \paragraph*{Renormalizing the encoding via robust oblivious amplitude amplification.} 
    We follow the construction in~\cite[Theorem 28]{GSLW18}, the full version of~\cite{GSLW19}, and perform a meticulous analysis of the complexity. 
    We observe that it suffices to consider $k \geq 3$, as for $U'$ is already a $(1,a+1,2\gamma\epsilon/\alpha)$-bitstring indexed encoding of $A\otimes \ket{0}\bra{0}$ when $k=1$. 
    Let $\varepsilon \coloneqq 2\gamma\epsilon/\alpha$, and for simplicity, we first start by considering the case with $\varepsilon=0$. By \Cref{def:bitstring-indexed-encoding}, we have $\tilde{\Pi}' U' \Pi' = \sin\rbra*{\frac{\pi}{2k}} A \otimes \ket{0}\bra{0}$.  
    Let $T_k\in\bbR[x]$ be the degree-$k$ Chebyshev polynomial (of the first kind). By employing \Cref{lemma:Chebyshev-poly-implementation}, we can apply the QSVT associated with $T_k$ to the bitstring indexed encoding $U'$, yielding: 
    \begin{equation*}
        \tilde{\Pi}' T_k^{\SV}(U') \Pi' 
        =  T_k\rbra[\Big]{\sin\rbra[\Big]{\frac{\pi}{2k}}} A  \otimes \ket{0}\bra{0}\\
        = \cos\Big( \frac{k-1}{2}\pi \Big) A \otimes \ket{0}\bra{0}\\
        = A \otimes \ket{0}\bra{0}.
    \end{equation*}
    Here, the second equality is due to $T_k\big(\sin\big(\frac{\pi}{2k}\big)\big) = T_k\big(\cos\big(\frac{\pi}{2}-\frac{\pi}{2k}\big)\big) = \cos\big( \frac{k-1}{2} \pi \big)$, and the last equality holds because $k\equiv 1 \pmod 4$.
    
    Next, we move on the case with $\varepsilon > 0$ and restrict it to $\varepsilon \leq 1/3$.\footnote{If $\varepsilon > 1/3$, then $\| \tilde{\Pi}' U' \Pi' - A\otimes \ket{0}\bra{0}\| \leq 2 < 6\varepsilon$ always holds, implying that we can directly use $U'$ as V.}
    Let $A' \coloneqq \tilde{\Pi}' U' \Pi'$ and $\hat{A} \coloneqq \sin\rbra*{\frac{\pi}{2k}} A \otimes \ket{0}\bra{0}$, then we have $\|A' - \hat{A}\| \leq \varepsilon$, indicating that $\big\|\frac{A'+\hat{A}}{2}\big\|^2 \leq \frac{4}{9}\coloneqq\zeta$\footnote{This is because $\|A'+\hat{A}\| \leq \|A'\|+\|A'\| + \|A'-\hat{A}\| \leq 2\sin(\pi/2k) + \varepsilon \leq 2\sin(\pi/6)+1/3 = 4/3$.} and $\|A'-\hat{A}\| + \big\|\frac{A'+\hat{A}}{2}\big\|^2 \leq \frac{1}{3} + \frac{4}{9} < 1$. By employing \Cref{lemma:Chebyshev-poly-implementation}, as well as the facts that $\frac{\sqrt{2}}{\sqrt{1-\zeta}} < 2$ and $2k\varepsilon = 4k\gamma\epsilon/\alpha \leq 36\epsilon$, we can construct a $(1,a+2,36\epsilon)$-bitstring indexed encoding of $A$, denoted by $V$. 

    \vspace{1em}
    Finally, we provide the computational resources required for implementing $V$. 
    As shown in \Cref{lemma:Chebyshev-poly-implementation}, the implementation of $V$ requires $O(\alpha)$ uses of $U$, $U^{\dagger}$, $\textsc{C}_{\Pi}\textsc{NOT}$, $\textsc{C}_{\tilde{\Pi}}\textsc{NOT}$, and $O(\alpha)$ single-qubit gates. Furthermore, the description of the resulting quantum circuit can be computed in deterministic time $O(\alpha)$ and space $O(s)$.
\end{proof}

\vspace{1em}
Finally, we combine \Cref{lemma:Chebyshev-poly-implementation,lemma:space-efficient-LCU,lemma:renormalizing-encodings} to proceed with the proof of \Cref{thm:LCU-averaged-chebyshev-truncation}. 

\begin{proof}[Proof of \Cref{thm:LCU-averaged-chebyshev-truncation}]
    By using \Cref{lemma:Chebyshev-poly-implementation}, we obtain $(1,a+1,4k\sqrt{\epsilon_1})$-bitstring indexed encodings $V_k$ corresponding to $T_k(A)$, where $1 \leq k \leq d'=2d-1$. 
    If $\hat{c}_0 \neq 0$, let $V_0$ denote the trivial exact bitstring indexed encoding of $T_0(A)/2 = I/2$ acting on the same $s+a+1$ qubits; this can be implemented by a single-qubit rotation on the additional ancillary qubit.
    The descriptions of quantum circuits $\{V_k\}_{k=0}^{d'}$ can be computed in $O(s(n))$ space and $\sum_{k=0}^{d'} k=O(d^2)$ time.

    If $U$ is a block-encoding, then all $V_k$ are associated with the same projections. Otherwise, since $P_{d'}$ is either even or odd, only one parity of Chebyshev polynomials appears in $P_{d'}$; hence all relevant $V_k$ are associated with the same projections by \Cref{prop:phase-modulation}.
    Employing \Cref{lemma:space-efficient-LCU}, we obtain a $(\|\hat{\bfc}\|_1, \hat{a}, 4d'\sqrt{\epsilon_1}\|\hat{\bfc}\|_1+\epsilon_2\|\hat{\bfc}\|_1)$-bitstring indexed encoding $V_{\rm unnorm}$ for $P_{d'}(A)=\hat{c}_0I/2+\sum_{k=1}^{d'} \hat{c}_k T_k(A)$, where $\hat{a}\coloneqq a+\lceil\log{d}\rceil+1$. 
    The remaining analysis depends on whether $P_{d'}(A)$ is a partial isometry (up to a normalization factor): 
    \begin{itemize}[itemsep=0.33em,topsep=0.33em,parsep=0.33em]
    \item \textbf{Partial isometry $P_{d'}(A)$:} Since $P_{d'}(A)$ is a partial isometry, we have $\norm{P_{d'}(A)}=1$. On the other hand, $\norm{P_{d'}(A)} \leq |\hat{c}_0|/2 + \sum_{k=1}^{d'} |\hat{c}_k| \leq \|\hat{\bfc}\|_1$, and hence $\|\hat{\bfc}\|_1 \geq 1$. If $\|\hat{\bfc}\|_1=1$, then we simply take $V_{\rm normed}\coloneqq V_{\rm unnorm}$. Otherwise, we can renormalize $V_{\rm unnorm}$ by utilizing \Cref{lemma:renormalizing-encodings} and obtain a $(1, a', 144d'\sqrt{\epsilon_1}\|\hat{\bfc}\|_1 + 36\epsilon_2\|\hat{\bfc}\|_1)$-bitstring indexed encoding $V_{\rm normed}$ that acts on $s+a'$ qubits, where $a' \coloneqq \hat{a}+2 = a+\lceil\log{d'}\rceil+3$.
    A direct calculation shows that the implementation of $V_{\rm normed}$ makes $\sum_{k=1}^{d'} k \cdot O(\|\hat{\bfc}\|_1)=O(d^2\|\hat{\bfc}\|_1)$ uses of $U$, $U^{\dagger}$, $\textsc{C}_{\Pi}\textsc{NOT}$, $\textsc{C}_{\tilde{\Pi}}\textsc{NOT}$, and multi-controlled single-qubit gates. 
    The description of the quantum circuit $V_{\rm normed}$ thus can be computed in deterministic time 
    \[\max\{\tilde{O}\big((d')^2\|\hat{\bfc}\|_1\log(d'/\epsilon_2)\big),O((d')^2\|\hat{\bfc}\|_1)\} = \tilde{O}(d^2\|\hat{\bfc}\|_1\log(d/\epsilon_2))\]
    and space $O(\max\{s(n), \log(d'/\epsilon_2^2)\}) = O(\max\{s(n), \log(d/\epsilon_2^2)\})$, as well as $O\big((d')^2\|\hat{\bfc}\|_1\big) = O(d^2\|\hat{\bfc}\|_1)$ oracle calls to $\Eval$ with precision $\varepsilon$. 
    \item \textbf{General $P_{d'}(A)$:} We simply use the bitstring indexed encoding $V_{\rm unnorm}$ without renormalizing it. Similarly, the implementation of $V_{\rm unnorm}$ makes $O(d^2)$ uses of $U$, $U^{\dagger}$, $\textsc{C}_{\Pi}\textsc{NOT}$, $\textsc{C}_{\tilde{\Pi}}\textsc{NOT}$, and multi-controlled single-qubit gates. Therefore, the description of the quantum circuit $V_{\rm unnorm}$ can be computed in deterministic time $\tilde{O}(d^2\log(d/\epsilon_2))$ and space $O(\max\{s(n), \log(d/\epsilon_2^2)\})$, as well as $O(d^2)$ oracle calls to $\Eval$ with precision $\varepsilon$.
    \end{itemize}
    
    Finally, we can extend our construction to any linear operator $A$ by replacing $P_{d'}(A)$ with $P_{d'}^{\SV}$ as defined in \Cref{def:matrix-SV-function}, taking into account that the Chebyshev polynomial (of the first kind) $T_k$ is either an even or an odd function.
\end{proof}

\subsection{Examples: the sign function and the normalized logarithmic function}
\label{subsec:space-efficient-QSVT-examples}

In this subsection, we provide explicit examples that illustrate the usage of the space-efficient quantum singular value transformation (QSVT) technique. We define two functions: 
\begin{equation*}
    \label{eq:space-efficient-QSVT-examples}
    \sign(x)\coloneqq{\scriptscriptstyle
    \begin{cases}
    1,& x > 0\\
    -1,& x < 0\\
    0,& x=0
    \end{cases}}
    ~~~~~\text{and}~~~~~
    \ln_{\beta}(x)\coloneqq \frac{\ln(1/x)}{2\ln(2/\beta)}.
\end{equation*}

In particular, the sign function is a bounded function, and we derive the corresponding bitstring indexed encoding with \textit{deterministic} space-efficient (classical) pre-processing in \Cref{corr:space-efficient-sign}. On the other hand, the logarithmic function is a piecewise-smooth function that is bounded by $1$, and we deduce the corresponding bitstring indexed encoding with \textit{randomized} space-efficient (classical) pre-processing in \Cref{corr:space-efficient-log}.

\begin{corollary}[Sign polynomial with space-efficient coefficients applied to bitstring indexed encodings]
    \label{corr:sign-polynomial-implementation}
    Let $A$ be an Hermitian matrix that acts on $s$ qubits, where $s(n) \geq \Omega(\log(n))$. Let $U$ be a $(1,a, \epsilon_1)$-bitstring indexed encoding of $A$ that acts on $s+a$ qubits. 
    Then, for any $d' \leq 2^{O(s(n))}$ and $\epsilon_2 \geq 2^{-O(s(n))}$, we have an 
    $\rbra[\big]{ 1, a+\lceil\log d'\rceil+3, 144\hat{C}_{\sign}d'\log d' \epsilon_1^{1/2} + (36\hat{C}_{\sign}\log d' +37)\epsilon_2 }$-bitstring indexed encoding $V$ of $P_{d'}^{\sign}(A)$, where $P_{d'}^{\sign}$ is a space-efficient bounded polynomial approximation of the sign function \emph{(}corresponding to some degree-$d$ averaged Chebyshev truncation\emph{)} specified in \Cref{corr:space-efficient-sign}, and $\hat{C}_{\sign}$ is a universal constant. 
    This implementation requires $O(d^2 \log{d})$ uses of $U$, $U^{\dagger}$, $\textsc{C}_{\Pi}\textsc{NOT}$,  $\textsc{C}_{\tilde{\Pi}}\textsc{NOT}$, and $O(d^2 \log{d})$ multi-controlled single-qubit gates. The description of $V$ can be computed in deterministic time $\tilde{O}(\epsilon_2^{-1} d^{9/2})$ and space $O(s(n))$.

    \noindent Furthermore, our construction directly extends to any non-Hermitian \emph{(}but linear\emph{)} matrix $A$ by simply replacing $P_{d'}^{\sign}(A)$ with $P_{\sign,d'}^{\SV}(A)$ defined in the same way as \Cref{def:matrix-SV-function}.
\end{corollary}

\begin{proof}
    Following \Cref{corr:space-efficient-sign}, we have $P_{d'}^{\sign}(x) = \hat{c}_0/2+\sum_{k=1}^{d'} \hat{c}_k T_k(x)$, where $d'=2d-1$ and $d'=O(\delta^{-1}\log{\epsilon^{-1}})$. The approximation error is given by: 
    \begin{equation}
        \label{eq:sign-approx-error}
        \forall x\in[-1,1] \setminus [-\delta,\delta],~|\sign(x)-P_{d'}^{\sign}(x)| \leq C_{\sign}\epsilon \coloneqq\epsilon_2.
    \end{equation}
    To implement $\Eval$ with precision $\varepsilon = O(\epsilon_2^2/d')$, we can compute the corresponding entry $\hat{c}_k$ of the coefficient vector, which requires deterministic time $\tilde{O}\big(\varepsilon^{-1/2} (d')^2\big)=\tilde{O}(\epsilon_2^{-1} d^{5/2})$ and space $O\big(\log(\varepsilon^{-3/2}(d')^3)\big)=O(\log(\epsilon_2^{-3}d^{9/2}))$.

    Note that $P_{d'}^{\sign}(A)$ is not a partial isometry (up to a normalization factor) and $\|\hat{\bfc}\|_1\leq \hat{C}_{\sign}\log d'$. Using \Cref{thm:LCU-averaged-chebyshev-truncation}, we have a $\rbra[\big]{\hat{C}_{\sign}\log d', a_{\rm un}, \rbra{4d'\epsilon_1^{1/2}+\epsilon_2}\hat{C}_{\sign}\log d'}$-bitstring indexed encoding $V_{\rm un}$, with projections $\tilde{\Pi}$ and $\Pi$, that acts on $s+a_{\rm un}$ qubits and $a_{\rm un} \coloneqq a+\lceil \log d' \rceil+1$. 

    \paragraph*{Renormalizing the encoding of $P^{\sign}_{d'}(A)$.}
    Notably, the renormalization procedure (\Cref{lemma:renormalizing-encodings}) is \textit{still applicable} when $P_{d'}^{\sign}(A)$ is restricted to appropriately chosen subspaces $\Gamma_L$ and $\Gamma_R$. 
    Let $A = \sum_{i=1}^{\rank(A)} \sigma_i \bfu_i \bfv_i^\dagger$ be the singular value decomposition of $A$. We define $A_{>\delta} \coloneqq \sum_{i: \sigma_i > \delta} \sigma_i \bfu_i \bfv_i^\dagger$, as well as subspaces $\Gamma_L\coloneqq\spanset\{\bfu_i|\sigma_i > \delta\}$ and $\Gamma_R\coloneqq\spanset\{\bfv_i | \sigma_i > \delta\}$. Consequently, we obtain the following for the bitstring indexed encoding $V_{\rm un}$ with projections $\tilde{\Pi}$ and $\Pi$:
    
    \begin{subequations}
    \label{eq:sign-subspace-isometry-unnorm}
    \begin{align}
        &\| \sign^{\SV}(A_{> \delta}) - \hat{C}_{\sign}\log d' \cdot \tilde{\Pi}|_{\Gamma_L} V_{\rm un} \Pi|_{\Gamma_R} \|\\
        \leq~&  \| \sign^{\SV}(A_{> \delta}) - P_{d'}^\sign(A_{> \delta})  \| + \| P_{d'}^\sign(A_{> \delta})  - \hat{C}_{\sign}\log d' \cdot \tilde{\Pi}|_{\Gamma_L} V_{\rm un} \Pi|_{\Gamma_R} \| \\
        \leq~& \epsilon_2 + \| P_{d'}^\sign(A_{> \delta})  - \hat{C}_{\sign}\log d' \cdot \tilde{\Pi} V_{\rm un} \Pi \| \\
        \leq~& \epsilon_2 + \rbra{4d'\epsilon_1^{1/2}+\epsilon_2}\hat{C}_{\sign}\log d'.
    \end{align}
    \end{subequations}
    
    Here, the second line owes to the triangle inequality, the third line uses the fact that restricting to subspaces cannot increase the operator norm together with \Cref{eq:sign-approx-error}, and the last line is because $V_{\rm un}$ is a $\rbra[\big]{\hat{C}_{\sign}\log d', a_{\rm un}, \rbra{4d'\epsilon_1^{1/2}+\epsilon_2}\hat{C}_{\sign}\log d'}$-bitstring indexed encoding of $P_{d'}^\sign(A_{> \delta})$. 
    Note that $\sign^{\SV}(A_{> \delta})$ is a partial isometry, and \Cref{eq:sign-subspace-isometry-unnorm} implies that $V_{\rm un}$ is a projected unitary encoding of $\sign^{\SV}(A_{> \delta})$. By applying \Cref{lemma:renormalizing-encodings} to $V_{\rm un}$ with projections $\tilde{\Pi}|_{\Gamma_L}$ and $\Pi|_{\Gamma_R}$, we can obtain a $\rbra[\big]{1, a_{\rm un}+2, 144\hat{C}_{\sign} d'\log d'\epsilon_1^{1/2} + 36(\hat{C}_{\sign}\log d'+1)\epsilon_2}$-projected unitary encoding of $\sign^{\SV}(A_{> \delta})$, denoted as $V$. Consequently, we can derive that: 
    
    \begin{subequations}
        \label{eq:sign-subspace-isometry-normed}
        \begin{align}
            &\| P^{\sign}_{d'}(A_{>\delta}) - \tilde{\Pi}|_{\Gamma_L} V \Pi|_{\Gamma_R}\| \\
            \leq~& \| P^{\sign}_{d'}(A_{>\delta}) - \sign^{\SV}(A_{>\delta})\| + \| \sign^{\SV}(A_{>\delta}) - \tilde{\Pi}|_{\Gamma_L} V \Pi|_{\Gamma_R}\| \\
            \leq~& \epsilon_2 + 144\hat{C}_{\sign}d'\log d' \epsilon_1^{1/2} + 36(\hat{C}_{\sign}\log d' +1)\epsilon_2.
        \end{align}
    \end{subequations}
    
    Here, the second line follows from the triangle inequality, and the third line additionally owes to \Cref{eq:sign-approx-error}. 
    Noting that \Cref{lemma:renormalizing-encodings} essentially applies a Chebyshev polynomial to $V_{\rm un}$ and preserves the projections $\tilde{\Pi}|_{\Gamma_L}$ and $\Pi|_{\Gamma_R}$, then we have $\|\tilde{\Pi} V \Pi\| \leq 1$. 
    Therefore, following \Cref{eq:sign-subspace-isometry-normed}, we conclude that $P^{\sign}_{d'}(A)$ has a $\rbra[\big]{1, a', 144\hat{C}_{\sign}d'\log d' \epsilon_1^{1/2} + (36\hat{C}_{\sign}\log d' +37)\epsilon_2 }$-bitstring indexed encoding $V$ that acts on $s+a'$ qubits, where $a'\coloneqq a+ \ceil*{\log{d'}} +3$.\footnote{We are somewhat abusing notations --- strictly speaking, $V$ corresponds to $\tilde{P}^{\sign}_{d'}(A)$, where $\tilde{P}^{\sign}_{d'}$ is another polynomial satisfying all requirements in \Cref{corr:space-efficient-sign} but does not necessarily exactly coincide with $P^{\sign}_{d'}$. } 
    
    \vspace{1em}
    Lastly, we can complete the remaining analysis similarly to \Cref{thm:LCU-averaged-chebyshev-truncation} with a partial isometry $P_{d'}(A)$. 
    Since $\|\hat{\bfc}\|_1 \leq \hat{C}_{\sign}\log d'$, the quantum circuit of $V$ makes $O(d^2\log d)$ uses of $U$, $U^{\dagger}$, $\textsc{C}_{\Pi}\textsc{NOT}$, and $\textsc{C}_{\tilde{\Pi}}\textsc{NOT}$ as well as $O(d^2\log d)$ multi-controlled single-qubit gates. 
    We note that $d'=O(d) \leq 2^{O(s(n))}$ and $\epsilon_2 \geq 2^{-O(s(n))}$. Moreover, we can compute the description of $V$ in $O(s(n))$ space since each oracle call to $\Eval$ with precision $\varepsilon$ can be computed in $O(\log(\epsilon_2^{-3} d^{9/2}))$ space. Additionally, the time complexity for computing the description of $V$ is 
    \[  \max\{\tilde{O}(d^2(\log d)\log(d/\epsilon_2)), O(d^2\log d) \cdot \tilde{O}(\epsilon_2^{-1} d^{5/2})\} = \tilde{O}(\epsilon_2^{-1} d^{9/2}).  \qedhere\]
\end{proof}

\begin{corollary}[Log polynomial with space-efficient coefficients applied to bitstring indexed encodings]
    \label{corr:log-polynomial-implementation}
    Let $A$ be an Hermitian matrix that acts on $s$ qubits, where $s(n) \geq \Omega(\log(n))$. Let $U$ be a $(1,a, \epsilon_1)$-bitstring indexed encoding of $A$ that acts on $s+a$ qubits. 
    Then, for any $d'=2d-1 \leq 2^{O(s(n))}$, $\epsilon_2 \geq 2^{-O(s(n))}$, and $\beta \geq 2^{-O(s(n))}$, we have a $\rbra[\big]{ \hat{C}_{\ln}d'^{1/2}, a+\ceil*{\log{d'}}+1, \rbra{4d' \epsilon_1^{1/2} + \epsilon_2} \hat{C}_{\ln}d'^{1/2} }$-bitstring indexed encoding $V$ of $P_{d'}^{\ln}(A)$, where $P_{d'}^{\ln}$ is a space-efficient bounded polynomial approximation of the normalized log function \emph{(}corresponding to some degree-$d$ averaged Chebyshev truncation\emph{)} specified in \Cref{corr:space-efficient-log}, and $\hat{C}_{\ln}$ is a universal constant. 
    This implementation requires $O(d^2)$ uses of $U$, $U^{\dagger}$, $\textsc{C}_{\Pi}\textsc{NOT}$,  $\textsc{C}_{\tilde{\Pi}}\textsc{NOT}$, and multi-controlled single-qubit gates. Moreover, we can compute the description of the resulting quantum circuit in bounded-error randomized time $\tilde{O}(\max\{\beta^{-5} \epsilon_2^{-4} d^4, \epsilon_2^{-1} d^{9/2}\})$ and space $O(s(n))$.
\end{corollary}

\begin{proof}
    Following \Cref{corr:space-efficient-log}, we have $P_{d'}^{\ln}(x) = c^{\ln}_0/2+\sum_{k=1}^{d'} c^{\ln}_k T_k(x)$, where $P_{d'}^{\ln}$ corresponds to some degree-$d$ averaged Chebyshev truncation and $d'=2d-1\leq \tilde{C}_{\ln}\beta^{-1}\log\rbra{\epsilon^{-1}}$. For any $\ln_{\beta}(x)$, we have $|\ln_{\beta}(x)-P_{d'}^{\ln}(x)| \leq C_{\ln}\epsilon \coloneqq\epsilon_2$ for all $x\in[\beta,1]$. To implement $\Eval$ with precision $\varepsilon = O(\epsilon_2^2/d)$, we can compute the corresponding entry $c^{\ln}_k$ of the coefficient vector by a bounded-error randomized algorithm. This requires $O(\log(\beta^{-4} \varepsilon^{-3/2} d^3))=O(\log(\beta^{-4} \epsilon_2^{-3} d^{9/2}))$ space and $\tilde{O}(\max\{\beta^{-5}\varepsilon^{-2}, \varepsilon^{-1/2} d^2\}) = \tilde{O}(\max\{\beta^{-5} \epsilon_2^{-4} d^2, \epsilon_2^{-1} d^{5/2}\})$ time. 
    Applying \Cref{thm:LCU-averaged-chebyshev-truncation} with $\|\bfc^{\ln}\|_1 \leq \hat{C}_{\ln}d'^{1/2}$, we conclude that $P_{d'}^{\ln}$ has a $\rbra[\big]{\hat{C}_{\ln}d'^{1/2}, \hat{a}, \rbra{4d'\epsilon_1^{1/2} + \epsilon_2} \hat{C}_{\ln} d'^{1/2}}$-bitstring indexed encoding $V$ that acts on $s+\hat{a}$ qubits, where $\hat{a}\coloneqq a+\lceil\log{d'}\rceil+1$. 

    Furthermore, the quantum circuit of $V$ makes $O((d')^2)=O(d^2)$ uses of $U$, $U^{\dagger}$, $\textsc{C}_{\Pi}\textsc{NOT}$, $\textsc{C}_{\tilde{\Pi}}\textsc{NOT}$, and multi-controlled single-qubit gates. 
    We note that $d'=2d-1 \leq 2^{O(s(n))}$, $\epsilon_2 \geq 2^{-O(s(n))}$, and $\beta \geq 2^{-O(s(n))}$.     
    Additionally, we can compute the description of $V$ in $O(s(n))$ space since each oracle call to $\Eval$ with precision $\varepsilon$ can be computed in $O(\log(\beta^{-4} \epsilon_2^{-3} d^{9/2}))$ space. 
    The time complexity for computing the description of $V$ is given by: 
    \begin{subequations}
    \label{eq:log-poly-time-complexity}
    \begin{align}
        & \max\cbra*{ \tilde{O}(d^2 \log(d/\epsilon_2)), O(d^2) \tilde{O}(\max\{\beta^{-5} \epsilon_2^{-4} d^2, \epsilon_2^{-1} d^{5/2}\}) }  \\
        =~& \tilde{O}(\max\{\beta^{-5} \epsilon_2^{-4} d^4, \epsilon_2^{-1} d^{9/2}\}).
    \end{align}
    \end{subequations}
    
    Finally, to guarantee that the probability that all $O((d')^2)=O(d^2)$ oracle calls to $\Eval$ succeed is at least $2/3$, we use a $\ceil*{4\ln(d'+1)}$-time sequential repetition of $\Eval$ for each oracle call. Together with the Chernoff--Hoeffding bound and the union bound, the resulting randomized algorithm succeeds with probability at least $1-(d')^2 \cdot 2e^{-4\ln(d'+1)} \geq 2/3$. We further note that the time complexity specified in \Cref{eq:log-poly-time-complexity} only increases by a $\ceil*{4\ln(d'+1)}$ factor. 
\end{proof}

\subsection{Application: space-efficient error reduction for unitary quantum computations}
\label{subsec:BQUL-error-reduction}
We provide a unified space-efficient error reduction for unitary quantum computations. In particular, one-sided error scenarios (e.g., $\RQUL$ and $\coRQUL$) have been proven in~\cite{Wat01}, and the two-sided error scenario (e.g., $\BQUL{}$) has been demonstrated in~\cite{FKLMN16}.

\begin{theorem}[Space-efficient error reduction for unitary quantum computations]
    \label{thm:space-efficient-error-reduction-unitaryQC}
    Let $s(n)$ be a space-constructible function, and let $a(n)$, $b(n)$, and $l(n)$ be deterministic $O(s(n))$ space computable functions such that $a(n)-b(n) \geq 2^{-O(s(n))}$, we know that for any $l(n)\leq O(s(n))$, there is $d\coloneqq l(n)/{\max\{\sqrt{a}-\sqrt{b},\sqrt{1-b}-\sqrt{1-a}\}}$ such that
    \[\BQUSPACE[s(n),a(n),b(n)] \subseteq \BQUSPACE\big[s(n)+\lceil\log{d}\rceil+1,1-2^{-l(n)},2^{-l(n)}\big].\]
    Furthermore, for one-sided error scenarios, we have that for any $l(n)\leq 2^{O(s(n))}$: 
    \begin{align*}
    \RQUSPACE[s(n),a(n)] \subseteq \RQUSPACE\big[s(n)+\lceil\log{d_0}\rceil+1,1-2^{-l(n)}\big] \text{ where } d_0\coloneqq \tfrac{l(n)}{{\max\{\sqrt{a},1-\sqrt{1-a}\}}},\\
    \coRQUSPACE[s(n),b(n)] \subseteq \coRQUSPACE\big[s(n)+\lceil\log{d_1}\rceil+1,2^{-l(n)}\big] \text{ where } d_1\coloneqq \tfrac{l(n)}{{\max\{1-\sqrt{b},\sqrt{1-b}\}}}.
    \end{align*}
\end{theorem}

By choosing $s(n)=\Theta(\log(n))$, we derive error reduction for logarithmic-space quantum computation in a unified approach:\footnote{It is noteworthy that the promise errors in~\cite[Corollary 2]{FKLMN16} can be reduced to $2^{-l(n)}$ for any function $l(n)$ that is polynomial in $n$, which is \textit{exponentially} smaller than the promise errors in \Cref{corr:error-reduction-untary-quantum-logspace}. Intuitively, the approaches in~\cite{FKLMN16} rely on probability behavior --- for example, repeating a procedure polynomially many times yields an exponentially small promise error --- whereas our approach relies on polynomial approximations, where the use of numerical integration is limited by logarithmic-bit numerical precision. \label{footnote:space-efficient-error-reduction-subtly}}  
\begin{corollary}[Error reduction for $\BQUL$, $\RQUL$, and $\coRQUL$]
    \label{corr:error-reduction-untary-quantum-logspace}
    For deterministic logspace computable functions $a(n)$, $b(n)$, and $l(n)$ satisfying $a(n)-b(n) \geq 1/\poly(n)$ and $l(n)\leq O(\log{n})$, we have the following inclusions: 
    \begin{align*}
    \BQUL[a(n),b(n)] &\subseteq \BQUL[1-2^{-l(n)},2^{-l(n)}], \\
    \RQUL[a(n)] &\subseteq \RQUL[1-2^{-l(n)}], \\
    \coRQUL[b(n)] &\subseteq \coRQUL[2^{-l(n)}].
    \end{align*}
\end{corollary}

% Theorem 32 in~\cite{GSLW18}
The construction described in \Cref{thm:space-efficient-error-reduction-unitaryQC} crucially relies on \Cref{lemma:space-efficient-singular-value-discrimination}, the proof of which follows directly from Theorem 20 in~\cite{GSLW19}.

\begin{lemma}[Space-efficient singular value discrimination]
\label{lemma:space-efficient-singular-value-discrimination}
Let $0 \leq \alpha < \beta \leq 1$ and $U$ be a $(1,0,0)$-bitstring indexed encoding of $A\coloneqq\tilde{\Pi} U \Pi$, where $U$ acts on $s$ qubits and $s(n)\geq \Omega(\log{n})$. 
Consider an unknown quantum state $\ket{\psi}$, with the promise that it is a right singular vector of $A$ with a singular value either above $\beta$ or below $\alpha$. 
There is a degree-$d'$ polynomial $P$, where $d'=O(\delta^{-1} \log \varepsilon^{-1})$ and $\delta \coloneqq \max\{\beta-\alpha,\sqrt{1-\alpha^2}-\sqrt{1-\beta^2}\}/2$, such that there is a singular value discriminator $U_P$ that distinguishes the two cases with error probability at most $\varepsilon \geq 2^{-O(s(n))}$. Moreover, the discriminator $U_P$ achieves one-sided error when $\alpha=0$ or $\beta=1$.

\noindent Furthermore, the quantum circuit implementation of $U_P$ requires $O(d^2\log{d})$ uses of $U$, $U^{\dagger}$, $\textsc{C}_{\Pi}\textsc{NOT}$,  $\textsc{C}_{\tilde{\Pi}}\textsc{NOT}$, and multi-controlled single-qubit gates. In addition, the description of the implementation can be computed in deterministic time $\tilde{O}(\varepsilon^{-1} \delta^{-9/2})$ and space $O(s(n))$. 
\end{lemma}

\begin{proof}
    Let the singular value decomposition of $A$ be $A = W\Sigma V^{\dagger} = \sum_{i} \sigma_i \ket{\tilde{\psi_i}}\bra{\psi_i}$.\footnote{When $\Pi'=I-\tilde{\Pi}$, this SVD notation $A = \tilde{\Pi}U\Pi = \sum_{i} \sigma_i \ket{\tilde{\psi_i}}\bra{\psi_i}$ is instead applied to $(I-\tilde{\Pi})U\Pi$. By \cite[Definition 12]{GSLW18}, the right singular vectors are the same as those of $\tilde{\Pi}U\Pi$, while each singular value $\sigma_i$ is replaced by $\sqrt{1-\sigma_i^2}$. The threshold projectors below are understood with respect to this complementary SVD.} 
    Note that $U$ is a $(1,0,0)$-bitstring indexed encoding, with projections $\tilde{\Pi}$ and $\Pi$, of $A$. 
    Let singular value threshold projectors $\Pi_{\geq\delta}$ and $\Pi'_{\geq\delta}$ be defined as $\Pi_{\geq\delta} \coloneqq \Pi V\Sigma_{\geq \delta} V^{\dagger} \Pi$ and ${\Pi'}_{\geq\delta}\coloneqq{\Pi'} W \Sigma_{\geq \delta} W^{\dagger} {\Pi'}$, respectively, with similar definitions for $\Pi_{\leq \delta}$ and ${\Pi'}_{\leq \delta}$.

    To discriminate whether the singular value corresponding to a given right singular vector of $A$ exceeds a certain threshold, we need an $\varepsilon$-\textit{singular value discriminator} $U_P$. Specifically, it suffices to construct a $(1,a,\tilde{\epsilon})$-bitstring indexed encoding $U_P$ of $P(A)$, associated with an appropriate odd polynomial $P$, that satisfies \Cref{eq:threshold-projections}. The parameters $a$ and $\tilde{\epsilon}$  will be specified later. 
    \begin{subequations}
    \label{eq:threshold-projections}
    \begin{align}
        \textstyle 
        \Big\| \big(\bra{0}^{\otimes a}\otimes{\Pi'}_{\geq t+\delta}\big) U_P \big(\ket{0}^{\otimes a}\otimes\Pi_{\geq t+\delta}\big) - \sum_{i \in \Lambda} \ket{\tilde{\psi}_i}\bra{\psi_i}\Big\| &\leq \varepsilon,\\
        \Big\| \big(\bra{0}^{\otimes a}\otimes{\Pi'}_{\leq t-\delta}\big) U_P \big(\ket{0}^{\otimes a}\otimes\Pi_{\leq t-\delta}\big) - 0 \Big\| &\leq \varepsilon.
    \end{align}
    \end{subequations}
    % \cite[Theorem 32]{GSLW18}
    Here, the index set $\Lambda \coloneqq \{ i \colon \sigma_i \geq t+\delta\}$. Additionally, following the proof in~\cite[Theorem 20]{GSLW19}, $\Pi'$ is defined as $\tilde{\Pi}$ if $\beta-\alpha \geq \sqrt{1-\alpha^2}-\sqrt{1-\beta^2}$, and as $I-\tilde{\Pi}$ otherwise.\footnote{By applying \cite[Definition 12]{GSLW18} (the full version of~\cite{GSLW19}) to $\Pi'\coloneqq I-\tilde{\Pi}$, we know that $\ket{\psi}$ is a right singular vector of $\Pi' U\Pi$ with singular value at least $\sqrt{1-\alpha^2}$ in the first case, or with a singular value of at most $\sqrt{1-\beta^2}$ in the second case. Additionally, in one-sided error scenarios, if $\alpha=0$, then $\beta-\alpha = \beta \geq 1-\sqrt{1-\beta^2}=\sqrt{1-\alpha^2} - \sqrt{1-\beta^2}$; while if $\beta=1$, then $\beta-\alpha=1-\alpha \leq \sqrt{1-\alpha^2} = \sqrt{1-\alpha^2} - \sqrt{1-\beta^2}$.}

    With the construction of this bitstring indexed encoding $U_P$, we can apply an $\varepsilon$-singular value discriminator with $\Pi' = \tilde{\Pi}$ by choosing $t \coloneqq (\alpha+\beta)/2$ and $\delta \coloneqq (\beta-\alpha)/2$; and with $\Pi'=I-\tilde{\Pi}$ by choosing $t\coloneq \rbra[\big]{\sqrt{1-\beta^2}+\sqrt{1-\alpha^2}}/2$ and $\delta \coloneqq \rbra[\big]{\sqrt{1-\alpha^2}-\sqrt{1-\beta^2}}/2$, respectively. Next, we measure $\ket{0}\bra{0}^{\otimes a} \otimes\Pi'$: If the final state is in $\Img\big(\ket{0}\bra{0}^{\otimes a} \otimes\Pi'\big)$, there exists a singular value $\sigma_i$ above $\beta$ (resp., $\sqrt{1-\alpha^2}$); otherwise, all singular values $\sigma_i$ must be below $\alpha$ (resp., $\sqrt{1-\beta^2}$).
    Furthermore, we can make the error one-sided when $\alpha=0$ or $\beta=1$, since a space-efficient QSVT associated with an \textit{odd} polynomial always preserves $0$ singular values (see \Cref{remark:QSVT-parity-preserving}). 
    
    It remains to implement an $\varepsilon$-singular value discriminator $U_P$ for some  odd polynomial $P$. 
    
    \paragraph*{Implementing $\varepsilon$-singular value discriminator.} We begin by considering the following odd function $Q(x)$ such that $Q(A) \approx U_P$ and $Q(A)$ satisfies \Cref{eq:threshold-projections}:   
    \[Q(x)\coloneqq \tfrac{1}{2}\left[\big(1-\tfrac{\varepsilon}{2}\big)\cdot \sign(x+t) + \big(1-\tfrac{\varepsilon}{2}\big)\cdot \sign(x-t) + \varepsilon\cdot \sign(x)\right].\] 
    
    Let $B \coloneqq 1 + \log\frac{1+t}{\delta} + \log\frac{1}{\varepsilon}$, and choose $\epsilon \coloneqq \frac{\varepsilon}{64B \cdot \max\cbra{1,36\hat{C}_{\sign}+37,C_{\sign}+38,\tilde{C}_{\sign}}}$.
    Using the space-efficient polynomial approximation $P_{d'}^{\sign}$ of the sign function (\Cref{corr:space-efficient-sign}), with threshold parameter $\delta/(1+t)$, error parameter $\epsilon$, and the affine maps $x\mapsto (x+t)/(1+t)$ and $x\mapsto (x-t)/(1+t)$, we obtain the following degree-$d'$ polynomial $P$ associated with some degree-$d$ averaged Chebyshev truncation: 
    \begin{equation*}
        P(x)=\tfrac{1}{2}\left[\big(1-\tfrac{\varepsilon}{2}\big)\cdot P_{d'}^{\sign}\rbra[\big]{\tfrac{x+t}{1+t}} + \big(1-\tfrac{\varepsilon}{2}\big)\cdot P_{d'}^{\sign}\rbra[\big]{\tfrac{x-t}{1+t}} + \varepsilon\cdot P_{d'}^{\sign}(x)\right].
    \end{equation*}    
    
    Noting that $P(x)$ is a convex combination of $P_{d'}^{\sign}\rbra[\big]{\tfrac{x+t}{1+t}}$, $P_{d'}^{\sign}\rbra[\big]{\tfrac{x-t}{1+t}}$, and $P_{d'}^{\sign}(x)$, the constant $\hat{C}_{\sign}$ specified in \Cref{corr:space-efficient-sign} remains the same, while $\tilde{C}_{\sign}$ absorbs the factor $2$.
    Hence, $P$ is a polynomial of degree $d'=2d-1 \leq \tilde{C}_{\sign}\frac{1+t}{\delta}\log\frac{1}{\epsilon}$, and the coefficient vector $\hat{\bfc}^{(P)}$ satisfies $\|\hat{\bfc}^{(P)}\|_1 \leq \hat{C}_{\sign} \log{d'}$. 

    Recall the notation $\|f\|_{\calI}$ defined in \Cref{subsubsec:piecewise-smooth}, namely $\|f\|_{\calI} \coloneqq \sup\{|f(x)|: x\in \calI\}$.
    Let $D(x) \coloneqq \sign(x) - P^{\sign}_{d'}$, $\calI_{0} \coloneqq (0,t-\delta]$, and $\calI_{1} \coloneqq [t+\delta,1]$. 
    Following \Cref{corr:space-efficient-sign}, we obtain: 
    \begin{subequations}
        \label{eq:singular-value-discriminator-poly-error}
        \begin{align}
            \|P(x)-Q(x)\|_{[\delta-t,0)} &= \|P(x) - Q(x)\|_{\calI_0}\\
            &\leq \tfrac{2-\varepsilon}{4}\|D\rbra[\big]{\tfrac{x+t}{1+t}}\|_{\calI_0} + \tfrac{2-\varepsilon}{4}\|D\rbra[\big]{\tfrac{x-t}{1+t}}\|_{\calI_0} + \tfrac{\varepsilon}{2} \|D(x)\|_{\calI_0}\\
            &\leq \big(1 - \tfrac{\varepsilon}{2}\big) C_{\sign} \epsilon + \tfrac{\varepsilon}{2},\\
            \|P(x)-Q(x)\|_{[-1,-t-\delta]} &= \|P(x) - Q(x)\|_{\calI_1}\\
            &\leq \tfrac{2-\varepsilon}{4}\|D\rbra[\big]{\tfrac{x+t}{1+t}}\|_{\calI_1} + \tfrac{2-\varepsilon}{4}\|D\rbra[\big]{\tfrac{x-t}{1+t}}\|_{\calI_1} + \tfrac{\varepsilon}{2} \|D(x)\|_{\calI_1}\\
            &\leq \big(1 -\tfrac{\varepsilon}{2} + \tfrac{\varepsilon}{2} \big) C_{\sign} \epsilon.
        \end{align}
    \end{subequations}
    Here, the equalities hold because both $P$ and $Q$ are odd functions. 

    Using \Cref{thm:LCU-averaged-chebyshev-truncation} with $P$, $\epsilon_1 \coloneqq 0$, and $\epsilon_2 \coloneqq \epsilon$, we obtain a $\rbra{\|\hat{\bfc}^{(P)}\|_1, \ceil*{\log{d'}+1}, \epsilon\|\hat{\bfc}^{(P)}\|_1}$-bitstring indexed encoding of $P(A)$. Applying the same renormalization argument as in the proof of \Cref{corr:sign-polynomial-implementation}, we obtain a $(1, a, \tilde{\epsilon})$-bitstring indexed encoding $U_P$ of $P(A)$, with $a \coloneqq \lceil \log{d'} \rceil+3$ and $\tilde{\epsilon} \coloneqq (36 \hat{C}_{\sign} \log{d'} +37) \epsilon$. 
    One can verify that $\rbra{36 \hat{C}_{\sign}\log d' + 37}\epsilon \leq \varepsilon/2$ and $\rbra{C_{\sign}+1}\epsilon \leq \varepsilon/2$.
    Together with \Cref{eq:singular-value-discriminator-poly-error}, we obtain:
    \begin{align*}
    \textstyle 
        \Big\| \big(\bra{0}^{\otimes a}\otimes{\Pi'}_{\geq t+\delta}\big) U_P \big(\ket{0}^{\otimes a}\otimes\Pi_{\geq t+\delta}\big) - \sum_{i\in\Lambda} \ket{\tilde{\psi}_i}\bra{\psi_i}\Big\| &\leq \big(1\!-\!\tfrac{\varepsilon}{2} \!+\! \tfrac{\varepsilon}{2} \big) C_{\sign} \epsilon + \tilde{\epsilon} \leq \varepsilon, \\
        \Big\| \big(\bra{0}^{\otimes a}{\Pi'}_{\leq t-\delta}\big) U_P \big(\ket{0}^{\otimes a}\otimes\Pi_{\leq t-\delta}\big) - 0 \Big\| &\leq C_{\sign} \epsilon + \tfrac{\varepsilon}{2} + \tilde{\epsilon} \leq \varepsilon.
    \end{align*}
    Hence, we conclude that our construction of $U_P$ indeed satisfies \Cref{eq:threshold-projections}. 

    Finally, we analyze the complexity of this $\varepsilon$-singular value discriminator $U_P$. Following \Cref{corr:sign-polynomial-implementation}, the quantum circuit implementation of $U_P$ requires $O(d^2 \log{d})$ uses of $U$, $U^{\dagger}$, $\textsc{C}_{\Pi}\textsc{NOT}$, $\textsc{C}_{\tilde{\Pi}}\textsc{NOT}$, and multi-controlled single-qubit gates. Moreover, we can compute the description of the circuit implementation in deterministic time $\tilde{O}(\epsilon^{-1} d^{9/2}) = \tilde{O}(\varepsilon^{-1} \delta^{-9/2})$ and space $O(s(n))$,
    where $\delta = \max\{\beta-\alpha, \sqrt{1-\alpha^2}-\sqrt{1-\beta^2}\}/2$. 
\end{proof}

\vspace{1em}

Finally, we provide the proof of \Cref{thm:space-efficient-error-reduction-unitaryQC}, which closely relates to Theorem 38 in~\cite{GSLW18} (the full version of~\cite{GSLW19}). 

\begin{proof}[Proof of \Cref{thm:space-efficient-error-reduction-unitaryQC}] 
It suffices to amplify the promise gap by QSVT. 
Note that the probability that a $\BQUSPACE{[s(n)]}$ circuit $C_x$ accepts is 
$\Pr{C_x \text{ accepts}} = \| \ket{1}\bra{1}_{\Out} C_x \ket{0^{k+m}} \|_2^2 \geq a$ for \textit{yes} instances, whereas $\Pr{C_x \text{ accepts }} = \| \ket{1}\bra{1}_{\Out} C_x \ket{0^{k+m}} \|_2^2 \leq b$ for \textit{no} instances. 
Then consider a $(1,0,0)$-bitstring indexed encoding $M_x\coloneqq\Pi_{\Out} C_x \Pi_{\In}$ such that
$\| M_x \| \geq \sqrt{a}$ for \textit{yes} instances while $\| M_x \| \leq \sqrt{b}$ for \textit{no} instances, where $\Pi_{\In}\coloneqq \ketbra{0}{0}^{\otimes k+m}$ and $\Pi_{\Out}\coloneqq \ket{1}\bra{1}_{\Out} \otimes I_{m+k-1}$. 
Since $\|M_x\|=\sigma_{\max}(M_x)$ where $\sigma_{\max}(M_x)$ is the largest singular value of $M_x$, it suffices to distinguish the largest singular value of $M_x$ is either above $\sqrt{a}$ or below $\sqrt{b}$. By setting $\alpha\coloneqq\sqrt{b}$, $\beta\coloneqq\sqrt{a}$, and $\varepsilon\coloneqq2^{-l(n)}$, this task is a direct corollary of \Cref{lemma:space-efficient-singular-value-discrimination}. 
\end{proof}

\section{Space-bounded quantum state testing}
\label{sec:space-bounded-quantum-state-testing}

We begin by defining the problem of quantum state testing in a space-bounded manner: 

\begin{definition}[Space-bounded Quantum State Testing]
    \label{def:space-bounded-quantum-state-testing}
    Given polynomial-size quantum circuits (devices) $Q_0$ and $Q_1$ that act on $O(\log{n})$ qubits and have a succinct description (the ``source code'' of devices), with $r(n)$ specified output qubits, where $r(n)$ is a deterministic logspace computable function such that $0 < r(n) \leq O(\log(n))$. For clarity, $n$ represents the (total) number of gates in $Q_0$ and $Q_1$.\footnote{It is noteworthy that in the time-bounded scenario, the input length of circuits, the size of circuit descriptions, and the number of gates in circuits are polynomially equivalent. However, in the space-bounded scenario, only the last two quantities are polynomially equivalent, and their dependence on the first quantity may be exponential.} Let $\rho_i$ denote the mixed state obtained by running $Q_i$ on the all-zero state $\ket{\bar{0}}$ and tracing out the non-output qubits. 
    
    \noindent We define a \textit{space-bounded quantum state testing} problem, with respect to a specified distance-like measure, to decide whether $\rho_0$ and $\rho_1$ are easily distinguished or almost indistinguishable. 
    Likewise, we also define a \textit{space-bounded quantum state certification} problem to decide whether $\rho_0$ and $\rho_1$ are easily distinguished or \textit{exactly} indistinguishable. 
\end{definition}

We remark that space-bounded quantum state certification, defined in \Cref{def:space-bounded-quantum-state-testing}, represents a ``white-box'' (log)space-bounded counterpart of quantum state certification~\cite{BOW19}. 

\begin{remark}[Lifting to exponential-size instances by succinct encodings]
    \label{remark:succinct-encodings}
    For $s(n)$ space-uniform quantum circuits $Q_0$ and $Q_1$ acting on $O(s(n))$ qubits, if these circuits admit a succinct encoding,\footnote{For instance, the construction in \cite[Remark 11]{FL18}, or~\cite{PY86,BLT92} in general.} namely there is a deterministic $O(s(n))$-space Turing machine with time complexity $\poly(s(n))$ can uniformly generate the corresponding gate sequences, then \Cref{def:space-bounded-quantum-state-testing} can be extended to any $s(n)$ satisfying $\Omega(\log{n}) \leq s(n) \leq \poly(n)$.\footnote{\Cref{def:space-bounded-quantum-state-testing} (mostly) coincides with the case of $s(n)=\Theta(\log{n})$ and directly takes the corresponding gate sequence of $Q_0$ and $Q_1$ as an input. }
\end{remark}

Next, we define space-bounded quantum state testing problems, based on \Cref{def:space-bounded-quantum-state-testing}, with respect to four commonplace distance-like measures.

\begin{definition}[Space-bounded Quantum State Distinguishability Problem, $\GapQSD_{\log}$]
    \label{def:space-bounded-QSD}
    Consider deterministic logspace computable functions $\alpha(n)$ and $\beta(n)$, satisfying $0 \leq \beta(n) < \alpha(n) \leq 1$ and $\alpha(n)-\beta(n) \geq 1/\poly(n)$. Then the promise is that one of the following holds:    
    \begin{itemize}
        \item \textit{Yes} instances: A pair of quantum circuits $(Q_0,Q_1)$ such that $\td(\rho_0,\rho_1) \geq \alpha(n)$;
        \item \textit{No} instances: A pair of quantum circuits $(Q_0,Q_1)$ such that $\td(\rho_0,\rho_1) \leq \beta(n)$. 
    \end{itemize}
    Moreover, we also define the certification counterpart of \GapQSDlog{}, referred to as \CertQSDlog{}, given that $\beta=0$.
    Specifically, $\CertQSDlog[\alpha(n)] \coloneqq \GapQSDlog[\alpha(n),0]$. 
\end{definition}

Likewise, we can define \GapQJSlog{} and \GapQHSlog{}, also the certification version \coCertQHSlog{}, in a similar manner to \Cref{def:space-bounded-QSD} by replacing the distance-like measure accordingly: 
\begin{itemize}
    \item $\GapQJS_{\log}[\alpha(n),\beta(n)]$: Decide whether $\QJS_2(\rho_0,\rho_1) \geq \alpha(n)$ or $\QJS_2(\rho_0,\rho_1) \leq \beta(n)$; 
    \item $\GapQHS_{\log}[\alpha(n),\beta(n)]$: Decide whether $\HS(\rho_0,\rho_1) \geq \alpha(n)$ or $\HS(\rho_0,\rho_1) \leq \beta(n)$. 
\end{itemize}
Furthermore, we use the notation \coCertQSDlog{} to indicate the \textit{complement} of \CertQSDlog{} with respect to the chosen parameter $\alpha(n)$, and so does \coCertQHSlog{}. 

\begin{definition}[Space-bounded Quantum Entropy Difference Problem, $\GapQED_{\log}$]
	\label{def:space-bounded-QED}
     Consider a deterministic logspace computable function $g:\bbN\rightarrow\bbR^+$, satisfying $g(n)\geq 1/\poly(n)$. Then the promise is that one of the following cases holds:     
    \begin{itemize}
	   \item \textit{Yes} instance: A pair of quantum circuits $(Q_0,Q_1)$ such that $\S(\rho_0)-\S(\rho_1) \geq g(n)$;
	   \item \textit{No} instance: A pair of quantum circuits $(Q_0,Q_1)$ such that $\S(\rho_1)-\S(\rho_0) \geq g(n)$.
    \end{itemize}
\end{definition}

\paragraph{Novel complete characterizations for space-bounded quantum computation.}
We now present the main theorems in this section and the paper. 
\Cref{thm:space-bounded-quantum-state-certification-RQL-complete} establishes the first family of natural $\coRQUL$-complete problems. 
By relaxing the error requirement from one-sided to two-sided, \Cref{thm:space-bounded-quantum-state-testing-BQL-complete} identifies a new family of natural \BQL{}-complete problems on space-bounded quantum state testing. It is noteworthy that \Cref{thm:space-bounded-quantum-state-certification-RQL-complete,thm:space-bounded-quantum-state-testing-BQL-complete} also have natural exponential-size up-scaling counterparts.\footnote{We can naturally extend \Cref{thm:space-bounded-quantum-state-certification-RQL-complete,thm:space-bounded-quantum-state-testing-BQL-complete} to their exponential-size up-scaling counterparts with $2^{-O(s(n))}$-precision, employing the extended version of \Cref{def:space-bounded-quantum-state-testing} outlined in \Cref{remark:succinct-encodings}, thus achieving the complete characterizations for $\coRQUSPACE[s(n)]$ and $\BQSPACE[s(n)]$, respectively.}

\begin{theorem}
    \label{thm:space-bounded-quantum-state-certification-RQL-complete}
    The computational hardness of the following \emph{(}log\emph{)}space-bounded quantum state certification problems, for any deterministic logspace computable $\alpha(n) \geq 1/\poly(n)$, is as follows:
    \begin{enumerate}[label={\upshape(\roman*)}]
        \item $\coCertQSDlog[\alpha(n)]$ is \coRQUL{}-complete;     
        \item $\coCertQHSlog[\alpha(n)]$ is \coRQUL{}-complete.             
    \end{enumerate}
\end{theorem}

\begin{theorem}
    \label{thm:space-bounded-quantum-state-testing-BQL-complete}
    The computational hardness of the following \emph{(}log\emph{)}space-bounded quantum state testing problems, where $\alpha(n)-\beta(n) \geq 1/\poly(n)$ or $g(n) \geq 1/\poly(n)$ as well as $\alpha(n)$, $\beta(n)$, $g(n)$ can be computed in deterministic logspace, is as follows:
    \begin{enumerate}[label={\upshape(\roman*)}]
        \item $\GapQSDlog[\alpha(n),\beta(n)]$ is \BQL{}-complete; 
        \item $\GapQEDlog[g(n)]$ is \BQL{}-complete; 
        \item $\GapQJSlog[\alpha(n),\beta(n)]$ is \BQL{}-complete; 
        \item $\GapQHSlog[\alpha(n),\beta(n)]$ is \BQL{}-complete. 
    \end{enumerate}
\end{theorem}

\vspace{1em}
To establish \Cref{thm:space-bounded-quantum-state-certification-RQL-complete,thm:space-bounded-quantum-state-testing-BQL-complete}, we introduce a general framework for space-bounded quantum state testing in \Cref{subsec:quantum-state-testing-framework}. Interestingly, \BQL{} and \coRQUL{} containments for these problems with respect to different distance-like measures, utilizing our general framework, correspond to approximate implementations of distinct two-outcome measurements. The main technical challenges then mostly involve \textit{parameter trade-offs} when using our space-efficient QSVT to construct these approximate measurement implementations. We summarize this correspondence in \Cref{table:correspondence-distances-measurements} and the associated subsection, which provides the detailed proof. 
\begin{table}[H]
\centering
\begin{tabular}{cccc}
    \toprule
    Distance-like measure & State testing & State certification & \makecell{Two-outcome measurement\\ $\Pi_b$ for $b\in\binset$}\\
    \midrule
    Trace distance
    & \makecell{\GapQSDlog{}\\ \footnotesize{\Cref{subsec:GapQSD-in-BQL}}} 
    & \makecell{\coCertQSDlog{}\\ \footnotesize{\Cref{subsec:coCertQSD-in-coRQL}}} 
    & \makecell{$\frac{I}{2}+\frac{(-1)^b}{2}\sign^{\SV}\big(\frac{\rho_0-\rho_1}{2}\big)$\\ \footnotesize{for $\rho_0$ and $\rho_1$}}\\
    \midrule
    \makecell{Quantum entropy difference\\ Quantum JS divergence} & 
    \makecell{\GapQEDlog{} \\ \GapQJSlog{} \\ \footnotesize{\Cref{subsec:GapQED-in-BQL}}} & None & \makecell{$\frac{I}{2}-\frac{(-1)^b}{2}\cdot \frac{\ln(\rho_i)}{2\ln(2/\beta)}$\\ \footnotesize{for $\rho_i$ where $i \in \binset$}\\ \footnotesize{and $\lambda(\rho_i) \in [-\beta,\beta]$}}\\
    \midrule
    Hilbert--Schmidt distance & 
    \makecell{\GapQHSlog{}\\ \footnotesize{\Cref{subsec:coCertQSD-in-coRQL}}} & 
    \makecell{\coCertQHSlog{}\\ \footnotesize{\Cref{subsec:coCertQSD-in-coRQL}}} & 
    \makecell{$\frac{I}{2}+\frac{(-1)^b}{2}\SWAP$\\ \footnotesize{for $\rho_0\otimes\rho_1$}}\\
    \bottomrule
\end{tabular}
\caption{The correspondence between the distance-like measures and measurements.}
\label{table:correspondence-distances-measurements}
\end{table}

Lastly, the corresponding hardness proof for all these problems is provided in \Cref{subsec:BQL-coRQL-hardness}.

\subsection{Space-bounded quantum state testing: a general framework}
\label{subsec:quantum-state-testing-framework}

In this subsection, we introduce a general framework for quantum state testing that utilizes a quantum tester $\calT$. Specifically, the space-efficient tester $\calT$ succeeds (outputting the value ``$0$'') with probability $x$, which is linearly dependent on some quantity closely related to the distance-like measure of interest. Consequently, we can obtain an additive-error estimation $\widetilde{x}$ of $x$ with high probability through sequential repetition (\Cref{lemma:success-probability-estimation}).

To construct $\calT$, we combine the one-bit precision phase estimation~\cite{Kitaev95}, commonly known as the Hadamard test~\cite{AJL09}, for block-encodings (\Cref{lemma:Hadamard-test}), with our space-efficient quantum singular value transformation (QSVT) technique, which we describe in \Cref{sec:space-efficient-QSVT}.

\begin{figure}[!htp]
    \centering
    \begin{quantikz}[wire types={q,b,b,b}, classical gap=0.06cm, row sep=0.75em]
        \lstick{$\ket{0}$} & \gate{H} & \ctrl{1} & \gate{H} & \meter{} & \setwiretype{c} \rstick{$x$}\\
        \lstick{$\ket{\bar 0}$} &  & \gate[2]{U_{P_{d'}(A)}}  &  &  \\
        \lstick{$\ket{0}^{\otimes r}$} & \gate[2]{\quad Q \quad}  &  &  &  \\
        \lstick{$\ket{\bar 0}$} & & & & 
    \end{quantikz}
    \caption{Quantum tester $\calT(Q,U_A,P_{d'},\epsilon)$: the circuit implementation.}
    \label{fig:framework-state-testing}
\end{figure}
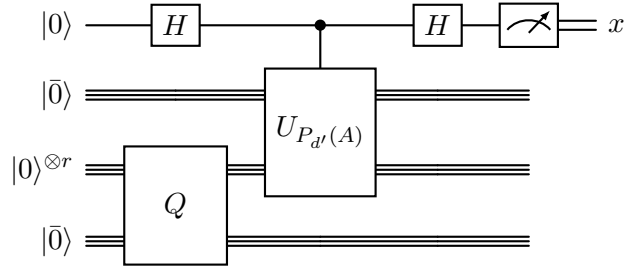

\paragraph*{Constructing a space-efficient quantum tester.}
We now provide a formal definition and the detailed construction of the quantum tester $\calT$. The quantum circuit shown in \Cref{fig:framework-state-testing} defines the quantum tester $\mathcal{T}(Q, U_A, P_{d'}, \epsilon)$ using the following parameters with $s(n)=\Theta(\log{n})$: 
\begin{itemize}[itemsep=0.33em,topsep=0.33em,parsep=0.33em]
    \item A $s(n)$-qubit quantum circuit $Q$ prepares the purification of an $r(n)$-qubit quantum state $\rho$ where $\rho$ is the quantum state of interest; 
    \item $U_A$ is a $(1,s(n)-r(n),0)$-block-encoding of an $r(n)$-qubit Hermitian operator $A$ where $A$ relates to the quantum states of interest and $r(n)\leq s(n)$; 
    \item $P_{d'}$ is the space-efficiently computable degree-$d'$ polynomial defined by \Cref{eq:averaged-Chebyshev-truncation} obtained from some degree-$d$ averaged Chebyshev truncation $P_{d'}$ with $d'=2d-1$, where $P_{d'}(x)=\hat c_0/2+\sum_{k=1}^{d'} \hat c_k T_k(x) \in \mathbb{R}[x]$ and $T_k$ is the $k$-th Chebyshev polynomial, with $d' \leq 2^{O(s(n))}$, such that the coefficients $\hat \bfc\coloneqq(\hat c_0,\cdots,\hat c_{d'})$ can be computed in bounded-error randomized space $O(s(n))$;     
    \item $\epsilon$ is the precision parameter used in the estimation of $x$, with $\epsilon \geq 2^{-O(s(n))}$. 
\end{itemize}

Leveraging our space-efficient QSVT, we assume that there is an $(\alpha,*,*)$-block-encoding of $P_{d'}(A)$, which is an approximate implementation of $U_{P_{d'}(A)}$ in \Cref{fig:framework-state-testing}.
Now, we can define the corresponding estimation procedure, $\hat{\calT}(Q, U_A, P_{d'}, \epsilon, \epsilon_H, \delta)$, namely a quantum algorithm that computes an additive-error estimate $\alpha\widetilde{x}$ of $\operatorname{Re}(\Tr(P_{d'}(A)\rho))$ from the tester $\calT(Q, U_A, P_{d'}, \epsilon)$. 
Technically speaking, $\hat{\calT}$ outputs $\widetilde{x}$ such that $\left| \alpha \widetilde x - \operatorname{Re}(\Tr(P_{d'}(A)\rho)) \right| \leq \|\hat\bfc\|_1 \epsilon + \alpha \epsilon_H$ with probability at least $1-\delta$. 
Now we will demonstrate that both the tester $\calT$ and the corresponding estimation procedure $\hat{\calT}$ are space-efficient: 

\begin{lemma}[Quantum tester $\calT$ and estimation procedure $\hat{\calT}$ are space-efficient]
    \label{lemma:space-efficient-quantum-tester}
    Assume that there is an $(\alpha,*,*)$-block-encoding of $P_{d'}(A)$ that approximately implements $U_{P_{d'}}(A)$, where $\alpha$ is either $\|\hat{\bfc}\|_1$ or $1$ based on conditions of $P_{d'}$ and $A$.
    The quantum tester $\calT(Q, U_A, P_{d'}, \epsilon)$, as specified in \Cref{fig:framework-state-testing}, accepts \emph{(}outputting the value ``$0$''\emph{)} with probability 
    \[\frac{1}{2}\big(1+ \frac{1}{\alpha}\operatorname{Re}(\Tr(P_{d'}(A)\rho))\big) \pm \frac 1 {2 \alpha} \|\hat \bfc\|_1 \epsilon.\]
    In addition, $\hat{\calT}(Q, U_A, P_{d'}, \epsilon, \epsilon_H, \delta)$ outputs $\widetilde x$ such that, with probability at least $1-\delta$, 
    \[\| \alpha \widetilde x - \operatorname{Re}(\Tr(P_{d'}(A)\rho)) \| \leq \|\hat \bfc\|_1 \epsilon + \alpha \epsilon_H.\]
    
    \noindent Moreover, we can compute the quantum circuit description of $\calT$ in deterministic space $O(s+\log(1/\epsilon))$ given the coefficient vector $\hat\bfc$ of $P_{d'}$. 
    Furthermore, we can implement the corresponding estimation procedure $\hat{\calT}$ in bounded-error quantum space $O(s+\log(1/\epsilon) + \log(1/\epsilon_H) + \log \log(1/\delta))$.
\end{lemma}

We first provide two useful lemmas for implementing our quantum tester $\mathcal{T}$. It is noteworthy that \Cref{lemma:purified-density-matrix} originates from~\cite{LC19}, as well as \Cref{lemma:Hadamard-test} is a specific version of one-bit precision phase estimation (or the Hadamard test)~\cite{Kitaev95,AJL09}. 

\begin{lemma} [Purified density matrix, {\cite[Lemma 25]{GSLW19}}]
\label{lemma:purified-density-matrix}
    Suppose $\rho$ is an $s$-qubit density operator and $U$ is an $(a+s)$-qubit unitary operator such that $U \ket{0}^{\otimes a} \ket{0}^{\otimes s} = \ket{\rho}$ and $\rho = \Tr_a(\ket{\rho}\bra{\rho})$.
    Then, we can construct an $O(a+s)$-qubit quantum circuit $\widetilde U$ that is a $(1, a+s, 0)$-block-encoding of $\rho$, using $O(1)$ queries to $U$ and $O(a+s)$ one- and two-qubit quantum gates. 
\end{lemma}

\begin{lemma} [Hadamard test for block-encodings, adapted from~{\cite[Lemma 9]{GP22}}]
\label{lemma:Hadamard-test}
    Suppose $U$ is an $(a+s)$-qubit unitary operator that is a block-encoding of $s(n)$-qubit operator $A$. 
    We can implement an $O(a+s)$-qubit quantum circuit that, on input $s(n)$-qubit quantum state $\rho$, outputs $0$ with probability $\frac{1+\operatorname{Re}(\Tr(A\rho))}{2}$.
\end{lemma}

Finally, we proceed with the actual proof of \Cref{lemma:space-efficient-quantum-tester}.

\begin{proof} [Proof of \Cref{lemma:space-efficient-quantum-tester}]
Note that $U_A$ is a $(1, a, 0)$-block-encoding of $A$, where $a = s - r$. 

We first consider the case where $\alpha=\|\hat{\bfc}\|_1$, which holds for any $P_{d'}$ and $A$. 
By \Cref{thm:LCU-averaged-chebyshev-truncation} with $\epsilon_1 \coloneqq 0$ and $\epsilon_2 \coloneqq \epsilon$, we can implement an $O(s)$-qubit quantum circuit $U'$ that is a $(\|\hat\bfc\|_1, \hat a, \epsilon\|\hat\bfc\|_1)$-block-encoding of $P_{d'}(A)$, using $O(d^2)$ queries to $U_A$, where $\hat a = a + \lceil\log d'\rceil + 1$.
Assume that $U'$ is a $(1, \hat a, 0)$-block-encoding of $A'$, then $\| 
\|\hat\bfc\|_1 A' - P_{d'}(A) \| \leq \|\hat\bfc\|_1\epsilon$.
Additionally, we can compute the quantum circuit description of $U'$ in deterministic space $O(s+\log(1/\epsilon))$ given the coefficient vector $\hat\bfc$ of $P_{d'}$. 
As the quantum tester $\mathcal{T}(Q, U_A, P_{d'}, \epsilon)$ is mainly based on the Hadamard test, by employing \Cref{lemma:Hadamard-test}, we have that $\calT$ outputs $0$ with probability 
\[\Pr{x = 0} = \frac{1}{2}\big(1+\operatorname{Re}(\Tr(A'\rho))\big) = \frac{1}{2}\left(1+\operatorname{Re}\!\left(\Tr\Big(\frac{P_{d'}(A)}{\|\hat\bfc\|_1}\rho\Big)\right)\right) \pm \frac 1 2 \epsilon.\] 
It is left to construct the estimation procedure $\hat{\calT}$. As detailed in \Cref{lemma:success-probability-estimation}, we obtain an estimation $\widetilde{x}$ by sequentially repeating the quantum tester $\calT(Q,U_A,P_{d'},\epsilon)$ for $O(1/\epsilon_H^2)$ times. This repetition ensures that $|\widetilde x - \operatorname{Re}(\Tr(A'\rho))| \leq \epsilon_H$ holds with probability at least $\Omega(1)$, and derives a further implication on $P_{d'}(A)$:
\[\Pr{ \left| \|\hat\bfc\|_1 \widetilde x - \operatorname{Re}(\Tr(P_{d'}(A)\rho)) \right| \leq (\epsilon + \epsilon_H) \|\hat\bfc\|_1 } \geq \Omega(1).\]
We thus conclude that construction of the estimation procedure $\hat{\calT}(Q, U_A, P_{d'}, \epsilon, \epsilon_H, \delta)$ by utilizing $O(\log(1/\delta)/\epsilon_H^2)$ sequential repetitions of $\calT(Q, U_A, P_{d'}, \epsilon)$. Similarly following \Cref{lemma:success-probability-estimation}, $\hat{\calT}(Q, U_A, P_d, \epsilon, \epsilon_H, \delta)$ outputs an estimate $\widetilde{x}$ satisfies the following condition: 
\[\Pr{ \left| \|\hat\bfc\|_1 \widetilde x - \operatorname{Re}(\Tr(P_{d'}(A)\rho)) \right| \leq (\epsilon + \epsilon_H) \|\hat\bfc\|_1 } \geq 1-\delta.\]
In addition, a direct calculation indicates that we can implement $\hat{\calT}(Q, U_A, P_{d'}, \epsilon, \epsilon_H, \delta)$ in quantum space $O(s+\log(1/\epsilon) + \log(1/\epsilon_H) + \log \log(1/\delta))$ as desired. 

\vspace{1em}
Next, we move to the case where $\alpha=1$, applicable to certain $P_{d'}$ and $A$, namely when $P_{d'}(A)$ is a partial isometry in \Cref{thm:LCU-averaged-chebyshev-truncation} or $P_{d'}=P^{\sign}_{d'}$ in \Cref{corr:sign-polynomial-implementation}.
The proof is similar, and we just sketch the key points as follows. 
Using \Cref{thm:LCU-averaged-chebyshev-truncation} with $\epsilon_1 \coloneqq 0$ and $\epsilon_2 \coloneqq \epsilon/36$, we can implement an $O(s)$-qubit quantum circuit $U'$ that is a $(1, a', \epsilon\|\hat\bfc\|_1)$-block-encoding of $P_{d'}(A)$, using $O(d^2\|\hat\bfc\|_1)$ queries to $U_A$, where $a' = a + \lceil\log d'\rceil + 3$.
Assume that $U'$ is a $(1, a', 0)$-block-encoding of $A'$, then $\| A' - P_{d'}(A) \| \leq \|\hat\bfc\|_1\epsilon$. Similarly, $\mathcal{T}$ outputs $0$ with probability 
\begin{equation}
    \label{eq:state-tester-pacc}
    \Pr{x = 0} = \frac{1}{2}\big(1+\operatorname{Re}(\Tr(A'\rho))\big) = \frac{1}{2}\left(1+\operatorname{Re}\left(\Tr\left({P_{d'}(A)}\rho\right)\right)\right) \pm \frac 1 2 \|\hat\bfc\|_1\epsilon.
\end{equation}
Therefore, we can obtain an estimate $\widetilde x$ such that
\begin{equation}
    \label{eq:state-tester-estimation-error}
    \Pr{ \left| \widetilde x - \operatorname{Re}(\Tr(P_{d'}(A)\rho)) \right| \leq \|\hat\bfc\|_1 \epsilon + \epsilon_H } \geq 1-\delta.
\end{equation}

For the special case $P_{d'} = P^{\sign}_{d'}$, use \Cref{corr:sign-polynomial-implementation} with $\epsilon_1\coloneqq 0$ and $\epsilon_2 \coloneqq \frac{\norm{\hat\bfc^{\sign}}_1}{36\hat C_{\sign}\log d' + 37}\,\epsilon$. Then the resulting block-encoding error is at most $(36\hat C_{\sign}\log d' + 37)\epsilon_2
    = \norm{\hat\bfc^{\sign}}_1\epsilon$.
Hence the corresponding formulas in \Cref{eq:state-tester-pacc,eq:state-tester-estimation-error} hold after replacing $\norm{\hat\bfc}_1$ by $\norm{\hat\bfc^{\sign}}_1$. Since $\norm{\hat\bfc^{\sign}}_1\leq \hat C_{\sign}\log d'$, as guaranteed by \Cref{corr:space-efficient-sign}, this quantity can finally be upper-bounded by $\hat C_{\sign}\log d'$ whenever a coefficient-free bound is preferred.
\end{proof}

\subsection{\GapQSDlog{} is in \BQL{}}
\label{subsec:GapQSD-in-BQL}

In this subsection, we demonstrate \Cref{thm:GapQSDlog-in-BQL} by constructing a quantum algorithm that incorporates testers $\calT(Q_i,U_{\frac{\rho_0-\rho_1}{2}},P_d^{\sign},\epsilon)$ for $i \in \binset$, where the construction of testers utilizes the space-efficient QSVT associated with the sign function.

\begin{theorem}
    \label{thm:GapQSDlog-in-BQL}
    For any functions $\alpha(n)$ and $\beta(n)$ that can be computed in deterministic logspace and satisfy $\alpha(n)-\beta(n) \geq 1/\poly(n)$, we have that $\GapQSDlog[\alpha(n),\beta(n)]$ is in \BQL{}. 
\end{theorem}

\begin{proof}
    Inspired by time-efficient algorithms for the low-rank variant of \GapQSD{}~\cite{WZ23}, we devise a space-efficient algorithm for \GapQSDlog{}, presented in \Cref{algo:GapQSD-in-BQL}. 

    \begin{algorithm}[ht!]
		\caption{Space-efficient algorithm for $\GapQSD_{\log}$.}
		\label{algo:GapQSD-in-BQL}
        \SetKwInOut{Input}{Input}
        \SetKwInOut{Output}{Output}
        \SetKwInOut{Parameter}{Params}
        \Input{Quantum circuits $Q_i$ that prepare the purification of $\rho_i$ for $i \in \binset$.}
        \Output{An additive-error estimation of $\td(\rho_0,\rho_1)$.}
        \Parameter{$\varepsilon \coloneqq \frac{\alpha - \beta}{4}$, $\delta \coloneqq \frac{\varepsilon}{2^{r+3}}$, 
        $\epsilon \coloneqq \frac{\varepsilon}{64\,s\,\max\{1,36 \hat{C}_{\sign},2 C_{\sign}+37,\tilde C_{\sign}\}}$,
        $d' \coloneqq \tilde C_{\sign} \cdot \frac{1}{\delta} \log\frac{1}{\epsilon} = 2d\!-\!1$, $\varepsilon_H \coloneqq \frac{\varepsilon}{4}$.}
        1. Construct block-encodings of $\rho_0$ and $\rho_1$, denoted by $U_{\rho_0}$ and $U_{\rho_1}$, respectively, using $O(1)$ queries to $Q_0$ and $Q_1$ and $O(s(n))$ ancillary qubits by \Cref{lemma:purified-density-matrix}\;
        2. Construct a block-encoding of $\frac{\rho_0-\rho_1}{2}$, denoted by $U_{\tfrac{\rho_0-\rho_1}{2}}$, using $O(1)$ queries to $U_{\rho_0}$ and $U_{\rho_1}$ and $O(s(n))$ ancillary qubits by \Cref{lemma:space-efficient-LCU}\;
        \emph{Let $P_{d'}^{\sign}$ be the degree-$d'$ polynomial specified in \Cref{corr:space-efficient-sign} with parameters $\delta$ and $\epsilon$, and its coefficients $\{\hat{c}_k\}_{k=0}^{d'}$ are computable in deterministic space $O(\log (d/\epsilon))$}\;
        3. Set $x_0 \coloneqq \hat{\calT}(Q_0, U_{\frac{\rho_0-\rho_1}{2}}, P_{d'}^{\sign}, \epsilon, \epsilon_H, 1/10)$, $x_1 \coloneqq \hat{\calT}(Q_1, U_{\frac{\rho_0-\rho_1}{2}}, P_{d'}^{\sign}, \epsilon, \epsilon_H, 1/10)$\;
        4. Compute $x = (x_0 - x_1) / 2$. Return ``yes'' if $x > (\alpha+\beta)/2$, and ``no'' otherwise. 
        \BlankLine
    \end{algorithm}

    Let us demonstrate the correctness of \Cref{algo:GapQSD-in-BQL} and analyze the computational complexity. We focus on the setting with $s(n)=\Theta(\log{n})$. We set $\varepsilon \coloneqq (\alpha-\beta)/4 \geq 2^{-O(s)}$ and assume that $Q_0$ and $Q_1$ are $s(n)$-qubit quantum circuits that prepare the purifications of $\rho_0$ and $\rho_1$, respectively. 
    According to \Cref{lemma:purified-density-matrix}, we can construct $O(s)$-qubit quantum circuits $U_{\rho_0}$ and $U_{\rho_1}$ that encode $\rho_0$ and $\rho_1$ as $(1,O(s), 0)$-block-encodings, using $O(1)$ queries to $Q_0$ and $Q_1$ as well as $O(1)$ one- and two-qubit quantum gates.
    Next, we apply \Cref{lemma:space-efficient-LCU} to construct a $(1,O(s),0)$-block-encoding $U_{\frac{\rho_0-\rho_1}{2}}$ of $\frac{\rho_0 - \rho_1}{2}$, using $O(1)$ queries to $U_{\rho_0}$ and $U_{\rho_1}$, as well as $O(1)$ one- and two-qubit quantum gates.

    Let $\delta \coloneqq \frac{\varepsilon}{2^{r+3}}$, 
    $\epsilon \coloneqq \frac{\varepsilon}{64\,s\,\max\{1,36 \hat{C}_{\sign},2 C_{\sign}+37,\tilde C_{\sign}\}}$,\footnote{The parameter choice above ensures the required error bound. Indeed, since $\varepsilon \geq 2^{-O(s)}$ and $r\leq s$, we have $\log(1/\delta)=O(s)$. The definition of $\epsilon$ also gives $\log(1/\epsilon)=O(s)$, and therefore $\log d'=\log\rbra[\big]{\tilde{C}_{\sign}\delta^{-1}\log(1/\epsilon)}=O(s)$. Thus the factor $\log d'$ can be absorbed by the $s(n)$ term in the denominator of $\epsilon$; after increasing the absolute constant $64$ if necessary, this choice gives
$(36\hat{C}_{\sign}\log d' + 2C_{\sign}+37)\epsilon \leq \varepsilon/2$.}
    and $d' \coloneqq \tilde C_{\sign} \cdot \frac{1}{\delta} \log\frac{1}{\epsilon} \leq 2^{O(s(n))}$, where $\tilde C_{\sign}$ comes from \Cref{corr:space-efficient-sign}.
    Let $P_{d'}^{\sign} \in \mathbb{R}[x]$ be the polynomial specified in \Cref{corr:space-efficient-sign} with $d' = 2d-1$.
    Let $\epsilon_H = \varepsilon/4$.
    By employing \Cref{corr:sign-polynomial-implementation} (with $\epsilon_1 \coloneqq 0$ and $\epsilon_2 \coloneqq \epsilon$) and the corresponding estimation procedure $\hat{\calT}(Q_i, U_{\frac{\rho_0-\rho_1}{2}}, P_{d'}^{\sign}, \Theta(\epsilon), \epsilon_H, 1/10)$ from \Cref{lemma:space-efficient-quantum-tester}, we obtain the values $x_i$ for $i\in\binset$, ensuring the following inequalities: 
    \begin{equation}
        \label{eq:GapQSDlog-tester-pacc}
        \Pr{ \left| x_i - \Tr\left(P_{d'}^{\sign}\Big(\frac{\rho_0-\rho_1}{2}\Big)\rho_i\right) \right| \leq (36 \hat{C}_{\sign}\log{d'}+37)\epsilon + \epsilon_H } \geq 0.9 \text{ for } i\in\binset.
    \end{equation}
    Here, the implementation uses $O(d^2 \log{d})$ queries to $U_{\frac{\rho_0-\rho_1}{2}}$ and $O(d^2 \log{d})$ multi-controlled single-qubit gates. 
    Moreover, the circuit descriptions of $\hat{\calT}(Q_i, U_{\frac{\rho_0-\rho_1}{2}}, P_{d'}^{\sign}, \epsilon, \epsilon_H, 1/10)$ can be computed in deterministic time $\tilde O(d^{9/2}/\epsilon)$ and space $O(s(n))$.  

    \vspace{1em}
    Now let $x \coloneqq (x_0-x_1)/2$. We will finish the correctness analysis of \Cref{algo:GapQSD-in-BQL} by showing $\Pr{\left|x - \td(\rho_0, \rho_1)\right| \leq \varepsilon} > 0.8$ through \Cref{eq:GapQSDlog-tester-pacc}. 
    By considering the approximation error of $P_{d'}^{\sign}$ in \Cref{corr:space-efficient-sign} and the QSVT implementation error in \Cref{corr:sign-polynomial-implementation}, we derive the following inequality in \Cref{prop:td-technical}, whose proof is presented immediately afterward:

    \begin{proposition}
        \label{prop:td-technical}
        $\Pr{\left|x - \td(\rho_0, \rho_1)\right| \leq (36 \hat{C}_{\sign} \log{d'} +2 C_{\sign}+37)\epsilon + \epsilon_H + 2^{r+1} \delta} > 0.8$.
    \end{proposition}

    Under the aforementioned choice of $\delta$, $\epsilon$, and $\epsilon_H$, we have $\epsilon_{H} = \varepsilon/4$, $2^{r+1}\delta = \varepsilon/4$, and $(36 \hat{C}_{\sign} \log{d'} +2 C_{\sign}+37)\epsilon \leq \varepsilon/2$, and thus, $\Pr{\left|x - \td(\rho_0, \rho_1)\right|} > 0.8$.
    
    \vspace{1em}
    Finally, we analyze the computational resources required for \Cref{algo:GapQSD-in-BQL}. 
    According to \Cref{lemma:space-efficient-quantum-tester}, we can compute $x$ in \BQL{}, with the resulting algorithm requiring $O(d^2 \log{d} /\epsilon_H^2) = \tilde O(2^{2r}/\varepsilon^4)$ queries to $Q_0$ and $Q_1$. 
    In addition, its circuit description can be computed in deterministic time $\tilde O(d^{9/2}/\varepsilon) = \tilde O(2^{4.5r}/\varepsilon^{5.5})$.     
\end{proof}

\begin{proof}[Proof of \Cref{prop:td-technical}]
    Using the triangle inequality, we obtain the following:
\begin{align*}
    \left| \frac{x_0-x_1}{2} - \td(\rho_0, \rho_1) \right|
    & = \left| \frac{x_0-x_1}{2} - \Tr \Big( \frac{\rho_0 - \rho_1}{2} \sign\Big(\frac{\rho_0 - \rho_1}{2}\Big) \Big) \right| \\
    & \leq \left| \frac{x_0-x_1}{2} - \Tr\Big(\frac{\rho_0-\rho_1}{2} P_{d'}^{\sign}\Big(\frac{\rho_0-\rho_1}{2}\Big)\Big) \right| \\
    & \qquad + \left| \Tr\Big(\frac{\rho_0-\rho_1}{2} P_{d'}^{\sign}\Big(\frac{\rho_0-\rho_1}{2}\Big)\Big) - \Tr \Big( \frac{\rho_0 - \rho_1}{2} \sign\Big(\frac{\rho_0 - \rho_1}{2}\Big) \Big) \right|.
\end{align*}

    For the first term, by noting the QSVT implementation error in \Cref{corr:sign-polynomial-implementation}, we know by \Cref{eq:GapQSDlog-tester-pacc} that, with probability at least $0.9^2 > 0.8$, it holds that
    \begin{equation}
    \label{eq:td-first-term}
    \left| \frac{x_0-x_1}{2} - \Tr\Big(\frac{\rho_0-\rho_1}{2} P_{d'}^{\sign}\Big(\frac{\rho_0-\rho_1}{2}\Big)\Big) \right| \leq (36 \hat{C}_{\sign} \log{d'} +37)\epsilon + \epsilon_H.
    \end{equation}
    For the second term, let $\frac{\rho_0-\rho_1}{2} = \sum_{j} \lambda_j \ket{\psi_j} \bra{\psi_j}$, where $\{ \ket{\psi_j} \}$ is an orthonormal basis. Then,
    \begin{equation}
    \label{eq:td-second-term}
    \left| \Tr\Big(\frac{\rho_0\!-\!\rho_1}{2} P_{d'}^{\sign}\Big(\frac{\rho_0\!-\!\rho_1}{2}\Big)\Big) - \Tr \Big( \frac{\rho_0 \!-\! \rho_1}{2} \sign\Big(\frac{\rho_0 \!-\! \rho_1}{2}\Big) \Big) \right| \leq \sum_j \left| \lambda_j P_{d'}^{\sign}(\lambda_j) \!-\! \lambda_j \sign(\lambda_j) \right|.
    \end{equation}
    We split the summation over $j$ into three separate summations:
    \[\sum_j = \sum_{\lambda_j < -\delta} + \sum_{\lambda_j > \delta} + \sum_{-\delta \leq \lambda_j \leq \delta}.\]
    By noticing the approximation error of $P_{d'}^{\sign}$ in \Cref{corr:space-efficient-sign} and $\sum_j |\lambda_j| = \td(\rho_0,\rho_1) \leq 1$, we obtain the following results for each of the three summations: 
    \[
    \sum_{\lambda_j > \delta} \left| \lambda_j P_{d'}^{\sign}(\lambda_j) - \lambda_j \sign(\lambda_j) \right| = \sum_{\lambda_j > \delta} |\lambda_j| \left| P_{d'}^{\sign}(\lambda_j) - 1 \right| \leq \sum_{\lambda_j > \delta} |\lambda_j| C_{\sign} \epsilon \leq C_{\sign} \epsilon,
    \]
    \[
    \sum_{\lambda_j < -\delta} \left| \lambda_j P_{d'}^{\sign}(\lambda_j) - \lambda_j \sign(\lambda_j) \right| = \sum_{\lambda_j < -\delta} |\lambda_j| \left| P_{d'}^{\sign}(\lambda_j) + 1 \right| \leq \sum_{\lambda_j < -\delta} |\lambda_j| C_{\sign} \epsilon \leq C_{\sign} \epsilon,
    \]
    \[
    \sum_{-\delta \leq \lambda_j \leq \delta} \left| \lambda_j P_{d'}^{\sign}(\lambda_j) - \lambda_j \sign(\lambda_j) \right| \leq \sum_{-\delta \leq \lambda_j \leq \delta} 2|\lambda_j| \leq 2^{r+1}\delta.
    \]
    Hence, we derive the following inequality by summing over the aforementioned three inequalities: 
    \begin{equation}
    \label{eq:td-last-term}
    \sum_j \left| \lambda_j P_{d'}^{\sign}(\lambda_j) - \lambda_j \sign(\lambda_j) \right| \leq 2^{r+1} \delta + 2 C_{\sign} \epsilon.
    \end{equation}
    By combining \Cref{eq:td-first-term,eq:td-second-term,eq:td-last-term}, we conclude that
    \[
    \left|\frac{x_0-x_1}{2} - \td(\rho_0, \rho_1)\right| \leq (36 \hat{C}_{\sign} \log{d'} + 37)\epsilon + \epsilon_H + 2 C_{\sign} \epsilon + 2^{r+1} \delta. \qedhere
    \]
\end{proof}

\subsection{\GapQEDlog{} and \GapQJSlog{} are in \BQL{}}
\label{subsec:GapQED-in-BQL}

In this subsection, we will demonstrate \Cref{thm:GapQEDlog-in-BQL} by devising a quantum algorithm that encompasses testers $\calT(Q_i, U_{\rho_i},P_{d'}^{\ln},\epsilon)$ for $i \in \binset$, where the construction of testers employs the space-efficient QSVT associated with the normalized logarithmic function. Consequently, we can deduce that \GapQJSlog{} is in \BQL{} via a reduction from \GapQJSlog{} to \GapQEDlog{}. 

\begin{theorem}
    \label{thm:GapQEDlog-in-BQL}
    For any deterministic logspace computable function $g(n)$ that satisfies $g(n) \geq 1/\poly(n)$, we have that $\GapQEDlog[g(n)]$ is in \BQL{}.
\end{theorem}

\begin{proof}
    We begin with a formal algorithm in \Cref{algo:GapQEDlog-in-BQL}.

    \begin{algorithm}[ht!]
		\caption{Space-efficient algorithm for $\GapQEDlog$.}
		\label{algo:GapQEDlog-in-BQL}
        \SetKwInOut{Input}{Input}
        \SetKwInOut{Output}{Output}
        \SetKwInOut{Parameter}{Params}
        \Input{Quantum circuits $Q_i$ that prepare the purification of $\rho_i$ for $i \in \binset$.}
        \Output{An additive-error estimation of $\S(\rho_0)-\S(\rho_1)$.}
        \Parameter{$\varepsilon \coloneqq \frac{g}{4}$, $\beta \coloneqq \min\{ \frac{\varepsilon}{2^{r+6} \ln(2^{r+6}/\varepsilon)}, \frac{1}{4} \}$, 
        $\epsilon \coloneqq \frac{\varepsilon^{3/2}}{64\cdot 2^{r/2}s^2\max\{1,\hat C_{\ln},C_{\ln},\tilde C_{\ln}\}}$,
        $d' \coloneqq \tilde C_{\ln} \!\cdot\! \frac{1}{\beta}\log\frac{1}{\epsilon} = 2d\!-\!1$,
        $\epsilon_H \coloneqq \frac{\varepsilon}{8\hat C_{\ln}\sqrt{d'}\ln(2/\beta)}$.
        }
        1. Construct block-encodings of $\rho_0$ and $\rho_1$, denoted by $U_{\rho_0}$ and $U_{\rho_1}$, respectively, using $O(1)$ queries to $Q_0$ and $Q_1$ and $O(s(n))$ ancillary qubits by \Cref{lemma:purified-density-matrix}\;
        \emph{Let $P_{d'}^{\ln}$ be the degree-$d'$ polynomial specified in \Cref{corr:space-efficient-log} with parameters $\beta$ and $\epsilon$, and its coefficients $\{c_k^{\ln}\}_{k=0}^{d'}$ are computable in bounded-error randomized space $O(\log (d/\epsilon))$}\;
        2. For each $i\in\binset$, set $x_i \coloneqq \hat C_{\ln}\sqrt{d'}\cdot \hat{\mathcal T}(Q_i,U_{\rho_i},P_{d'}^{\ln},\epsilon,\epsilon_H,1/10)$\;
        3. Compute $x = 2\ln(\frac{2}{\beta})(x_0 - x_1)$. Return ``yes'' if $x > 0$, and ``no'' otherwise. 
        \BlankLine
    \end{algorithm}

    Let us now demonstrate the correctness and computational complexity of \Cref{algo:GapQEDlog-in-BQL}. We concentrate on the scenario with $s(n) = \Theta(\log{n})$ and $\varepsilon = g/4 \geq 2^{-O(s)}$. Our strategy is to estimate the entropy of each of $\rho_0$ and $\rho_1$, respectively. We assume that $Q_0$ and $Q_1$ are $s$-qubit quantum circuits that prepare the purifications of $\rho_0$ and $\rho_1$, respectively. By \Cref{lemma:purified-density-matrix}, we can construct $(1,O(s), 0)$-block-encodings $U_{\rho_0}$ and $U_{\rho_1}$ of $\rho_0$ and $\rho_1$, respectively, using $O(1)$ queries to $Q_0$ and $Q_1$ as well as $O(1)$ one- and two-qubit quantum gates.

    Let $\beta = \min\{ \frac{\varepsilon}{2^{r+6} \ln(2^{r+6}/\varepsilon)}, \frac{1}{4} \}$, 
    $\epsilon \coloneqq \frac{\varepsilon^{3/2}}{64\cdot 2^{r/2}s^2\max\{1,\hat C_{\ln},C_{\ln},\tilde C_{\ln}\}}$,
    and $d' \coloneqq \tilde C_{\ln} \cdot \frac{1}{\beta} \log \frac{1}{\epsilon} = 2^{O(s(n))}$, 
    where $\tilde C_{\ln}$ comes from \Cref{corr:space-efficient-log}.
    Let $P_{d'}^{\ln} \in \mathbb{R}[x]$ be the polynomial specified in \Cref{corr:space-efficient-log} with $d' = 2d-1$.
    Let $\epsilon_H \coloneqq \frac{\varepsilon}{8\hat C_{\ln}\sqrt{d'}\ln(2/\beta)}$. 
    By utilizing \Cref{corr:log-polynomial-implementation} (with $\epsilon_1 \coloneqq 0$ and $\epsilon_2 \coloneqq \epsilon$) and the corresponding estimation procedure $\hat{\calT}(Q_i, U_{\rho_i}, P_{d'}^{\ln}, \epsilon, \epsilon_H, 1/10)$ from \Cref{lemma:space-efficient-quantum-tester}, we obtain the values $x_i$ for $i\in\binset$, ensuring the following inequalities: 
    \begin{equation}
        \label{eq:GapQEDlog-tester-pacc}
        \Pr{ \left| x_i - \Tr\left(P_{d'}^{\ln}\!\left(\rho_i\right)\rho_i\right) \right| \leq \hat{C}_{\ln} \sqrt{d'} \rbra{\epsilon + \epsilon_H} } \geq 0.9 \text{ for } i\in\binset. 
    \end{equation}
    Here, the implementation uses $O(d^2)$ queries to $U_{\rho_i}$ and $O(d^2)$ multi-controlled single-qubit gates.
    Moreover, the circuit descriptions of $\hat{\calT}(Q_i, U_{\rho_i}, P_{d'}^{\ln}, \epsilon, \epsilon_H, 1/10)$ can be computed in bounded-error time $\tilde O(d^{9}/\epsilon^4)$ and space $O(s(n))$. 

    \vspace{1em}
    We will finish the correctness analysis of \Cref{algo:GapQEDlog-in-BQL} by showing $\Pr{\left| 2 \ln\rbra[\big]{\tfrac{2}{\beta}} x_i - \S(\rho_i)\right| \leq \varepsilon} \geq 0.9$ through \Cref{eq:GapQEDlog-tester-pacc}. By considering the approximation error of $P_{d'}^{\ln}$ in \Cref{corr:space-efficient-log} and the QSVT implementation error in \Cref{corr:log-polynomial-implementation}, we derive the following inequality in \Cref{prop:entropy-technical}, the proof of which is presented immediately afterward: 
    \begin{proposition}
        \label{prop:entropy-technical}
        The following inequality holds for $i \in \binset$: 
        \[\mathrm{Pr}\Big[\left| 2 \ln\!\big(\tfrac{2}{\beta}\big) x_i - \S(\rho_i) \right| \leq 2 \ln\!\big(\tfrac{2}{\beta}\big) \rbra*{ \hat C_{\ln}\sqrt{d'}(\epsilon+\epsilon_H)+C_{\ln}\epsilon+2^{r+1}\beta}  \Big] \geq 0.9.\] 
    \end{proposition}

    Consequently, it is left to show that $2 \ln\!\big(\tfrac{2}{\beta}\big) \rbra*{ \hat C_{\ln}\sqrt{d'}(\epsilon+\epsilon_H)+C_{\ln}\epsilon+2^{r+1}\beta}  \leq \varepsilon$ for the specified value of $\beta$, $\epsilon$, and $\epsilon_H$. This can be seen by noting that 
    \[ 2\ln\rbra[\big]{\tfrac{2}{\beta}} \hat{C}_{\ln} \sqrt{d'} \epsilon_{H} = \tfrac{\varepsilon}{4}, \quad 2\ln\rbra[\big]{\tfrac{2}{\beta}} (\hat C_{\ln}\sqrt{d'} + C_{\ln}) \epsilon \leq \tfrac{\varepsilon}{2}, \quad\text{and}\quad 2 \ln\rbra[\big]{\tfrac{2}{\beta}} \cdot 2^{r+1} \beta \leq \tfrac{\varepsilon}{4}.\]
    The first identity and the second inequality are trivial. The third inequality, stated in \Cref{prop:GapQEDlog-beta-bound}, is accompanied by a proof presented immediately afterward: 
    \begin{proposition}
        \label{prop:GapQEDlog-beta-bound}
        $2 \ln(\frac{2}{\beta}) \cdot 2^{r+1} \beta \leq \tfrac{\varepsilon}{4}$.
    \end{proposition}

    Finally, we analyze the computational resources required for \Cref{algo:GapQEDlog-in-BQL}. 
    As per \Cref{lemma:space-efficient-quantum-tester}, we can compute $x$ in \BQL{}, with the resulting algorithm requiring $O(d^2/\epsilon_H^2) = \tilde O(2^{3r}/\varepsilon^5)$ queries to $Q_0$ and $Q_1$. Furthermore, its circuit description can be computed in bounded-error randomized time $\tilde{O}(d^{9}/\epsilon^4) = \tilde O(2^{11r}/\varepsilon^{15})$. 
\end{proof}

\paragraph*{\GapQJSlog{} is in \BQL{}.}
We can achieve $\GapQJSlog \in \BQL$ by employing the estimation procedure $\hat{\calT}$ in \Cref{algo:GapQEDlog-in-BQL} for \textit{three corresponding states}, given that the quantum Jensen-Shannon divergence $\QJS(\rho_0,\rho_1)$ is a linear combination of $\S(\rho_0), \S(\rho_1)$, and $\S\big(\frac{\rho_0+\rho_1}{2}\big)$. Nevertheless, the logspace Karp reduction from \GapQJSlog{} to \GapQEDlog{} (\Cref{corr:GapQJSlog-in-BQL}) allows us to utilize $\hat{\calT}$ for only \textit{two} states.\footnote{It is worth noting that a more direct approach was recently presented in~\cite[Equation (4)]{LW25entropy}, where $\QJS(\rho_0,\rho_1)$ is expressed as the difference between the von Neumann entropies of the quantum states $\rbra[\big]{\frac{\rho_0+\rho_1}{2}} \otimes \rbra[\big]{\frac{\rho_0+\rho_1}{2}}$ and $\rho_0 \otimes \rho_1$. Consequently, the proof is much simpler than that of \Cref{corr:GapQJSlog-in-BQL}.} Furthermore, our construction is adapted from the time-bounded scenario~\cite[Lemma 4.4]{Liu23}. 

\begin{corollary}
    \label{corr:GapQJSlog-in-BQL}
    For any functions $\alpha(n)$ and $\beta(n)$ that can be computed in deterministic logspace and satisfy $\alpha(n)-\beta(n) \geq 1/\poly(n)$, we have that $\GapQJSlog[\alpha(n),\beta(n)]$ is in \BQL{}. 
\end{corollary}

\begin{proof}
Let $Q_0$ and $Q_1$ be the given $s(n)$-qubit quantum circuits where $s(n)=\Theta(\log{n})$. 
Consider a classical-quantum mixed state on a classical register $\sfB$ and a quantum register $\sfY$, denoted by $\rho'_1 \coloneqq \frac{1}{2}\ket{0}\bra{0}\otimes \rho_0 + \frac{1}{2} \ket{1}\bra{1}\otimes \rho_1$, where $\rho_0$ and $\rho_1$ are states obtained by running $Q_0$ and $Q_1$, respectively, and tracing out the non-output qubits. 
We utilize our reduction to output classical-quantum mixed states $\rho'_0$ and $\rho'_1$, which are the output of $(s+2)$-qubit quantum circuits $Q'_0$ and $Q'_1$,\footnote{To construct $Q'_1$, we follow these steps: 
We start by applying a \Had{} gate on $\sfB$ followed by a $\CNOT_{\sfB \rightarrow \sfR}$ gate where $\sfB$ and $\sfR$ are single-qubit quantum registers initialized on $\ket{0}$. 
Next, we apply the controlled-$Q_1$ gate on the qubits from $\sfB$ to $\sfS$, where $\sfS=(\sfY,\sfZ)$ is an $s(n)$-qubit register initialized on $\ket{\bar{0}}$. We then apply $X$ gate on $\sfB$ followed by the controlled-$Q_0$ gate on the qubits from $\sfB$ to $\sfS$, and we apply $X$ gate on $\sfB$ again. 
Finally, we obtain $\rho'_1$ by tracing out $\sfR$ and the qubits in $\sfZ$. 
In addition, we can construct $Q'_0$ similarly.} respectively, where $\rho'_0\coloneqq(p_0\ket{0}\bra{0}+p_1\ket{1}\bra{1})\otimes(\frac{1}{2}\rho_0+\frac{1}{2}\rho_1)$ and $\sfB'\coloneqq(p_0,p_1)$ is an independent random bit with entropy $\H(\sfB')=1-\frac{1}{2}[\alpha(n)+\beta(n)]$. 
Let $\S_2(\rho)\coloneqq\S(\rho)/\ln{2}$ for any quantum state $\rho$, we then have derived that: 
\begin{subequations}
\label{eq:QJS-to-QED}
\begin{align}
	\S_2(\rho'_0)-\S_2(\rho'_1) &= \S_2(\sfB',\sfY)_{\rho'_0} - \S_2(\sfB,\sfY)_{\rho'_1}\\
	&= [\H(\sfB')+\S_2(\sfY|\sfB')_{\rho'_0}]-[\H(\sfB)+\S_2(\sfY|\sfB)_{\rho'_1}]\\
	&= \S_2(\sfY)_{\rho'_0} - \S_2(\sfY|\sfB)_{\rho'_1} + \H(\sfB') - \H(\sfB)\\
	&= \S_2(\sfY)_{\rho'_0}-\S_2(\sfY|\sfB)_{\rho'_1} -\tfrac{1}{2}[\alpha(n)+\beta(n)]\\
	&= \S_2 \rbra*{ \tfrac{1}{2}\rho_0+\tfrac{1}{2}\rho_1 } - \tfrac{1}{2}(\S_2(\rho_0)+\S_2(\rho_1)) -\tfrac{1}{2}[\alpha(n)+\beta(n)]\\
	&= \QJS_2(\rho_0,\rho_1) - \tfrac{1}{2}[\alpha(n)+\beta(n)].
\end{align}
\end{subequations}
Here, the second line derives from the definition of quantum conditional entropy and acknowledges that both $\sfB$ and $\sfB'$ are classical registers. The third line owes to the independence of $\sfB'$ as a random bit. Furthermore, the fifth line relies on the Joint entropy theorem (\Cref{lemma:joint-entropy-theorem}).	

By plugging \Cref{eq:QJS-to-QED} into the conditions in the promise of $\GapQJS_{\log}[\alpha(n),\beta(n)]$, we can define $g(n')\coloneqq\frac{\ln{2}}{2} \big( \alpha(n)-\beta(n) \big)$ and conclude that:
\begin{itemize}
	\item If $\QJS_2(\rho_0,\rho_1) \geq \alpha(n)$, then $\S(\rho'_0)-\S(\rho'_1) \geq \frac{\ln{2}}{2} \big( \alpha(n)-\beta(n) \big) = g(n')$;
	\item If $\QJS_2(\rho_0,\rho_1) \leq \beta(n)$, then $\S(\rho'_0)-\S(\rho'_1) \leq -\frac{\ln{2}}{2} \big(\alpha(n)-\beta(n) \big) = -g(n')$.
\end{itemize}
As $\rho'_1$ and $\rho'_0$ are $r'(n')$-qubit states where $n'\coloneqq3n$\footnote{By inspecting the circuit description of $Q'_0$ and $Q'_1$ (see~\cite[Section 4.2]{Liu23} for details), the maximum number of gates in $Q'_0$ and $Q'_1$ is $2n+9+\polylog(1/\epsilon) \leq 3n$ for large enough $n$. Specifically, the implementation of $R_{\theta}$ in~\cite[Figure 1]{Liu23} requires $\polylog(1/\epsilon) = \polylog(n)$ gates due to the space-efficient Solovay--Kitaev theorem~\cite[Theorem 4.3]{vMW12}.} and $r'(n')\coloneqq r(n)+1$, the output length of the corresponding space-bounded quantum circuits $Q'_0$ and $Q'_1$ is $r'(n')$. Hence, $\GapQJS_{s}[\alpha,\beta]$ is logspace Karp reducible to $\GapQED_{s+3}[g]$ by mapping $(Q_0,Q_1)$ to $(Q'_0,Q'_1)$.
\end{proof}

We next present the proofs of \Cref{prop:entropy-technical,prop:GapQEDlog-beta-bound}:

\begin{proof}[Proof of \Cref{prop:entropy-technical}]
    We only prove the case with $i=0$ while the case with $i=1$ follows straightforwardly. By applying the triangle inequality with $i=0$, we have: 
    \[ \left| 2\ln\!\big(\tfrac{2}{\beta}\big) x_0 - \S(\rho_0) \right|
     \leq \left| 2\ln\!\big(\tfrac{2}{\beta}\big) x_0 - 2\ln\!\big(\tfrac{2}{\beta}\big) \Tr\left(P_{d'}^{\ln}(\rho_0)\rho_0\right) \right| + \left| 2\ln\!\big(\tfrac{2}{\beta}\big) \Tr\left(P_{d'}^{\ln}(\rho_0)\rho_0\right) - \S(\rho_0) \right|. \]
    For the first term, by noting the QSVT implementation error in \Cref{corr:log-polynomial-implementation}, we have by \Cref{eq:GapQEDlog-tester-pacc} that with probability at least $0.9$, it holds that
    \begin{equation}
    \label{eq:entropy-first-term}
    \left| 2\ln\!\big(\tfrac{2}{\beta}\big) x_0 - 2 \ln\!\big(\tfrac{2}{\beta}\big) \Tr\rbra[\big]{P_{d'}^{\ln}(\rho_0)\rho_0} \right| \leq 2 \ln\!\big(\tfrac{2}{\beta}\big)  \hat{C}_{\ln} \sqrt{d'} (\epsilon + \epsilon_H). 
    \end{equation}
    For the second term, let $\rho_0 = \sum_{j} \lambda_j \ket{\psi_j} \bra{\psi_j}$, where $\{ \ket{\psi_j} \}$ is an orthonormal basis.
    Then,
    \begin{equation}
    \label{eq:entropy-second-term}
    \left| 2\ln\!\big(\tfrac{2}{\beta}\big) \Tr\rbra[\big]{P_{d'}^{\ln}(\rho_0)\rho_0} - \S(\rho_0) \right| \leq \sum_j \left| 2 \ln\!\big(\tfrac{2}{\beta}\big) \lambda_j P_{d'}^{\ln}(\lambda_j) - \lambda_j \ln(1/\lambda_j) \right|.
    \end{equation}
    We split the summation over $j$ into two separate summations:
    $\sum_j = \sum_{\lambda_j > \beta} + \sum_{\lambda_j \leq \beta}.$
    By noticing the approximation error of $P_{d'}^{\ln}$ in \Cref{corr:space-efficient-log} and $\sum_{j} |\lambda_j| = \Tr(\rho_0) \leq 1$, we can then obtain the following results for each of the two summations: 
    \begin{align*}
    \sum_{\lambda_j > \beta} \left| 2\ln\!\big(\tfrac{2}{\beta}\big) \lambda_j P_{d'}^{\ln}(\lambda_j) - \lambda_j \ln(1/\lambda_j) \right| 
    & = \sum_{\lambda_j > \beta} |\lambda_j| \cdot \left| 2 \ln\!\big(\tfrac{2}{\beta}\big) P_{d'}^{\ln}(\lambda_j) - \ln(1/\lambda_j) \right| \\
    & \leq \sum_{\lambda_j > \beta} |\lambda_j| \cdot 2 \ln\!\big(\tfrac{2}{\beta}\big) C_{\ln} \epsilon \\
    & \leq 2 \ln\!\big(\tfrac{2}{\beta}\big) C_{\ln} \epsilon,
    \end{align*}
    \begin{align*}
    \sum_{\lambda_j \leq \beta} \left| 2 \ln\!\big(\tfrac{2}{\beta}\big) \lambda_j P_{d'}^{\ln}(\lambda_j) - \lambda_j \ln(1/\lambda_j) \right|
    & \leq \sum_{\lambda_j \leq \beta} \left( 2 \ln\!\big(\tfrac{2}{\beta}\big) |\lambda_j| + |\lambda_j| \ln(1/\lambda_j) \right) \\
    & \leq 2 \ln\!\big(\tfrac{2}{\beta}\big) \cdot 2^r \beta + 2^r \beta\ln(1/\beta)\\
    & \leq  2\ln\!\big(\tfrac{2}{\beta}\big) 2^{r+1} \beta.
    \end{align*}
    Hence, we have derived the following inequality by summing over the aforementioned inequalities: 
    \begin{equation}
    \label{eq:entropy-last-term}
    \sum_j \left| 2 \ln\!\big(\tfrac{2}{\beta}\big) \lambda_j P_{d'}^{\ln}(\lambda_j) - \lambda_j \ln(1/\lambda_j) \right| \leq 2 \ln\!\big(\tfrac{2}{\beta}\big) C_{\ln} \epsilon + 2\ln\!\big(\tfrac{2}{\beta}\big) 2^{r+1} \beta.
    \end{equation}
    By combining \Cref{eq:entropy-first-term,eq:entropy-second-term,eq:entropy-last-term}, we conclude that
    \[
    \left| 2\ln\!\big(\tfrac{2}{\beta}\big) x_0 - \S(\rho_0) \right| \leq 2 \ln\!\big(\tfrac{2}{\beta}\big) \left(  \hat{C}_{\ln}\sqrt{d'} \rbra[\big]{\epsilon + \epsilon_H}  + C_{\ln} \epsilon + 2^{r+1} \beta \right). \qedhere
    \]
\end{proof}

\begin{proof}[Proof of \Cref{prop:GapQEDlog-beta-bound}]
    Note that the choice of $\beta$ is given by $\beta \coloneqq \frac{\varepsilon}{2^{r+6} \ln(\frac{2^{r+6}}{\varepsilon})}$.
    Then, to demonstrate the inequality $2 \ln(\frac{2}{\beta}) \cdot 2^{r+1} \beta \leq \tfrac{\varepsilon}{4}$, it suffices to prove that
    \begin{equation}
        \label{eq:GapQEDlog-beta-bound}
        2 \ln\left(\frac{2^{r+7} \ln(\frac{2^{r+6}}{\varepsilon})}{\varepsilon}\right) \cdot \frac{\varepsilon}{2^{5} \ln(\frac{2^{r+6}}{\varepsilon})} \leq \frac{\varepsilon}{4}.
    \end{equation}
    
    Let $x \coloneqq 2^{-r-6} \varepsilon \in (0, 1)$, then \Cref{eq:GapQEDlog-beta-bound} becomes $\ln\left(\frac{2}{x} \ln\left(\frac{1}{x}\right)\right) \leq 4 \ln\left(\frac{1}{x}\right)$.
    This simplifies further to $2x^3\ln\left(\frac{1}{x}\right) \leq 1$.
    
    To complete the proof, let $f(x) = 2x^3\ln(\frac{1}{x})$, then its first derivative is $f'(x) = 2x^2 \left(3\ln\left(\frac{1}{x}\right)-1\right)$. 
    Note that $f'(x) > 0$ for $x \in (0, e^{-1/3})$ and $f'(x) < 0$ for $x \in (e^{-1/3}, 1)$. Thus, $f(x)$ is monotonically increasing for $x \in (0, e^{-1/3})$ and monotonically decreasing for $x \in (e^{-1/3}, 1)$.
    Therefore, $f(x)$ takes the maximum value at $x = e^{-1/3}$, and consequently, $f(x) \leq f(e^{-1/3}) = \frac{2}{3e} \leq 1$.
\end{proof}

\subsection{\coCertQSDlog{} and \coCertQHSlog{} are in \coRQUL{}}
\label{subsec:coCertQSD-in-coRQL}

To make the error one-sided, we adapt the Grover search when the number of solutions is one quarter~\cite{BBHT98}, also known as the exact amplitude amplification~\cite{BHMT02}. 

\begin{lemma} [Exact amplitude amplification, adapted from~{\cite[Equation 8]{BHMT02}}]
\label{lemma:exact-amplitude-amplification}
    Suppose $U$ is a unitary of interest such that $U\ket{\bar 0} = \sin(\theta) \ket{\psi_0} + \cos(\theta) \ket{\psi_1}$, where $\ket{\psi_0}$ and $\ket{\psi_1}$ are normalized pure states and $\innerprod{\psi_0}{\psi_1} = 0$.
    Let $G = - U (I - 2\ket{\bar 0}\bra{\bar 0}) U^\dag (I - 2\ket{\psi_0}\bra{\psi_0})$ be the Grover operator. Then, for every integer $j \geq 0$, we have
    \[G^j U \ket{\bar 0} = \sin((2j+1)\theta) \ket{\psi_0} + \cos((2j+1)\theta) \ket{\psi_1}.\]

    \noindent In particular, with a single application of $G$, we obtain 
    \[G U \ket{\bar 0} = \sin(3\theta) \ket{\psi_0} + \cos(3\theta) \ket{\psi_1},\]
    signifying that $G U \ket{\bar 0} = \ket{\psi_0}$ when $\sin(\theta) = 1/2$.
\end{lemma}

Notably, when dealing with the unitary of interest with the property specified in \Cref{lemma:exact-amplitude-amplification}, which is typically a quantum algorithm with acceptance probability linearly dependent on the chosen distance-like measure (e.g., a tester $\calT$ from \Cref{lemma:space-efficient-quantum-tester}), \Cref{lemma:exact-amplitude-amplification} guarantees that the resulting algorithm $\calA$ accepts with probability exactly $1$ for \textit{yes} instances ($\rho_0=\rho_1$). However, achieving $\calA$ to accept with probability polynomially deviating from $1$ for \textit{no} instances requires additional efforts, leading to the \coRQUL{} containment established through error reduction for \coRQUL{} (\Cref{corr:error-reduction-untary-quantum-logspace}).
In a nutshell, demonstrating \coRQUL{} containment entails satisfying the desired property, which is achieved differently for \coCertQSDlog{} and \coCertQHSlog{}.

\subsubsection{\coCertQSDlog{} is in \coRQUL{}}

Our algorithm in \Cref{thm:coCertQSDlog-is-in-coRQUL} relies on the quantum tester $\calT(Q_i, U_{\frac{\rho_0-\rho_1}{2}},P_{d'}^{\sign},\epsilon)$, which is specified in \Cref{algo:GapQSD-in-BQL}. 
According to \Cref{remark:QSVT-parity-preserving}, the exact implementation of the space-efficient QSVT associated with odd polynomials preserves the original point. As a consequence, $\calT(Q_i, U_{\frac{\rho_0-\rho_1}{2}}, P_{d'}^{\sign},\epsilon)$ outputs $0$ with probability exactly $1/2$ when $\rho_0=\rho_1$, enabling us to derive the \coRQUL{} containment through a relatively involved analysis for cases when $\td(\rho_0,\rho_1) \geq \alpha$: 

\begin{theorem}
\label{thm:coCertQSDlog-is-in-coRQUL}
    For any deterministic logspace computable function $\alpha(n) \geq 1/\poly(n)$, we have that $\coCertQSDlog[\alpha(n)]$ is in \coRQUL{}.
\end{theorem}
\begin{proof}
    We first present a formal algorithm in \Cref{algo:coCertQSDlog-in-RQL}:
    \begin{algorithm}[!htp]
    	\caption{Space-efficient algorithm for \coCertQSDlog{}.}
		\label{algo:coCertQSDlog-in-RQL}
        \SetKwInOut{Input}{Input}
        \SetKwInOut{Output}{Output}
        \SetKwInOut{Parameter}{Params}
        \Input{Quantum circuits $Q_i$ that prepare the purification of $\rho_i$ for $i \in \binset$.}
        \Output{Return ``yes'' if $\rho_0=\rho_1$, and ``no'' otherwise.}
        \Parameter{$\varepsilon \coloneqq \frac{\alpha}{2}$, $\delta \coloneqq \frac{\varepsilon}{2^{r+3}}$, 
        $\epsilon \coloneqq \frac{\varepsilon}{64\,s\,\max\{1,36 \hat{C}_{\sign},2 C_{\sign}+37,\tilde C_{\sign}\}}$,
        $d' \coloneqq \tilde C_{\sign} \, \frac{1}{\delta} \log\frac{1}{\epsilon} = 2d-1$.}
        1. Construct block-encodings of $\rho_0$ and $\rho_1$, denoted by $U_{\rho_0}$ and $U_{\rho_1}$, respectively, using $O(1)$ queries to $Q_0$ and $Q_1$ and $O(s(n))$ ancillary qubits by \Cref{lemma:purified-density-matrix}\;
        2. Construct a block-encoding of $\frac{\rho_0-\rho_1}{2}$, denoted by $U_{\frac{\rho_0-\rho_1}{2}}$, using $O(1)$ queries to $U_{\rho_0}$ and $U_{\rho_1}$ and $O(s(n))$ ancillary qubits by \Cref{lemma:space-efficient-LCU}\;
        \emph{Let $P_{d'}^{\sign}$ be the degree-$d'$ \emph{odd} polynomial specified in \Cref{corr:space-efficient-sign} with parameters $\!\delta\!$ and $\!\epsilon$, and its coefficients $\{\hat{c}_k\}_{k=0}^{d'}$ are computable in deterministic space $O(\log (d/\epsilon))$}\;
        3. Let $U_0 \coloneqq \calT(Q_0, U_{\frac{\rho_0-\rho_1}{2}}, P_{d'}^{\sign}, \epsilon)$ and $U_1 \coloneqq \calT(Q_1, U_{\frac{\rho_0-\rho_1}{2}}, P_{d'}^{\sign}, \epsilon)$\;
        4. Let $G_i \coloneqq - (H \otimes U_i) (I - 2\ket{\bar 0}\bra{\bar 0}) (H \otimes U_i^\dag) (I - 2\Pi_0)$ for $i \in \binset$, where $\Pi_0$ is the projection onto the subspace spanned by $\{ \ket{0}\ket{0}\ket{\varphi} \}$ over all $\ket{\varphi}$\;
        5. For each $i\in\binset$, measure the first two qubits of $G_i (H \otimes U_i) \ket{0}\ket{0}\ket{\bar 0}$, and let $x_{i0}$ and $x_{i1}$ be the outcomes, respectively. Return ``yes'' if $x_{00} = x_{01} = x_{10} = x_{11} = 0$, and ``no'' otherwise. 
        \BlankLine
    \end{algorithm}

    \paragraph*{Constructing the unitary of interest via the space-efficient QSVT.}
    We consider the setting with $s(n) = \Theta(\log{n})$ and  $\varepsilon = \alpha / 2$.
    Suppose $Q_0$ and $Q_1$ are $s(n)$-qubit quantum circuits that prepare the purifications of $\rho_0$ and $\rho_1$, respectively. 
    Similar to \Cref{algo:GapQSD-in-BQL}, we first construct an $O(s)$-qubit quantum circuit $U_{\frac{\rho_0 - \rho_1}{2}}$ that is a $(1,O(s), 0)$-block-encoding of $\frac{\rho_0-\rho_1}{2}$, using $O(1)$ queries to $Q_0$ and $Q_1$ and $O(1)$ one- and two-qubit quantum gates. 

    Let $\delta = \frac{\varepsilon}{2^{r+3}}$, $\epsilon \coloneqq \frac{\varepsilon}{64\,s\,\max\{1,36 \hat{C}_{\sign},2 C_{\sign}+37,\tilde C_{\sign}\}}$, and $d' \coloneqq \tilde C_{\sign} \cdot \frac{1}{\delta} \log\frac{1}{\epsilon} = 2^{O(s(n))}$, where $\tilde C_{\sign}$ comes from \Cref{corr:space-efficient-sign}.
    Let $P_{d'}^{\sign} \in \mathbb{R}[x]$ be the odd polynomial specified in \Cref{corr:space-efficient-sign}.
    Let $U_i \coloneqq \calT(Q_i, U_{\frac{\rho_0-\rho_1}{2}}, P_{d'}^{\sign}, \epsilon)$ for $i \in \binset$, then we have the following equalities:
    \[
        \forall i\in\binset, \quad U_i \ket{0}\ket{\bar 0} = \sqrt{p_i} \ket{0} \ket{\psi_{i,0}} + \sqrt{1 - p_i} \ket{1} \ket{\psi_{i,1}}, \quad\text{where } 0 \leq p_i \leq 1.
    \]

    Let $H$ be the \textsc{Hadamard} gate, then we derive the following equality for $i\in\binset$: 
    \[
    (H \otimes U_i) \ket{0} \ket{0}\ket{\bar 0} = \sqrt{\frac{p_i}{2}} \ket{0}\ket{0}\ket{\psi_{i,0}} + \underbrace{\sqrt{\frac{p_i}{2}} \ket{1}\ket{0}\ket{\psi_{i,0}} +
    \sqrt{\frac{1-p_i}{2}} \ket{0}\ket{1}\ket{\psi_{i,1}} +
    \sqrt{\frac{1-p_i}{2}} \ket{1}\ket{1}\ket{\psi_{i,1}}}_{\sqrt{1 - \frac{p_i}{2}}\ket{\perp_i}}.
    \]

    \paragraph*{Making the error one-sided by exact amplitude amplification.} 
    Consider the Grover operator $G_i \coloneqq - (H \otimes U_i) (I - 2\ket{\bar 0}\bra{\bar 0}) (H \otimes U_i^\dag) (I - 2\Pi_0)$, where $\Pi_0$ is the projection onto the subspace spanned by $\{ \ket{0}\ket{0}\ket{\varphi} \}$ over all $\ket{\varphi}$. 
    By employing the exact amplitude amplification (\Cref{lemma:exact-amplitude-amplification}), we can obtain that:  
    \begin{subequations}
        \label{eq:exact-AA-coCertQSDlog}
        \begin{align}
            G_i (H \otimes U_i) \ket{0}\ket{0}\ket{\bar 0} &= \sin(3\theta_i) \ket{0}\ket{0}\ket{\psi_{i,0}} + \cos(3\theta_i) \ket{\perp_i},\\
            \text{where } \sin^2(\theta_i) &= \frac{p_i}{2} \text{ when } \theta_i \in \sbra*{0, \frac{\pi}{4}}.
        \end{align}
    \end{subequations}
    Let $x_{i0}$ and $x_{i1}$ be the measurement outcomes of the first two qubits of $G_i (H \otimes U_i) \ket{0}\ket{0}\ket{\bar 0}$ for $i\in\binset$.
    \Cref{algo:coCertQSDlog-in-RQL} returns ``yes'' if $x_{00} = x_{01} = x_{10} = x_{11} = 0$, and ``no'' otherwise.
    Let $U_{P_{d'}^{\sign}\left(\frac{\rho_0 - \rho_1}{2}\right)}$ be the unitary operator being controlled in the implementation of $U_i \coloneqq \calT(Q_i, U_{\frac{\rho_0-\rho_1}{2}}, P_{d'}^{\sign}, \epsilon)$, and note that by \Cref{corr:sign-polynomial-implementation}, $U_{P_{d'}^{\sign}\left(\frac{\rho_0 - \rho_1}{2}\right)}$ is a $(1, O(s), (36 \hat{C}_{\sign} \log{d'} + 37)\epsilon)$-block-encoding of $P_{d'}^{\sign}\left(\frac{\rho_0-\rho_1}{2}\right)$.
    We will show the correctness of our algorithm as follows: 

    \begin{itemize}[leftmargin=2em]
        \item For \textit{yes} instances ($\rho_0=\rho_1$), $U_{P_{d'}^{\sign}\left(\frac{\rho_0 - \rho_1}{2}\right)}$ is a $(1,O(s), 0)$-block-encoding of the zero operator, following from \Cref{remark:QSVT-parity-preserving}. Consequently, $\calT(Q_i, U_{\frac{\rho_0-\rho_1}{2}}, P_{d'}^{\sign}, \epsilon)$ outputs $0$ with probability $1/2$ for $i \in \binset$, i.e., $p_0=p_1=1/2$. As a result, we have $\theta_0 = \theta_1 = \pi/6$ and $\sin^2(3\theta_0) = \sin^2(3\theta_1) = 1$. Substituting these values into \Cref{eq:exact-AA-coCertQSDlog}, we can conclude that $x_{00} = x_{01} = x_{10} = x_{11} = 0$ with certainty, which completes the analysis.

        \item For \textit{no} instances ($\td(\rho_0, \rho_1) \geq \alpha$), $U_{P_{d'}^{\sign}\left(\frac{\rho_0 - \rho_1}{2}\right)}$ is a $(1,O(s), 0)$-block-encoding of $A$ satisfying 
        \[\left\| A - P_{d'}^{\sign}\left(\frac{\rho_0-\rho_1}{2}\right) \right\| \leq (36 \hat{C}_{\sign} \log{d'} +37)\epsilon.\] 
        Let $p_i$ be the probability that $\calT(Q_i, U_{\frac{\rho_0-\rho_1}{2}}, P_{d'}^{\sign}, \epsilon)$ outputs $0$ for $i \in \binset$, then $p_i = \frac{1}{2}\big(1+\operatorname{Re}(\Tr(\rho_i A))\big)$ following from \Cref{lemma:space-efficient-quantum-tester}. 
        A direct calculation similar to \Cref{prop:td-technical} indicates that: 
        \[ \left| (p_0 - p_1) - \td(\rho_0, \rho_1) \right| \leq (36 \hat{C}_{\sign}\log{d'}+2 C_{\sign}+37)\epsilon + 2^{r+1} \delta. \]
        Under the choice of $\delta$ and $\epsilon$ (the same as in the proof of \Cref{thm:GapQSDlog-in-BQL}), we obtain that $|(p_0 - p_1) - \td(\rho_0, \rho_1)| \leq \varepsilon$ which yields that $\max\{ |p_0-1/2|, |p_1-1/2| \} \geq \varepsilon/2$.\footnote{This inequality is because $|p_0 - p_1| \geq \td(\rho_0, \rho_1) - \varepsilon \geq 2\varepsilon - \varepsilon = \varepsilon$.} 
        
        Noting that $\Pr{x_{i0}=x_{i1}=0}=\sin^2(3\theta_i)$ for $i \in \binset$, \Cref{algo:coCertQSDlog-in-RQL} will return ``yes'' with probability $p_{\yes} = \sin^2(3\theta_0) \sin^2(3\theta_1)$. We provide an upper bound for $p_{\yes}$ in \Cref{prop:coCertQSDlog-soundness-upper-bound}, with the proof provided immediately afterward: 
        \begin{proposition}
        \label{prop:coCertQSDlog-soundness-upper-bound}
        Let $f(\theta_0,\theta_1) \coloneqq \sin^2(3\theta_0) \sin^2(3\theta_1)$ be a function such that $\sin^2(\theta_i) = p_i/2$ for $i\in\binset$ and $\max\{ |p_0-1/2|, |p_1-1/2| \} \geq \varepsilon/2$, then $f(\theta_0,\theta_1) \leq 1 - \varepsilon^2/4$. 
        \end{proposition}

        Consequently, we finish the analysis by noticing $p_{\yes} = f(\theta_0,\theta_1) \leq 1 - \varepsilon^2/4 = 1 - \alpha^2/16$.
    \end{itemize}

    Now we analyze the complexity of \Cref{algo:coCertQSDlog-in-RQL}. 
    Following \Cref{lemma:space-efficient-quantum-tester}, we can compute $x_{00}, x_{01}, x_{10}, x_{11}$ in \BQL{}.
    The quantum circuit that computes $x_{00}, x_{01}, x_{10}, x_{11}$ takes $O(d^2 \log{d}) = \tilde O(2^{2r}/\alpha^2)$ queries to $Q_0$ and $Q_1$, and its circuit description can be computed in deterministic time $\tilde O(d^{9/2}/\alpha) = \tilde O(2^{4.5r}/\alpha^{5.5})$. 
    Finally, we conclude the \coRQUL{} containment of \coCertQSDlog{} by applying error reduction for \coRQUL{} (\Cref{corr:error-reduction-untary-quantum-logspace}) to \Cref{algo:coCertQSDlog-in-RQL}. 
\end{proof}

\begin{proof}[Proof of \Cref{prop:coCertQSDlog-soundness-upper-bound}]
    We begin by stating the facts that $\sin^2(\theta_i) = p_i/2$ for $i \in \binset$ and $\sin^2(3\theta) = \sin^6(\theta) - 6\cos^2(\theta) \sin^4(\theta) + 9\cos^4(\theta)\sin^2(\theta)$. 
    Then we notice that $0 \leq p_0,p_1 \leq 1$ and complete the proof by a direct calculation:    
    \begin{align*}
        f(\theta_0,\theta_1)
        & = \left( 2p_0^3 - 6p_0^2 + \tfrac{9}{2} p_0 \right) \left( 2p_1^3 - 6p_1^2 + \tfrac{9}{2} p_1 \right) \\
        & \leq \left( 1 - \left( p_0 - \tfrac{1}{2} \right)^2 \right) \left( 1 - \left( p_1 - \tfrac{1}{2} \right)^2 \right) \\
        & \leq 1 - \left( \max\left\{\left|p_0 - \tfrac{1}{2}\right|, \left|p_1 - \tfrac{1}{2}\right|\right\} \right)^2 \\
        & \leq 1 - \tfrac{\varepsilon^2}{4}. \qedhere 
    \end{align*} 
\end{proof}

\subsubsection{\coCertQHSlog{} is in \coRQUL{}}

As a warm-up, we observe that $\HS(\rho_0,\rho_1)$ can be written as a summation of $\frac{1}{2}\Tr(\rho_0^2)$, $\frac{1}{2}\Tr(\rho_1^2)$, and $-\Tr(\rho_0\rho_1)$. It follows that \GapQHSlog{} is contained in \BQL{}: 

\begin{theorem}
    \label{thm:GapQHSlog-in-BQL}
    For any functions $\alpha(n)$ and $\beta(n)$ that can be computed in deterministic logspace and satisfy $\alpha(n)-\beta(n) \geq 1/\poly(n)$, we have that $\GapQHS_{\log}[\alpha(n),\beta(n)]$ is in \BQL{}. 
\end{theorem}
\begin{proof}
    Note that $\HS(\rho_0,\rho_1) = \frac{1}{2} \left(\Tr(\rho_0^2)+\Tr(\rho_1^2)\right)-\Tr(\rho_0\rho_1)$. 
    Let $\varepsilon \coloneqq (\alpha - \beta)/100$. According to \Cref{lemma:swap-test}, we can use the SWAP test to estimate $\Tr(\rho_0^2)$, $\Tr(\rho_1^2)$, and $\Tr(\rho_0\rho_1)$, and hence $\HS(\rho_0,\rho_1)$, within additive error $\varepsilon$ with high probability by performing $O(1/\varepsilon^2)$ sequential repetitions. Therefore, we can conclude that $\GapQHS_{\log}[\alpha(n),\beta(n)]$ is in \BQL{}.
\end{proof}

The observing algorithm in \Cref{thm:GapQHSlog-in-BQL} corresponds to a single quantum circuit that uses \textit{two random coins}, implementable by single-qubit measurements in the computational basis, and whose acceptance probability is proportional to $\HS(\rho_0,\rho_1)$. 
However, to guarantee \textit{unitarity}, we construct an alternative algorithm using the LCU technique, which serves as the unitary of interest with the desired property. 

\begin{theorem}
\label{thm:coCertQHSlog-is-in-coRQUL}
    For any deterministic logspace computable function $\alpha(n) \geq 1/\poly(n)$, we have that $\coCertQHSlog[\alpha(n)]$ is in \coRQUL{}.
\end{theorem}

\begin{proof}
    We first provide a formal algorithm in \Cref{algo:CertQHS-in-RQL}.
    
    \begin{algorithm}[!htp]
		\caption{Space-efficient algorithm for \coCertQHSlog{}.}
		\label{algo:CertQHS-in-RQL}
        \SetKwInOut{Input}{Input}
        \SetKwInOut{Output}{Output}
        \SetKwInOut{Parameter}{Params}
        \Input{Quantum circuits $Q_i$ that prepare the purification of $\rho_i$ for $i \in \binset$.}
        \Output{Return ``yes'' if $\rho_0=\rho_1$, and ``no'' otherwise.}
        1. Construct subroutines $T_{ij}\coloneqq\SWAP(\rho_i,\rho_j)$ for $(i,j)\in\{(0,0),(1,1),(0,1)\}$, which output $0$ with probability $p_{ij}$. In particular, the subroutine $\SWAP(\rho_i,\rho_j)$ involves applying $Q_i$ and $Q_j$ to prepare quantum states $\rho_i$ and $\rho_j$, respectively, and then employing the SWAP test (\Cref{lemma:swap-test}) on these states $\rho_i$ and $\rho_j$\;
        2. Construct a block-encoding of $\smash{\varrho\big(\frac 1 2 + \frac{\HS(\rho_0, \rho_1)}{4}\big)}$ where $\varrho(p) \coloneqq p \ket{0}\bra{0} + (1-p)\ket{1}\bra{1}$, denoted by $U$, using $O(1)$ queries to $T_{00}$, $T_{11}$, and $T_{01}$ by \Cref{lemma:space-efficient-LCU}\;
        3. Let $G \coloneqq - U (I - 2\ket{\bar 0}\bra{\bar 0}) U^\dag (I - 2\ket{\bar 0}\bra{\bar 0})$\;
        4. Measure all qubits of $G U \ket{\bar 0}$ in the computational basis. Return ``yes'' if the measurement outcome is an all-zero string, and ``no'' otherwise. 
        \BlankLine
    \end{algorithm}

    \paragraph*{Constructing the unitary of interest via the SWAP test.}
    We consider the setting with $s(n)=\Theta(\log(n))$. 
    Our main building block is the circuit implementation of the SWAP test (\Cref{lemma:swap-test}). Specifically, we utilize the subroutine $\SWAP(\rho_i,\rho_j)$ for $i,j\in \binset$, which involves applying $Q_i$ and $Q_j$ to prepare quantum states $\rho_i$ and $\rho_j$, respectively, and then employing the SWAP test on these states $\rho_i$ and $\rho_j$. We denote by $p_{ij}$ the probability that $\SWAP(\rho_i,\rho_j)$ outputs $0$ based on the measurement outcome of the control qubit in the SWAP test. 
    Following \Cref{lemma:swap-test}, we have $p_{ij} = \frac{1}{2}\big(1+\Tr(\rho_i\rho_j) \big)$ for $i,j \in \binset$. 

    We define $T_{ij}\coloneqq\SWAP(\rho_i,\rho_j)$ for $(i,j) \in \calI\coloneqq\{(0,0), (1,1), (0,1)\}$, with the control qubit in $\SWAP(\rho_i,\rho_j)$ serving as the output qubit of $T_{ij}$. By introducing another ancillary qubit, we construct $T'_{ij} \coloneqq \CNOT (I \otimes T_{ij})$ for $(i,j) \in \calI$, where $\CNOT$ is controlled by the output qubit of $T_{ij}$ and targets on the new ancillary qubit. It is effortless to see that $T'_{ij}$ prepares the purification of $\varrho(p_{ij})$ with $\varrho(p_{ij})\coloneqq p_{ij} \ket{0}\bra{0} + (1-p_{ij})\ket{1}\bra{1}$ for $(i,j)\in\calI$. 

    By applying \Cref{lemma:purified-density-matrix}, we can construct quantum circuits $T''_{ij}$ for $(i,j)\in\calI$ that serve as $(1, O(s), 0)$-block-encoding of $\varrho(p_{ij})$, using $O(1)$ queries to $T'_{ij}$ and $O(1)$ one- and two-qubit quantum gates.
    Notably, $(X \otimes I) T''_{01}$, with $X$ acting on the qubit of $\varrho(p_{01})$, prepares the purification of $X \varrho(p_{01}) X^\dag = p_{01}\ket{1}\bra{1} + (1-p_{01})\ket{0}\bra{0} = \varrho(1-p_{01})$, leading to the equality: 
    \[\varrho(\rho_0,\rho_1)\coloneqq \frac{1}{4} \varrho(p_{00}) + \frac{1}{4} \varrho(p_{11}) + \frac{1}{2} \varrho(1-p_{01}) = \varrho\!\left(\frac 1 2 + \frac{\HS(\rho_0, \rho_1)}{4}\right).\]
    
    Consequently, we employ \Cref{lemma:space-efficient-LCU} to construct a unitary quantum circuit $U$ that is a $(1,m, 0)$-block-encoding of $\varrho\big(\frac{1}{2} + \frac{\HS(\rho_0, \rho_1)}{4}\big)$ using $O(1)$ queries to $T''_{00}$, $T''_{11}$, $T''_{01}$, and $O(1)$ one- and two-qubit quantum gates, where $m \coloneqq O(s)$. The construction ensures the following:  
    \begin{equation}
        \label{eq:coCertQHSlog-unitary}
        U \ket{0} \ket{0}^{\otimes m} = \underbrace{\left( \frac 1 2 + \frac{\HS(\rho_0, \rho_1)}{4} \right)}_{\sin(\theta)} \ket{0}\ket{0}^{\otimes m} + \cos(\theta) \ket{\perp}, \text{ where } \bra{0}\bra{0}^{\otimes m} \ket{\perp} = 0.
    \end{equation}
    
    \paragraph*{Making the error one-sided.} Let us consider the Grover operator $G \coloneqq - U (I - 2\ket{\bar 0}\bra{\bar 0}) U^\dag (I - 2\ket{\bar 0}\bra{\bar 0})$. 
    By applying \Cref{lemma:exact-amplitude-amplification}, we derive that 
    \[G U \ket{0}\ket{0}^{\otimes m} = \sin(3\theta) \ket{0}\ket{0}^{\otimes m} + \cos(3\theta) \ket{\perp}.\] 
    Subsequently, we measure all qubits of $G U \ket{0}\ket{0}^{\otimes m}$ in the computational basis, represented as $x\in\binset^{m+1}$. 
    Hence, \Cref{algo:CertQHS-in-RQL} returns ``yes'' if the outcome $x$ is $0^{m+1}$ and ``no'' otherwise. \Cref{algo:CertQHS-in-RQL} accepts with probability $\sin^2(3\theta)$. 
    Now we analyze the correctness of the algorithm: 
    \begin{itemize}[leftmargin=2em]
        \item For \textit{yes} instances ($\rho_0 = \rho_1$), we have $\HS(\rho_0, \rho_1) = 0$. Following \Cref{eq:coCertQHSlog-unitary}, we obtain $\sin(\theta) = 1/2$ and thus $\sin^2(3\theta) = 1$. We conclude that \Cref{algo:CertQHS-in-RQL} will always return ``yes''. 
        \item For \textit{no} instances, we have $\HS(\rho_0, \rho_1) \geq \alpha$. 
        According to \Cref{eq:coCertQHSlog-unitary}, we derive that: 
        \begin{subequations}
            \label{eq:coCertQHSlog-in-coRQUL-soundness}
            \begin{align}
            &\sin(\theta) = \frac{1}{2} + \frac{\HS(\rho_0,\rho_1)}{4} \geq \frac{1}{2} + \frac{\alpha}{4},\\ 
            \frac{1}{4} \leq &\sin^2(\theta) = \rbra*{ \frac{1}{2} + \frac{\HS(\rho_0,\rho_1)}{4} }^2 \leq \rbra*{ \frac{1}{2} + \frac{1}{4} }^2 = \frac{9}{16}.
            \end{align}
        \end{subequations}
        As a result, considering the fact that $\sin^2(3\theta) = f(\sin^2(\theta))$ where $f(x) \coloneqq 16x^3 - 24x^2 + 9x$, % 9x - 24x^2 + 16x^3
        we require \Cref{prop:coCertQHSsoundness} and the proof is given immediately afterward:
        \begin{proposition}
        \label{prop:coCertQHSsoundness}
            The polynomial function $f(x)\coloneqq 16x^3 - 24x^2 + 9x$ is monotonically decreasing in $[1/4,9/16]$. Moreover, we have 
            \[\forall \alpha\in[0,1],\quad f\rbra*{\rbra[\Big]{\frac{1}{2}+\frac{\alpha}{4}}^2} \leq 1-\frac{\alpha^2}{2}.\]   
        \end{proposition}
        Combining \Cref{eq:coCertQHSlog-in-coRQUL-soundness} and \Cref{prop:coCertQHSsoundness}, we have that $\sin^2(3\theta) = f(\sin^2(\theta)) \leq f\big(\big(\frac{1}{2}+\frac{\alpha}{4}\big)^2\big) \leq 1-\frac{\alpha^2}{2}$. Hence, \Cref{algo:CertQHS-in-RQL} will return ``no'' with probability at least $\alpha^2/2$.
    \end{itemize}
    
    Regarding the computational complexity of \Cref{algo:CertQHS-in-RQL}, this algorithm requires $O(s(n))$ qubits and performs $O(1)$ queries to $Q_0$ and $Q_1$. 
    Finally, we finish the proof by applying error reduction from \coRQUL{} (\Cref{corr:error-reduction-untary-quantum-logspace}) to \Cref{algo:CertQHS-in-RQL}.
\end{proof}

\begin{proof}[Proof of \Cref{prop:coCertQHSsoundness}]
    Through a direct calculation, we have  $f'(x) = 48x^2 - 48x + 9 \leq 0$ for $x \in [1/4, 3/4]$, then $f(x)$ is monotonically decreasing in $[1/4, 9/16] \subseteq [1/4, 3/4]$. Moreover, it is left to show that: 
    \[
        f\rbra*{ \rbra[\Big]{ \frac{1}{2} + \frac{\alpha}{4} }^2 } = \frac{\alpha^6}{256}+\frac{3 \alpha^5}{64}+\frac{9 \alpha^4}{64}-\frac{\alpha^3}{8}-\frac{3 \alpha^2}{4}+1 \leq 1-\frac{\alpha^2}{2}. 
    \]
    Equivalently, it suffices to show that $g(x) \coloneqq -\frac{x^4}{256}-\frac{3 x^3}{64}-\frac{9 x^2}{64}+\frac{x}{8}+\frac{1}{4} \geq 0$ for $0 \leq x \leq 1$.
    We first compute the first derivative of $g(x)$, which is $g'(x)=-\frac{x^3}{64}-\frac{9 x^2}{64}-\frac{9 x}{32}+\frac{1}{8}$. Setting $g'(x)$ equal to zero, we obtain three roots: $x_1=-4$, $x_2=\frac{1}{2} (-\sqrt{33}-5) < 0$, and $x_3=\frac{1}{2} (\sqrt{33}-5) \in (0,1)$.

    Since $g'(0)=1/8 > 0$ and $g'(1)=-5/16 < 0$, we conclude that $g(x)$ is monotonically increasing in $[0,x_3]$ and monotonically decreasing in $[x_3,1]$. Therefore, we can determine the minimum value of $g(x)$ by evaluating $g(0)=\frac{1}{4}$ and $g(1)=\frac{47}{256}$. Since both values are greater than zero, we conclude that $\min\{g(0),g(1)\} = \min\big\{\frac{1}{4},\frac{47}{256}\big\} > 0$, as desired.
\end{proof}

\subsection{\BQL{}- and \coRQUL{}-hardness for space-bounded state testing problems}
\label{subsec:BQL-coRQL-hardness}

We will prove that space-bounded state testing problems mentioned in \Cref{thm:space-bounded-quantum-state-testing-BQL-complete} are \BQUL{}-hard, which implies their \BQL{}-hardness since \BQL{}=\BQUL{}~\cite{FR21}. Similarly, all space-bounded state certification problems mentioned in \Cref{thm:space-bounded-quantum-state-certification-RQL-complete} are \coRQUL{}-hard.

\subsubsection{Hardness results for \GapQSDlog{}, \GapQHSlog{}, and their certification version}
Employing similar constructions, we can establish the \BQUL{}-hardness of both \GapQSDlog{} and \GapQHSlog{}. The former involves a single-qubit pure state and a single-qubit mixed state, while the latter involves two pure states.

\begin{lemma}[\GapQSDlog{} is \BQUL{}-hard]
    \label{lemma:GapQSDlog-BQLhard}
    For any deterministic logspace computable functions $a(n)$ and $b(n)$ such that $1-\sqrt{1-a(n)}-\sqrt{b(n)} \geq 1/\poly(n)$, we have that 
    \[\GapQSDlog[1-\sqrt{1-a(n)},\sqrt{b(n)}] \text{ is } \BQUL[a(n),b(n)]\text{-hard}.\]
    Furthermore, there exists some polynomial $p(n)$ such that $\GapQSDlog[\alpha(n),\beta(n)]$ is \BQUL{}-hard for all $\alpha(n) \leq 1-1/p(n)$ and $\beta(n) \geq 1/p(n)$.
\end{lemma}

\begin{proof}
    Consider a promise problem $(\calL_{yes},\calL_{no})\in \BQUL[a(n),b(n)]$, then we know that the acceptance probability $\Pr{C_x \text{ accepts}} \geq a(n)$ if $x \in \calL_{yes}$, whereas $\Pr{C_x \text{ accepts}} \leq b(n)$ if $x \in \calL_{no}$. 
    Now we notice that the acceptance probability can be expressed in terms of the fidelity between a single-qubit pure state $\rho_0$ and a single-qubit mixed state $\rho_1$ that is prepared by two logarithmic-qubit quantum circuits $Q_0$ and $Q_1$, respectively:
    \begin{subequations}
    \label{eq:GapQSDlog-BQL-hard}
    \begin{align}
    \Pr{C_x \text{ accepts}} =& \left\|\ket{1}\bra{1}_{\Out}C_x\ket{\bar{0}}\right\|^2_2\\
    =&\Tr\rbra*{ \ket{1}\bra{1}_{\Out}\Tr_{\overline{\Out}}\rbra*{ C_x \ket{\bar{0}}\bra{\bar{0}} C_x^{\dagger} } }\\
    =& 1 - \Tr\rbra*{ \ket{0}\bra{0}_{\Out}\Tr_{\overline{\Out}}\rbra*{ C_x \ket{\bar{0}}\bra{\bar{0}} C_x^{\dagger} } }\\
    =& 1 - \F^2(\rho_0,\rho_1)
    \end{align}
    \end{subequations}
    In particular, the corresponding $Q_0$ is simply the identity, while $Q_1$ is exactly the circuit $C_x$. We then prepare $\rho_0$ and $\rho_1$, specifically $\rho_0 \coloneqq \ketbra{0}{0}_{\Out}$ and $\rho_1 \coloneqq \Tr_{\overline{\Out}}\rbra[\big]{C_x \ket{\bar{0}}\bra{\bar{0}} C_x^{\dagger}}$, by tracing out all non-output qubits. 
    By utilizing \Cref{lemma:traceDist-vs-fidelity}, we have derived that: 
    \begin{itemize}[itemsep=0.33em,topsep=0.33em,parsep=0.33em]
        \item For \textit{yes} instances, $\Pr{C_x \text{ accepts}} = 1-\F^2(\rho_0,\rho_1) \geq a(n)$ deduces that 
        \[\td(\rho_0,\rho_1) \geq 1-\F(\rho_0,\rho_1) \geq 1 - \sqrt{1-a(n)};\]
        \item For \textit{no} instances, $\Pr{C_x \text{ accepts}} = 1-\F^2(\rho_0,\rho_1) \leq b(n)$ yields that 
        \[\td(\rho_0,\rho_1) \leq \sqrt{1-\F^2(\rho_0,\rho_1)} \leq \sqrt{b(n)}.\]
    \end{itemize}
    Therefore, we prove that $\GapQSD_{\log}[1-\sqrt{1-a(n)},\sqrt{b(n)}]$ is $\BQUL[a(n),b(n)]$-hard, and we complete the argument by combining this with error reduction for \BQUL{} (\Cref{corr:error-reduction-untary-quantum-logspace}).
\end{proof}

To construct pure states, adapted from the construction in \Cref{lemma:GapQSDlog-BQLhard}, we replace the final measurement in the \BQL{} circuit $C_x$ with a quantum gate (\CNOT{}) and design a new algorithm based on $C_x$ with the final measurement on \textit{all} qubits in the computational basis.

\begin{lemma}[\GapQHSlog{} is \BQUL{}-hard]
    \label{lemma:GapQHSlog-BQLhard}
    For any deterministic logspace computable functions $a(n)$ and $b(n)$ such that $a(n)-b(n) \geq 1/\poly(n)$, we have that 
    \[\GapQHSlog\sbra*{2a(n)-a^2(n), 2b(n)-b^2(n)} \text{ is } \BQUL[a(n),b(n)]\text{-hard}.\]
    Furthermore, there exists some polynomial $p(n)$ such that $\GapQHSlog[\alpha(n),\beta(n)]$ is \BQUL{}-hard for all $\alpha(n) \leq 1-1/p(n)$ and $\beta(n) \geq 1/p(n)$.
\end{lemma}

\begin{proof}
    For any promise problem $(\calL_{yes},\calL_{no})\in \BQUL[a(n),b(n)]$, we have that the acceptance probability $\Pr{C_x \text{ accepts}} \geq a(n)$ if $x \in \calL_{yes}$, whereas $\Pr{C_x \text{ accepts}} \leq b(n)$ if $x \in \calL_{no}$.
    For convenience, let the output qubit be the register $\sfO$. 
    Now we construct a new quantum circuit $C'_x$ with an additional ancillary qubit on the register $\sfF$ initialized to zero: 
    \begin{equation}
        \label{eq:GapQHSlog-hardness}
        C'_x \coloneqq C^{\dagger}_x \CNOT_{\sfO\rightarrow \sfF} C_x.
    \end{equation}
    And we say that $C'_x$ accepts if the measurement outcome of all qubits (namely the working qubit of $C_x$ and $\sfF$) are all zero. Through a direct calculation, we obtain: 
    \begin{subequations}
        \label{eq:GapQHSlog-BQL-hard-pacc}
        \begin{align}
        \Pr{C'_x \text{ accepts}} &= \big\| (\ket{\bar{0}}\bra{\bar{0}}\otimes \ket{0}\bra{0}_{\sfF}) C^{\dagger}_x \CNOT_{\sfO\rightarrow \sfF} C_x (\ket{\bar{0}}\otimes \ket{0}_{\sfF})\big\|_2^2 \\
        &= \big| \bra{\bar{0}} C_x^{\dagger} \ket{0}\bra{0}_{\sfO} C_x \ket{\bar{0}} \big|^2\\
        &= \rbra*{1- \bra{\bar{0}} C_x^{\dagger} \ket{1}\bra{1}_{\sfO} C_x \ket{\bar{0}}}^2\\
        &= \rbra*{1-{\rm Pr}\sbra*{ C_x \text{ accepts} }}^2.
        \end{align}
    \end{subequations}
    Here, the second line owes to $\CNOT_{\sfO \rightarrow \sfF} = \ket{0}\bra{0}_{\sfO}\otimes I_{\sfF} + \ket{1}\bra{1}_{\sfO} \otimes X_{\sfF}$, and the last line is because of \Cref{eq:GapQSDlog-BQL-hard}. Interestingly, by defining two pure states $\ketbra{\psi_0}{\psi_0} \coloneqq \ket{\bar{0}}\bra{\bar{0}}\otimes \ket{0}\bra{0}_{\sfF}$ and $\ketbra{\psi_1}{\psi_1} \coloneqq C'_x (\ket{\bar{0}}\bra{\bar{0}}\otimes \ket{0}\bra{0}_{\sfF}) {C'}^{\dagger}_x$ corresponding to $Q_0 = I$ and $Q_1=C'_x$, respectively, we deduce the following from \Cref{eq:GapQHSlog-BQL-hard-pacc}: 
    \begin{equation}
        \label{eq:GapQHSlog-BQL-hard}
        \Pr{C'_x \text{ accepts}} = \abs{\innerprod{\psi_0}{\psi_1}}^2 = 1-\HS(\ketbra{\psi_0}{\psi_0},\ketbra{\psi_1}{\psi_1}).
    \end{equation}
    Combining \Cref{eq:GapQHSlog-BQL-hard-pacc,eq:GapQHSlog-BQL-hard}, and writing $p_{\rm acc} \coloneqq {\rm Pr}\sbra*{ C_x\text{ accepts} }$, we obtain
    \[ \HS(\ketbra{\psi_0}{\psi_0},\ketbra{\psi_1}{\psi_1}) = 1-(1-p_{\rm acc})^2 = 2 p_{\rm acc} - p_{\rm acc}^2.\] 
    Since the map $p_{\rm acc} \mapsto 2p_{\rm acc}-p_{\rm acc}^2$ is monotonically non-decreasing on $[0,1]$, it follows that: 
    \begin{itemize}[itemsep=0.33em,topsep=0.33em,parsep=0.33em]
        \item For \textit{yes} instances, $\Pr{C_x \text{ accepts}} \geq a(n)$ implies that 
        \[\HS(\ketbra{\psi_0}{\psi_0},\ketbra{\psi_1}{\psi_1}) \geq 2a(n)-a(n)^2.\]
        \item For \textit{no} instances, $\Pr{C_x \text{ accepts}} \leq b(n)$ yields that 
        \[\HS(\ketbra{\psi_0}{\psi_0},\ketbra{\psi_1}{\psi_1}) \leq 2b(n)-b(n)^2.\]
    \end{itemize}
    We thus establish that $\GapQHSlog[2a(n)-a^2(n), 2b(n)-b^2(n)]$ is $\BQUL[a(n),b(n)]$-hard, and we complete the argument by combining this with error reduction for \BQUL{} (\Cref{corr:error-reduction-untary-quantum-logspace}).
\end{proof}

Our constructions in the proofs of \Cref{lemma:GapQSDlog-BQLhard,lemma:GapQHSlog-BQLhard} are, respectively, somewhat analogous to~\cite[Theorems 12 and 13]{ARS+21}. 
We then present the \coRQUL{}-hardness results, which are essentially adapted from \Cref{lemma:GapQSDlog-BQLhard,lemma:GapQHSlog-BQLhard}: 

\begin{lemma}[\coCertQSDlog{} and \coCertQHSlog{} are \coRQUL{}-hard]
    \label{lemma:coCertQSDlog-coCertQHSlog-coRQL-hard}
    For any function $\gamma(n)$ that is computable in deterministic logspace and satisfies $\gamma(n) \geq 1/\poly(n)$, the following holds for some polynomial $p(n)$ that can also be computed in deterministic logspace: 
    \begin{enumerate}[label={\upshape(\roman*)}]
        \item $\coCertQSDlog[\gamma(n)]$ is \coRQUL{}-hard for $\gamma(n) \leq 1-1/p(n)$; 
        \item $\coCertQHSlog[\gamma(n)]$ is \coRQUL{}-hard for $\gamma(n) \leq 1-1/p(n)$.  
    \end{enumerate}
\end{lemma}

\begin{proof}
    To prove the first item, we follow the construction in \Cref{lemma:GapQSDlog-BQLhard} and replace $Q_0$ with a circuit that flips the output qubit. Consequently, the corresponding quantum states are 
    \[\widehat{\rho}_0 = \ketbra{1}{1}_{\Out} \quad \text{and} \quad \widehat{\rho}_1 = \Tr_{\overline{\Out}}\rbra[\big]{C_x \ket{\bar{0}}\bra{\bar{0}} C_x^{\dagger}}.\] 
    Therefore, we obtain $\Pr{C_x \text{ accepts}} = \F^2(\widehat{\rho}_0,\widehat{\rho}_1)$, which implies $\td(\widehat{\rho}_0,\widehat{\rho}_1)=0$ (due to perfect completeness) for \textit{yes} instances, while $\td(\widehat{\rho}_0,\widehat{\rho}_1) \geq 1-\F(\widehat{\rho}_0,\widehat{\rho}_1) \geq 1-\sqrt{b(n)}$ for \textit{no} instances. We then complete the proof by combining this with error reduction for \coRQUL{} (\Cref{corr:error-reduction-untary-quantum-logspace}). 

    For the second item, we proceed as in the construction of \Cref{lemma:GapQHSlog-BQLhard} and replace $C'_x$ with 
    \[\widetilde{C}'_x \coloneqq C^{\dagger}_x X^{\dagger}_{\sfO} \CNOT_{\sfO\rightarrow \sfF} X_{\sfO} C_x.\] 
    The corresponding pure states are then
    \[\ketbra{\widetilde{\psi}_0}{\widetilde{\psi}_0}\coloneqq\ket{\bar{0}}\bra{\bar{0}}\otimes \ket{0}\bra{0}_{\sfF} \quad\text{and}\quad \ketbra{\widetilde{\psi}_1}{\widetilde{\psi}_1}\coloneqq \widetilde{C}'_x (\ket{\bar{0}}\bra{\bar{0}}\otimes \ket{0}\bra{0}_{\sfF}) \rbra[\big]{\widetilde{C}'_x}^{\dagger}.\] 
    It follows that $\mathrm{Pr}^2[C_x\text{ accepts}] = 1-\HS\rbra*{\ketbra{\widetilde{\psi}_0}{\widetilde{\psi}_0}, \ketbra{\widetilde{\psi}_1}{\widetilde{\psi}_1}}$, which implies 
    \[\HS\rbra*{\ketbra{\widetilde{\psi}_0}{\widetilde{\psi}_0}, \ketbra{\widetilde{\psi}_1}{\widetilde{\psi}_1}} = 0\] 
    for \textit{yes} instances, while $\HS\rbra*{\ketbra{\widetilde{\psi}_0}{\widetilde{\psi}_0}, \ketbra{\widetilde{\psi}_1}{\widetilde{\psi}_1}} \geq 1-b^2(n)$ for \textit{no} instances. We thus similarly conclude the proof by integrating this with \Cref{corr:error-reduction-untary-quantum-logspace}. 
\end{proof}

\subsubsection{Hardness results for \GapQJSlog{} and \GapQEDlog{}}

We demonstrate the \BQUL{}-hardness of \GapQJSlog{} by reducing \GapQSDlog{} to \GapQJSlog{}, following a similar approach to that in~\cite[Lemma 4.15]{Liu23}.\footnote{Lemma 4.13 in the last arXiv version of~\cite{Liu23}.}
\begin{lemma}[\GapQJSlog{} is \BQUL{}-hard]
    \label{lemma:GapQJSlog-BQL-hard}
    For any functions $\alpha(n)$ and $\beta(n)$ that are computable in deterministic logspace, we have 
    \[ \GapQJS_{\log}[\alpha(n),\beta(n)] \text{ is } \BQUL{}\text{-hard} \]
    for $\alpha(n) \leq 1-\sqrt{2}/\sqrt{p(n)}$ and $\beta(n) \geq 1/p(n)$, where $p(n)$ is some deterministic logspace computable polynomial. 
\end{lemma}

\begin{proof}
    By \Cref{lemma:GapQSDlog-BQLhard}, it suffices to reduce $\GapQSD_{\log}[1-1/p(n),1/p(n)]$ to $\GapQJS_{\log}[\alpha(n),\beta(n)]$. 
    Consider logarithmic-qubit quantum circuits $Q_0$ and $Q_1$, which is a $\GapQSD_{\log}$ instance. We can obtain $\rho_k$ for $k\in\binset$ by performing $Q_k$ on $\ket{0^n}$ and tracing out the non-output qubits. We then have the following: 
    \begin{itemize}
        \item If $\td(\rho_0,\rho_1) \geq 1-1/p(n)$, then \Cref{lemma:QJS-vs-traceDist} yields that
        \[\QJS_2(\rho_0,\rho_1) \geq 1-\binH\left(\tfrac{1-\td(\rho_0,\rho_1)}{2}\right) \geq 1-\binH\left(\tfrac{1}{2p(n)}\right) \geq 1-\tfrac{\sqrt{2}}{\sqrt{p(n)}} \geq \alpha(n),\]
        where the third inequality owing to $\binH(x)\leq 2\sqrt{x}$ for all $x\in[0,1]$. 
        
        \item If $\td(\rho_0,\rho_1) \leq 1/p(n)$, then \Cref{lemma:QJS-vs-traceDist} indicates that
        \[\QJS_2(\rho_0,\rho_1) \leq \td(\rho_0,\rho_1) \leq \tfrac{1}{p(n)} \leq \beta(n).\]
    \end{itemize}

    Therefore, we can utilize the same quantum circuits $Q_0$ and $Q_1$, along with their corresponding quantum states $\rho_0$ and $\rho_1$, respectively, to establish a logspace Karp reduction from $\GapQSD_{\log}[1-1/p(n),1/p(n)]$ to $\GapQJS_{\log}[\alpha(n),\beta(n)]$, as required.
\end{proof}

By combining the reduction from \GapQSDlog{} to \GapQJSlog{} (\Cref{lemma:GapQJSlog-BQL-hard}) and the reduction from \GapQJSlog{} to \GapQEDlog{} (\Cref{corr:GapQJSlog-in-BQL}), we obtain the \BQUL{}-hardness for \GapQEDlog{} through reducing \GapQSDlog{} to \GapQEDlog{}. This proof resembles the approach outlined in \cite[Corollary 4.3]{Liu23}. 
\begin{corollary}[\GapQEDlog{} is \BQUL{}-hard]
    \label{corr:GapQEDlog-BQL-hard}
    For any function $g(n)$ that is computable in deterministic logspace, we have 
    $\GapQEDlog[g(n)]$ is \BQUL{}-hard for $g(n) \leq \frac{\ln{2}}{2}\big(1-\frac{\sqrt{2}}{\sqrt{p(n/3)}}-\frac{1}{p(n/3)}\big)$, where $p(n)$ is some polynomial that can be computed in deterministic logspace. 
\end{corollary}

\begin{proof}
    By combining \Cref{lemma:GapQSDlog-BQLhard,lemma:GapQJSlog-BQL-hard}, we establish that $\GapQJS_{\log}[\alpha(n),\beta(n)]$ is \BQUL{}-hard for $\alpha(n) \leq 1-\sqrt{2}/\sqrt{p(n)}$ and $\beta(n) \geq 1/p(n)$, where $p(n)$ is some deterministic logspace computable polynomial. The \GapQSDlog{}-hard (and simultaneously \GapQJSlog{}-hard) instances, as specified in \Cref{lemma:GapQSDlog-BQLhard}, consist of $s(n)$-qubit quantum circuits $Q_0$ and $Q_1$ that prepare a purification of $r(n)$-qubit quantum (mixed) states $\rho_0$ and $\rho_1$, respectively, where $1 \leq r(n) \leq s(n) = \Theta(\log{n})$.
    
    Subsequently, by employing \Cref{corr:GapQJSlog-in-BQL}, we construct $(s+3)$-qubit quantum circuits $Q'_0$ and $Q'_1$ that prepare a purification of $(r+1)$-qubit quantum states $\rho'_0=\big(p\ket{0}\bra{0}+(1-p)\ket{1}\bra{1}\big)\otimes (\frac{1}{2}\rho_0 + \frac{1}{2}\rho_1)$ satisfying $\binH(p)=1- \rbra*{ \alpha(n)+\beta(n) }/2$ and $\rho'_1=\frac{1}{2}\ket{0}\bra{0}\otimes \rho_0 + \frac{1}{2}\ket{1}\bra{1}\otimes \rho_1$, respectively. Following \Cref{corr:GapQJSlog-in-BQL}, $\GapQEDlog[g(n)]$ is \BQUL{}-hard as long as 
    \[g(n)=\tfrac{\ln{2}}{2} \big(\alpha(n/3)-\beta(n/3)\big) \leq \tfrac{\ln{2}}{2}\Big(1-\tfrac{\sqrt{2}}{\sqrt{p(n/3)}} - \tfrac{1}{p(n/3)}\Big).\]
    Therefore, $\GapQSD_{s}$ is logspace Karp reducible to $\GapQED_{s+1}$ by mapping $(Q_0,Q_1)$ to $(Q'_0,Q'_1)$. 
\end{proof}

%%%%%%%%%%%%%%%%%%%%%%%%%%%%%%%%%%%%%%%%%%%

\section*{Acknowledgments}
\noindent
An earlier version of this work was included in the second-named author's PhD thesis~\cite{Liu25}. 
We express our gratitude to anonymous reviewers for providing detailed suggestions on the space-efficient quantum singular value transformation, particularly improved norm bounds for the coefficient vector in \Cref{lemma:coefficient-vector-norm-bound} (and consequently \Cref{lemma:space-efficient-bounded-funcs}) by leveraging the smoothness property of functions, and for suggesting to add discussion on space-bounded distribution testing.  

This work was partially supported by MEXT Q-LEAP Grant No.~\mbox{JPMXS0120319794}. 
FLG was also supported by JSPS KAKENHI Grants Nos.~\mbox{JP19H04066}, \mbox{JP20H05966}, \mbox{JP20H00579}, \mbox{JP20H04139}, and \mbox{JP21H04879}.
YL was also supported by JST, the establishment of University fellowships towards the creation of science technology innovation, Grant No.~\mbox{JPMJFS2125}, and in part by funding from the Swiss State Secretariat for Education, Research and Innovation (SERI).
QW was also supported in part by a startup funding from Shanghai Jiao Tong University. 
In addition, ChatGPT was used only for proofreading the manuscript, including identifying possible calculation errors and suggesting changes, with all corresponding changes verified and made by the authors.
Circuit diagrams were drawn by the Quantikz package~\cite{Kay18}. 

%%%%%%%%%%%%%%%%%%%%%%%%%%%%%%%%%%%%%%%%%%%

% Reference
\bibliographystyle{alphaurlQ}
\bibliography{space-bounded-state-testing}

\newcommand{\etalchar}[1]{$^{#1}$}
\DeclareRobustCommand{\dutchPrefix}[2]{#2}\providecommand{\dutchPrefix}[2]{#2}\renewcommand{\dutchPrefix}[2]{#2}\newcommand{\prelimVersion}[2]{Preliminary version in \textit{\MakeUppercase{#1} #2}}\newcommand{\arXiv}[1]{\href{https://arxiv.org/abs/#1}{\texttt{arXiv:\allowbreak#1}}}
\begin{thebibliography}{{\dutchPrefix{Apeldoorn}{v}}AGG{\dutchPrefix{Wolf}{d}}W20}

\bibitem[ABIS19]{ABIS19}
Jayadev Acharya, Sourbh Bhadane, Piotr Indyk, and Ziteng Sun.
\newblock Estimating entropy of distributions in constant space.
\newblock {\em Proceedings of the 33rd International Conference on Neural Information Processing Systems}. 32. 2019.
\newblock URL: \url{https://dl.acm.org/doi/10.5555/3454287.3454751}. \href {https://arxiv.org/abs/1911.07976} {\nolinkurl{arXiv:1911.07976}}. Appearances:\!

\bibitem[AH91]{AH91}
William Aiello and Johan H{\aa}stad.
\newblock Statistical zero-knowledge languages can be recognized in two rounds.
\newblock {\em Journal of Computer and System Sciences}. 42(3):327--345. 1991.
\newblock \href {https://doi.org/10.1016/0022-0000(91)90006-Q} {\nolinkurl{doi:10.1016/0022-0000(91)90006-Q}}. \prelimVersion{FOCS}{1987}. Appearances:\!

\bibitem[AISW20]{AISW20}
Jayadev Acharya, Ibrahim Issa, Nirmal~V Shende, and Aaron~B Wagner.
\newblock Estimating quantum entropy.
\newblock {\em IEEE Journal on Selected Areas in Information Theory}. 1(2):454--468. 2020.
\newblock \href {https://doi.org/10.1109/JSAIT.2020.3015235} {\nolinkurl{doi:10.1109/JSAIT.2020.3015235}}. \href {https://arxiv.org/abs/1711.00814} {\nolinkurl{arXiv:1711.00814}}. Appearances:\!

\bibitem[AJL09]{AJL09}
Dorit Aharonov, Vaughan Jones, and Zeph Landau.
\newblock A polynomial quantum algorithm for approximating the {Jones} polynomial.
\newblock {\em Algorithmica}. 55(3):395--421. 2009.
\newblock \href {https://doi.org/10.1007/s00453-008-9168-0} {\nolinkurl{doi:10.1007/s00453-008-9168-0}}. \prelimVersion{STOC}{2006}. \href {https://arxiv.org/abs/quant-ph/0511096} {\nolinkurl{arXiv:quant-ph/0511096}}. Appearances:\!

\bibitem[AKN98]{AKN98}
Dorit Aharonov, Alexei Kitaev, and Noam Nisan.
\newblock Quantum circuits with mixed states.
\newblock In {\em Proceedings of the 30th Annual ACM Symposium on Theory of Computing}. pages 20--30. 1998.
\newblock \href {https://doi.org/10.1145/276698.276708} {\nolinkurl{doi:10.1145/276698.276708}}. \href {https://arxiv.org/abs/quant-ph/9806029} {\nolinkurl{arXiv:quant-ph/9806029}}. Appearances:\!

\bibitem[AMNW22]{AMNW22}
Maryam Aliakbarpour, Andrew McGregor, Jelani Nelson, and Erik Waingarten.
\newblock Estimation of entropy in constant space with improved sample complexity.
\newblock {\em Proceedings of the 36th International Conference on Neural Information Processing Systems}. 35:32474--32486. 2022.
\newblock URL: \url{https://dl.acm.org/doi/10.5555/3600270.3602623}. \href {https://arxiv.org/abs/2205.09804} {\nolinkurl{arXiv:2205.09804}}. Appearances:\!

\bibitem[{\dutchPrefix{Apeldoorn}{v}}AGG{\dutchPrefix{Wolf}{d}}W20]{vAGGdW17}
Joran {\dutchPrefix{Apeldoorn}{v}}an~Apeldoorn, Andr{\'a}s Gily{\'e}n, Sander Gribling, and Ronald {\dutchPrefix{Wolf}{d}}e~Wolf.
\newblock Quantum {SDP}-solvers: {Better} upper and lower bounds.
\newblock {\em Quantum}. 4:230. 2020.
\newblock \href {https://doi.org/10.22331/q-2020-02-14-230} {\nolinkurl{doi:10.22331/q-2020-02-14-230}}. \prelimVersion{FOCS}{2017}. \href {https://arxiv.org/abs/1705.01843} {\nolinkurl{arXiv:1705.01843}}. Appearances:\!

\bibitem[AS16]{AS16}
Noga Alon and Joel~H Spencer.
\newblock {\em The Probabilistic Method}.
\newblock John Wiley \& Sons. 2016.
\newblock \href {https://doi.org/10.1002/0471722154} {\nolinkurl{doi:10.1002/0471722154}}. Appearances:\!

\bibitem[AZLO16]{AZLO16}
Zeyuan Allen-Zhu, Yin~Tat Lee, and Lorenzo Orecchia.
\newblock Using optimization to obtain a width-independent, parallel, simpler, and faster positive {SDP} solver.
\newblock In {\em Proceedings of the 27th Annual ACM-SIAM Symposium on Discrete Algorithms}. pages 1824--1831. SIAM. 2016.
\newblock \href {https://doi.org/10.1137/1.9781611974331.ch127} {\nolinkurl{doi:10.1137/1.9781611974331.ch127}}. \href {https://arxiv.org/abs/1507.02259} {\nolinkurl{arXiv:1507.02259}}. Appearances:\!

\bibitem[BASTS10]{BASTS10}
Avraham Ben-Aroya, Oded Schwartz, and Amnon Ta-Shma.
\newblock Quantum expanders: Motivation and construction.
\newblock {\em Theory of Computing}. 6(1):47--79. 2010.
\newblock \href {https://doi.org/10.4086/toc.2010.v006a003} {\nolinkurl{doi:10.4086/toc.2010.v006a003}}. \prelimVersion{CCC}{2008}. Appearances:\!

\bibitem[BBHT98]{BBHT98}
Michel Boyer, Gilles Brassard, Peter H{\o}yer, and Alain Tapp.
\newblock Tight bounds on quantum searching.
\newblock {\em Fortschritte der Physik: Progress of Physics}. 46(4--5):493--505. 1998.
\newblock \href {https://doi.org/10.1002/(SICI)1521-3978(199806)46:4/5<493::AID-PROP493>3.0.CO;2-P} {\nolinkurl{doi:10.1002/(SICI)1521-3978(199806)46:4/5<493::AID-PROP493>3.0.CO;2-P}}. \href {https://arxiv.org/abs/quant-ph/9605034} {\nolinkurl{arXiv:quant-ph/9605034}}. Appearances:\!

\bibitem[BCC{\etalchar{+}}15]{BCC+15}
Dominic~W Berry, Andrew~M Childs, Richard Cleve, Robin Kothari, and Rolando~D Somma.
\newblock Simulating {Hamiltonian} dynamics with a truncated {Taylor} series.
\newblock {\em Physical Review Letters}. 114(9):090502. 2015.
\newblock \href {https://doi.org/10.1103/PhysRevLett.114.090502} {\nolinkurl{doi:10.1103/PhysRevLett.114.090502}}. \href {https://arxiv.org/abs/1412.4687} {\nolinkurl{arXiv:1412.4687}}. Appearances:\!

\bibitem[BCH{\etalchar{+}}19]{BCHTV19}
Adam Bouland, Lijie Chen, Dhiraj Holden, Justin Thaler, and Prashant~Nalini Vasudevan.
\newblock On the power of statistical zero knowledge.
\newblock {\em SIAM Journal on Computing}. 49(4):FOCS17--1. 2019.
\newblock \href {https://doi.org/10.1137/17M1161749} {\nolinkurl{doi:10.1137/17M1161749}}. \prelimVersion{FOCS}{2017}. \href {https://arxiv.org/abs/1609.02888} {\nolinkurl{arXiv:1609.02888}}. Appearances:\!

\bibitem[BCW{\dutchPrefix{Wolf}{d}}W01]{BCWdW01}
Harry Buhrman, Richard Cleve, John Watrous, and Ronald {\dutchPrefix{Wolf}{d}}e~Wolf.
\newblock Quantum fingerprinting.
\newblock {\em Physical Review Letters}. 87(16):167902. 2001.
\newblock \href {https://doi.org/10.1103/PhysRevLett.87.167902} {\nolinkurl{doi:10.1103/PhysRevLett.87.167902}}. \href {https://arxiv.org/abs/quant-ph/0102001} {\nolinkurl{arXiv:quant-ph/0102001}}. Appearances:\!

\bibitem[BDRV19]{BDRV19}
Itay Berman, Akshay Degwekar, Ron~D Rothblum, and Prashant~Nalini Vasudevan.
\newblock Statistical difference beyond the polarizing regime.
\newblock In {\em Theory of Cryptography Conference}. pages 311--332. Springer. 2019.
\newblock \href {https://doi.org/10.1007/978-3-030-36033-7\_12} {\nolinkurl{doi:10.1007/978-3-030-36033-7\_12}}. \href {https://eccc.weizmann.ac.il/report/2019/038} {\nolinkurl{ECCC:TR19-038}}. Appearances:\!

\bibitem[BH09]{BH09}
Jop Bri{\"e}t and Peter Harremo{\"e}s.
\newblock Properties of classical and quantum {Jensen-Shannon} divergence.
\newblock {\em Physical Review A}. 79(5):052311. 2009.
\newblock \href {https://doi.org/10.1103/PhysRevA.79.052311} {\nolinkurl{doi:10.1103/PhysRevA.79.052311}}. \href {https://arxiv.org/abs/0806.4472} {\nolinkurl{arXiv:0806.4472}}. Appearances:\!

\bibitem[Bha96]{Bhatia96}
Rajendra Bhatia.
\newblock {\em Matrix Analysis}. volume 169.
\newblock Springer Science \& Business Media. 1996.
\newblock \href {https://doi.org/10.1007/978-1-4612-0653-8} {\nolinkurl{doi:10.1007/978-1-4612-0653-8}}. Appearances:\!

\bibitem[BHMT02]{BHMT02}
Gilles Brassard, Peter H{\o}yer, Michele Mosca, and Alain Tapp.
\newblock Quantum amplitude amplification and estimation.
\newblock {\em Quantum Computation and Information}. 305:53--74. 2002.
\newblock \href {https://doi.org/10.1090/conm/305/05215} {\nolinkurl{doi:10.1090/conm/305/05215}}. \href {https://arxiv.org/abs/quant-ph/0005055} {\nolinkurl{arXiv:quant-ph/0005055}}. Appearances:\!

\bibitem[BL13]{BL13}
Andrej Bogdanov and Chin~Ho Lee.
\newblock Limits of provable security for homomorphic encryption.
\newblock In {\em Annual Cryptology Conference}. pages 111--128. Springer. 2013.
\newblock \href {https://doi.org/10.1007/978-3-642-40041-4\_7} {\nolinkurl{doi:10.1007/978-3-642-40041-4\_7}}. \href {https://eprint.iacr.org/2013/344} {\nolinkurl{IACR ePrint:2013/344}}. Appearances:\!

\bibitem[BLT92]{BLT92}
Jos{\'e}~L Balc{\'a}zar, Antoni Lozano, and Jacobo Tor{\'a}n.
\newblock The complexity of algorithmic problems on succinct instances.
\newblock In {\em Computer Science: Research and Applications}. pages 351--377. Springer. 1992.
\newblock \href {https://doi.org/10.1007/978-1-4615-3422-8\_30} {\nolinkurl{doi:10.1007/978-1-4615-3422-8\_30}}. Appearances:\!

\bibitem[BOW19]{BOW19}
Costin B{\u{a}}descu, Ryan O'Donnell, and John Wright.
\newblock Quantum state certification.
\newblock In {\em Proceedings of the 51st Annual ACM SIGACT Symposium on Theory of Computing}. pages 503--514. 2019.
\newblock \href {https://doi.org/10.1145/3313276.3316344} {\nolinkurl{doi:10.1145/3313276.3316344}}. \href {https://arxiv.org/abs/1708.06002} {\nolinkurl{arXiv:1708.06002}}. Appearances:\!

\bibitem[BZ10]{BZ10}
Richard~P Brent and Paul Zimmermann.
\newblock {\em Modern Computer Arithmetic}. volume~18.
\newblock Cambridge University Press. 2010.
\newblock \href {https://doi.org/10.1017/CBO9780511921698} {\nolinkurl{doi:10.1017/CBO9780511921698}}. Appearances:\!

\bibitem[Can20]{Canonne20}
Cl{\'e}ment~L Canonne.
\newblock A survey on distribution testing: Your data is big. but is it blue?
\newblock {\em Theory of Computing}. pages 1--100. 2020.
\newblock \href {https://doi.org/10.4086/toc.gs.2020.009} {\nolinkurl{doi:10.4086/toc.gs.2020.009}}. \href {https://eccc.weizmann.ac.il/report/2015/063} {\nolinkurl{ECCC:TR15-063}}. Appearances:\!

\bibitem[CDG{\etalchar{+}}20]{CDG+20}
Rui Chao, Dawei Ding, Andras Gilyen, Cupjin Huang, and Mario Szegedy.
\newblock Finding angles for quantum signal processing with machine precision.
\newblock {\em arXiv preprint}. 2020.
\newblock \href {https://arxiv.org/abs/2003.02831} {\nolinkurl{arXiv:2003.02831}}. Appearances:\!

\bibitem[CDSTS23]{CDSTS23}
Gil Cohen, Dean Doron, Ori Sberlo, and Amnon Ta-Shma.
\newblock Approximating iterated multiplication of stochastic matrices in small space.
\newblock In {\em Proceedings of the 55th Annual ACM Symposium on Theory of Computing}. pages 35--45. 2023.
\newblock \href {https://doi.org/10.1145/3564246.3585181} {\nolinkurl{doi:10.1145/3564246.3585181}}. \href {https://eccc.weizmann.ac.il/report/2022/149} {\nolinkurl{ECCC:TR22-149}}. Appearances:\!

\bibitem[CDVV14]{CDVV14}
Siu-On Chan, Ilias Diakonikolas, Paul Valiant, and Gregory Valiant.
\newblock Optimal algorithms for testing closeness of discrete distributions.
\newblock In {\em Proceedings of the 25th Annual ACM-SIAM Symposium on Discrete Algorithms}. pages 1193--1203. SIAM. 2014.
\newblock \href {https://doi.org/10.1137/1.9781611973402.88} {\nolinkurl{doi:10.1137/1.9781611973402.88}}. \href {https://arxiv.org/abs/1308.3946} {\nolinkurl{arXiv:1308.3946}}. Appearances:\!

\bibitem[CGKZ05]{CGKZ05}
Howard Cheng, Barry Gergel, Ethan Kim, and Eugene Zima.
\newblock Space-efficient evaluation of hypergeometric series.
\newblock {\em ACM SIGSAM Bulletin}. 39(2):41--52. 2005.
\newblock \href {https://doi.org/10.1145/1101884.1101886} {\nolinkurl{doi:10.1145/1101884.1101886}}. Appearances:\!

\bibitem[CLM10]{CLM10}
Steve Chien, Katrina Ligett, and Andrew McGregor.
\newblock Space-efficient estimation of robust statistics and distribution testing.
\newblock In {\em Proceedings of the First Innovations in Computer Science Conference}. pages 251--265. 2010.
\newblock URL: \url{http://conference.iiis.tsinghua.edu.cn/ICS2010/content/papers/21.html}. Appearances:\!

\bibitem[CLW20]{CLW20}
Anirban~N Chowdhury, Guang~Hao Low, and Nathan Wiebe.
\newblock A variational quantum algorithm for preparing quantum {Gibbs} states.
\newblock {\em arXiv preprint}. 2020.
\newblock \href {https://arxiv.org/abs/2002.00055} {\nolinkurl{arXiv:2002.00055}}. Appearances:\!

\bibitem[DGKR19]{DGKR19}
Ilias Diakonikolas, Themis Gouleakis, Daniel~M Kane, and Sankeerth Rao.
\newblock Communication and memory efficient testing of discrete distributions.
\newblock In {\em Proceedings of the Thirty-Second Conference on Learning Theory}. pages 1070--1106. PMLR. 2019.
\newblock URL: \url{https://proceedings.mlr.press/v99/diakonikolas19a.html}. \href {https://arxiv.org/abs/1906.04709} {\nolinkurl{arXiv:1906.04709}}. Appearances:\!

\bibitem[DMWL21]{DMWL21}
Yulong Dong, Xiang Meng, K~Birgitta Whaley, and Lin Lin.
\newblock Efficient phase-factor evaluation in quantum signal processing.
\newblock {\em Physical Review A}. 103(4):042419. 2021.
\newblock \href {https://doi.org/10.1103/PhysRevA.103.042419} {\nolinkurl{doi:10.1103/PhysRevA.103.042419}}. \href {https://arxiv.org/abs/2002.11649} {\nolinkurl{arXiv:2002.11649}}. Appearances:\!

\bibitem[F{\dutchPrefix{Graaf}{v}}dG99]{FvdG99}
Christopher~A Fuchs and Jeroen {\dutchPrefix{Graaf}{v}}an~de Graaf.
\newblock Cryptographic distinguishability measures for quantum-mechanical states.
\newblock {\em IEEE Transactions on Information Theory}. 45(4):1216--1227. 1999.
\newblock \href {https://doi.org/10.1109/18.761271} {\nolinkurl{doi:10.1109/18.761271}}. \href {https://arxiv.org/abs/quant-ph/9712042} {\nolinkurl{arXiv:quant-ph/9712042}}. Appearances:\!

\bibitem[FKL{\etalchar{+}}16]{FKLMN16}
Bill Fefferman, Hirotada Kobayashi, Cedric Yen-Yu Lin, Tomoyuki Morimae, and Harumichi Nishimura.
\newblock Space-efficient error reduction for unitary quantum computations.
\newblock In {\em Proceedings of the 43rd International Colloquium on Automata, Languages, and Programming}. volume~55. page~14. 2016.
\newblock \href {https://doi.org/10.4230/LIPIcs.ICALP.2016.14} {\nolinkurl{doi:10.4230/LIPIcs.ICALP.2016.14}}. \href {https://arxiv.org/abs/1604.08192} {\nolinkurl{arXiv:1604.08192}}. Appearances:\!

\bibitem[FKSV02]{FKSV02}
Joan Feigenbaum, Sampath Kannan, Martin~J Strauss, and Mahesh Viswanathan.
\newblock An approximate $l_1$-difference algorithm for massive data streams.
\newblock {\em SIAM Journal on Computing}. 32(1):131--151. 2002.
\newblock \href {https://doi.org/10.1137/S0097539799361701} {\nolinkurl{doi:10.1137/S0097539799361701}}. \prelimVersion{FOCS}{1999}. Appearances:\!

\bibitem[FL11]{FL11}
Steven~T Flammia and Yi-Kai Liu.
\newblock Direct fidelity estimation from few {Pauli} measurements.
\newblock {\em Physical Review Letters}. 106(23):230501. 2011.
\newblock \href {https://doi.org/10.1103/PhysRevLett.106.230501} {\nolinkurl{doi:10.1103/PhysRevLett.106.230501}}. \href {https://arxiv.org/abs/1104.4695} {\nolinkurl{arXiv:1104.4695}}. Appearances:\!

\bibitem[FL18]{FL18}
Bill Fefferman and Cedric Yen-Yu Lin.
\newblock A complete characterization of unitary quantum space.
\newblock In {\em Proceedings of the 9th Innovations in Theoretical Computer Science Conference}. volume~94. page~4. 2018.
\newblock \href {https://doi.org/10.4230/LIPIcs.ITCS.2018.4} {\nolinkurl{doi:10.4230/LIPIcs.ITCS.2018.4}}. \href {https://arxiv.org/abs/1604.01384} {\nolinkurl{arXiv:1604.01384}}. Appearances:\!

\bibitem[For87]{Fortnow87}
Lance Fortnow.
\newblock The complexity of perfect zero-knowledge.
\newblock In {\em Proceedings of the 19th Annual ACM Symposium on Theory of Computing}. pages 204--209. 1987.
\newblock \href {https://doi.org/10.1145/28395.28418} {\nolinkurl{doi:10.1145/28395.28418}}. Appearances:\!

\bibitem[FR21]{FR21}
Bill Fefferman and Zachary Remscrim.
\newblock Eliminating intermediate measurements in space-bounded quantum computation.
\newblock In {\em Proceedings of the 53rd Annual ACM SIGACT Symposium on Theory of Computing}. pages 1343--1356. 2021.
\newblock \href {https://doi.org/10.1145/3406325.3451051} {\nolinkurl{doi:10.1145/3406325.3451051}}. \href {https://arxiv.org/abs/2006.03530} {\nolinkurl{arXiv:2006.03530}}. Appearances:\!

\bibitem[GH20]{GH20}
Alexandru Gheorghiu and Matty~J Hoban.
\newblock Estimating the entropy of shallow circuit outputs is hard.
\newblock {\em arXiv preprint}. 2020.
\newblock \href {https://arxiv.org/abs/2002.12814} {\nolinkurl{arXiv:2002.12814}}. Appearances:\!

\bibitem[GHS21]{GHS21}
Tom Gur, Min-Hsiu Hsieh, and Sathyawageeswar Subramanian.
\newblock Sublinear quantum algorithms for estimating von {Neumann} entropy.
\newblock {\em arXiv preprint}. 2021.
\newblock \href {https://arxiv.org/abs/2111.11139} {\nolinkurl{arXiv:2111.11139}}. Appearances:\!

\bibitem[Gil19]{Gilyen19}
Andr{\'a}s Gily{\'e}n.
\newblock {\em Quantum Singular Value Transformation \& Its Algorithmic Applications}.
\newblock PhD thesis. University of Amsterdam. 2019. Appearances:\!

\bibitem[GL20]{GL20}
Andr{\'a}s Gily{\'e}n and Tongyang Li.
\newblock Distributional property testing in a quantum world.
\newblock In {\em Proceedings of the 11th Innovations in Theoretical Computer Science Conference}. volume 151. pages 25:1--25:19. 2020.
\newblock \href {https://doi.org/10.4230/LIPIcs.ITCS.2020.25} {\nolinkurl{doi:10.4230/LIPIcs.ITCS.2020.25}}. \href {https://arxiv.org/abs/1902.00814} {\nolinkurl{arXiv:1902.00814}}. Appearances:\!

\bibitem[GMV06]{GMV06}
Sudipto Guha, Andrew McGregor, and Suresh Venkatasubramanian.
\newblock Streaming and sublinear approximation of entropy and information distances.
\newblock In {\em Proceedings of the 17th Annual ACM-SIAM Symposium on Discrete Algorithm}. pages 733--742. 2006.
\newblock \href {https://doi.org/10.1145/1109557.1109637} {\nolinkurl{doi:10.1145/1109557.1109637}}. \href {https://arxiv.org/abs/cs/0508122} {\nolinkurl{arXiv:cs/0508122}}. Appearances:\!

\bibitem[Gol17]{Goldreich17}
Oded Goldreich.
\newblock {\em Introduction to property testing}.
\newblock Cambridge University Press. 2017.
\newblock \href {https://doi.org/10.1017/9781108135252} {\nolinkurl{doi:10.1017/9781108135252}}. Appearances:\!

\bibitem[Gol19]{Goldreich19}
Oded Goldreich.
\newblock Errata {(3-Feb-2019)}.
\newblock \url{http://www.wisdom.weizmann.ac.il/~/oded/entropy.html}. 2019. Appearances:\!

\bibitem[GP22]{GP22}
Andr{\'a}s Gily{\'e}n and Alexander Poremba.
\newblock Improved quantum algorithms for fidelity estimation.
\newblock {\em arXiv preprint}. 2022.
\newblock \href {https://arxiv.org/abs/2203.15993} {\nolinkurl{arXiv:2203.15993}}. Appearances:\!

\bibitem[GR02]{GR02}
Lov Grover and Terry Rudolph.
\newblock Creating superpositions that correspond to efficiently integrable probability distributions.
\newblock {\em arXiv preprint}. 2002.
\newblock \href {https://arxiv.org/abs/quant-ph/0208112} {\nolinkurl{arXiv:quant-ph/0208112}}. Appearances:\!

\bibitem[GR22]{GR22}
Uma Girish and Ran Raz.
\newblock Eliminating intermediate measurements using pseudorandom generators.
\newblock In {\em Proceedings of the 13th Innovations in Theoretical Computer Science Conference}. volume 215. pages 76:1--76:18. 2022.
\newblock \href {https://doi.org/10.4230/LIPIcs.ITCS.2022.76} {\nolinkurl{doi:10.4230/LIPIcs.ITCS.2022.76}}. \href {https://arxiv.org/abs/2106.11877} {\nolinkurl{arXiv:2106.11877}}. Appearances:\!

\bibitem[GRZ21]{GRZ21}
Uma Girish, Ran Raz, and Wei Zhan.
\newblock Quantum logspace algorithm for powering matrices with bounded norm.
\newblock In {\em Proceedings of the 48th International Colloquium on Automata, Languages, and Programming}. volume 198. pages 73:1--73:20. 2021.
\newblock \href {https://doi.org/10.4230/LIPIcs.ICALP.2021.73} {\nolinkurl{doi:10.4230/LIPIcs.ICALP.2021.73}}. \href {https://arxiv.org/abs/2006.04880} {\nolinkurl{arXiv:2006.04880}}. Appearances:\!

\bibitem[GSLW18]{GSLW18}
Andr{\'a}s Gily{\'e}n, Yuan Su, Guang~Hao Low, and Nathan Wiebe.
\newblock Quantum singular value transformation and beyond: exponential improvements for quantum matrix arithmetics.
\newblock {\em arXiv preprint}. 2018.
\newblock \href {https://arxiv.org/abs/1806.01838} {\nolinkurl{arXiv:1806.01838}}. Appearances:\!

\bibitem[GSLW19]{GSLW19}
Andr{\'a}s Gily{\'e}n, Yuan Su, Guang~Hao Low, and Nathan Wiebe.
\newblock Quantum singular value transformation and beyond: exponential improvements for quantum matrix arithmetics.
\newblock In {\em Proceedings of the 51st Annual ACM SIGACT Symposium on Theory of Computing}. pages 193--204. 2019.
\newblock \href {https://doi.org/10.1145/3313276.3316366} {\nolinkurl{doi:10.1145/3313276.3316366}}. \href {https://arxiv.org/abs/1806.01838} {\nolinkurl{arXiv:1806.01838}}. Appearances:\!

\bibitem[GSV98]{GSV98}
Oded Goldreich, Amit Sahai, and Salil Vadhan.
\newblock Honest-verifier statistical zero-knowledge equals general statistical zero-knowledge.
\newblock In {\em Proceedings of the 30th Annual ACM Symposium on Theory of Computing}. pages 399--408. 1998.
\newblock \href {https://doi.org/10.1145/276698.276852} {\nolinkurl{doi:10.1145/276698.276852}}. Appearances:\!

\bibitem[GV99]{GV99}
Oded Goldreich and Salil Vadhan.
\newblock Comparing entropies in statistical zero knowledge with applications to the structure of {SZK}.
\newblock In {\em Proceedings of the Fourteenth Annual IEEE Conference on Computational Complexity}. pages 54--73. IEEE. 1999.
\newblock \href {https://doi.org/10.1109/ccc.1999.766262} {\nolinkurl{doi:10.1109/ccc.1999.766262}}. \href {https://eccc.weizmann.ac.il/report/1998/063} {\nolinkurl{ECCC:TR98-063}}. Appearances:\!

\bibitem[GV11]{GV11}
Oded Goldreich and Salil~P Vadhan.
\newblock On the complexity of computational problems regarding distributions.
\newblock {\em Studies in Complexity and Cryptography}. 6650:390--405. 2011.
\newblock \href {https://doi.org/10.1007/978-3-642-22670-0\_27} {\nolinkurl{doi:10.1007/978-3-642-22670-0\_27}}. \href {https://eccc.weizmann.ac.il/report/2011/004} {\nolinkurl{ECCC:TR11-004}}. Appearances:\!

\bibitem[Haa19]{Haah19}
Jeongwan Haah.
\newblock Product decomposition of periodic functions in quantum signal processing.
\newblock {\em Quantum}. 3:190. 2019.
\newblock \href {https://doi.org/10.22331/q-2019-10-07-190} {\nolinkurl{doi:10.22331/q-2019-10-07-190}}. \href {https://arxiv.org/abs/1806.10236} {\nolinkurl{arXiv:1806.10236}}. Appearances:\!

\bibitem[Hal87]{Halmos87}
Paul~R. Halmos.
\newblock {\em Finite-Dimensional Vector Spaces}.
\newblock Undergraduate Texts in Mathematics. Springer-Verlag New York, NY. 1987.
\newblock \href {https://doi.org/10.1007/978-1-4612-6387-6} {\nolinkurl{doi:10.1007/978-1-4612-6387-6}}. Appearances:\!

\bibitem[HHL09]{HHL09}
Aram~W Harrow, Avinatan Hassidim, and Seth Lloyd.
\newblock Quantum algorithm for linear systems of equations.
\newblock {\em Physical Review Letters}. 103(15):150502. 2009.
\newblock \href {https://doi.org/10.1103/PhysRevLett.103.150502} {\nolinkurl{doi:10.1103/PhysRevLett.103.150502}}. \href {https://arxiv.org/abs/0811.3171} {\nolinkurl{arXiv:0811.3171}}. Appearances:\!

\bibitem[HJ12]{HJ12}
Roger~A Horn and Charles~R Johnson.
\newblock {\em Matrix analysis}.
\newblock Cambridge University Press. 2012.
\newblock \href {https://doi.org/10.1017/cbo9781139020411} {\nolinkurl{doi:10.1017/cbo9781139020411}}. Appearances:\!

\bibitem[Hol73]{Holevo73JS}
Alexander~S Holevo.
\newblock Bounds for the quantity of information transmitted by a quantum communication channel.
\newblock {\em Problemy Peredachi Informatsii}. 9(3):3--11. 1973. Appearances:\!

\bibitem[JVHW15]{JVHW15}
Jiantao Jiao, Kartik Venkat, Yanjun Han, and Tsachy Weissman.
\newblock Minimax estimation of functionals of discrete distributions.
\newblock {\em IEEE Transactions on Information Theory}. 61(5):2835--2885. 2015.
\newblock \href {https://doi.org/10.1109/TIT.2015.2412945} {\nolinkurl{doi:10.1109/TIT.2015.2412945}}. \href {https://arxiv.org/abs/1406.6956} {\nolinkurl{arXiv:1406.6956}}. Appearances:\!

\bibitem[JY11]{JY11}
Rahul Jain and Penghui Yao.
\newblock A parallel approximation algorithm for positive semidefinite programming.
\newblock In {\em Proceedings of the 52nd Annual IEEE Symposium on Foundations of Computer Science}. pages 463--471. IEEE. 2011.
\newblock \href {https://doi.org/10.1109/FOCS.2011.25} {\nolinkurl{doi:10.1109/FOCS.2011.25}}. \href {https://arxiv.org/abs/1104.2502} {\nolinkurl{arXiv:1104.2502}}. Appearances:\!

\bibitem[Kay18]{Kay18}
Alastair Kay.
\newblock Tutorial on the quantikz package.
\newblock {\em arXiv preprint}. 2018.
\newblock \href {https://arxiv.org/abs/1809.03842} {\nolinkurl{arXiv:1809.03842}}. Appearances:\!

\bibitem[Kit95]{Kitaev95}
Alexei~Yu Kitaev.
\newblock Quantum measurements and the {Abelian} stabilizer problem.
\newblock {\em arXiv preprint}. 1995.
\newblock \href {https://arxiv.org/abs/quant-ph/9511026} {\nolinkurl{arXiv:quant-ph/9511026}}. Appearances:\!

\bibitem[Kit97]{Kitaev97}
Alexei~Yu Kitaev.
\newblock Quantum computations: algorithms and error correction.
\newblock {\em Russian Mathematical Surveys}. 52(6):1191. 1997.
\newblock \href {https://doi.org/10.1070/RM1997v052n06ABEH002155} {\nolinkurl{doi:10.1070/RM1997v052n06ABEH002155}}. Appearances:\!

\bibitem[KM01]{KM01}
Phillip Kaye and Michele Mosca.
\newblock Quantum networks for generating arbitrary quantum states.
\newblock In {\em Optical Fiber Communication Conference and International Conference on Quantum Information}. Optica Publishing Group. 2001.
\newblock \href {https://doi.org/10.1364/ICQI.2001.PB28} {\nolinkurl{doi:10.1364/ICQI.2001.PB28}}. \href {https://arxiv.org/abs/quant-ph/0407102} {\nolinkurl{arXiv:quant-ph/0407102}}. Appearances:\!

\bibitem[KMY09]{KMY09}
Hirotada Kobayashi, Keiji Matsumoto, and Tomoyuki Yamakami.
\newblock Quantum {Merlin-Arthur} proof systems: Are multiple {Merlins} more helpful to {Arthur}?
\newblock {\em Chicago Journal of Theoretical Computer Science}. 2009:3. 2009.
\newblock \href {https://doi.org/10.4086/cjtcs.2009.003} {\nolinkurl{doi:10.4086/cjtcs.2009.003}}. \prelimVersion{ISAAC}{2003}. \href {https://arxiv.org/abs/quant-ph/0306051} {\nolinkurl{arXiv:quant-ph/0306051}}. Appearances:\!

\bibitem[LC17]{LC17}
Guang~Hao Low and Isaac~L Chuang.
\newblock Hamiltonian simulation by uniform spectral amplification.
\newblock {\em arXiv preprint}. 2017.
\newblock \href {https://arxiv.org/abs/1707.05391} {\nolinkurl{arXiv:1707.05391}}. Appearances:\!

\bibitem[LC19]{LC19}
Guang~Hao Low and Isaac~L Chuang.
\newblock Hamiltonian simulation by qubitization.
\newblock {\em Quantum}. 3:163. 2019.
\newblock \href {https://doi.org/10.22331/q-2019-07-12-163} {\nolinkurl{doi:10.22331/q-2019-07-12-163}}. \href {https://arxiv.org/abs/1610.06546} {\nolinkurl{arXiv:1610.06546}}. Appearances:\!

\bibitem[Liu25a]{Liu25}
Yupan Liu.
\newblock {\em Complexity-theoretic perspectives on quantum state testing}.
\newblock PhD thesis. Nagoya University. 2025.
\newblock URL: \url{https://nagoya.repo.nii.ac.jp/records/2012662}. Appearances:\!

\bibitem[Liu25b]{Liu23}
Yupan Liu.
\newblock Quantum state testing beyond the polarizing regime and quantum triangular discrimination.
\newblock {\em computational complexity}. 34(11):1--67. 2025.
\newblock \href {https://doi.org/10.1007/s00037-025-00273-8} {\nolinkurl{doi:10.1007/s00037-025-00273-8}}. \href {https://arxiv.org/abs/2303.01952} {\nolinkurl{arXiv:2303.01952}}. Appearances:\!

\bibitem[LLNW25]{LLNW24}
Fran{\c{c}}ois {Le Gall}, Yupan Liu, Harumichi Nishimura, and Qisheng Wang.
\newblock Space-bounded quantum interactive proof systems.
\newblock In {\em Proceedings of the 40th Computational Complexity Conference (CCC 2025)}. volume 339 of {\em LIPIcs}. pages 17:1--17:18. Schloss Dagstuhl - Leibniz-Zentrum f{\"{u}}r Informatik. 2025.
\newblock \href {https://doi.org/10.4230/LIPICS.CCC.2025.17} {\nolinkurl{doi:10.4230/LIPICS.CCC.2025.17}}. \href {https://arxiv.org/abs/2410.23958} {\nolinkurl{arXiv:2410.23958}}. Appearances:\!

\bibitem[LLW26]{LLW25}
Fran{\c{c}}ois {Le Gall}, Yupan Liu, and Qisheng Wang.
\newblock A slightly improved upper bound for quantum statistical zero-knowledge.
\newblock In {\em Proceedings of the 51st International Symposium on Mathematical Foundations of Computer Science ({MFCS} 2026)}. volume 386 of {\em LIPIcs}. pages 6:1--6:20. Schloss Dagstuhl - Leibniz-Zentrum f{\"{u}}r Informatik. 2026.
\newblock \href {https://arxiv.org/abs/2512.11597} {\nolinkurl{arXiv:2512.11597}}. Appearances:\!

\bibitem[LW25]{LW25Lalpha}
Yupan Liu and Qisheng Wang.
\newblock On estimating the quantum $\ell_{\alpha}$ distance.
\newblock In {\em Proceedings of the 33rd Annual European Symposium on Algorithms ({ESA} 2025)}. volume 351 of {\em LIPIcs}. pages 105:1--105:20. Schloss Dagstuhl - Leibniz-Zentrum f{\"{u}}r Informatik. 2025.
\newblock \href {https://doi.org/10.4230/LIPIcs.ESA.2025.105} {\nolinkurl{doi:10.4230/LIPIcs.ESA.2025.105}}. \href {https://arxiv.org/abs/2505.00457} {\nolinkurl{arXiv:2505.00457}}. Appearances:\!

\bibitem[LW26]{LW25entropy}
Yupan Liu and Qisheng Wang.
\newblock On estimating the trace of quantum state powers.
\newblock {\em IEEE Transactions on Information Theory}. pages 1--1. 2026.
\newblock \href {https://doi.org/10.1109/TIT.2026.3683891} {\nolinkurl{doi:10.1109/TIT.2026.3683891}}. \prelimVersion{SODA}{2025}. \href {https://arxiv.org/abs/2410.13559v3} {\nolinkurl{arXiv:2410.13559v3}}. Appearances:\!

\bibitem[{\dutchPrefix{Melkebeek}{v}}MW12]{vMW12}
Dieter {\dutchPrefix{Melkebeek}{v}}an~Melkebeek and Thomas Watson.
\newblock Time-space efficient simulations of quantum computations.
\newblock {\em Theory of Computing}. 8(1):1--51. 2012.
\newblock \href {https://doi.org/10.4086/toc.2012.v008a001} {\nolinkurl{doi:10.4086/toc.2012.v008a001}}. Preliminary version in \arXiv{0712.2545}. Appearances:\!

\bibitem[Mez12]{Mezzarobba12}
Marc Mezzarobba.
\newblock A note on the space complexity of fast {D}-finite function evaluation.
\newblock In {\em Computer Algebra in Scientific Computing: 14th International Workshop (CASC 2012)}. pages 212--223. Springer. 2012.
\newblock \href {https://doi.org/10.1007/978-3-642-32973-9\_18} {\nolinkurl{doi:10.1007/978-3-642-32973-9\_18}}. \href {https://arxiv.org/abs/1209.5097} {\nolinkurl{arXiv:1209.5097}}. Appearances:\!

\bibitem[MLP05]{MLP05}
Ana~P Majtey, Pedro~W Lamberti, and Domingo~P Prato.
\newblock {Jensen-Shannon} divergence as a measure of distinguishability between mixed quantum states.
\newblock {\em Physical Review A}. 72(5):052310. 2005.
\newblock \href {https://doi.org/10.1103/PhysRevA.72.052310} {\nolinkurl{doi:10.1103/PhysRevA.72.052310}}. \href {https://arxiv.org/abs/quant-ph/0508138} {\nolinkurl{arXiv:quant-ph/0508138}}. Appearances:\!

\bibitem[MP16]{MP16}
Ashley Montanaro and Sam Pallister.
\newblock Quantum algorithms and the finite element method.
\newblock {\em Physical Review A}. 93(3):032324. 2016.
\newblock \href {https://doi.org/10.1103/PhysRevA.93.032324} {\nolinkurl{doi:10.1103/PhysRevA.93.032324}}. \href {https://arxiv.org/abs/1512.05903} {\nolinkurl{arXiv:1512.05903}}. Appearances:\!

\bibitem[MRTC21]{MRTC21}
John~M Martyn, Zane~M Rossi, Andrew~K Tan, and Isaac~L Chuang.
\newblock Grand unification of quantum algorithms.
\newblock {\em PRX Quantum}. 2(4):040203. 2021.
\newblock \href {https://doi.org/10.1103/PRXQuantum.2.040203} {\nolinkurl{doi:10.1103/PRXQuantum.2.040203}}. \href {https://arxiv.org/abs/2105.02859} {\nolinkurl{arXiv:2105.02859}}. Appearances:\!

\bibitem[MS24]{MS23}
Ashley Montanaro and Changpeng Shao.
\newblock Quantum and classical query complexities of functions of matrices.
\newblock In {\em Proceedings of the 56th Annual {ACM} Symposium on Theory of Computing}. pages 573--584. {ACM}. 2024.
\newblock \href {https://doi.org/10.1145/3618260.3649665} {\nolinkurl{doi:10.1145/3618260.3649665}}. \href {https://arxiv.org/abs/2311.06999} {\nolinkurl{arXiv:2311.06999}}. Appearances:\!

\bibitem[MU17]{MU17}
Michael Mitzenmacher and Eli Upfal.
\newblock {\em Probability and computing: {Randomization} and probabilistic techniques in algorithms and data analysis}.
\newblock Cambridge University Press. 2017.
\newblock \href {https://doi.org/10.1017/CBO9780511813603} {\nolinkurl{doi:10.1017/CBO9780511813603}}. Appearances:\!

\bibitem[MW05]{MW05}
Chris Marriott and John Watrous.
\newblock Quantum {Arthur--Merlin} games.
\newblock {\em computational complexity}. 14(2):122--152. 2005.
\newblock \href {https://doi.org/10.1007/s00037-005-0194-x} {\nolinkurl{doi:10.1007/s00037-005-0194-x}}. \prelimVersion{CCC}{2004}. \href {https://arxiv.org/abs/cs/0506068} {\nolinkurl{arXiv:cs/0506068}}. Appearances:\!

\bibitem[M{\dutchPrefix{Wolf}{d}}W16]{MdW16}
Ashley Montanaro and Ronald {\dutchPrefix{Wolf}{d}}e~Wolf.
\newblock A survey of quantum property testing.
\newblock {\em Theory of Computing}. pages 1--81. 2016.
\newblock \href {https://doi.org/10.4086/toc.gs.2016.007} {\nolinkurl{doi:10.4086/toc.gs.2016.007}}. \href {https://arxiv.org/abs/1310.2035} {\nolinkurl{arXiv:1310.2035}}. Appearances:\!

\bibitem[MY23]{MY23}
Tony Metger and Henry Yuen.
\newblock $\mathsf{stateQIP}=\mathsf{statePSPACE}$.
\newblock In {\em Proceedings of the 64th Annual IEEE Symposium on Foundations of Computer Science}. pages 1349--1356. {IEEE}. 2023.
\newblock \href {https://doi.org/10.1109/FOCS57990.2023.00082} {\nolinkurl{doi:10.1109/FOCS57990.2023.00082}}. \href {https://arxiv.org/abs/2301.07730} {\nolinkurl{arXiv:2301.07730}}. Appearances:\!

\bibitem[NC10]{NC10}
Michael~A Nielsen and Isaac~L Chuang.
\newblock {\em Quantum computation and quantum information}.
\newblock Cambridge University Press. 2010.
\newblock \href {https://doi.org/10.1017/CBO9780511976667} {\nolinkurl{doi:10.1017/CBO9780511976667}}. Appearances:\!

\bibitem[OW21]{OW21}
Ryan O'Donnell and John Wright.
\newblock Quantum spectrum testing.
\newblock {\em Communications in Mathematical Physics}. 387(1):1--75. 2021.
\newblock \href {https://doi.org/10.1007/s00220-021-04180-1} {\nolinkurl{doi:10.1007/s00220-021-04180-1}}. \prelimVersion{STOC}{2015}. \href {https://arxiv.org/abs/1501.05028} {\nolinkurl{arXiv:1501.05028}}. Appearances:\!

\bibitem[PP23]{PP23}
Aaron Putterman and Edward Pyne.
\newblock Near-optimal derandomization of medium-width branching programs.
\newblock In {\em Proceedings of the 55th Annual ACM Symposium on Theory of Computing}. pages 23--34. 2023.
\newblock \href {https://doi.org/10.1145/3564246.3585108} {\nolinkurl{doi:10.1145/3564246.3585108}}. \href {https://eccc.weizmann.ac.il/report/2022/150} {\nolinkurl{ECCC:TR22-150}}. Appearances:\!

\bibitem[PY86]{PY86}
Christos~H Papadimitriou and Mihalis Yannakakis.
\newblock A note on succinct representations of graphs.
\newblock {\em Information and Control}. 71(3):181--185. 1986.
\newblock \href {https://doi.org/10.1016/S0019-9958(86)80009-2} {\nolinkurl{doi:10.1016/S0019-9958(86)80009-2}}. Appearances:\!

\bibitem[RASW23]{ARS+21}
Soorya Rethinasamy, Rochisha Agarwal, Kunal Sharma, and Mark~M Wilde.
\newblock Estimating distinguishability measures on quantum computers.
\newblock {\em Physical Review A}. 108(1):012409. 2023.
\newblock \href {https://doi.org/10.1103/PhysRevA.108.012409} {\nolinkurl{doi:10.1103/PhysRevA.108.012409}}. \href {https://arxiv.org/abs/2108.08406} {\nolinkurl{arXiv:2108.08406}}. Appearances:\!

\bibitem[Rei08]{Reingold08}
Omer Reingold.
\newblock Undirected connectivity in log-space.
\newblock {\em Journal of the ACM}. 55(4):1--24. 2008.
\newblock \href {https://doi.org/10.1145/1391289.1391291} {\nolinkurl{doi:10.1145/1391289.1391291}}. \prelimVersion{STOC}{2005}. \href {https://eccc.weizmann.ac.il/report/2004/094} {\nolinkurl{ECCC:TR04-094}}. Appearances:\!

\bibitem[Riv90]{Rivlin90}
Theodore~J Rivlin.
\newblock {\em Chebyshev polynomials: from approximation theory to algebra and number theory}.
\newblock Courier Dover Publications. 1990. Appearances:\!

\bibitem[Sak96]{Saks96}
Michael Saks.
\newblock Randomization and derandomization in space-bounded computation.
\newblock In {\em Proceedings of Computational Complexity (Formerly Structure in Complexity Theory)}. pages 128--149. IEEE. 1996.
\newblock \href {https://doi.org/10.1109/CCC.1996.507676} {\nolinkurl{doi:10.1109/CCC.1996.507676}}. Appearances:\!

\bibitem[SH21]{SH21}
Sathyawageeswar Subramanian and Min-Hsiu Hsieh.
\newblock Quantum algorithm for estimating $\alpha$-{Renyi} entropies of quantum states.
\newblock {\em Physical Review A}. 104(2):022428. 2021.
\newblock \href {https://doi.org/10.1103/PhysRevA.104.022428} {\nolinkurl{doi:10.1103/PhysRevA.104.022428}}. \href {https://arxiv.org/abs/1908.05251} {\nolinkurl{arXiv:1908.05251}}. Appearances:\!

\bibitem[SM03]{SM03}
Endre S{\"u}li and David~F Mayers.
\newblock {\em An Introduction to Numerical Analysis}.
\newblock Cambridge University Press. 2003.
\newblock \href {https://doi.org/10.1017/CBO9780511801181} {\nolinkurl{doi:10.1017/CBO9780511801181}}. Appearances:\!

\bibitem[SS03]{SS03}
Elias~M Stein and Rami Shakarchi.
\newblock {\em Fourier Analysis: An Introduction}. volume~1.
\newblock Princeton University Press. 2003. Appearances:\!

\bibitem[SV03]{SV97}
Amit Sahai and Salil Vadhan.
\newblock A complete problem for statistical zero knowledge.
\newblock {\em Journal of the ACM}. 50(2):196--249. 2003.
\newblock \href {https://doi.org/10.1145/636865.636868} {\nolinkurl{doi:10.1145/636865.636868}}. \prelimVersion{FOCS}{1997}. \href {https://eccc.weizmann.ac.il/report/2000/084} {\nolinkurl{ECCC:TR00-084}}. Appearances:\!

\bibitem[SZ99]{SZ99}
Michael Saks and Shiyu Zhou.
\newblock $\mathrm{BP_HSPACE}(s)\subseteq \mathrm{DSPACE}(s^{3/2})$.
\newblock {\em Journal of Computer and System Sciences}. 58(2):376--403. 1999.
\newblock \href {https://doi.org/10.1006/jcss.1998.1616} {\nolinkurl{doi:10.1006/jcss.1998.1616}}. \prelimVersion{FOCS}{1995}. Appearances:\!

\bibitem[TS13]{TS13}
Amnon Ta-Shma.
\newblock Inverting well conditioned matrices in quantum logspace.
\newblock In {\em Proceedings of the 45th Annual ACM Symposium on Theory of Computing}. pages 881--890. 2013.
\newblock \href {https://doi.org/10.1145/2488608.2488720} {\nolinkurl{doi:10.1145/2488608.2488720}}. Appearances:\!

\bibitem[Wat99]{Wat99}
John Watrous.
\newblock Space-bounded quantum complexity.
\newblock {\em Journal of Computer and System Sciences}. 59(2):281--326. 1999.
\newblock \href {https://doi.org/10.1006/jcss.1999.1655} {\nolinkurl{doi:10.1006/jcss.1999.1655}}. \prelimVersion{CCC}{1998}. Appearances:\!

\bibitem[Wat01]{Wat01}
John Watrous.
\newblock Quantum simulations of classical random walks and undirected graph connectivity.
\newblock {\em Journal of Computer and System Sciences}. 62(2):376--391. 2001.
\newblock \href {https://doi.org/10.1006/jcss.2000.1732} {\nolinkurl{doi:10.1006/jcss.2000.1732}}. \prelimVersion{CCC}{1999}. \href {https://arxiv.org/abs/cs/9812012} {\nolinkurl{arXiv:cs/9812012}}. Appearances:\!

\bibitem[Wat02]{Wat02}
John Watrous.
\newblock Limits on the power of quantum statistical zero-knowledge.
\newblock In {\em Proceedings of the 43rd Annual IEEE Symposium on Foundations of Computer Science}. pages 459--468. IEEE. 2002.
\newblock \href {https://doi.org/10.1109/SFCS.2002.1181970} {\nolinkurl{doi:10.1109/SFCS.2002.1181970}}. \href {https://arxiv.org/abs/quant-ph/0202111} {\nolinkurl{arXiv:quant-ph/0202111}}. Appearances:\!

\bibitem[Wat03]{Wat03}
John Watrous.
\newblock On the complexity of simulating space-bounded quantum computations.
\newblock {\em computational complexity}. 12:48--84. 2003.
\newblock \href {https://doi.org/10.1007/s00037-003-0177-8} {\nolinkurl{doi:10.1007/s00037-003-0177-8}}. \prelimVersion{FOCS}{1999}. \href {https://arxiv.org/abs/cs/9911008} {\nolinkurl{arXiv:cs/9911008}}. Appearances:\!

\bibitem[Wat09a]{Watrous08}
John Watrous.
\newblock Quantum computational complexity.
\newblock {\em Encyclopedia of Complexity and Systems Science}. pages 7174--7201. 2009.
\newblock \href {https://doi.org/10.1007/978-0-387-30440-3\_428} {\nolinkurl{doi:10.1007/978-0-387-30440-3\_428}}. \href {https://arxiv.org/abs/0804.3401} {\nolinkurl{arXiv:0804.3401}}. Appearances:\!

\bibitem[Wat09b]{Wat09}
John Watrous.
\newblock Zero-knowledge against quantum attacks.
\newblock {\em SIAM Journal on Computing}. 39(1):25--58. 2009.
\newblock \href {https://doi.org/10.1137/060670997} {\nolinkurl{doi:10.1137/060670997}}. \prelimVersion{STOC}{2006}. \href {https://arxiv.org/abs/quant-ph/0511020} {\nolinkurl{arXiv:quant-ph/0511020}}. Appearances:\!

\bibitem[WGL{\etalchar{+}}24]{WGL+22}
Qisheng Wang, Ji~Guan, Junyi Liu, Zhicheng Zhang, and Mingsheng Ying.
\newblock New quantum algorithms for computing quantum entropies and distances.
\newblock {\em IEEE Transactions on Information Theory}. 70(8):5653--5680. 2024.
\newblock \href {https://doi.org/10.1109/TIT.2024.3399014} {\nolinkurl{doi:10.1109/TIT.2024.3399014}}. \href {https://arxiv.org/abs/2203.13522} {\nolinkurl{arXiv:2203.13522}}. Appearances:\!

\bibitem[{\dutchPrefix{Wolf}{d}}W19]{deWolf19}
Ronald {\dutchPrefix{Wolf}{d}}e~Wolf.
\newblock Quantum computing: Lecture notes.
\newblock {\em arXiv preprint}. 2019.
\newblock \href {https://arxiv.org/abs/1907.09415} {\nolinkurl{arXiv:1907.09415}}. Appearances:\!

\bibitem[WY16]{WY16}
Yihong Wu and Pengkun Yang.
\newblock Minimax rates of entropy estimation on large alphabets via best polynomial approximation.
\newblock {\em IEEE Transactions on Information Theory}. 62(6):3702--3720. 2016.
\newblock \href {https://doi.org/10.1109/TIT.2016.2548468} {\nolinkurl{doi:10.1109/TIT.2016.2548468}}. \href {https://arxiv.org/abs/1407.0381} {\nolinkurl{arXiv:1407.0381}}. Appearances:\!

\bibitem[WZ24]{WZ23}
Qisheng Wang and Zhicheng Zhang.
\newblock Fast quantum algorithms for trace distance estimation.
\newblock {\em IEEE Transactions on Information Theory}. 70(4):2720--2733. 2024.
\newblock \href {https://doi.org/10.1109/TIT.2023.3321121} {\nolinkurl{doi:10.1109/TIT.2023.3321121}}. \href {https://arxiv.org/abs/2301.06783} {\nolinkurl{arXiv:2301.06783}}. Appearances:\!

\bibitem[WZC{\etalchar{+}}23]{WZC+22}
Qisheng Wang, Zhicheng Zhang, Kean Chen, Ji~Guan, Wang Fang, Junyi Liu, and Mingsheng Ying.
\newblock Quantum algorithm for fidelity estimation.
\newblock {\em IEEE Transactions on Information Theory}. 69(1):273--282. 2023.
\newblock \href {https://doi.org/10.1109/TIT.2022.3203985} {\nolinkurl{doi:10.1109/TIT.2022.3203985}}. \href {https://arxiv.org/abs/2103.09076} {\nolinkurl{arXiv:2103.09076}}. Appearances:\!

\bibitem[WZL24]{LWZ22}
Xinzhao Wang, Shengyu Zhang, and Tongyang Li.
\newblock A quantum algorithm framework for discrete probability distributions with applications to {R\'{e}nyi} entropy estimation.
\newblock {\em IEEE Transactions on Information Theory}. 70(5):3399--3426. 2024.
\newblock \href {https://doi.org/10.1109/TIT.2024.3382037} {\nolinkurl{doi:10.1109/TIT.2024.3382037}}. \href {https://arxiv.org/abs/2212.01571} {\nolinkurl{arXiv:2212.01571}}. Appearances:\!

\bibitem[Zal98]{Zalka98}
Christof Zalka.
\newblock Simulating quantum systems on a quantum computer.
\newblock {\em Proceedings of the Royal Society of London. Series A: Mathematical, Physical and Engineering Sciences}. 454(1969):313--322. 1998.
\newblock \href {https://doi.org/10.1098/rspa.1998.0162} {\nolinkurl{doi:10.1098/rspa.1998.0162}}. \href {https://arxiv.org/abs/quant-ph/9603026} {\nolinkurl{arXiv:quant-ph/9603026}}. Appearances:\!

\end{thebibliography}

%%%%%%%%%%%%%%%%%%%%%%%%%%%%%%%%%%%%%%%%%%%
%%%%%%%%%%%%%%%%%%%%%%%%%%%%%%%%%%%%%%%%%%%
%%%%%%%%%%%%%%%%%%%%%%%%%%%%%%%%%%%%%%%%%%%

\end{document}